\documentclass[12pt]{article} 
\usepackage[authoryear]{natbib}
\usepackage[resetlabels]{multibib}
\newcites{S}{References}  
\usepackage{setspace}
\usepackage{etoolbox}
\usepackage{adjustbox}
\usepackage{epsfig,color,amsmath,amssymb,euscript}
\usepackage{mathabx} 
\usepackage{subfigure}
\usepackage{fancyhdr}
\usepackage{rotating}
\usepackage{amsmath} 
\usepackage{amsfonts} 
\usepackage{fancyhdr}
\usepackage{lscape}
\usepackage{amsthm}
\usepackage{enumerate} 
\usepackage{threeparttable}
\usepackage{diagbox}
\usepackage{multirow}
\usepackage{makecell}
\usepackage{xcolor}
\usepackage{colortbl} 
\usepackage{url}  
\usepackage{bbm}
\usepackage{booktabs} 
\usepackage{graphicx} 
\usepackage{enumerate}
\usepackage{mathrsfs}
\usepackage{verbatim}
\usepackage{blkarray}
\usepackage{cancel}
 
 \usepackage{graphicx}

\usepackage{authblk} 
\usepackage{threeparttable}
\usepackage{caption}

\usepackage{algorithm,algcompatible} 
 
\AtBeginEnvironment{thebibliography}{\setstretch{1}}

\allowdisplaybreaks

\usepackage[colorlinks=true,
            linkcolor=blue,
            anchorcolor=blue,
            citecolor=blue,
            urlcolor=blue]{hyperref}

\DeclareMathOperator*{\argmax}{\arg\!\max} 
\algnewcommand\INPUT{\item[\textbf{Input:}]} 
\algnewcommand\OUTPUT{\item[\textbf{Output:}]} 
\makeatletter 
\newcommand{\fdsy@scale}{1.0}
\newcommand\fdsy@mweight@normal{Book}
\newcommand\fdsy@mweight@small{Book}
\newcommand\fdsy@bweight@normal{Medium}
\newcommand\fdsy@bweight@small{Medium}

\DeclareFontFamily{U}{FdSymbolC}{}

\DeclareFontShape{U}{FdSymbolC}{m}{n}{
	<-7.1> s * [\fdsy@scale] FdSymbolC-\fdsy@mweight@small
	<7.1-> s * [\fdsy@scale] FdSymbolC-\fdsy@mweight@normal
}{}
\DeclareFontShape{U}{FdSymbolC}{b}{n}{
	<-7.1> s * [\fdsy@scale] FdSymbolC-\fdsy@bweight@small
	<7.1-> s * [\fdsy@scale] FdSymbolC-\fdsy@bweight@normal
}{}

\DeclareSymbolFont{arrows}{U}{FdSymbolC}{m}{n}
\SetSymbolFont{arrows}{bold}{U}{FdSymbolC}{b}{n}
 
\DeclareMathSymbol{\upvDash}{\mathrel}{arrows}{233}

\DeclareMathSymbol{\upmodels}{\mathrel}{arrows}{237}
\makeatother

\def\T{{\mathrm{\scriptscriptstyle \top} }}

\theoremstyle{definition}

\newcommand{\bA}{{\mathbf A}}
\newcommand{\bB}{{\mathbf B}}

\newcommand{\bX}{{\mathbf X}}

\newcommand{\bZ}{{\mathbf Z}}

\newcommand{\bI}{{\mathbf I}}
\newcommand{\bz}{{\mathbf z}}

\newcommand{\bH}{{\mathbf H}}
\newcommand{\bS}{{\mathbf S}}

\newcommand{\bE}{{\mathbf E}}

\newcommand{\bu}{{\mathbf u}}
\newcommand{\bv}{{\mathbf v}}
\newcommand{\bw}{{\mathbf w}}

\newcommand{\bh}{{\mathbf h}}
\newcommand{\bg}{{\mathbf g}}

\newcommand{\bs}{{\mathbf s}}
\newcommand{\ba}{{\mathbf a}}
\newcommand{\bb}{{\mathbf b}}

\newcommand{\be}{{\mathbf e}}

\newcommand{\bchi}{{\boldsymbol \chi}} 
\let\smallhat\hat
\let\hat\widehat

\let\tilde\widetilde
\let\smalltilde\originaltilde 
\let\check\widecheck
\newcommand{\RR}{{\mathbb R}}
\newcommand{\EE}{{\mathbb E}}
\newcommand{\PP}{{\mathbb P}}
\newcommand{\BB}{{\mathbb B}}
\newcommand{\GG}{{\mathbb G}}
 
\newcommand{\sn}{\sum_{i=1}^n} 
\newcommand{\cL}{\mathcal L}
\newcommand{\cX}{\mathcal X}

\newcommand{\cE}{\mathcal E} 
\newcommand{\cF}{\mathcal F}
\newcommand{\cS}{\mathcal S}
\newcommand{\cG}{\mathcal G}
\newcommand{\cN}{\mathcal N}
\newcommand{\cM}{\mathcal M}
\newcommand{\SSS}{\mathbb S}

\ifx\theorem\undefined
\newtheorem{theorem}{Theorem}
\fi
\ifx\example\undefined
\newtheorem{example}{Example}
\fi
\ifx\property\undefined
\newtheorem{property}{Property}
\fi
\ifx\lemma\undefined
\newtheorem{lemma}{Lemma}
\fi
\ifx\proposition\undefined
\newtheorem{proposition}{Proposition}
\fi
\ifx\remark\undefined
\newtheorem{remark}{Remark}
\fi
\ifx\corollary\undefined
\newtheorem{corollary}{Corollary}
\fi
\ifx\definition\undefined
\newtheorem{definition}{Definition}
\fi
\ifx\conjecture\undefined
\newtheorem{conjecture}{Conjecture}
\fi
\ifx\fact\undefined
\newtheorem{fact}{Fact}
\fi
\ifx\claim\undefined
\newtheorem{claim}{Claim}
\fi
\ifx\assumption\undefined
\newtheorem{assumption}{Assumption}
\fi
\ifx\cond\undefined
\newtheorem{cond}{Condition}
\fi

\newcommand{\bx}{\mathbf{x}}

\newcommand{\bbeta}  {\boldsymbol{\beta}}

\newcommand{\bdelta} {\boldsymbol{\delta}}

\newcommand{\bOmega}{\boldsymbol{\Omega}}

\newcommand{\bSigma}{\boldsymbol{\Sigma}}

\newcommand{\bPsi} {\boldsymbol{\Psi}}

\newcommand{\bxi} {\boldsymbol{\xi}}

\newcommand{\bzeta} {\boldsymbol{\zeta}}

\newcommand{\bzero}{{\mathbf 0}}

\newcommand{\bXi}{\boldsymbol{\Xi}}

\newcommand{\var}{\mbox{Var}}

\def\spacingset#1{\renewcommand{\baselinestretch} 
	{#1}\small\normalsize} \spacingset{1}

\newcommand{\blind}{1}
\makeatletter

\newcommand{\nn}{\nonumber}

\def\singlespace{\def\baselinestretch{1}\@normalsize}

\def\wt{\widetilde}

\begin{document} 
\if1\blind
{
\spacingset{1.25}
  \title{\bf \Large Adapting to Noise Tails in Private Linear Regression}
\author[1,2]{Jinyuan Chang}
\author[1]{Lin Yang}
\author[3]{Mengyue Zha}
\author[4]{Wen-Xin Zhou}

\affil[1]{\it \small Joint Laboratory of Data Science and Business
Intelligence, Institute of Statistical Interdisciplinary Research, Southwestern University of Finance and Economics, Chengdu, China}
\affil[2]{\it \small State Key Laboratory of Mathematical Sciences, Academy of Mathematics and Systems Science, Chinese Academy of Sciences, Beijing, China} 
\affil[3]{\it \small Department of Mathematics, Hong Kong University of Science and Technology, Hong Kong}
\affil[4]{\it \small College of Business Administration, University of Illinois Chicago, Chicago, IL, USA}

		\setcounter{Maxaffil}{0}
		
		\renewcommand\Affilfont{\itshape\small}
	 
		\date{\vspace{-5ex}}
		\maketitle
	 
	} \fi
	\if0\blind
	{
		\bigskip
		\bigskip
		\bigskip
            
            \begin{center}
			{
			\Large \bf Adapting to Noise Tails in Private Linear Regression
			}
		\end{center}
       
		\medskip
	} \fi 
\spacingset{1}
\begin{abstract}
While the traditional goal of statistics is to infer population parameters, modern practice increasingly demands protection of individual privacy. One way to address this need is to adapt classical statistical procedures into privacy-preserving algorithms. In this paper, we develop differentially private tail-robust methods for linear regression. The trade-off among bias, privacy, and robustness is controlled by a tunable robustification parameter in the Huber loss. We implement noisy clipped gradient descent for low-dimensional settings and noisy iterative hard thresholding for high-dimensional sparse models. Under sub-Gaussian errors, our method achieves near-optimal convergence rates while relaxing several assumptions required in earlier work. For heavy-tailed errors, we explicitly characterize how the non-asymptotic convergence rate depends on the moment index, privacy parameters, sample size, and intrinsic dimension. Our analysis shows how the moment index influences the choice of robustification parameters and, in turn, the resulting statistical error and privacy cost. By quantifying the interplay among bias, privacy, and robustness, we extend classical perspectives on privacy-preserving robust regression. The proposed methods are evaluated through simulations and two real datasets.
\end{abstract}


\spacingset{1.69}
\setlength{\abovedisplayskip}{0.2\baselineskip}
\setlength{\belowdisplayskip}{0.2\baselineskip}
\setlength{\abovedisplayshortskip}{0.2\baselineskip}
\setlength{\belowdisplayshortskip}{0.2\baselineskip}

\section{Introduction}
\label{sec:1} 

The increasing demand for privacy-preserving statistical methods has brought differential privacy \citep{dwork2006calibrating} to the forefront of statistical research. Informally, a differentially private (DP) algorithm ensures that an attacker cannot determine whether a particular data point is present in the dataset. Pioneering works, such as those by \cite{hardt2010}, \cite{chaudhuri2012convergence}, \cite{duchi2013local}, and \cite{duchi2018minimax},  quantified the cost of privacy in a range of statistical estimation problems. More recently, the role of privacy in statistical inference has been rigorously investigated \citep{sheffet2017differentially, cai2017priv,  awan2018differentially,  karwa2017finite,chang2024}. Interestingly, even prior to these developments, \cite{nissim2007} and \cite{dwork2009robust} recognized a fundamental connection between privacy and robustness \citep{hampel1986}. In particular, \cite{dwork2009robust} introduced the propose-test-release (PTR) framework, which derives DP algorithms from principles of robust statistics. This philosophy of leveraging robustness to achieve privacy has since motivated a series of follow-up works, including \cite{lei2011differentially}, \cite{smith2011}, \cite{chaudhuri2012convergence}, \cite{avella2021inference}, \cite{liu2022differential}, and \cite{GDP-mean}, to name a few. Specifically, \cite{avella2021inference} and \cite{liu2022differential} focused on robustness against small fractions of contamination in the data, whereas \cite{GDP-mean} considered robustness to heavy-tailed sampling distributions.

In this work, we focus on both linear and sparse linear regressions, which are among the most fundamental statistical problems and serve as building blocks for more advanced methodologies. Consider the linear model $y = \bx^{\T}\bbeta^* + \varepsilon$, where $y \in \RR$ is the response variable, $\bx \in \RR^{p}$ denotes the (random) vector of covariates, $\bbeta^* \in \RR^{p}$ is the unknown vector of regression coefficients, and $\varepsilon$ is a mean-zero error variable. Assuming that $\varepsilon$ follows a normal or sub-Gaussian distribution, significant progress has been made in developing DP ordinary least squares (OLS) estimators for $\bbeta^*$ since the works of \cite{vu2009differential} and \cite{kifer2012private}. When the sample size $n$ is much larger than $p$, several studies have made notable advances in terms of accuracy, sample efficiency, and computational efficiency \citep{sheffet2017differentially, wang2018revisiting, sheffet2019old, cai2021cost, varshney2022nearly, brown2024insufficient}. In the case where the covariates and noise are independent and both possess finite moments of polynomial order, \cite{liu2022differential} established the existence of a DP algorithm. However, their proposed high-dimensional PTR algorithm is not computationally efficient. Furthermore, \cite{kifer2012private} and \cite{cai2021cost} have extended DP OLS methods to high-dimensional settings.

In the presence of heavy-tailed errors, the finite-sample performance of the OLS estimator becomes sub-optimal, particularly in terms of how its error bounds depend on $\delta$, which may scale logarithmically or polynomially at a confidence level of $1-\delta$ \citep{catoni2012}. From a privacy-preserving perspective, the construction of DP procedures relies critically on the notion of sensitivity \citep{dwork2006calibrating}, which is closely tied to the boundedness of the loss function’s gradient \citep{bassily2014private}. This characteristic makes OLS-based methods inadequate, as their gradients exhibit heavy tails, which prevents them from being properly bounded. To achieve both robustness against heavy-tailed error distributions and strong privacy guarantees, we consider employing loss functions with bounded derivatives, such as the well-known Huber loss \citep{huber1973robust}. Many other loss functions, particularly smoothed variants of the Huber loss, share desirable properties such as Lipschitz continuity and local strong convexity. In this work, we focus on the Huber loss to more effectively illustrate the core ideas, without pursuing purely technical generalizations. When the noise variable is independent of the covariates and follows a symmetric distribution, using the Huber loss with a fixed parameter (independent of the data scale) is typically sufficient. In more general settings with heteroscedasticity or asymmetry, \cite{fanliwang2017} and \cite{sun2020adaptive} introduced adaptive Huber regression, in which the robustification parameter $\tau$ is adjusted according to the sample size $n$, dimensionality $p$, and noise scale to balance bias and robustness effectively. We provide a brief review of Huber regression in Section~\ref{sec:adaptive-huber} and Section~\ref{review-Huber-reg} of the supplementary material. Building on these works, we aim to develop a unified DP tail-robust regression framework that applies to `general' linear regression models in both low- and high-dimensional settings. Here, `generality' refers to the minimal assumptions that $\EE(\varepsilon \,|\, \bx) = 0$ and $\EE(\varepsilon^2 \,|\, \bx) \leq \sigma_0^2$, without requiring independence between $\varepsilon$ and $\bx$, or symmetry of the noise distribution. To elucidate the effect of tail behavior on the privacy cost, we explicitly characterize the choice of the robustification parameter when the noise variable exhibits either bounded higher-order moments or sub-Gaussian tails. This parameter offers new insight by effectively bridging the trinity of robustness (quantified via non-asymptotic tail bounds), bias (arising from skewness in the response or error distribution), and privacy.

In low-dimensional settings where $n\gg p$, we propose a DP Huber regression estimator implemented via noisy clipped gradient descent, with noise carefully calibrated to ensure the desired privacy guarantee. Our method builds on the framework of private empirical risk minimization (ERM) via gradient perturbation \citep{bassily2014private}. This line of work, together with earlier studies based on objective function perturbation \citep{chaudhuri2011differentially, kifer2012private}, provides utility guarantees in terms of excess risk while preserving privacy, under the assumption that the loss function satisfies specific convexity and differentiability conditions.  In particular, the objective function perturbation framework requires an objective composed of a convex loss with bounded derivative and a differentiable, strongly convex regularizer. The  Huber loss, by contrast, has a derivative bounded in magnitude by $\tau$, which is allowed to increase with the sample size, and neither the low- or high-dimensional Huber regressions employ a strongly convex penalty. As a result, the objective function perturbation method proposed by \cite{chaudhuri2011differentially} is not applicable to our setting. Under standard assumptions such as strong convexity of the loss function and boundedness of the parameter space, \citet{wang2020} and \citet{kamath2022improved} developed general theoretical frameworks for DP stochastic convex optimization problems with heavy-tailed data. \cite{avella2021differentially} incorporated robust statistics into noisy gradient descent to facilitate DP estimation and inference. Nevertheless, the presence of the robustification parameter $\tau$ introduces additional complexity. Existing techniques and theoretical results do not seamlessly extend to this setting, particularly under general noise distributions. To address this gap, we adapt the framework of \cite{avella2021differentially}, explicitly quantifying the influence of  $\tau$ on both the statistical error and the error induced by privacy constraints. We first show that the DP Huber estimator lies, with high probability, in a small neighborhood of its non-private counterpart. We then establish its convergence rate to the true parameter $\bbeta^*$ in the $\ell_2$-norm. Beyond standard differential privacy, we further evaluate the performance of the DP Huber estimator under the framework of Gaussian differential privacy (GDP) \citep{dong2021gaussian}. Compared to $(\epsilon, \delta)$-DP (see Definition \ref{def.DP} in Section \ref{sec:2}), GDP has attracted growing attention in the statistics community due to  
its elegant interpretation of privacy through the lens of hypothesis testing \citep{WassermanZhouDP}.

In high-dimensional settings where $\|\bbeta^*\|_0 \ll \min\{n, p\}$, ensuring privacy becomes more challenging due to the sparse structure of $\bbeta^*$. \cite{wang2019ijcai} and \cite{cai2021cost} proposed a noisy variant of the iterative hard thresholding (HT) algorithm \citep{blumensath2009iterative, jain2014iterative}  for implementing a privatized sparse OLS estimator. 
They demonstrated that, under Gaussian errors and certain conditions on the covariates, their methods may achieve (near-)optimal statistical performance when applied to sparse linear models. In contrast, DP sparse regression that is robust to heavy-tailed errors has remained largely understudied until recent work by \cite{liu2022differential} and \cite{hu2022high}. \cite{liu2022differential} adapted soft thresholding with the absolute loss and $\ell_1$-regularization to achieve a convergence rate that scales with $\sqrt{p}$. \cite{hu2022high} proposed truncating the original data at a fixed threshold to mitigate the influence of heavy-tailed distributions. However, neither \cite{liu2022differential} nor \cite{hu2022high} achieved the optimal convergence rate that accounts for sparsity. Moreover, the interaction among bias, privacy, and robustness remains insufficiently explored. To bridge this gap, we propose a sparse DP Huber estimator based on noisy iterative hard thresholding. Our analysis establishes the convergence rate of the sparse DP Huber estimator to the true parameter $\bbeta^*$ in the $\ell_2$-norm and examines the bias-privacy-robustness triad. Extending our approach to other robust $M$-estimators and iterative algorithms for $\ell_0$-constrained $M$-estimation \citep[e.g.,][]{liu2021high,she2023} is feasible but beyond the scope of the current work and left for future investigation.

Alongside statistical performance, we examine the interaction among estimation accuracy, privacy, and robustness.   
We demonstrate that the robustification parameter $\tau$ affects global sensitivity and, consequently, the privacy cost. In other words, $\tau$ governs the trade-off among bias, privacy, and robustness.  Our theoretical results suggest choosing $\tau$ based on the effective sample size under privacy constraints. Similar to \cite{barber2014privacy} and \cite{kamath2020private}, we study how the moment condition parameter $\iota \geq 0$ influences this trade-off when the noise variables have bounded $2+\iota$ moments. We find that $\iota$ affects the optimal choice of $\tau$, which in turn impacts both the statistical error and the privacy cost. In high-dimensional settings, we further investigate how sparsity shapes the bias-privacy-robustness trade-off. By quantifying the effects of this trade-off on convergence rates and sample size requirements, we provide a new perspective on the relationship between privacy and robustness. This perspective complements prior work that explores the interplay between these two properties \citep{dwork2009robust, avella2021inference, liu2021robust, georgiev2022privacy, asi2023, liu2023labelrobust, hopkins2023robustness}.

The rest of the paper is organized as follows. Section~\ref{sec:2} provides a brief review of $(\epsilon, \delta)$-DP and $\epsilon$-GDP. Section~\ref{sec:3} presents a noisy gradient descent algorithm for low-dimensional DP Huber regression and a noisy iterative hard thresholding method for sparse Huber regression in high dimensions. Section~\ref{sec:4} provides theoretical guarantees for DP Huber regression, including comprehensive non-asymptotic convergence results under either polynomial-moment or sub-Gaussian noise, followed by a discussion of the bias-privacy-robust trade-off. As a byproduct, we also establish the statistical convergence of the sparse  Huber estimator (in the absence of privacy constraints), computed via iterative HT. Section~\ref{sec:5} demonstrates the proposed methods on simulated datasets. The real data analysis, all proofs, and the construction of DP confidence intervals are provided in the supplementary material.  The used
real data and the code for implementing our proposed methods are available at the GitHub repository: \href{https://github.com/JinyuanChang-Lab/DifferentiallyPrivateHuberRegression}{https://github.com/JinyuanChang-Lab/DifferentiallyPrivateHuberRegression}.

{\it Notation}. For every integer $k \geq 1$, denote by $\RR^{k}$ the $k$-dimensional Euclidean space.
Let $\|\bu\|_1$, $\|\bu\|_2$ and $\|\bu\|_{\infty}$ denote the $\ell_1$-norm, the $\ell_2$-norm and the $\ell_{\infty}$-norm of the vector $\bu$, respectively. We denote the (pseudo) $\ell_0$-norm of $\bu=(u_1,\ldots,u_k)^{\T}$ as $ \|\bu\|_0 = \sum_{i=1}^k \mathbbm{1}( u_i \neq 0)$, where $\mathbbm{1}(\cdot)$ is the indicator function. Write $\mathbb{S}^{k-1} := \{\bu \in \RR^k : \|\bu\|_2 = 1\}$, $[k] := \{1, \ldots, k\}$, and $\mathbb{B}^k(r):=\{\bu  \in \RR^k: \|\bu\|_2\leq r\}$.  For a subset $S\subseteq [k]$ with cardinality $|S|$ and a $k$-dimensional vector $\bu \in \RR^k$, we write $\bu_S \in \RR^k$ as the vector obtained by setting all entries in $\bu$ to $0$, except for those indexed by $S$. The inner product of any two vectors $\bu = (u_1, \ldots, u_k)^{\T}$ and $\bv = (v_1, \ldots, v_k)^{\T}$ is defined by $\langle \bu, \bv \rangle = \sum_{i=1}^k u_i v_i$. For a symmetric  matrix $\bA\in \RR^{k \times k}$, we denote the smallest and largest eigenvalues of $ \bA $ by $\lambda_{\min}(\bA)$ and $\lambda_{\max}(\bA)$, respectively. For two sequences of non-negative numbers $\{a_n\}_{n\geq 1}$ and $\{b_n\}_{n\geq 1}$, we say $a_n \lesssim b_n$ or $b_n\gtrsim a_n$ if there exists a constant $C>0$ independent of $n$ such that $a_n\leq Cb_n$.  We write $a_n \asymp b_n$ if $a_n \lesssim b_n$ and $b_n \lesssim a_n$ hold simultaneously, and $a_n\ll b_n$ or $b_n\gg a_n$ if $\lim_{n\rightarrow\infty} a_n/b_n = 0$.   For a    matrix $\bA\in \RR^{k \times k}$, let $\|\bA\|$ denote the spectral norm of $\bA$. For any two matrices  $\bA$ and $\bB$,  we write $\bA \succeq \bB$  if $\bA-\bB$ is positive semi-definite. For any $x\in\RR$, we write $\lceil x\rceil  = \inf\{y \in \mathbb{Z}:y \geq x\}$ and $\Phi(x)=(2\pi)^{-1/2} \int_{-\infty}^x e^{-u^2/2}\, {\rm d} u$.

\section{Background on Differential Privacy}
\label{sec:2}

Differential privacy was originally introduced to provide a formal framework for data privacy. It ensures that a randomized mechanism $\mathcal{M}$ produces similar output distributions for datasets $\bX$ and $\bX'$ that differ by only a single data point. Intuitively, this implies that an attacker cannot determine whether a particular data point $x$ is included in the dataset $\bX$ based on the output of the mechanism. The formal definition is given below.

\begin{definition}[\cite{dwork2006calibrating}]
\label{def.DP}
A randomized mechanism $\cM:\cX\to\RR^d$ is said to be $(\epsilon, \delta)$-DP if, for every pair of adjacent datasets $\bX, \bX'\in\cX$ that differ in exactly one data point, and for every measurable subset $S\subseteq\RR^d$, it holds that 
$\PP\{ \cM(\bX)\in S\} \leq e^\epsilon\cdot\PP\{ \cM(\bX')\in S \} +\delta$. 
\end{definition}

In addition to its privacy guarantees, differential privacy is valued for the simplicity and versatility in the design of private algorithms. Typically, such algorithms are constructed by adding random noise to the output of a non-private algorithm.  Among the various types of noise, Gaussian and Laplace noise are the most commonly used; see Theorems 3.6 and 3.22 of \cite{dwork2014algorithmic}. Let $\cM$ be an algorithm that maps a dataset $\bX$ to $\RR^{d}$. For any $q\geq 1$, the $\ell_q$-sensitivity of $\cM$ is defined as ${\rm sens}_q(\cM) =  \sup_{\bX, \bX'} \| \cM(\bX) - \cM(\bX') \|_q$, where the supremum is taken over all pairs of datasets $\bX$ and $\bX'$ that differ by a single data point.

\begin{lemma}(Gaussian mechanism) 
\label{glmechanism} 
Assume sens$_2(\cM)\leq B$ for some $B>0$. Then $\cM(\bX)+\bg$, with $\bg\sim \cN(0,\sigma^2\bI_{d})$ and $\sigma = B \epsilon^{-1}\sqrt{2 \log(1.25/\delta)}$, is $(\epsilon,\delta)$-DP. 
\end{lemma}

The construction of multi-step differential private algorithms is further enhanced by the post-processing property \citep{dwork2014algorithmic} and the composition property \citep{dwork2014algorithmic, kairouz2017} of differential privacy. The composition property characterizes the evolution of privacy parameters under composition. Heuristically, the post-processing property ensures that no external entity can undo the privatization.

\begin{lemma} 
\label{post}
(Post-processing property). Let $\mathcal{M}$ be an $(\epsilon,\delta)$-DP algorithm, and let $g$ be an arbitrary deterministic mapping that takes the output of $\mathcal{M}$ as an input. Then $g(\mathcal{M}(\bX))$ is also $(\epsilon, \delta)$-DP.
\end{lemma}

\begin{lemma} 
\label{standardcompositiontheorem}
(Standard composition property \citep{dwork2014algorithmic}). Let $\cM_1$ and $\cM_2$ be $(\epsilon_1,\delta_1)$-DP and $(\epsilon_2,\delta_2)$-DP algorithms, respectively. The composition $\cM_1\circ \cM_2$ is $(\epsilon_1+\epsilon_2, \delta_1+\delta_2)$-DP.
\end{lemma}

\begin{lemma} 
\label{compositiontheorem}
(Advanced composition property \citep{dwork2014algorithmic}). If $T$  algorithms are run sequentially, each satisfying $(\epsilon,\delta)$-DP, then for any $\delta'>0$, the combined procedure is $(\epsilon\sqrt{2T\log(1/{\delta'})}+T\epsilon(e^{\epsilon}-1),T\delta+\delta')$-DP.
\end{lemma}

By the standard composition property, if each step of an iterative algorithm is $(\epsilon/T,\delta/T)$-DP, then after $T$ iterations, the resulting algorithm will be $(\epsilon,\delta)$-DP. Moreover, if each step of an iterative algorithm is $(\epsilon\sqrt{2/\{5T\log(2/\delta)\}},\delta/(2T))$-DP, where 
\begin{align}\label{privacy.constraints}
   0< \epsilon \leq 1~~\mbox{and}~~0<\delta \leq 0.01\,,
\end{align} 
then, by the advanced composition property (with $\delta' = \delta/2$) and a straightforward computation, the combined algorithm after $T$ iterations is also $(\epsilon, \delta)$-DP.

\cite{WassermanZhouDP} presented a statistical perspective that connects differential privacy to hypothesis testing. Briefly, consider the following hypothesis testing problem
\begin{align}
    H_0:\text{ the underlying data is $\bX\quad$ versus} \quad H_1:\text{ the underlying data is $\bX'$}\,.   \label{hypothesis.test}
\end{align} 
Suppose $x_1$ is the only data point present in $\bX$ but not in $\bX'$. Rejecting $H_0$ would reveal $x_1$'s absence, while accepting $H_0$ would confirm its presence. When an $(\epsilon,\delta)$-DP algorithm $\cM$ is used, the power of any test at significance level $\alpha$ is bounded by $\min\{ e^\epsilon \alpha + \delta, 1 - e^{-\epsilon} (1-\alpha-\delta) \}$. When both $\epsilon$ and $\delta$ are small, any $\alpha$-level test becomes nearly powerless. To address this limitation, \cite{dong2021gaussian} introduced the trade-off function to characterize the trade-off between Type I and Type II errors.

\begin{definition}[Trade-off function]
\label{tradeoff.func}
For any two probability distributions $\mathbb{P}$ and $\mathbb{Q}$ on the same space, the trade-off function $T(\mathbb{P}, \mathbb{Q}): [0, 1] \to [0, 1]$ is defined as $T(\mathbb{P}, \mathbb{Q}) (\alpha) = \inf\{ \beta_\phi : \alpha_\phi \leq \alpha\}$, where $\alpha_\phi= \mathbb{E}_{\mathbb{P}}(\phi)$ and $\beta_\phi= 1-\mathbb{E}_{\mathbb{Q}}(\phi)$ represent the Type I and Type II errors associated with the test $\phi$, respectively. The infimum is taken over all measurable tests.
\end{definition}

The larger the trade-off function, the more difficult it becomes to distinguish between the two distributions via hypothesis testing. Building on trade-off functions, \cite{dong2021gaussian} proposed a generalization of differential privacy, referred to as $f$-DP. With a slight abuse of notation, we identify $\cM(\bX)$ and $\cM(\bX')$ with their corresponding probability distributions.

\begin{definition}[$f$-DP and GDP]
\label{def:f-DP}
   (i) Let $f$ be a trade-off function. A mechanism $\cM$ is said to be $f$-DP if $T(\cM(\bX), \cM(\bX')) \geq f$ for all adjacent data sets $\bX, \bX'\in\cX$ that differ by a single data point. (ii) Let $\epsilon >0$. A mechanism $\cM$ is said to be $\epsilon$-Gaussian differentially private ($\epsilon$-GDP) if it is $G_\epsilon$-DP, where $G_\epsilon(\alpha)= \Phi(\Phi^{-1}(1-\alpha) - \epsilon)$.
\end{definition}

A mechanism $\cM$ is $(\epsilon, \delta)$-DP if and only if it is $f_{\epsilon, \delta}$-DP, where $f_{\epsilon, \delta}(\alpha) = \max\{ 0, 1-e^{\epsilon} \alpha - \delta, e^{-\epsilon}(1-\alpha -\delta) \}$ \citep{WassermanZhouDP}. We refer to Figure~3 in \cite{dong2021gaussian} for a visual comparison between  $G_\epsilon(\cdot)$ introduced in Definition~\ref{def:f-DP} and $f_{\epsilon, \delta}(\cdot)$. GDP is a core single-parameter family within the $f$-DP framework, and is defined via testing two shifted Gaussian distributions; see Definition 2.6 in \cite{dong2021gaussian}. Lemma~\ref{gdp.mechanism} below extends Theorem~1 in \cite{dong2021gaussian} from the univariate to the multivariate setting.

\begin{lemma} \label{gdp.mechanism}
The Gaussian mechanism, $\cM(\bX)+ \epsilon^{-1}{\rm sens}_2(\cM)\cdot\bg$ with $\bg\sim \cN(\bzero, \bI_{d})$, is $\epsilon$-GDP.
\end{lemma}

Under the GDP framework, the composition of $T$ mechanisms with parameters $\{\epsilon_t\}_{t=1}^T$ yields overall $\epsilon$-GDP with $\epsilon= \sqrt{\epsilon_1^2+\cdots+\epsilon_T^2} $. This implies that the individual privacy level $\epsilon_t$ can be set to $\epsilon/\sqrt{T}$.

\section{Differentially Private Huber Regression}
\label{sec:3} 

This section introduces algorithms for DP Huber regression in both low- and high-dimensional settings. We begin by reviewing non-private Huber regression, accompanied by the required notation and relevant background.

\subsection{Huber regression}
\label{sec:adaptive-huber}

Suppose we observe $n$ independent samples $\{ (y_i, \bx_i) \}_{i=1}^n$ satisfying the linear model
\begin{align} \label{linear.model}
	y_i = \bx_i ^\T   \bbeta^* + \varepsilon_i  = \sum_{j=1}^p x_{i,j} \beta^*_j + \varepsilon_i \,, \quad  i\in[n]\,, 
\end{align}
where $\bx_i = (x_{i,1},\ldots, x_{i,p})^\T$ is a $p$-dimensional feature vector with $x_{i,1}\equiv 1$, and $\bbeta^* = (\beta^*_1, \ldots, \beta^*_p)^\T \in \RR^p$ denotes the unknown true coefficient vector.
In the random design setting, let $\bx_1, \ldots, \bx_n$ be independent copies of a random vector $\bx \in \RR^p$. The noise variable $\varepsilon_i$ satisfies $\EE(\varepsilon_i\, | \,\bx_i) = 0$ and $\EE(\varepsilon_i^2\,|\, \bx_i) \leq \sigma_0^2 <\infty$. Under this moment condition, while the OLS estimator, obtained by minimizing $\bbeta \mapsto \sn (y_i - \bx_i^\T \bbeta)^2$, exhibits desirable asymptotic properties as $n\to \infty$ with $p$ fixed, its finite sample performance, particularly in terms of high-probability tail bounds, is suboptimal compared to the case when $\varepsilon_i$ is sub-Gaussian \citep{catoni2012, sun2020adaptive}.

To obtain an estimator that is both asymptotically efficient and exhibits exponential-type concentration bounds in finite-sample settings, we employ the robust loss $\rho_\tau(u) = \tau^2 \rho(u/\tau)$, where $\rho:\RR\to [0,\infty)$ is a continuously differentiable convex function, and $\tau>0$ serves as a robustification parameter. Assume that $\psi(u) = \rho'(u)$ is Lipschitz continuous, concave, and differentiable almost everywhere. We consider the following $M$-estimator:
\begin{align}
	 \hat{\bbeta}_\tau   \in \arg\min_{\bbeta\in \RR^p}  \bigg\{  \hat \cL_\tau(\bbeta)  := \frac{1}{n} \sn \rho_\tau( y_i - \bx_i^\T \bbeta) \bigg\}\, .  \label{def:m.est}
\end{align}
Define  $\psi_\tau(u) = \rho'_\tau(u) = \tau \psi(u/\tau )$ for $u \in \RR$. Due to the convexity of $\rho_\tau(\cdot)$ and, consequently, of $\hat \cL_\tau(\cdot)$, the $M$-estimator $ \hat{\bbeta}_\tau$ can alternatively be characterized as the solution to the first-order condition $ \sn \psi_\tau (   y_i - \bx_i^\T \bbeta   )  \bx_i = \textbf{0}$. For simplicity, we focus on the Huber loss \citep{huber1964robust}, defined as $\rho(u) =  (u^2/2) {\mathbbm{1}}(|u| \leq 1) + (|u| - 1/2) {\mathbbm{1}}(|u| >1)$ for $u\in \RR$. The associated score function is $\psi(u) = {\rm{sign}}(u)\min\{|u|, 1\}$.  To conserve space, we defer the details on Huber regression to Section~\ref{review-Huber-reg} of the supplementary material. The next section demonstrates the construction of a DP Huber estimator in low dimensions.

\subsection{Randomized estimator via noisy gradient descent}
\label{noisyGD.lowd}

We begin by considering the low-dimensional setting where $n\gg p$. To compute the $M$-estimator $\widehat{\bbeta}_{\tau}$ defined in \eqref{def:m.est}, one of the most widely used algorithms is gradient descent, which generates iterates as follows:
$$
 \bbeta^{(t+1)}  = \bbeta^{(t)}  - \eta_0 \nabla \hat{\cL}_\tau( \bbeta^{(t)}) =  \bbeta^{(t)}  + \frac{\eta_0}{n} \sn \psi_\tau(y_i - \bx_i^\T \bbeta^{(t)} ) \bx_i\,  , \quad t\geq 0\, ,  
$$
where $\bbeta^{(0)}$ is an initial estimator and $\eta_0>0$ is the learning rate. For each $t\geq 0$, given the previous iterate, the gradient descent update can be viewed as an algorithm that maps the dataset $\{ (y_i, \bx_i)  \}_{i=1}^n$ to $\RR^p$. However, in the case of random features, the sensitivity of this algorithm may not be bounded. To address this, we use a clipping function $w_\gamma(t) = \min\{ \gamma / t, 1\}$, $t\geq 0$, enabling the privatization of gradient descent.

Motivated by the general concept of noisy gradient descent in the differential privacy literature \citep{avella2021differentially}, we consider a noisy clipped gradient descent method that incorporates the Gaussian mechanism. Let $T\geq 1$ be a prespecified number of iterations, $\gamma>0$ a truncation parameter, and $\sigma>0$ a noise level to be determined. Starting with an initial estimate $ \bbeta^{(0)}$, the noisy clipped gradient descent computes updates as follows:
\begin{align}
 \bbeta^{(t+1)}  = \bbeta^{(t)}  +  \eta_0 \bigg\{   \frac{1}{n}  \sn  \psi_\tau ( y_i - \bx_i^\T \bbeta^{(t)}  ) \bx_i w_ \gamma ( \| \bx_i \|_2 )    +    \sigma    \bg_t \bigg\}\, , \quad t \in  \{0\}\cup [T-1]\,, \label{noisy.gd}
\end{align} 
where $\bg_0, \bg_1, \ldots, \bg_{T-1} \in \RR^p$ are i.i.d. standard normal random vectors.

Recall that $\psi_\tau(u) = \tau  {\rm sign}(u) \min\{ |u/\tau| ,1 \}$ is the derivative of the Huber loss, and it satisfies $\sup_{u\in \RR}  | \psi_\tau (u) | \leq \tau$. Therefore, we have $\sup_{ (y, \bx) \in \RR \times \RR^p , \,  \bbeta \in \RR^p} \|  \psi_\tau ( y - \bx^\T \bbeta  ) \bx w_\gamma( \| \bx \|_2 )  \|_2 \leq \gamma \tau$, which implies that the $\ell_2$-sensitivity of each clipped gradient descent step is bounded by $2\eta_0 \gamma \tau/n$. By Lemmas~\ref{glmechanism} and \ref{gdp.mechanism}, each noisy clipped gradient descent step with noise scales
\begin{align} \label{dp.noise.scale}
    \sigma_{{\rm dp}} =   \frac{2  \gamma \tau}{n\epsilon }\cdot T \sqrt{2 \log\bigg(\frac{1.25 T}{\delta}\bigg) }  ~~~\textrm{and}~~~\sigma_{{\rm gdp}} = \frac{ 2 \gamma \tau}{n\epsilon}\cdot\sqrt{T}
\end{align}
is, respectively, $(\epsilon/T , \delta/T)$-DP and $(\epsilon/\sqrt{T})$-GDP. After $T$ iterations, the standard composition property  implies that the $T$-th iterate $\bbeta^{(T)}$ is either $(\epsilon,\delta)$-DP or $\epsilon$-GDP, depending on the choice of the noise scale. The complete algorithm is presented in Algorithm~\ref{alg:DPHuberlow}.  Alternatively, a gradient descent step with noise scale  $$ \sigma_{{\rm dp}} = \frac{ 2 \gamma \tau }{  n\epsilon } \sqrt{5 T \log\bigg(\frac{2}{\delta}\bigg) \log\bigg(\frac{5T}{2\delta}\bigg) }$$  ensures  $(\epsilon\sqrt{2/\{5T\log(2/\delta)\}},\delta/(2T))$-DP per iteration. By Lemma~\ref{compositiontheorem}, the final iterate $\bbeta^{(T)}$ is $(\epsilon, \delta)$-DP after $T$ iterations, provided that the privacy constraints in \eqref{privacy.constraints} are met. In practice, for a given triplet $(T, \epsilon, \delta)$, we choose the configuration that yields the smaller value of $\sigma_{{\rm dp}}$ between the two alternatives. 

\begin{algorithm}
\caption{DP Huber Regression} \label{alg:DPHuberlow}
\vspace{.2cm}
\textbf{Input:} Dataset $\{(y_i, \bx_i)\}_{i=1}^n$, initial value $\bbeta^{(0)}$, learning rate $\eta_0$, number of iterations $T$, truncation level $\gamma$, robustification parameter $\tau$, and privacy parameters $(\epsilon, \delta)$.

\begin{algorithmic}[1]
 \renewcommand{\algorithmicindent}{0pt}
\FOR{ $t=0,\ldots, T-1$} 
\STATE ~ Generate $\bg_t\sim \cN(\bzero,\bI_p)$ and compute  
\[
\bbeta^{(t+1)}  = \bbeta^{(t)}  +  \eta_0 \bigg\{   \frac{ 1}{n}  \sn  \psi_\tau( y_i - \bx_i^\T \bbeta^{(t)} ) \bx_i   w_ \gamma ( \| \bx_i \|_2 )    +   \sigma  \bg_t \bigg\}\,,
\]
where $\sigma>0$ is set to either $\sigma_{{\rm dp}}$ or $\sigma_{{\rm gdp}}$ specified in \eqref{dp.noise.scale}; 
\ENDFOR
\end{algorithmic}
\textbf{Output:} $\bbeta^{(T)}$.
\end{algorithm}

\begin{remark}
\label{compare.cwz2021}  
Algorithm~\ref{alg:DPHuberlow} extends and strengthens the $(\epsilon, \delta)$-DP least squares algorithm (Algorithm~4.1 in \cite{cai2021cost}) in several important aspects. First, it permits the noise variable in \eqref{linear.model} to follow a general distribution, which may be non-Gaussian and may exhibit heavy tails or asymmetry; see Theorem~\ref{thm:low} in Section \ref{sec:4.1}. This added flexibility is enabled by the careful choice of $\tau$, which governs the trade-off among robustness, bias, and privacy. Second, by combining the Huber loss with an adaptively selected $\tau$ and incorporating covariate clipping, the proposed algorithm accommodates a much broader range of design and parameter settings, while requiring fewer tuning parameters tied to the underlying data-generating process. This stands in contrast to the bounded design and parameter assumptions (D1) and (P1) in \cite{cai2021cost}. Specifically, these assumptions require $\| \bx \|_2 \leq  c_{\bx}$ almost surely and $\|\bbeta^*\|_2 \leq c_0$, where the constants $c_{\bx}, c_0 >0$ are not only essential for the theoretical analysis but also appear explicitly in the algorithm, introducing nontrivial practical constraints. Moreover, the covariate vector $\bx$ must have zero mean with covariance matrix $\bSigma = \mathbb{E}(\bx\bx^{\T})$ satisfying $(p L)^{-1} \leq  \lambda_{\min}(\bSigma) \leq \lambda_{\max}(\bSigma) \leq Lp^{-1}$ for some constant $L>1$, which likewise enters the algorithm.
\end{remark}

\begin{remark}
\label{compare.abl2023}
Algorithm~\ref{alg:DPHuberlow} fits within the general framework of DP $M$-estimation considered in \cite{avella2021differentially}. The key distinction is that, when applied to linear models with heavy-tailed and asymmetric errors, the robustification parameter $\tau$ is no longer a fixed constant. Instead, it is chosen as a function of the sample size, dimensionality, privacy level, and noise scale to balance bias, tail-robustness and privacy. In Section~\ref{sec:4.1}, we demonstrate that under certain assumptions, achieving an optimal trade-off among bias, robustness, and privacy requires selecting $\tau$ to be of order $\sigma_0 (n \epsilon / p)^{1/(2+\iota)}$ in the case of $\epsilon$-GDP.
\end{remark}

\begin{remark}
Algorithm~\ref{alg:DPHuberlow} has the potential to save privacy budget by leveraging privacy amplification techniques \citep{kasiviswanathan2008, bassily2014private, abadi2016, feldman2018, wang2019}. However, to ensure a fair and consistent comparison with existing methods under similar settings \citep{cai2021cost, avella2021differentially}, we have not incorporated privacy amplification into the theoretical analysis presented in this work. We leave the incorporation of privacy amplification into our theoretical framework as a direction for future research. Moreover, we note that the theoretical framework in \citet{feldman2018} relies on the assumption that the underlying feasible set is convex. This assumption is violated in the high-dimensional setting when applying NoisyHT (see Algorithm \ref{alg:NHT} in Section \ref{noisyGD.highd}), as the feasible set defined by sparsity constraints is inherently non-convex. Developing both theoretical and empirical justifications for privacy amplification in non-convex settings remains an open challenge and may require fundamentally different analytical techniques.
\end{remark}

Define $\bXi:=(\Xi^{jk})_{j,k\in[p]}= \bSigma^{-1}\bOmega\bSigma^{-1}$ with $\bSigma = \EE(\bx_i \bx_i^\T)$ and $\bOmega= \EE(\varepsilon_i^2 \bx_i \bx_i^{\T})$. 
We will show in Theorem \ref{thm.DPGA} in the supplementary material that, under certain conditions, the DP Huber estimator is asymptotically normal:
\begin{align*}
     \sqrt{n}(\Xi^{jj})^{-1/2}({\beta}^{(T)}_j-\beta^*_j)  \overset{{\rm d}}{\rightarrow} \cN(0,1)~~\mbox{as}~~n \rightarrow\infty\,.
\end{align*}
 To construct confidence intervals, for some $\tau_1>0$ and $\gamma_1>0$, we define the private estimators  of $\bSigma$  and $\bOmega$, respectively, as $\widehat{\bSigma}_{ \gamma_1,\epsilon}= \widehat\bSigma_{\gamma_1}+\varsigma_{1 }\bE$ 
 and $\widehat\bOmega_{{\tau_1}, \gamma_1,\epsilon}=\widehat\bOmega_{{\tau_1}, \gamma_1}(\bbeta^{(T)})+\varsigma_{2 }\bE$,
 where $\varsigma_{1 }$ and $\varsigma_{2 }$ are two noise scale parameters depending on $\epsilon$ or $(\epsilon, \delta)$, and $\bE\in \RR^{p\times p}$ is a symmetric random matrix whose upper-triangular and diagonal
 entries are i.i.d. $\cN(0,1)$. Here,  $\widehat\bSigma_{ \gamma_1 } = {n}^{-1}\sn  \bx_i\bx_i^{\T}w_{\gamma_1}^2(\|\bx_i\|_2)$, and $\widehat\bOmega_{ \tau_1 , \gamma_1 }(\bbeta) =  {n}^{-1}\sn \psi_{ \tau_1 }^2(y_i-\bx_i^{\T}\bbeta)\bx_i\bx_i^{\T}w_{\gamma_1}^2(\|\bx_i\|_2)$, and their sensitivities are given in Section~\ref{sec.CI} of the supplementary material. To ensure positive definiteness, we further project both $\widehat{\bSigma}_{ \gamma_1 ,\epsilon}$ and $\widehat\bOmega_{ \tau_1 , \gamma_1 ,\epsilon}$ onto the cone of positive definite matrices $\{\bH : \bH \succeq \zeta \bI\}$. Specifically, we obtain $\widehat{\bSigma}_{ \gamma_1 , \epsilon}^{+}= \arg\min_{\bH \succeq \zeta \bI}\|\bH-	\widehat{\bSigma}_{ \gamma_1 ,\epsilon} \|$ and $\widehat\bOmega_{\tau_1, \gamma_1,\epsilon}^{+}= \arg\min_{\bH \succeq \zeta \bI}\|\bH-\widehat\bOmega_{ \tau_1, \gamma_1 ,\epsilon} \|$ for a sufficiently small constant $\zeta>0$. Consequently, we take  
  \begin{align}\label{tildeXi}
 \widetilde\bXi_{\tau_1,\gamma_1,\epsilon} := (\widetilde\Xi_{\tau_1,\gamma_1,\epsilon}^{jk})_{j,k\in[p]} =   (\widehat{\bSigma}_{\gamma_1,\epsilon }^{+} )^{-1}\widehat\bOmega_{\tau_1,\gamma_1,\epsilon }^{+}(\widehat{\bSigma}_{\gamma_1,\epsilon }^{+} )^{-1}
 \end{align}
 as a private estimator of $\bXi$.
 The privacy guarantee and statistical consistency of $\widetilde\bXi_{\tau_1,\gamma_1,\epsilon}$, as well as the theoretical requirements for $\tau_1$ and $\gamma_1$, are established in Section \ref{sec.CI} of the supplementary material.   For any $\alpha \in (0,1)$, we construct the $100(1-\alpha)\%$ confidence interval of $\beta^*_j$ as $\beta^{(T)}_j \pm z_{\alpha/2}(\widetilde\Xi_{\tau_1,\gamma_1,\epsilon}^{jj})^{1/2} / \sqrt{n}$, where $z_{\alpha/2}$ denotes the $(1-\alpha/2)$-th quantile of $\cN(0,1)$.

\subsection{Noisy iterative hard thresholding for sparse Huber regression}
\label{noisyGD.highd}

In this section, we consider the high-dimensional setting, where $\bbeta^* \in \RR^p$ in \eqref{linear.model} is sparse with $\|\bbeta^*\|_0 \ll \min\{n, p\}$. A common approach to induce sparsity is $\ell_1$-penalization, which is well-known for its computational efficiency and desirable theoretical properties. For sparse least absolute deviation (LAD) regression, that is, quantile regression with $\tau=1/2$, \cite{liu2024DP} proposed a private estimator by reformulating the sparse LAD problem as a penalized least squares estimation and adopting a three-stage noise injection mechanism to ensure $(\epsilon, \delta)$-DP. However, the convergence rate of this private estimator is suboptimal, as it scales with $\sqrt{p}$.

From a different perspective, \cite{cai2021cost} proposed a noisy variant of the iterative HT algorithm \citep{blumensath2009iterative,jain2014iterative} for least squares regression and established its statistical (near-)optimality for sparse linear models with Gaussian errors. Instead of relying on soft thresholding as in \cite{liu2024DP}, which is closely associated with $\ell_1$-penalization, a key step in \cite{cai2021cost} is the adoption of the `peeling' procedure from \cite{dwork2018dpfdr}; see Algorithm~\ref{alg:NHT} below. 
To emphasize its connection to HT, we refer to it as NoisyHT throughout the remainder of this paper. This method involves adding independent Laplace random variables to the absolute values of the entries in a given vector  and then selecting the top $s$ largest coordinates from the resulting vector. As demonstrated by Lemma~\ref{lem.HD.noiseht} in the supplementary material, when $0<\epsilon \leq 0.5$, $0<\delta \leq 0.011$ and $s\geq 10$, Algorithm~\ref{alg:NHT} is an $(\epsilon,\delta)$-DP algorithm if the involved parameter $\lambda$ satisfies $\|\bv(\bZ)-\bv(\bZ')\|_{\infty} < \lambda$  for every pair of adjacent datasets $\bZ$ and $\bZ'$. 

\begin{algorithm}
\caption{NoisyHT$(\bv , \bZ, s, \epsilon, \delta, \lambda)$ Algorithm \citep{dwork2018dpfdr}}\label{alg:NHT}
\vspace{.2cm}
\textbf{Input:} Dataset $\bZ$, vector-valued function $\bv=\bv(\bZ)=(v_1,\ldots,v_p)^{\T} \in \mathbb{R}^{p}$, sparsity level $s$, privacy parameters $(\epsilon, \delta)$, and noise scale $\lambda$. 
 	\begin{algorithmic}[1]
    \renewcommand{\algorithmicindent}{0pt}
 			\STATE Initialize $S=\emptyset$;
 			\FOR{$i\in[s]$}
 			\STATE ~ Generate $\bw_{i}=(w_{i, 1},  \ldots, w_{i, p})^{\T}\in \mathbb{R}^{p}$ with $w_{i, 1},   \ldots, w_{i, p} \stackrel{\text{i.i.d.}}{\sim}$  Laplace$({2\epsilon^{-1}\lambda \sqrt{5  s \log (1 / \delta)}})$;
 			\STATE ~ Append $j^{*}=\arg\max_{j \in[p] \setminus S} ( |v_{j}|+w_{i, j} )$ to $S$;
 			\ENDFOR
 			\STATE Set $\tilde{P}_{s}(\bv)=\bv_{S}$;
 			\STATE Generate $\tilde{\bw} =(\tilde{w}_{  1},  \ldots, \tilde{w}_{ p})^{\T}\in \mathbb{R}^{p}$  with $\tilde{w}_{1}, \ldots, \tilde{w}_{p} \stackrel{\text{i.i.d.}}{\sim}$ Laplace$({2\epsilon^{-1}\lambda \sqrt{ 5 s \log (1 / \delta)}})$;
 		\end{algorithmic}
\textbf{Output:} $\tilde{P}_{s}(\bv)+\tilde{\bw}_{S} \in \RR^p$.
\end{algorithm}

In high-dimensional settings, the noisy clipped gradient decent step involves calculating the intermediate update and the noisy update sequentially. The intermediate update involving clipped gradients is defined as
\begin{align}\label{noisy.gdnew}
\tilde \bbeta^{(t+1)}  = \bbeta^{(t)}  +     \frac{ \eta_0}{n}  \sn  \psi_\tau( y_i - \bx_i^\T \bbeta^{(t)} ) \bx_i  w_ \gamma ( \| \bx_i \|_\infty )\,,   \quad t \in  \{0\}\cup [T-1]\,. 
\end{align}
The noisy update $\bbeta^{(t+1)}$ is then obtained by using $\tilde \bbeta^{(t+1)}$ from \eqref{noisy.gdnew} as the input to Algorithm~\ref{alg:NHT}.  
Due to $\sup_{u\in \RR}  | \psi_\tau (u) | \leq \tau$, we know $\sup_{ (y, \bx) \in \RR \times \RR^p , \,  \bbeta \in \RR^p} \|    \psi_\tau( y - \bx^\T \bbeta ) \bx  w_\gamma( \| \bx \|_\infty )  \|_\infty \leq \gamma \tau$, which implies the $\ell_{\infty}$-sensitivity of the procedure for calculating the intermediate update is bounded by $2\eta_0 \gamma \tau/n$. By setting the privacy parameters to $(\epsilon/T, \delta/ T)$   and choosing the noise scale as $\lambda = 2\eta_0 \gamma \tau/n$ in Algorithm~\ref{alg:NHT}, Lemma~\ref{lem.HD.noiseht} in the supplementary material indicates that each noisy clipped gradient descent step is $(\epsilon/T, \delta/T)$-DP. After $T$ iterations, the standard composition property (Lemma \ref{standardcompositiontheorem})  implies that the $T$-th iterate $ \bbeta^{(T)}$ is $(\epsilon, \delta)$-DP.  The complete algorithm is presented in Algorithm~\ref{alg:DPHuberhigh}. Alternatively, by setting the privacy parameters to $(\epsilon\sqrt{2/\{5T\log(2/\delta)\}},\delta/(2T))$ and choosing the noise scale as $\lambda = 2\eta_0 \gamma \tau/n$ in Algorithm~\ref{alg:NHT},  Lemma~\ref{lem.HD.noiseht} and the  advanced composition property (Lemma \ref{compositiontheorem}) imply that the $T$-th iterate $ \bbeta^{(T)}$ is also $(\epsilon, \delta)$-DP,  provided that the privacy constraints  \eqref{privacy.constraints} hold. In the practical implementation, the choice of privacy parameters is guided by which configuration results in a smaller scale of the Laplace random variables injected in NoisyHT.

\begin{algorithm}
\caption{DP Sparse Huber Regression}\label{alg:DPHuberhigh}
\vspace{.2cm}
\textbf{Input:} Dataset $\{(y_i, \bx_i)\}_{i=1}^n$, learning rate $\eta_0$, number of iterations $T$, truncation parameter $\gamma$, robustification parameter $\tau$, privacy parameters $(\epsilon, \delta)$, sparsity level $s$, and initial value $\bbeta^{(0)}$.
\begin{algorithmic}[1]
 \renewcommand{\algorithmicindent}{0pt}
\FOR{ $t=0,\ldots, T-1$}
\STATE ~ Compute  
\begin{align*}
\tilde \bbeta^{(t+1)}  =&~ \bbeta^{(t)}  +     \frac{\eta_0}{n}\sn  \psi_\tau( y_i - \bx_i^\T \bbeta^{(t)}   ) \bx_i w_ \gamma ( \| \bx_i \|_\infty )\,,\\
\bbeta^{(t+1)}  =&~ \text{NoisyHT}\big( \tilde \bbeta^{(t+1)}, \{(y_i, \bx_i)\}_{i=1}^n, s,  \epsilon/T ,   \delta/T ,  2 \eta_0\gamma\tau /n \big)\,.
\end{align*}
\ENDFOR
\end{algorithmic}
\textbf{Output:} $\bbeta^{(T)} $.
\end{algorithm}

\begin{remark}
Algorithm~\ref{alg:DPHuberhigh} integrates a noisy gradient-based descent method with Huber regression to simultaneously ensure differential privacy and robustness against heavy-tailed errors, provided that $\tau$ is properly tuned, as discussed in Section~\ref{sec:4}. It is applicable to a broader class of random features, including those following a normal distribution. Since the Huber loss has a bounded derivative, the residuals are truncated by the function $\psi_\tau(\cdot)$, eliminating the need to truncate the response variables. In practice, truncating the response variable can be problematic, particularly when it is positive and follows a right-skewed distribution, such as wages or prices. In contrast, residuals are expected to fluctuate around zero, making their truncation more reasonable.
\end{remark}

\section{Statistical Analysis of Private Huber Regression}
\label{sec:4}

For the theoretical analysis, we impose the following assumptions on the distributions of the random covariates and errors under the linear model \eqref{linear.model}.

 \begin{assumption}
 \label{assump:design}
 For each $i\in[n]$, the random covariate vector $\bx_i=(x_{i,1}, \ldots, x_{i,p})^\T$ with $x_{i,1}\equiv 1$ satisfies: (i) $\EE(x_{i,j})=0$ for $j\in[p]\setminus \{1\}$, and (ii) $\PP(| \langle \bu, \bSigma^{-1/2} \bx_i\rangle | \geq \upsilon_1 z ) \leq 2 e^{-z^2/2}$ for all $\bu \in \mathbb{S}^{p-1}$ and $z\geq 0$,  where  $\upsilon_1 \geq 1$ is a  dimension-free constant and $\bSigma = \EE(\bx_i \bx_i^\T)$. Moreover, there exist constants $\lambda_1 \geq \lambda_p>0$ such that
 $ \lambda_p \leq \lambda_{\min}(\bSigma) \leq \lambda_{\max} (\bSigma) \leq \lambda_1$. 
 \end{assumption}

\begin{assumption}
\label{assump:heavy-tail}
The regression errors $\{\varepsilon_i\}_{i=1}^n$ are independent random variables satisfying $\mathbb{E}(\varepsilon_i \,|\, \bx_i)=0$ and $\mathbb{E}(\varepsilon_i^{2} \,| \,\bx_i)\leq \sigma_0^{2}$ almost surely. Moreover, $  \mathbb{E}( |\varepsilon_i|^{2+\iota} \,|\, \bx_i)    \leq \sigma_\iota^{2+\iota}$ almost surely  for some $\iota \geq 0$ and $\sigma_0\leq \sigma_\iota <\infty$.
\end{assumption}

Assumption~\ref{assump:design} imposes a sub-Gaussian condition on the random covariates, generalizing the standard one-dimensional sub-Gaussian assumption to random vectors. This condition also complements the bounded design assumptions (D1) and (D$1'$) in \cite{cai2021cost}.   Various types of random vectors fulfill this assumption, for example: (i) Gaussian and Bernoulli random vectors, (ii) random vectors uniformly distributed on the Euclidean sphere or ball centered at the origin with radius $\sqrt{p}$, and (iii) random vectors uniformly distributed on the unit cube $[-1,1]^p$. For more detailed discussions of high-dimensional sub-Gaussian distributions, including discrete cases, we refer to Chapter 3.4 in \cite{vershynin2018high}.  Let $\kappa _4 = \sup_{\bu\in \mathbb{S}^{p-1}}   \EE (\langle \bSigma^{-1/2} \bx_i, \bu \rangle^4 )$, 
which serves as an upper bound on the kurtoses of all marginal projections of the normalized covariate vectors. It can be shown that  $\kappa_4 \leq C \upsilon_1^4$ for some absolute constant $C>1$. 
Assumption~\ref{assump:heavy-tail} relaxes the Gaussian error condition (4.1) in \cite{cai2021cost} by allowing for heavy-tailed error distributions, such as the $t_\nu$-distribution with $\nu>2$, the centered lognormal distribution, and the centered Pareto distribution with shape parameter greater than $2$. Our theoretical results show that finite conditional second moment condition (i.e., $\iota = 0$) suffices for coefficient estimation (see Theorems~\ref{thm:low} and~\ref{thm:high}), while the stronger moment condition $\iota >0$ is required only for inference (see Theorem~C.3 in the supplementary material). Given a dataset $\{ (y_i, \bx_i) \}_{i=1}^n$ satisfying Assumptions~\ref{assump:design} and \ref{assump:heavy-tail}, recall   
$$
\hat \cL_\tau(\bbeta)  = \frac{1}{n} \sn \rho_\tau  ( y_i - \bx_i^\T \bbeta )   ~~\mbox{with}~~ \rho_\tau(u) = \frac{u^2}{2}\mathbbm{1}(|u|\leq \tau) + \bigg(\tau |u| - \frac{\tau^2}{2}\bigg) \mathbbm{1}(|u| > \tau) 
$$
denotes the empirical Huber loss. In the following analysis, all constants depending on $(\upsilon_1, \lambda_1, \lambda_p, \kappa_4)$ are absorbed into the notation $\lesssim$, $\gtrsim$, and $\asymp$.

\subsection{Low-dimensional setting}
\label{sec:4.1}

In this section, we establish the statistical properties of the DP Huber estimator $\bbeta^{(T)}$  as defined in Algorithm~\ref{alg:DPHuberlow}. As a benchmark, let ${\hat{\bbeta}_{\tau}} \in \arg\min_{\bbeta \in \RR^p} \hat{\cL}_\tau(\bbeta)$ denote the non-private Huber estimator in the low-dimensional setting. Given an initial estimate $\bbeta^{(0)}$, the statistical analysis of the noisy clipped gradient descent iterates $\{ \bbeta^{(t)} \}_{t=1}^T$ relies heavily on the properties of the empirical Huber loss $\hat{\cL}_\tau(\cdot)$,  including its local strong convexity and global smoothness. The key observation is that the Huber loss $\rho_\tau(u)$ is strongly convex only when $|u| < \tau$, with its second-order derivative $\rho_\tau''(u) = 1$ in this region. Our analysis consists of two parts. First, we establish that, by conditioning on a series of `good events' associated with the empirical Huber loss, the noisy clipped gradient descent iterates exhibit favorable convergence properties. Second, we prove that these good events occur with high probability under Assumptions~\ref{assump:design} and \ref{assump:heavy-tail}. Due to space limitations, intermediate results that hold conditioned on good events are provided in the supplementary material.

Given that the initial value $\bbeta^{(0)}$ lies within a neighborhood of the non-private Huber estimator $\hat\bbeta_{\tau}$, Theorem \ref{thm:low} below establishes a high-probability bound for $\| \bbeta^{(T)} - {\bbeta}^* \|_2$.  

\begin{theorem} \label{thm:low}
Let Assumptions \ref{assump:design} and \ref{assump:heavy-tail} hold.  Assume that $\bbeta^{(0)} \in \hat \bbeta_{\tau} + \BB^p(r_0)$ for some  $r_0\asymp \tau$ , and that the learning rate satisfies $\eta_0 =\eta/(2\lambda_1) $ for some $\eta\in(0,1]$. Moreover, let $\gamma \asymp \sqrt{p+\log n}$, $T\asymp \log\{r_0 n\epsilon(\sigma_0 p)^{-1}\} $, and  $\tau \asymp  \tau_0\{n\epsilon(p+\log n)^{-1}\}^{1/(2+\iota)}$ for some $\tau_0\geq \sigma_0$. Then, if the noise scale $\sigma$ defined in \eqref{dp.noise.scale} satisfies $\sigma \sqrt{p+\log T+\log n}      \lesssim  r_0$,
the DP Huber estimator $\bbeta^{(T)}$ obtained in Algorithm~\ref{alg:DPHuberlow} satisfies 
\begin{align*}
\|\bbeta^{(T)} -   \bbeta^* \|_2   
   \lesssim &~    \underbrace{ \frac{\sigma_0 p}{n\epsilon} +    \sigma\sqrt{p+ \log n}}_{\|\bbeta^{(T)} - \hat \bbeta_{\tau} \|_2} \\
   &+\underbrace{  \max\{ \sigma_\iota^{2+\iota} \tau_0^{-1-\iota} , \tau_0 \}   \bigg( \frac{p+\log n}{n \epsilon} \bigg)^{(1+\iota)/({2+\iota})} + \sigma_0  \sqrt{\frac{p+\log n} {n }}  }_{\|\hat \bbeta_{\tau} -  \bbeta^* \|_2 }  
\end{align*}
with probability at least $1-Cn^{-1}$, provided that $n\epsilon \gtrsim C_{\tau_0,\sigma_\iota}(p+\log n)$, where $C_{\tau_0,\sigma_\iota}$ is a positive
constant depending only on ($\tau_0,\sigma_\iota$).
\end{theorem}

The selection of robustification parameter $\tau$ in Theorem \ref{thm:low} is intended to balance bias and robustness, while accounting for the effective sample size $n\epsilon$. See the proof of Proposition \ref{prop:events1} in Section \ref{proof.e1} of the supplementary material for details. The initial condition $\bbeta^{(0)} \in \hat \bbeta_{\tau} + \BB^p(r_0)$ in Theorem~\ref{thm:low} is not restrictive. As shown in Theorem~\ref{prop:initial_value} of the supplementary material, for any initial point $\bbeta^{(0)}$ satisfying $\|\bbeta^{(0)} - \hat \bbeta_{\tau} \|_2 \geq r_0$, we have $\|\bbeta^{(T_0)}- \hat \bbeta_{\tau} \|_2\leq r_0$ with high probability for some $T_0\in \mathbb{N}_{+}$, provided that the sample size $n$ is sufficiently large. Moreover, since $T \asymp \log (n\epsilon)$ and each iteration in  Algorithm~\ref{alg:DPHuberlow} incurs a computational cost of $O(np)$, the total complexity of Algorithm~\ref{alg:DPHuberlow} is $O\{np \log (n\epsilon)\}$.

Let $\tau_0 \asymp \sigma_0$ in Theorem \ref{thm:low}. The robustification parameter $\tau$  is then required to satisfy $\tau \asymp \sigma_0  \{   n\epsilon ( p + \log n )^{-1} \}^{1/(2+\iota)}$. Based on Theorem \ref{thm:low}, we now present the final convergence guarantees for the DP Huber estimator $\bbeta^{(T)}$ obtained from Algorithm~\ref{alg:DPHuberlow}, under different choices of the noise scale $\sigma$ as specified in \eqref{dp.noise.scale}. The results are summarized as follows:
\begin{itemize}
    \item[(i)] ($\epsilon$-GDP) As long as $n\epsilon \gtrsim \sqrt{T}(p+\log n)$, the $\epsilon$-GDP Huber estimator $\bbeta^{(T)}$ obtained in Algorithm \ref{alg:DPHuberlow} by setting $\sigma=\sigma_{\rm gdp}$ satisfies, with probability at least $1-C n^{-1}$, that
    \begin{align} \label{gdp.lowd.rate}
        \|\bbeta^{(T)}-\bbeta^*\|_2 \lesssim  \sigma_0\bigg\{\bigg(\frac{\sigma_{\iota}}{\sigma_0}\bigg)^{2+\iota}+\sqrt{T}\bigg\}   \bigg( \frac{p+\log n}{n \epsilon} \bigg)^{(1+\iota)/(2+\iota)}  +  \sigma_0\sqrt{\frac{p+\log n}{n}}\,.
    \end{align}

    \item[(ii)] ($(\epsilon, \delta)$-DP) As long as $n \epsilon  \gtrsim  T  (p+\log n)\sqrt{ \log(T/\delta)} $, the $(\epsilon, \delta)$-DP Huber estimator $\bbeta^{(T)}$ obtained in Algorithm \ref{alg:DPHuberlow} by setting $\sigma=\sigma_{\rm dp}$ satisfies, with probability at least $1-C n^{-1}$, \begin{align} \label{dp.lowd.rate} 
        \|\bbeta^{(T)}-\bbeta^*\|_2 \lesssim &~  \sigma_0\bigg\{\bigg(\frac{\sigma_{\iota}}{\sigma_0}\bigg)^{2+\iota}+  T\sqrt{   \log\bigg(\frac{T}{\delta}\bigg) }  \bigg\}  \bigg( \frac{p+\log n}{n \epsilon} \bigg)^{(1+\iota)/(2+\iota)}    + \sigma_0 \sqrt{\frac{p+\log n}{n}} \,.
    \end{align}
\end{itemize}

In the error bounds \eqref{gdp.lowd.rate} and \eqref{dp.lowd.rate}, the second term reflects the convergence of the non-private Huber estimator $\widehat{\bbeta}_{\tau}$, while the first term represents the privacy cost. Notably, $n \epsilon$ can be interpreted as the {\rm effective sample size} under privacy constraints, and the choice of $\tau$ is influenced by this effective sample size. For normally distributed errors, \cite{cai2021cost} established a minimax lower bound for $(\epsilon, \delta)$-DP estimation of $\bbeta^*$ under $\ell_2$-risk. The lower bound is of order $\sigma_0\{\sqrt{pn^{-1}} + p(n \epsilon)^{-1}\}$ when $\epsilon \in (0, 1)$ and $\delta < n^{-1-\omega}$ for some fixed constant  $\omega>0$. In comparison, the slower term $\{p(n \epsilon)^{-1}\}^{1-1/(2+\iota)}$, ignoring logarithmic factors, in \eqref{gdp.lowd.rate} and \eqref{dp.lowd.rate} explicitly captures the combined influence of heavy-tailedness and privacy.

\begin{remark}
\label{final.convergence.lowd.subgaussian}
Another notable implication of our results is that DP Huber regression achieves a near-optimal convergence rate for sub-Gaussian errors. In addition to Assumption~\ref{assump:design}, assume that the regression error $\varepsilon_i$ satisfies $\EE (e^{\lambda \varepsilon_i} \,|\, \bx_i) \leq e^{\sigma_0^2 \lambda^2/2}$ for all $\lambda \in \RR$. In this case, we instead choose
$\tau \asymp \sigma_0 \sqrt{\log \{ n \epsilon (p+\log n)^{-1}\}}$. The $\ell_2$-error of the resulting $\epsilon$-GDP Huber estimator is, up to constant factors, upper bounded by
$$
 \frac{p+\log n}{n \epsilon }\cdot\sigma_0  \sqrt{T\log\bigg(\frac{n\epsilon}{p+\log n}\bigg)}  + \sigma_0 \sqrt{\frac{p+\log n}{n}}
$$
with probability at least $1-C n^{-1}$ as long as $n\epsilon \gtrsim \sqrt{T}(p + \log n)$. Compared to Theorem~4.2 in \cite{cai2021cost}, the above results hold without requiring (i) $\| \bbeta^* \|_2 < c_0$ for some constant $c_0>0$ that  appears in Algorithm 4.1 and Theorem 4.2 of \cite{cai2021cost}, and (ii) $\| \bx_i \|_2 < c_{\bx}$ with probability one for some dimension-free constant $c_{\bx}$. Moreover, the sample size requirement is relaxed from $n \epsilon \gtrsim p^{3/2}$ to $n\epsilon \gtrsim p$, ignoring logarithmic factors.
\end{remark}

\subsection{High-dimensional setting}
\label{sec:HD}

In this section, we establish the statistical properties of the sparse DP Huber estimator  $\bbeta^{(T)} $ as defined in Algorithm~\ref{alg:DPHuberhigh}. Compared to the low-dimensional setting, the sparsity level $s$ is prespecified to perform HT on the high-dimensional gradient vector after the injection of Laplace noises. Analogous to the low-dimensional case, our analysis proceeds in two parts.  First, we establish that by conditioning on a series of `good events' associated with the empirical Huber loss, the noisy clipped gradient descent iterates exhibit favorable convergence properties. Second, we prove that these good events occur with high probability under Assumptions \ref{assump:design} and \ref{assump:heavy-tail}. All intermediate results that hold conditionally on the good events are provided in the supplementary material.

Let $\mathbb{H}(s):= \{ \bbeta\in\mathbb{R}^p:\|\bbeta\|_0\leq s \}$ denote the set of $s$-sparse vectors in $\mathbb{R}^p$, and let $\Theta(r) = \{ \bbeta \in \RR^p : \| \bbeta - \bbeta^* \|_2 \leq r \}$ be the ball of radius $r$ centered at $\bbeta^*$. Theorem~\ref{thm:high} below provides a high-probability bound for $\| \bbeta^{(T)} - \bbeta^* \|_2$, given that the initial value $\bbeta^{(0)}$ lies within some neighborhood of $\bbeta^*$.   
 
\begin{theorem} \label{thm:high} 
   Let Assumptions \ref{assump:design} and \ref{assump:heavy-tail} hold. Assume that $\bbeta^{(0)} \in \mathbb{H}(s)\cap \Theta(r_0)$ for some $r_0 \asymp \tau \geq 16\sigma_0$, and the learning rate satisfies $\eta_0 =\eta/(2\lambda_1) $ for some $\eta\in(0,1)$. Moreover, let $\gamma\asymp \sqrt{\log(pn)}$ and $T\asymp \log\{  r_0    n \epsilon (\sigma_0\log p )^{-1}\}$. Write  $s^*:= \|\bbeta^*\|_0$. Then, the sparse DP Huber estimator $\bbeta^{(T)}$  obtained from Algorithm~\ref{alg:DPHuberhigh} satisfies
        \begin{align*}
           \|\bbeta^{(T)}-\bbeta^*\|_2  \lesssim &~\underbrace{\frac{\sigma_{\iota}^{2+\iota}} {\tau^{1+\iota} }   +     \sigma_0 \sqrt{\frac{s\log (ep/s) +\log n}{n}} +  \tau \frac{s\log (ep/s) + \log n}{n}}_{\tilde{r}}  \\
           &~+\frac{\sigma_0\log p}{  n\epsilon}   +C_{\eta,1}\frac{sT\tau\{\log(pn)\}^{3/2}}{n\epsilon}\sqrt{  \log\bigg(\frac{T}{\delta}\bigg) }   
        \end{align*}
     with probability at least $1-Cn^{-1}$,    provided that  $s \geq s^*\max \{ 192(1+\eta^{-2})(8\lambda_1/\lambda_p)^2 , 16\eta/(1-\eta)  \}$, $\smalltilde{r} \lesssim  r_0$, and $n\epsilon\gtrsim C_{\eta,2}sT\{\log(pn)\}^{3/2}\sqrt{ \log(T/\delta)}$,  
    where $C_{\eta,1}$ and $C_{\eta,2}$ are  two  positive constants depending only on $\eta$.
\end{theorem}

By setting $\tau \asymp \sigma_0 (n\epsilon )^{1/(2+\iota)}  \{ s\log(ep/s)+\log n \}^{-1/(2+\iota)}\,$, Theorem \ref{thm:high} implies that as long as $n \epsilon  \gtrsim s T   \{\log(pn)\}^{3/2} \sqrt{ \log(T/\delta)}$, the sparse DP Huber estimator $\bbeta^{(T)}$ satisfies, with probability at least $1-C n^{-1}$, that 
    \begin{align}\label{dp.highd.rate}
          \|\bbeta^{(T)}-\bbeta^*\|_2  
         \lesssim &~ \frac{\sigma_\iota^{2+\iota}}{\sigma_0^{1+\iota}}  \bigg\{\frac{s\log(ep/s) +\log n}{n \epsilon} \bigg\}^{(1+\iota)/(2+\iota)}  \{\log (p n  )\}^{3/2}T\sqrt{  \log\bigg(\frac{T}{\delta}\bigg)}    \nn\\
        &+  \sigma_0\sqrt{\frac{s\log(ep/s) + \log n}{n}}\,.
            \end{align}

In the estimation error bound given by \eqref{dp.highd.rate}, the second term represents the convergence rate of the non-private Huber estimator, while the first term accounts for the privacy cost, influenced by a carefully chosen robustification parameter $\tau$. The robustification parameter $\tau$ is specifically chosen based on the effective sample size $n\epsilon$, along with the intrinsic dimensionality and noise level. This selection aims to achieve a balance among robustness, bias, and privacy. For normally distributed errors, \cite{cai2021cost} derived a minimax lower bound for $(\epsilon, \delta)$-DP estimation of $\bbeta^*$ under $\ell_2$-risk in the high dimensions. The lower bound is of the order $\sigma_0 \{ \sqrt{s^*n^{-1}\log p} + s^*({n \epsilon})^{-1}\log p\}$, which holds under the conditions $\epsilon \in (0, 1)$, $s^* \ll  p^{1-\omega} $, and $\delta < n^{-1-\omega}$ for some fixed constant $\omega>0$. In contrast, the slower term $\{s({n \epsilon})^{-1}\log p\}^{1-1/(2+\iota)}$ (ignoring  logarithmic factors) in \eqref{dp.highd.rate} highlights the joint effects of heavy-tailedness and privacy considerations more explicitly. Moreover,  each iteration in Algorithm \ref{alg:DPHuberhigh}
 consists of an intermediate update with complexity $O(np)$ and a noisy update (NoisyHT) with
 complexity $O(ps)$. Since $s \ll n$ and $T\asymp \log (n\epsilon)$, the overall complexity for Algorithm \ref{alg:DPHuberhigh} is $O\{np\log (n\epsilon)\}$.

\begin{remark}
\label{final.convergence.highd.subgaussian} 
The result in \eqref{dp.highd.rate} shows that the sparse $(\epsilon, \delta)$-DP Huber estimator achieves a near-optimal convergence rate for sub-Gaussian errors in high dimensions. Similar to Remark~\ref{final.convergence.lowd.subgaussian}, assume that the regression error $\varepsilon_i$ is sub-Gaussian, satisfying $\EE (e^{\lambda \varepsilon_i}\, |\, \bx_i) \leq e^{\sigma_0^2 \lambda^2/2}$ for all $\lambda \in \RR$. In this case, we set the robustification parameter as 
$$
    \tau \asymp \sigma_0 \sqrt{\log\bigg\{\frac{n \epsilon}{s\log(ep/s)+\log n}\bigg\}}\, .
$$
The $\ell_2$-error of the resulting sparse DP Huber estimator $\bbeta^{(T)}$ obtained in Algorithm~\ref{alg:DPHuberhigh}, up to constant and logarithmic factors, is upper bounded by
$$
 \sigma_0  \frac{s\log(ep/s)+\log n}{n \epsilon } + \sigma_0 \sqrt{\frac{s\log(ep/s)+\log n}{n}}
$$
with probability at least $1-C n^{-1}$, provided that $n\epsilon \gtrsim  s T  \{\log(pn) \}^{3/2}\sqrt{ \log(T/\delta)}$. Compared to Theorem~4.4 in \cite{cai2021cost}, our results do not require the following assumptions: (i) $\| \bbeta^* \|_2  <c_0$ for some constant $c_0>0$, which is needed by Algorithm 4.2 and Theorem 4.4 of \cite{cai2021cost}, and (ii) $\sqrt{|I|}\| \bx_{i,I} \|_{\infty} <c_{\bx}$ with probability one for all subsets $I \subseteq [p]$ with $|I| \ll n$, where $c_{\bx}$ is a dimension-free constant. Additionally, we relax the sample size requirement from $n \epsilon \gtrsim (s^{*})^{3/2}$ to $n\epsilon \gtrsim s^* $, up to logarithmic factors.
\end{remark}

We have analyzed the trade-offs among bias, privacy, and robustness in $\bbeta^{(T)}$, and compared Algorithm~\ref{alg:DPHuberhigh} with the privatized OLS method proposed by \cite{cai2021cost} under sub-Gaussian noise conditions. Our analysis shows that when $\varepsilon_i$ follows a sub-Gaussian distribution, the estimator obtained from Algorithm~\ref{alg:DPHuberhigh} achieves a near-optimal convergence rate, while requiring less restrictive sample size conditions and accommodating broader classes of $\bbeta^*$ and $\bx_i$, all while maintaining differential privacy guarantees. We now examine the statistical convergence of Algorithm~\ref{alg:DPHuberhigh} in a non-private setting. Let $\breve{\bbeta}^{(T)}$ be the non-private sparse Huber estimator obtained from Algorithm~\ref{alg:DPHuberhigh} by eliminating the noise terms $\bw^t_1, \ldots, \bw^t_s, \widetilde{\bw}^t$ for all iterations. Corollary \ref{cor:non.dp.HD.convergence}, which is of independent interest, establishes the convergence rate of $\breve{\bbeta}^{(T)}$ under conditions similar to those in Theorem~\ref{thm:high}. The proof of Corollary \ref{cor:non.dp.HD.convergence} follows closely that of Theorem \ref{thm:high}, with the noise terms $\{\bw^t_1, \ldots, \bw^t_s, \widetilde{\bw}^t\}_{t=0}^{T-1}$ set to zero.

\begin{corollary} \label{cor:non.dp.HD.convergence} 
Assume that $\bbeta^{(0)} \in \mathbb{H}(s)\cap \Theta(r_0)$ for some $r_0>0$, and the learning rate satisfies $\eta_0 = \eta/(2\lambda_1) $ for some $\eta\in(0,1)$. Then,  with $\tau \asymp  \sigma_0   n^{1/(2+\iota)} \{ s\log(ep/s)+ \log n \}^{- 1/(2+\iota)}$, $r_0 \asymp \tau$, and $T\asymp \log(n / \log p)$, the non-private sparse Huber estimator $\breve{\bbeta}^{(T)}$ satisfies  
\begin{align*}
\|\breve{\bbeta}^{(T)}-\bbeta^*\|_2  \lesssim \sigma_0\sqrt{\frac{s\log(ep/s) + \log n}{n} }
\end{align*}
with probability at least  $1-Cn^{-1}$, provided that $n \gtrsim s \log p+ \log n$ and  $s  \gtrsim  s^*$.
\end{corollary} 

\vspace{-.3cm}
\section{Numerical Studies}
\label{sec:5}

In this section, we conduct simulation studies to evaluate the numerical performance of the DP Huber regression in both low- and high-dimensional settings. The data $\{( y_i,\bx_i)\}_{i=1}^n$ are generated from the linear model $y_i = \bx_i^{\T} \bbeta^* + \sqrt{b}\varepsilon_i$, where $\bx_i = (1, \bx_{i,-1}^{\T})^{\T}$ and $\bx_{i,-1}\in \mathbb{R}^{p-1}$ follows a distribution that varies across different settings. The noise variables $\varepsilon_i$ are i.i.d. and follow either a standard normal distribution $\cN(0, 1)$ or a heavy-tailed $t_{2.25}$ distribution; the constant $b$ determines the noise scale. 
The construction of $\bbeta^*$ differs across scenarios. In the low-dimensional case, each component of $\bbeta^*$ is independently set to $a$ or $-a$ with equal probability. In the high-dimensional setting, only the first $s^*$ entries are generated in the same manner and the remaining $p-s^*$ components are set to zero. The signal-to-noise ratio (SNR) is defined as ${\rm Var}(\bx_i^{\T}\bbeta^*)/\{b{\rm Var}(\varepsilon_i)\}$.  
For a fixed design distribution and noise law, larger values of $a^2/b$ correspond to higher SNR.

\subsection{Selection of tuning parameters}

Recognizing that Algorithms \ref{alg:DPHuberlow} and \ref{alg:DPHuberhigh} both rely on a sequence of tuning parameters, we begin by outlining the selection principles adopted in our implementation, as these choices are essential for practical performance and reproducibility. Although our theory permits a broad class of initial values $\bbeta^{(0)}$ satisfying $\|\bbeta^{(0)} - \bbeta^*\|_2 \leq r_0$ for some $r_0 > 0$ that may even diverge with $n$, our numerical studies, especially in high-dimensional settings, show that a well-chosen initialization can markedly improve finite-sample behavior. Motivated by \cite{liu2024DP}, we therefore use, in both low- and high-dimensional scenarios, a private initial estimator obtained by solving a ridge-penalized Huber regression followed by output perturbation. Additional details on this private initialization are provided below. Throughout this section, we fix $\delta = 10 n^{-1.1}$ and present tuning choices under the $(\epsilon,\delta)$-DP framework. The corresponding rules for $\epsilon$-GDP, along with proofs of the privacy guarantee for the initialization step, are deferred to Section  \ref{privacy.guarantee}  of the supplementary material.

\subsubsection{Selection of tuning parameters in Algorithm \ref{alg:DPHuberlow}}
\label{selection.low}

We choose the initial value $\bbeta^{(0)}$ in a data-dependent manner. To ensure that the overall procedure remains $(\epsilon,\delta)$-DP, we divide the privacy budget between the initialization step and Algorithm~\ref{alg:DPHuberlow}, which we denote by $(\epsilon_{\rm init}, \delta_{\rm init})$ and $(\epsilon_{\rm main}, \delta_{\rm main})$, respectively. Specifically, we set $\epsilon_{\rm init} = \epsilon_{\rm main}/5= \epsilon /6 $ and $\delta_{\rm init} = \delta_{\rm main}/5= \delta /6 $. 

\indent
\underline{{\textbf{Initialization of}} $\bbeta^{(0)}$.}  We estimate $\bbeta^{(0)}$ by solving a ridge-penalized Huber regression with row-wise clipping: 
\begin{align}\label{ridge.beta}
    \widehat\bbeta^{(0)} = \arg\min_{\bbeta \in \mathbb{R}^p}  \bigg\{ \frac{1}{n}\sum_{i=1}^n \rho_{\tau_0}(y_i - \tilde{\bx}_{i}^{\T}\bbeta)  + \frac{\lambda_0}{2} \|\bbeta\|_2^2 \bigg\}\,,
\end{align}
where $\lambda_0 = 0.2$ and $\tilde{\bx}_{i} = (1, \tilde{\bx}_{i,-1}^{\T})^{\T}$ with $\tilde{\bx}_{i,-1} = \bx_{i,-1} \min \{ \sqrt{p}/(6\|\bx_{i,-1}\|_2),1 \} $. For this initialization step, the privacy parameters are allocated as $(\epsilon_{\rm init}/4, 0)$ for selecting $\tau_0$ and $(3\epsilon_{\rm init}/4, \delta_{\rm init})$ for the output perturbation.

\begin{itemize}

    \item Selection of $\tau_0$. For each $i\in[n]$, define $\smalltilde{y}_i = \min\{ \log n, \max\{ -\log n ,y_i \} \}$. Let $m_1 = n^{-1}\sum_{i=1}^n\smalltilde{y}_i  + w_1$ and $m_2 =  n^{-1}\sum_{i=1}^n\smalltilde{y}_i^2 + w_2$, with $w_1 \sim {\rm Laplace}(16(n\epsilon_{\rm init})^{-1}\log n)$ and $w_2 \sim {\rm Laplace}(8(n\epsilon_{\rm init})^{-1}(\log n)^2)$. Set $\tau_0$ as $\sqrt{ m_2 - m_1^2}$ if $m_2 - m_1^2 > 0$ and $2$ otherwise.

    \item Output perturbation. Let  $\smalltilde{\sigma} = 8B\tau_0(3n\epsilon_{\rm init}\lambda_0)^{-1}\sqrt{2\log(1.25/\delta_{\rm init})}$ and $B= \sqrt{1+p/36}$. We then set $\bbeta^{(0)} = \widehat{\bbeta}^{(0)} + \smalltilde{\sigma}\cdot\cN(\bzero,\bI_p)$.
\end{itemize} 

\indent
\underline{{\textbf{Selection of other parameters}}.} 
We set $\eta_0 = 0.2$, $\gamma = 0.5 \sqrt{p+\log n}$, and $T= \lceil 2 \log n \rceil$. Recall that the optimal choice of $\tau$ satisfies $\tau \asymp \sigma_0 \{n\epsilon(p+\log n )^{-1}\}^{1/2}$ (here we fix $\iota=0$ for simplicity). We then set $\tau =0.04 \tau_0\sqrt{ n\epsilon  (p+\log n)^{-1}}$, where $\tau_0$ defined in \eqref{ridge.beta} serves as a rough upper bound for $\sigma_0$.

\subsubsection{Selection of tuning parameters in Algorithm \ref{alg:DPHuberhigh}}
\label{selection.high}

Similar to Section~\ref{selection.low}, we divide the overall privacy budget between the initialization step and Algorithm \ref{alg:DPHuberhigh}, using $(\epsilon_{\rm init}, \delta_{\rm init})$ for initialization and $(\epsilon_{\rm main}, \delta_{\rm main})$ for the main algorithm. We set     $\epsilon_{\rm init}=2\epsilon_{\rm main}=2\epsilon /3 $ and $\delta_{\rm init}=\delta_{\rm main }= \delta /2$, which ensures that the resulting estimator satisfies the desired $(\epsilon,\delta)$-DP guarantee.

\indent
\underline{{\textbf{Initialization of}} $\bbeta^{(0)}$.}  In high-dimensional sparse models, only the covariates corresponding to the true support of $\bbeta^*$ carry meaningful signal. Guided by this observation, we construct the initial estimator $\bbeta^{(0)}$ in two stages. First, we obtain a DP estimate of the support using a clipped gradient procedure. Second, conditional on the selected support, we compute a private initial estimator restricted to that subset of covariates. To ensure privacy, we allocate the budget for these two components as $(\epsilon_{\rm init}/2, 0)$ for the support recovery step and $(\epsilon_{\rm init}/2, \delta_{\rm init})$ for subsequent private estimation step.

\begin{itemize}
     \item Support recovery. For each $j\in \{2,\ldots, p\}$, define  $u_{i,j}= y_i x_{i,j} $. Set ${g}_j =  | {n}^{-1}\sum_{i=1}^n  \smalltilde{u}_{i,j}   |$, $\smalltilde{u}_{i,j}= u_{i,j}\min\{ \sqrt{\log(pn)}/|u_{i,j}|,1 \}$, with the convention $\smalltilde u_{i,j}=0$ when $u_{i,j}=0$. Let $s_0=s-1$. We then apply the Report Noisy Max procedure \citep[Claim 3.9]{dwork2014algorithmic} to $({g}_2,\ldots, {g}_p)$ in a peeling scheme: run it $s_0$ times, each time adding independent Laplace noise ${\rm Laplace}(2s_0 \Delta/ \epsilon_{\rm init})$ with $\Delta= 2n^{-1}\sqrt{\log (pn)}$, and select the top-$s_0$ indices without replacement. This yields the support $\hat{\cS}_0=(j_{(1)},\ldots,j_{(s_0 )})$.

    \item Initialize $\bbeta^{(0)}$. We first obtain an $(\epsilon_{\rm init}/2,\delta_{\rm init})$-DP estimator $\hat{\bbeta}_{\mathrm{init}}$ by applying the initialization procedure described in Section~\ref{selection.low} to $\{(y_i,\bx_{i,\hat{\cS}})\}_{i=1}^n$, where $\hat{\cS}=\{1\} \cup \hat{\cS}_0$. We then embed this estimator into $\bbeta^{(0)} \in \mathbb{R}^p$ by setting $\bbeta^{(0)}_{\hat{\cS}} = \hat{\bbeta}_{\mathrm{init}}$ and $\bbeta^{(0)}_{[p]\setminus \hat{\cS}} = \bzero$.
\end{itemize}

\indent
\underline{{\textbf{Selection of other parameters}}.} We set $\gamma = 0.5\sqrt{ \log (pn)}$, take $T= \lceil 2 \log n \rceil$, and choose the working sparsity level as $s=\lceil 1.2s^*\rceil $. The learning rate is fixed at $\eta_0 = 0.01$. Recall that the optimal order of the robustification parameter satisfies $\tau \asymp \sigma_0 \{n\epsilon(s\log p+\log n )^{-1}\}^{1/2}$ (here we fix $\iota =0$  for simplicity). We then set $\tau = 0.04\tau_0 \sqrt{  n\epsilon(s\log p+\log n)^{-1}}$, where the choice of $\tau_0$ is described in Section~\ref{selection.low}; consistent with the low-dimensional procedure, this same quantity is also used in constructing the initial estimator $\bbeta^{(0)}$.

\subsection{Low-dimensional setting}
\label{sec.simulation.low}

In the low-dimensional setting, each entry of $\bx_{i,-1}$ is independently drawn from either $\cN(0,1)$ or ${\rm Uniform}( -\sqrt{3},\sqrt{3})$. The signal and noise scale parameters $a$ and $b$ take values in $\{0.5, 1, 2\}$. As a benchmark, we consider the non-private Huber estimator obtained by running Algorithm~\ref{alg:DPHuberlow} without initialization, clipping, and noise injection. For this estimator, we set $\eta_0 = 0.5$, $T=\lceil 2\log n\rceil$, and update the robustification parameter according to $\tau  = 0.2\smallhat{\sigma}_0\sqrt{n (p+\log n )^{-1}}$, where $\smallhat{\sigma}_0^2 = n^{-1}\sum_{i=1}^n(y_i-\bar{y})^2$ with $\bar{y}=n^{-1}\sum_{i=1}^n y_i$.

We set the dimension $p\in \{5,10,20\}$ and let the sample size $n$ range from $2500$ to $10000$ to examine how the relative $\ell_2$-error, $\|\mathring\bbeta-\bbeta^*\|_2/\|\bbeta^*\|_2$, varies with both dimension and sample size. All results are computed over 300 repetitions, and $\mathring{\bbeta}$ denotes the estimator being evaluated. We first investigate the effect of the initialization scheme on the accuracy of  the $(\epsilon,\delta)$-DP estimator using the relative 
$\ell_2$-error across different sample sizes; see Figure~\ref{fig.low.init}. Figure~\ref{fig.low.eps} further compares the $(\epsilon,\delta)$-DP Huber estimators at multiple privacy levels with the non-private Huber estimator. As expected, weaker privacy protection (large $\epsilon$) yields improved estimation accuracy. In addition, when the dimension is fixed at $p=10$, Table~\ref{tab.low.est_p10} shows that the estimation error of the $(\epsilon, \delta)$-DP Huber estimator decreases as the SNR increases. Additional results at $p\in \{5, 20\}$, along with finite-sample performance of the $\epsilon$-GDP Huber estimators, are given in supplementary material (see Tables \ref{tab.low.est_p5}, \ref{tab.low.est_p20} and \ref{tab.low.est_GDP_p5}--\ref{tab.low.est_GDP_p20}).

Next, we compare the coverage performance of the privatized confidence intervals (CIs) introduced in Section~\ref{noisyGD.lowd} with both their non-private counterparts and the private bootstrap procedure described in Algorithm~2 of \cite{Ferrando2022}\footnote{\href{https://github.com/ceciliaferrando/PB-DP-CIs}{https://github.com/ceciliaferrando/PB-DP-CIs}}, for which the required data ranges are calibrated using the empirical maxima of each simulated dataset. The parameters $(\tau_1, \gamma_1)$ used to construct $\widetilde{\bXi}_{\tau_1,\gamma_1,\epsilon}$ in \eqref{tildeXi} are set as $\gamma_1 = 0.5\sqrt{p+\log n}$ and $\tau_1 = 0.95\tau_0 \sqrt{n\epsilon (p+\log n)^{-1}}$, where $\tau_0$ is selected according to Section~\ref{selection.low}. The pair 
$$
   (\varsigma_1, \varsigma_2) =  \big( 2\gamma_1^2(n\epsilon)^{-1} \sqrt{2\log(1.25/\delta)}\,, 2\gamma_1^2 \tau_1^2 (n\epsilon)^{-1}\sqrt{2\log(1.25/\delta)}  \big)
$$ 
follows from Lemma~\ref{lem.dp} in the supplementary material. We allocate the total privacy budget $(\epsilon,\delta)$ across the initialization, main estimation, and inference steps as $(\epsilon_{{\rm init}}, \delta_{{\rm init}})$, $(\epsilon_{{\rm main}}, \delta_{{\rm main}})$, and $(\epsilon_{{\rm infer}}, \delta_{{\rm infer}})$, respectively, with $\epsilon_{\rm init}=\epsilon_{\rm infer}=\epsilon_{\rm main}/4=\epsilon/6$, and the same proportions applied to $\delta$. The comparison results for the setting $(n,p,a,b)=(10000,5,1,1)$ are reported in Table~\ref{tab.inference}. All methods achieve nominal coverage across a range of covariate designs and noise distributions. Notably, under heavy-tailed noise, the privatized CIs are substantially shorter than the private bootstrap intervals of \cite{Ferrando2022}.

\begin{figure}[h]
    \centering
\includegraphics[width=0.8\textwidth]{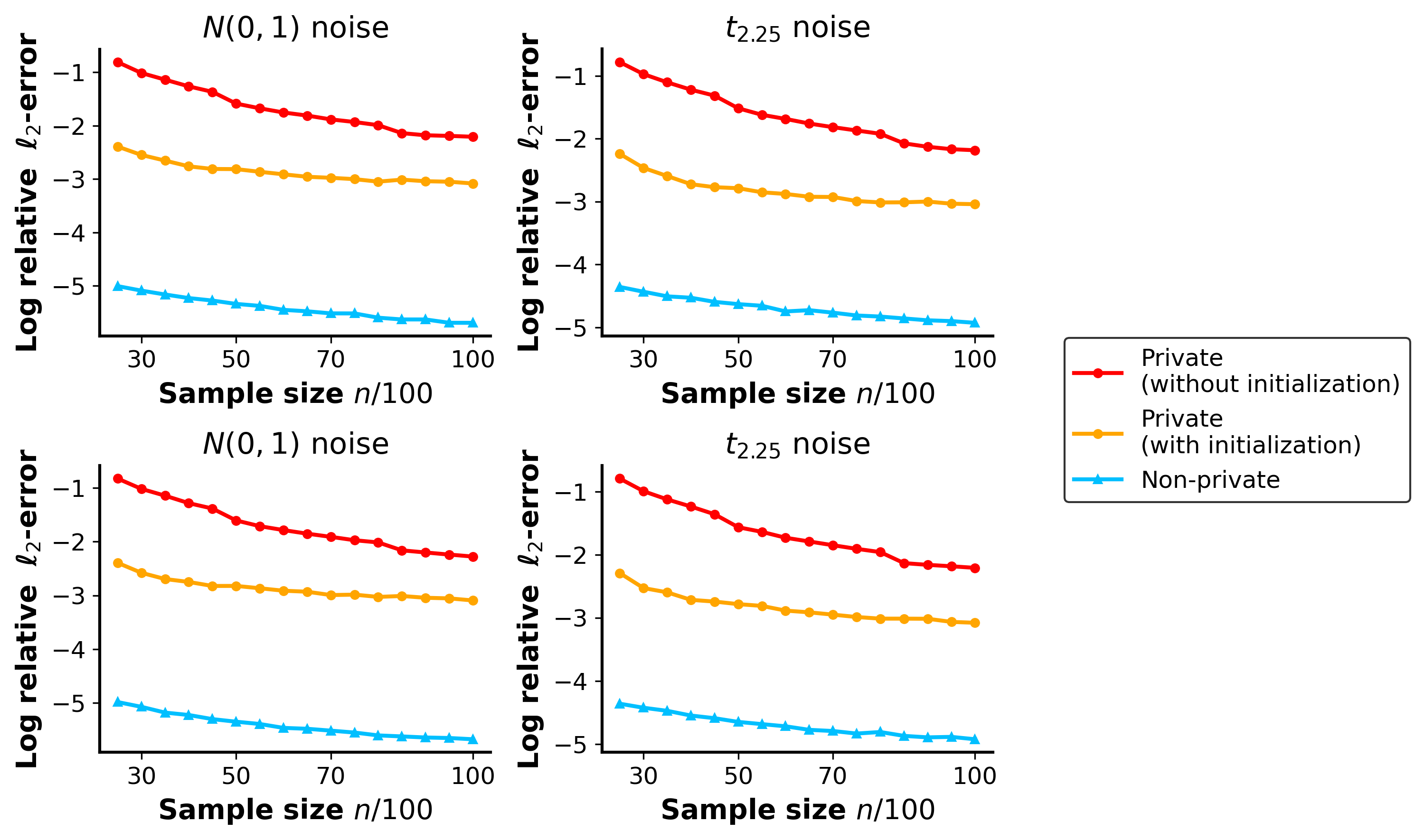}
    \caption{Plots of the average logarithmic relative $\ell_2$-error over 300 repetitions as a function of the sample size under different design settings. The top two panels correspond to Gaussian covariates, and the bottom two panels to uniform covariates, with $\epsilon=0.9$, $p=10$, $a=2$, and $b=0.5$. The noise distribution is either $\cN(0,1)$ or $t_{2.25}$.}
\label{fig.low.init}
\end{figure}

\begin{figure}[h]
    \centering
\includegraphics[width=0.9\textwidth]{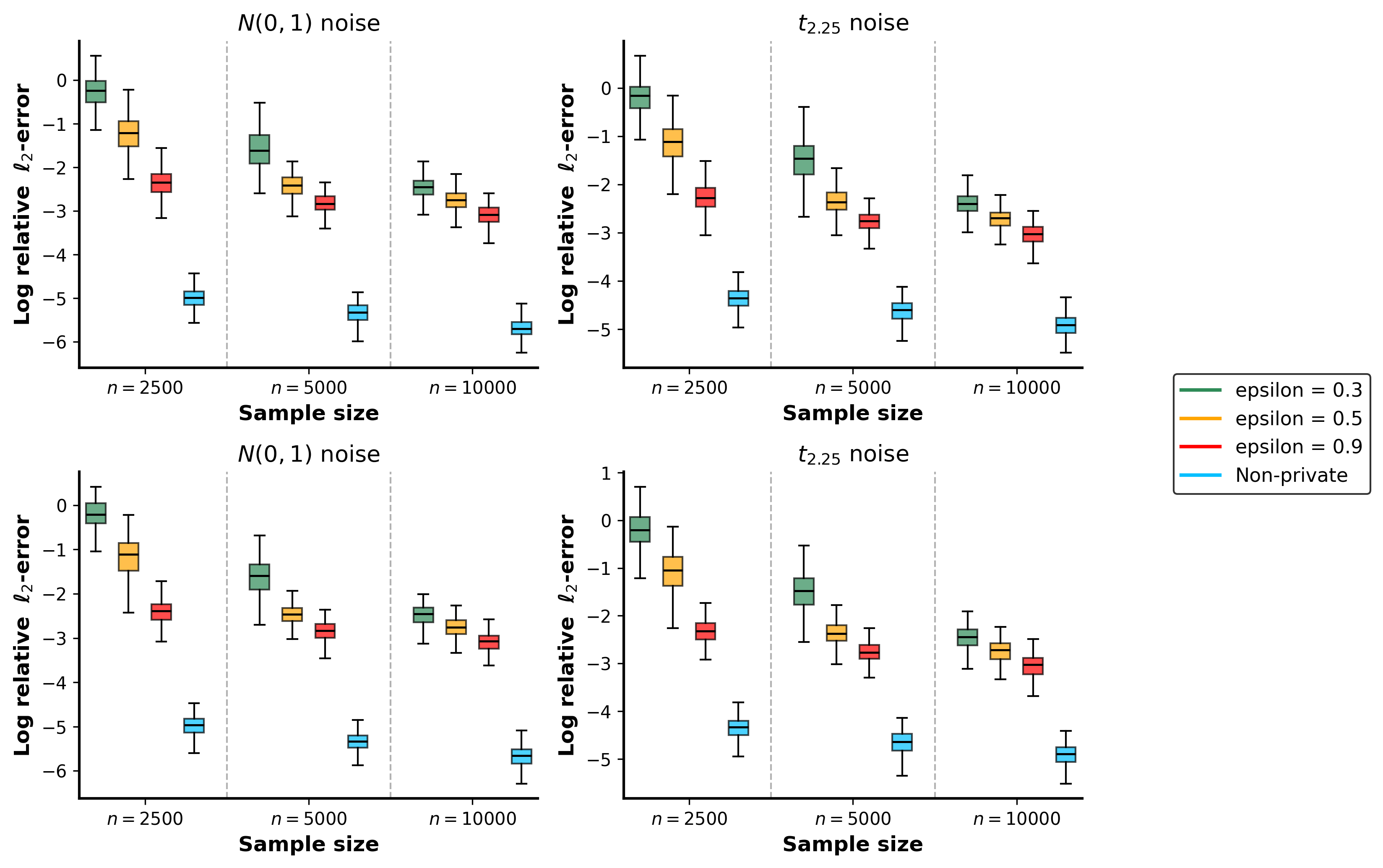}
    \caption{Boxplots of the logarithmic relative $\ell_2$-error, based on 300 repetitions, across three sample sizes under four design settings. The top two panels correspond to Gaussian covariates, and the bottom two to uniform covariates, with $p=10$, $a=2$, and $b=0.5$. The noise distribution is either $\cN(0,1)$ or $t_{2.25}$.}
\label{fig.low.eps}
\end{figure}

\begin{sidewaystable} 
\centering 
\scriptsize 
\setlength{\tabcolsep}{2.5pt}
\caption{Average logarithmic relative $\ell_2$-error over 300 repetitions for the $(\epsilon,\delta)$-DP Huber estimators under various combinations of $(a, b, n, \epsilon)$, with $p = 10$.} 
\renewcommand{\arraystretch}{1.3} 
\begin{tabular}{ccc cccccccc cccccccc} \toprule &&&\multicolumn{8}{c}{Gaussian design} &\multicolumn{8}{c}{ Uniform design} \\\cmidrule(lr){4-11} \cmidrule(lr){12-19} &&&\multicolumn{4}{c}{$ \cN(0,1)$ noise} &\multicolumn{4}{c}{ $ t_{2.25}$ noise}& \multicolumn{4}{c}{$ \cN(0,1)$ noise} &\multicolumn{4}{c}{ $ t_{2.25}$ noise} \\\cmidrule(lr){4-7} \cmidrule(lr){8-11}\cmidrule(lr){12-15} \cmidrule(lr){16-19} $a$&$b$ &$n$ & non-private&$\epsilon=0.3$&$\epsilon=0.5$&$\epsilon=0.9$&non-private&$\epsilon=0.3$&$\epsilon=0.5$&$\epsilon=0.9$&non-private&$\epsilon=0.3$&$\epsilon=0.5$&$\epsilon=0.9$&non-private&$\epsilon=0.3$&$\epsilon=0.5$&$\epsilon=0.9$\\ 
\midrule 0.5&0.5 &2500 &$-3.624$ & $0.272$ & $-0.650$ & $-1.612$ & $-3.018$ & $0.350$ & $-0.489$ & $-1.451$ & $-3.603$ & $0.262$ & $-0.611$ & $-1.693$ & $-3.026$ & $0.305$ & $-0.476$ & $-1.510$ \\
&&5000&$-3.955$ & $-0.950$ & $-1.745$ & $-2.366$ & $-3.290$ & $-0.780$ & $-1.609$ & $-2.168$ & $-3.966$ & $-0.993$ & $-1.819$ & $-2.393$ & $-3.309$ & $-0.788$ & $-1.618$ & $-2.158$ \\
&&10000&$-4.310$ & $-1.890$ & $-2.392$ & $-2.805$ & $-3.578$ & $-1.754$ & $-2.176$ & $-2.533$ & $-4.292$ & $-1.957$ & $-2.408$ & $-2.811$ & $-3.570$ & $-1.788$ & $-2.202$ & $-2.551$ \\
&1&2500&$-3.277$ & $0.309$ & $-0.540$ & $-1.408$ & $-2.692$ & $0.437$ & $-0.348$ & $-1.234$ & $-3.256$ & $0.313$ & $-0.513$ & $-1.485$ & $-2.703$ & $0.404$ & $-0.304$ & $-1.273$ \\
&&5000&$-3.609$ & $-0.797$ & $-1.540$ & $-2.193$ & $-2.964$ & $-0.598$ & $-1.376$ & $-1.951$ & $-3.619$ & $-0.848$ & $-1.607$ & $-2.212$ & $-2.982$ & $-0.600$ & $-1.385$ & $-1.946$ \\
&&10000&$-3.963$ & $-1.708$ & $-2.225$ & $-2.667$ & $-3.249$ & $-1.518$ & $-1.964$ & $-2.342$ & $-3.945$ & $-1.754$ & $-2.238$ & $-2.676$ & $-3.238$ & $-1.568$ & $-1.987$ & $-2.357$ \\
&2&2500&$-2.931$ & $0.387$ & $-0.386$ & $-1.175$ & $-2.360$ & $0.572$ & $-0.144$ & $-0.974$ & $-2.909$ & $0.375$ & $-0.373$ & $-1.247$ & $-2.372$ & $0.566$ & $-0.101$ & $-1.014$ \\
&&5000&$-3.262$ & $-0.613$ & $-1.314$ & $-1.955$ & $-2.631$ & $-0.368$ & $-1.113$ & $-1.707$ & $-3.273$ & $-0.667$ & $-1.367$ & $-1.975$ & $-2.649$ & $-0.373$ & $-1.116$ & $-1.697$ \\
&&10000&$-3.617$ & $-1.480$ & $-1.998$ & $-2.463$ & $-2.913$ & $-1.262$ & $-1.722$ & $-2.118$ & $-3.598$ & $-1.520$ & $-2.021$ & $-2.473$ & $-2.901$ & $-1.315$ & $-1.741$ & $-2.133$ \\
\midrule
1&0.5&2500&$-4.317$ & $-0.069$ & $-0.971$ & $-2.041$ & $-3.650$ & $-0.006$ & $-0.836$ & $-1.905$ & $-4.296$ & $-0.051$ & $-0.914$ & $-2.113$ & $-3.649$ & $0.003$ & $-0.803$ & $-1.943$ \\
&&5000&$-4.648$ & $-1.322$ & $-2.170$ & $-2.635$ & $-3.924$ & $-1.163$ & $-2.035$ & $-2.503$ & $-4.659$ & $-1.366$ & $-2.217$ & $-2.645$ & $-3.940$ & $-1.163$ & $-2.036$ & $-2.491$ \\
&&10000&$-5.003$ & $-2.272$ & $-2.611$ & $-2.934$ & $-4.223$ & $-2.130$ & $-2.478$ & $-2.813$ & $-4.985$ & $-2.307$ & $-2.616$ & $-2.934$ & $-4.218$ & $-2.185$ & $-2.504$ & $-2.825$ \\
&1&2500&$-3.970$ & $-0.039$ & $-0.893$ & $-1.889$ & $-3.335$ & $0.064$ & $-0.734$ & $-1.734$ & $-3.949$ & $-0.019$ & $-0.845$ & $-1.967$ & $-3.341$ & $0.070$ & $-0.692$ & $-1.773$ \\
&&5000&$-4.302$ & $-1.212$ & $-2.039$ & $-2.565$ & $-3.610$ & $-1.023$ & $-1.874$ & $-2.390$ & $-4.312$ & $-1.254$ & $-2.086$ & $-2.579$ & $-3.627$ & $-1.029$ & $-1.878$ & $-2.378$ \\
&&10000&$-4.656$ & $-2.162$ & $-2.555$ & $-2.897$ & $-3.902$ & $-1.984$ & $-2.374$ & $-2.726$ & $-4.638$ & $-2.203$ & $-2.563$ & $-2.900$ & $-3.897$ & $-2.046$ & $-2.398$ & $-2.736$ \\
&2&2500&$-3.624$ & $0.017$ & $-0.782$ & $-1.690$ & $-3.018$ & $0.164$ & $-0.596$ & $-1.529$ & $-3.603$ & $0.036$ & $-0.739$ & $-1.761$ & $-3.026$ & $0.161$ & $-0.556$ & $-1.568$ \\
&&5000&$-3.955$ & $-1.063$ & $-1.847$ & $-2.431$ & $-3.290$ & $-0.856$ & $-1.672$ & $-2.232$ & $-3.966$ & $-1.100$ & $-1.890$ & $-2.449$ & $-3.309$ & $-0.869$ & $-1.677$ & $-2.221$ \\
&&10000&$-4.310$ & $-1.997$ & $-2.441$ & $-2.819$ & $-3.578$ & $-1.803$ & $-2.228$ & $-2.601$ & $-4.292$ & $-2.030$ & $-2.453$ & $-2.825$ & $-3.570$ & $-1.863$ & $-2.250$ & $-2.612$ \\
\midrule
2&0.5&2500&$-5.010$ & $-0.253$ & $-1.233$ & $-2.362$ & $-4.278$ & $-0.204$ & $-1.138$ & $-2.278$ & $-4.989$ & $-0.228$ & $-1.170$ & $-2.418$ & $-4.267$ & $-0.215$ & $-1.102$ & $-2.320$ \\
&&5000&$-5.341$ & $-1.591$ & $-2.430$ & $-2.831$ & $-4.556$ & $-1.485$ & $-2.368$ & $-2.779$ & $-5.352$ & $-1.635$ & $-2.465$ & $-2.836$ & $-4.561$ & $-1.500$ & $-2.374$ & $-2.769$ \\
&&10000&$-5.696$ & $-2.473$ & $-2.771$ & $-3.098$ & $-4.861$ & $-2.409$ & $-2.723$ & $-3.046$ & $-5.678$ & $-2.494$ & $-2.775$ & $-3.094$ & $-4.855$ & $-2.465$ & $-2.742$ & $-3.054$ \\
&1&2500&$-4.664$ & $-0.245$ & $-1.195$ & $-2.296$ & $-3.964$ & $-0.185$ & $-1.079$ & $-2.172$ & $-4.642$ & $-0.221$ & $-1.136$ & $-2.357$ & $-3.957$ & $-0.194$ & $-1.044$ & $-2.217$ \\
&&5000&$-4.995$ & $-1.536$ & $-2.388$ & $-2.820$ & $-4.239$ & $-1.401$ & $-2.286$ & $-2.736$ & $-5.006$ & $-1.579$ & $-2.425$ & $-2.826$ & $-4.249$ & $-1.419$ & $-2.293$ & $-2.727$ \\
&&10000&$-5.349$ & $-2.449$ & $-2.764$ & $-3.090$ & $-4.543$ & $-2.350$ & $-2.687$ & $-3.019$ & $-5.331$ & $-2.472$ & $-2.769$ & $-3.088$ & $-4.536$ & $-2.409$ & $-2.707$ & $-3.027$ \\
&2&2500&$-4.317$ & $-0.230$ & $-1.134$ & $-2.175$ & $-3.650$ & $-0.153$ & $-0.998$ & $-2.023$ & $-4.296$ & $-0.206$ & $-1.081$ & $-2.239$ & $-3.649$ & $-0.163$ & $-0.966$ & $-2.070$ \\
&&5000&$-4.648$ & $-1.448$ & $-2.295$ & $-2.785$ & $-3.924$ & $-1.290$ & $-2.158$ & $-2.661$ & $-4.659$ & $-1.488$ & $-2.337$ & $-2.794$ & $-3.940$ & $-1.311$ & $-2.168$ & $-2.653$ \\
&&10000&$-5.003$ & $-2.388$ & $-2.742$ & $-3.074$ & $-4.223$ & $-2.253$ & $-2.627$ & $-2.971$ & $-4.985$ & $-2.413$ & $-2.748$ & $-3.075$ & $-4.218$ & $-2.315$ & $-2.645$ & $-2.981$ \\

 \bottomrule \end{tabular}  \label{tab.low.est_p10} 
\end{sidewaystable}

\begin{table}[h]
\centering
\scriptsize              
\setlength{\tabcolsep}{2.5pt} 
 \caption{Empirical coverage and interval widths (with standard deviations) of the $100(1-\alpha)\%$ CIs for $\bbeta^*$, constructed using the $(\epsilon,\delta)$-DP Huber estimator, its non-private counterpart, and the private bootstrap  method. Results are reported for $n=10000$, $p=5$, $\epsilon=0.5$, and $a=b=1$.}
 
\renewcommand{\arraystretch}{1.1}
\resizebox{0.7\textwidth}{!}{
\begin{tabular}{cc cccccccc}
\toprule
&&\multicolumn{4}{c}{Gaussian design} &\multicolumn{4}{c}{ Uniform design}  \\\cmidrule(lr){3-6} \cmidrule(lr){7-10}
&&\multicolumn{2}{c}{$ \cN(0,1)$ noise} &\multicolumn{2}{c}{ $ t_{2.25}$ noise}& \multicolumn{2}{c}{$ \cN(0,1)$ noise} &\multicolumn{2}{c}{ $ t_{2.25}$ noise} \\\cmidrule(lr){3-4} \cmidrule(lr){5-6}\cmidrule(lr){7-8} \cmidrule(lr){9-10}
$\alpha$&   & coverage & width (sd)& coverage & width (sd)& coverage & width (sd)& coverage & width (sd)\\
\midrule
$0.05$ & private & 0.942 
 & $ {0.352 
} ~(0.009)$ & 0.943 
& $ {0.430 
} ~(0.017)$ & 0.941 
 &  $ {0.349 
}~(0.005)$ 
	&0.938 
 & $ {0.421 
}~(0.007)$  \\
    & non-private & 0.954 
 &$ ~~~{0.039 
}~(<0.001)$ 
	&0.953 
 &$ ~~~{0.080 
}~(<0.001)$
	&0.949 
 & $~~~ {0.039 
}~(<0.001)$
&	0.952 
 &$ ~~~ {0.080
}~(<0.001)$
  \\
  & bootstrap &0.935 
 & $ {0.557 
}~(0.001  
)$ 
&	0.949 
 & $ {1.827 
}~(0.011
)$
 &	0.949 
 &  $ ~~~{0.101  
}~(<0.001  
)$ &	0.941 
 & $ {0.693  
}~(0.001  
)$\\
\midrule
$0.1$ & private & 0.909 
  
 & $ {0.296  
}~(0.008)$ & 0.916 
& $ {0.361 
}~(0.015)$ & 0.905  
 &  $ {0.293 
}~(0.004)$ 
	&0.912  
 & $ {0.354  
}~(0.005)$  \\
    & non-private & 0.903  
 &$ ~~~{0.033 
}~(<0.001)$ 
	&0.896 
 &$ ~~~ {0.067 
}~(<0.001)$
	&0.889 
 & $~~~{0.033 
}~(<0.001)$
&	0.897 
 &$~~~{0.067 
}~(<0.001)$
  \\
  & bootstrap &0.882  
 & $ {0.453  
}~(0.001  
)$ 
&	0.891  
 & $ {1.434  
}~(0.004
)$
 &	0.899  
 &  $~~~{0.083  
}~(<0.001  
)$ &	0.899  
 & $ {0.540
}~(0.002 
)$\\
\bottomrule
    \end{tabular} 
    }
    \label{tab.inference}
\end{table}

\subsection{High-dimensional setting} 
\label{highdim:simulation}

In the high-dimensional setting, the covariate vectors $\bx_{i,-1}$ are independently generated from $\cN(\bzero, \bPsi)$, where the covariance matrix $\bPsi=(\Psi_{j,k})_{j,k \in [p-1]}$ with $\Psi_{j,k} = 0.1^{|j-k|}$. We set $a=b=1$ and $s^*=10$. As a benchmark, we consider the non-private sparse Huber estimator, Algorithm~\ref{alg:DPHuberhigh} run without initialization, clipping, and  noise injection. For this estimator, we take $\eta_0 = 0.2$, $T=\lceil 2\log n\rceil$, and $s=\lceil 1.2s^*\rceil$, and update the robustification parameter according to  $\tau  =  0.1\smallhat{\sigma}_0 \sqrt{   n  (s\log p+\log n )^{-1}}$, where $\smallhat{\sigma}_0^2= n^{-1}\sum_{i=1}^n(y_i-\bar{y})^2$. 

Figure~\ref{fig.high.eps_p10000} displays the logarithmic relative $\ell_2$-error as a function of the sample size for both the $(\epsilon, \delta)$-DP and non-private sparse Huber estimators, with $p  = 10000$ and $\epsilon \in \{0.5, 0.9\}$. An additional plot for $p = 5000$ is provided in the supplementary material (see Figure~\ref{fig.high.eps_p5000}).

To assess robustness under heavy-tailed noise, we also implement the sparse DP least squares estimator (sparse DP LS) following Algorithm 4.2 of \cite{cai2021cost}. Because their method is designed for models with zero-mean covariates and no intercept, we restrict attention to the DP estimation of the slope coefficients. Accordingly, we center both the response variables and the covariates before applying Algorithm 4.2 of  \cite{cai2021cost}. In addition to $(\eta_0,s,T,\bbeta^{(0)})$, the algorithm requires three extra tuning parameters: the truncation level $R$, the feasibility parameter $c_0$, and the noise scale $B$. As specified in Theorem 4.4 of \cite{cai2021cost}, we set $R={C}_R\sigma\sqrt{2\log n}$ with $ \sigma = 2$, ${C}_R \in \{0.1,0.5,1\}$,  $c_0=1.01\sqrt{s^*}$, $B =  {4(R + c_0 c_{\bx}) c_{\bx}}/{\sqrt{s}}$, and $ c_{\bx} = 0.5\sqrt{ \log (pn)}$. The results for $p=10000$ and $n\in \{5000, 10000, 15000\}$ are summarized in Table~\ref{tab.high.est_p10000}. The proposed sparse DP Huber estimator achieves substantially higher statistical accuracy than the sparse DP LS method, while also requiring far less tuning. Moreover, under sub-Gaussian errors, it attains the same privacy level with a smaller noise scale than the sparse DP LS estimator. Additional results for $p = 5000$ are provided in the supplementary material (see Table~\ref{tab.high.est_p5000}).

 \begin{figure}[h]
    \centering
\includegraphics[width=0.8\textwidth]{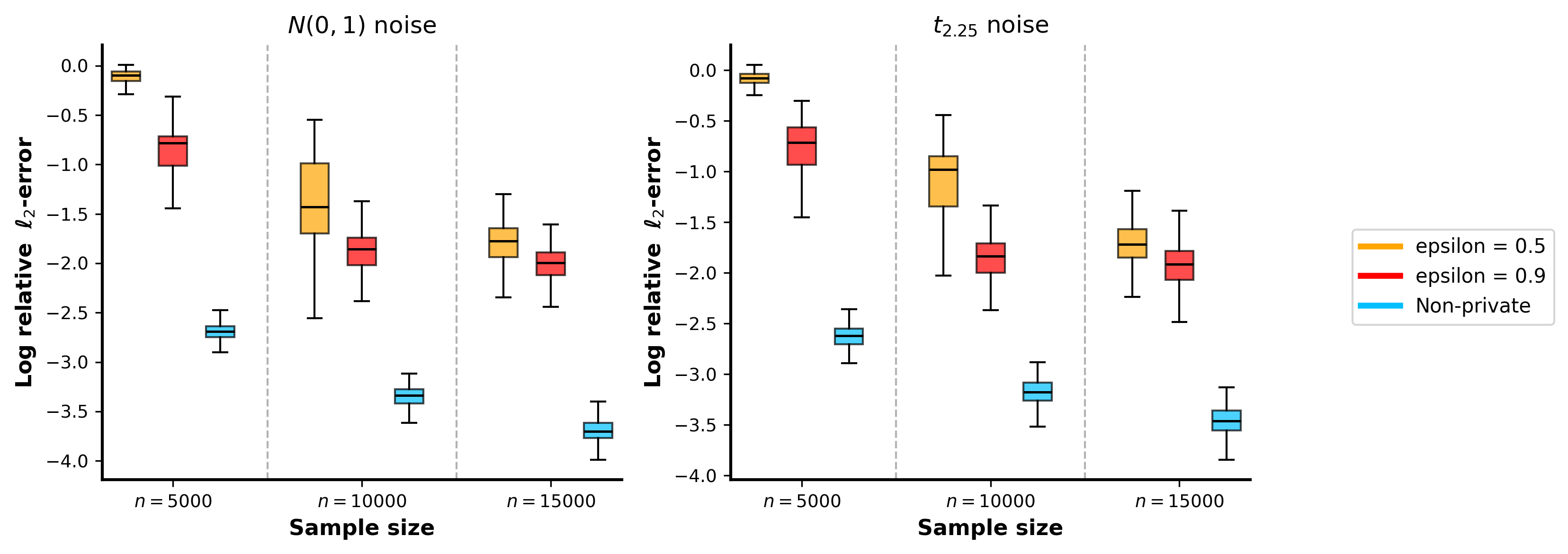}
\caption{Boxplots of the logarithmic relative $\ell_2$-error over 300 repetitions for private and non-private sparse Huber estimators across sample sizes, with $p = 10000$ and $s^* = 10$.}
\label{fig.high.eps_p10000}
\end{figure}

\begin{table}[ht]
\centering
\scriptsize              
\caption{Average logarithmic relative $\ell_2$-error (for slope coefficients) over 300 repetitions across sample sizes, with privacy levels $(0.5, 10n^{-1.1})$ and $p = 10000$.}
\renewcommand{\arraystretch}{1.1}
\resizebox{0.6\textwidth}{!}{
\begin{tabular}{c c c c ccc }
\toprule
\multirow{2}{*}{Noise} & \multirow{2}{*}{$n$} & \multirow{2}{*}{ non-private sparse  Huber} &  \multirow{2}{*}{sparse DP Huber}  & \multicolumn{3}{c}{ sparse DP LS} \\
\cline{5-7}
& & & & $C_R=0.1$ & $C_R=0.5$ & $C_R=1$ \\
\midrule
  $\cN(0,1)$       & 5000  & $-2.671$ &$-0.063$  & $0.066$ & $0.107$ & $0.164$ \\

&10000 & $-3.327$  & $-1.337$  & $-0.079$ & $0.046$ & $0.084$ \\

&15000 & $-3.665$  &$-1.799$  & $-0.301$ & $-0.043$ & $0.047$ \\

\midrule
$t_{2.25}$ &5000 & $-2.609$  &$-0.039$  & $0.067$ & $0.113$ & $0.166$ \\

&10000 & $-3.158$  &$-1.047$  & $-0.078$ & $0.039$ & $0.091$ \\

&15000  & $-3.437$ &$-1.725$  & $-0.308$ & $-0.038$ & $0.047$ \\

\bottomrule
\end{tabular}}
\label{tab.high.est_p10000}
\end{table}

{\small{
 \setlength{\bibsep}{5pt}
\bibliographystyle{my_temp} 
\bibliography{refs_final}
}
}

\newpage
\appendix  
\spacingset{1.7}\selectfont
\setlength{\abovedisplayskip}{0.2\baselineskip}
\setlength{\belowdisplayskip}{0.2\baselineskip}
\setlength{\abovedisplayshortskip}{0.2\baselineskip}
\setlength{\belowdisplayshortskip}{0.2\baselineskip} 
\begin{center}
	{\noindent \bf \Large Supplementary Material for 
                              ``Adapting to Noise Tails in Private Linear Regression"}\\
\end{center} 
\bigskip

\renewcommand{\thepage}{S\arabic{page}}  
\setcounter{page}{1}
\setcounter{section}{0}
\renewcommand\thesection{\Alph{section}}  
\setcounter{lemma}{0}
\renewcommand{\theHlemma}{\thesection.\arabic{lemma}} 
\setcounter{theorem}{0}  
\renewcommand{\theHtheorem}{\thesection.\arabic{theorem}} 
\setcounter{equation}{0}  
\renewcommand{\theHequation}{S.\arabic{equation}}
\renewcommand{\theequation}{S.\arabic{equation}}
\setcounter{definition}{0}
\renewcommand{\theHdefinition}{\thesection.\arabic{definition}}
\setcounter{proposition}{0}
\renewcommand{\theHproposition}{\thesection.\arabic{proposition}}
\setcounter{corollary}{0}
\renewcommand{\theHcorollary}{\thesection.\arabic{corollary}}
\setcounter{remark}{0}
\renewcommand{\theHremark}{\thesection.\arabic{remark}}
\setcounter{table}{0} 
\setcounter{figure}{0}  
\renewcommand{\theHfigure}{S\arabic{figure}}
\renewcommand{\theHtable}{S\arabic{table}}
\renewcommand{\thefigure}{S\arabic{figure}}
\renewcommand{\thetable}{S\arabic{table}}

 \renewcommand{\thetheorem}{\thesection.\arabic{theorem}}
\renewcommand{\thelemma}{\thesection.\arabic{lemma}}
\renewcommand{\thedefinition}{\thesection.\arabic{definition}}
\renewcommand{\theproposition}{\thesection.\arabic{proposition}}
\renewcommand{\thecorollary}{\thesection.\arabic{corollary}}
\renewcommand{\theremark}{\thesection.\arabic{remark}}

In this supplementary material, we denote $\hat \bbeta=\hat \bbeta_\tau$ as the non-private Huber estimator for simplicity and introduce the following additional notation. For $\bu,\bv \in \RR^k$ and a subset $S\subseteq [k]$, we use $\langle \bu, \bv \rangle_S = \bu_S^{\T}\bv_S = \sum_{i\in S} u_i v_i$ to denote the inner product between $\bu_S$ and $\bv_S$. For a    matrix $\bA\in \RR^{k \times k}$, let $\|\bA\|$ denote the spectral norm of $\bA$. Moreover, let $\|\cdot\|_{\bA}$ denote the vector norm induced by the positive semi-definite matrix $\bA$, that is, $\|\bu\|_{\bA} = \|\bA^{1/2}\bu\|_2$ for all $\bu \in \RR^k$. For any $r>0$, we define $\partial \mathbb{B}^k(r) =  \{ \mathbf{u} \in \mathbb{R}^k : \|\mathbf{u}\|_2 = r \}$  and $\Theta(r) = \{ \bbeta \in \RR^p : \| \bbeta - \bbeta^* \|_2 \leq r \}$.

To highlight both the commonalities and differences between the private methods proposed in this paper and the non-private adaptive methods in  \citeS{sun2020adaptiveS}, Table~\ref{tab.tau} summarizes the choice of $\tau$ in the low-dimensional and high-dimensional sparse settings under heavy-tailed noise, along with the corresponding convergence rates. For clarity, note that the results in \citeS{sun2020adaptiveS} are derived under the assumption that the noise variables have bounded $(1+\theta)$-th moments ($\theta>0$), whereas our analysis assumes bounded $(2+\iota)$-th moments ($\iota\geq 0$). Accordingly, we evaluate both methods assuming the existence of moments of order at least two  (i.e., setting $\theta=1+\iota$). 
 
\begin{table}[h]
\centering
\scriptsize
\caption{Comparison of the private and non-private adaptive Huber estimators in terms of the choice of $\tau$ and the high-probability error bound $\|\mathring{\bbeta}-\bbeta^*\|_2$ (up to factors logarithmic in $n$ and $1/\delta$), when the noise variables have polynomial-order moments. Here, $\mathring{\bbeta}$ denotes the estimator of $\bbeta^*$.}
\setlength{\tabcolsep}{4pt}
\begin{tabular}{  c c c c}
\toprule
  Setting & Method \& noise assumption 
& $\tau$ 
&   $\|\mathring{\bbeta}-\bbeta^*\|_2$ \\
\midrule
  \multirow[t]{2}{*}{Low-dim}
& Our work ($\iota \geq 0$) 
& $\displaystyle\Big(\frac{n\epsilon}{p+\log n}\Big)^{1/(2+\iota)}$
& $\displaystyle\Big(\frac{p}{n\epsilon}\Big)^{ (1+\iota)/ ({2+\iota})} 
   + \sqrt{\frac{p}{n}}$ \\[2mm]
  
& \citeS{sun2020adaptiveS}  
& $\displaystyle \sqrt{\frac{n}{p+\log n} }$
& $\displaystyle \sqrt{ \frac{p}{n}  } $ \\[2mm]
\midrule
  \multirow[t]{2}{*}{High-dim}
& Our work ($\iota \geq 0$)
& $\displaystyle\Big(\frac{n\epsilon}{s\log p+\log n}\Big)^{1/(2+\iota)}$
& $\displaystyle\Big(\frac{s\log p}{n\epsilon}\Big)^{ (1+\iota) /({2+\iota})} 
   + \sqrt{\frac{s\log p}{n}}$ \\[2mm]
 
& \citeS{sun2020adaptiveS}  
& $\displaystyle \sqrt{\frac{n}{\log p } }$
& $\displaystyle \sqrt{\frac{s\log p}{n}} $ \\[2mm]  
\bottomrule
\end{tabular}
\label{tab.tau}
\end{table}

\section{Real Data Analysis}
\label{sec:6}
To illustrate the performance of our methods, we consider two real datasets: (i) the California housing price dataset \citepS{pace1997sparseS}\footnote{\href{http://lib.stat.cmu.edu/datasets/houses.zip}{http://lib.stat.cmu.edu/datasets/houses.zip}}, which contains data from the 1990 California census on 20640 block groups,  and (ii) the Communities and Crime dataset \citepS{dataUCIS}\footnote{\href{https://archive.ics.uci.edu/static/public/183/communities+and+crime.zip}{https://archive.ics.uci.edu/static/public/183/communities+and+crime.zip}}, which combines socio-economic data from the 1990 US census, law enforcement data from the 1990 US LEMAS survey, and crime data from the 1995 FBI UCR. For each dataset, we consider the linear model \eqref{linear.model} and evaluate the performance of our proposed method via evaluating the predictive accuracy. More specifically, given an positive integer $m$, we first randomly select $\lceil0.8m\rceil$ subjects and $\lceil0.2m\rceil$ subjects from the dataset, respectively, as the training set $\mathcal{D}_{\rm train}$ and the test set $\mathcal{D}_{\rm test}$. We then estimate the unknown parameter $\bbeta^*$ in \eqref{linear.model} by various methods based on $\mathcal{D}_{\rm train}$ and compute the associated mean squared prediction error (MSPE) of the response variables on $\mathcal{D}_{\rm test}$, i.e., ${\rm MSPE} =|\mathcal{D}_{\rm test}|^{-1} \sum_{i\in\mathcal{D}_{\rm test}} (y_i - \smallhat{y}_i)^2$,  where $\smallhat{y}_i$ represents the predicted value of the response variable $y_i$ based on the considered estimate of $\bbeta^*$. For each $m$, we set $\delta_m = 10(\lceil 0.8m\rceil)^{-1.1}$.

In the California housing price dataset, the response variable is the logarithm of the median house value per block group, and the covariates include  median income, housing median age, total population, number of households, and total rooms. We vary $m\in \{5000,10000,15000,20000\}$. The privacy parameters are set to $\epsilon\in\{0.3,0.5,0.9\}$. For each $\mathcal{D}_{\rm train}$, we compute the  $(\epsilon,\delta_m)$-DP Huber estimator ${\bbeta}^{(T)}$, obtained using Algorithm \ref{alg:DPHuberlow} with $\sigma=\sigma_{\rm dp}$,  and the non-private Huber estimator $\check\bbeta$, where the   parameters involved in Algorithm \ref{alg:DPHuberlow} are determined by the principle given in   Section \ref{selection.low}. The design matrix is standardized before the analysis is performed. This process is repeated 300 times, and Figure~\ref{fig:house_RD} reports the resulting boxplots of MSPE as $m$ increases. It can be seen that, while preserving privacy, the $(\epsilon,\delta_m)$-DP estimates gradually approach the non-private benchmark, demonstrating the effectiveness of Algorithm \ref{alg:DPHuberlow}. Moreover, a lower privacy level leads to more accurate estimates.

\begin{figure}[h]  
    \centering  
\includegraphics[width=0.7\textwidth]{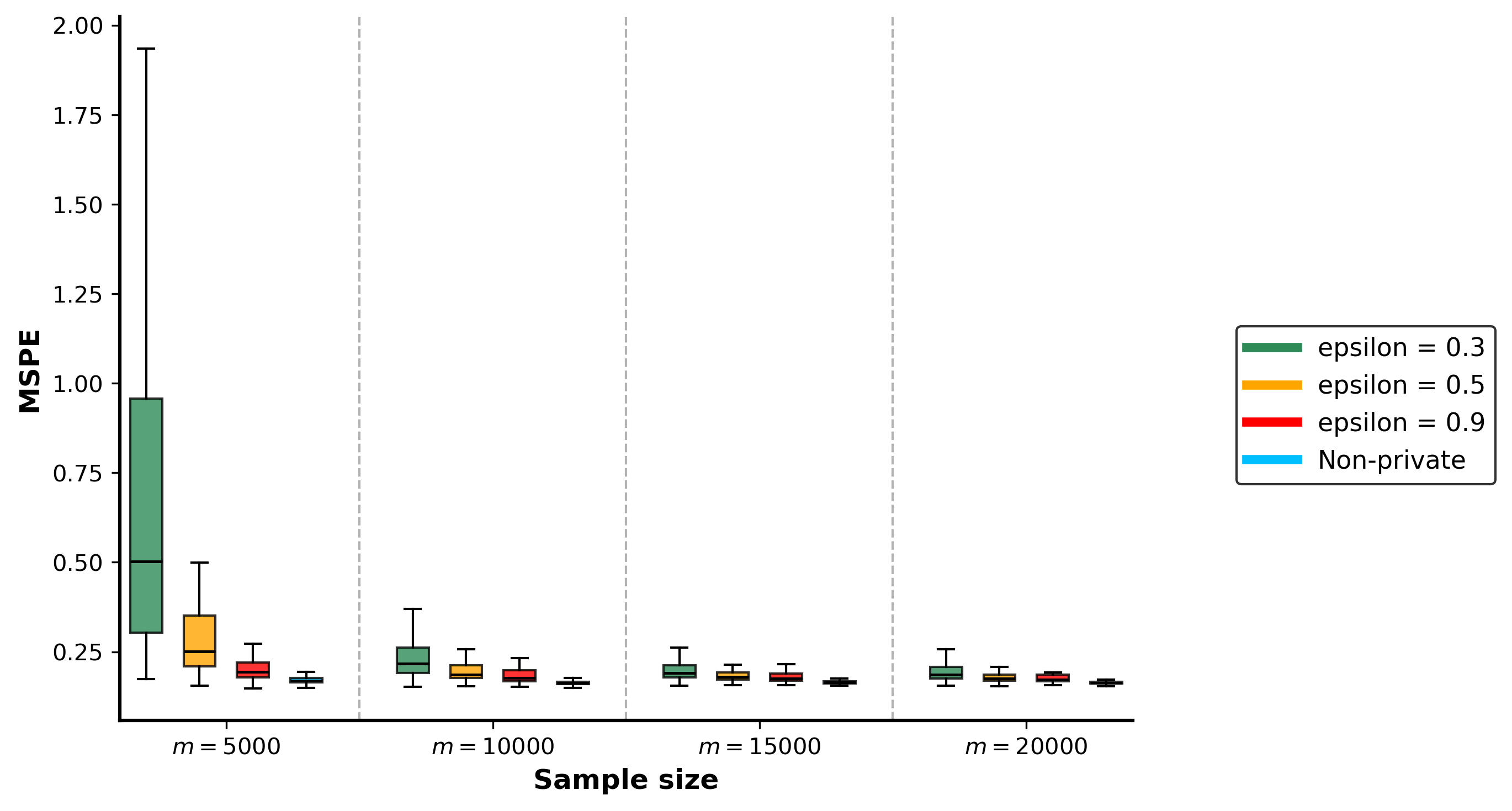}
    \caption{Boxplots of the MSPE  over 300 replications, comparing the $(\epsilon,\delta_m)$-DP Huber estimator at different privacy levels with its non-private counterpart.
    }   
    \label{fig:house_RD}  
\end{figure}

  Moreover, we report the estimated coefficients obtained from three methods applied to the California housing price dataset: the private $(0.5,10n^{-1.1})$-DP Huber estimator, the corresponding non-private Huber estimator, and the OLS  estimator. These comparisons are conducted using two forms of the response variable: the logarithm of the median house value  ($\log$(Med. Value)) and the median house value (in \$25,000). 
 From Table~\ref{tab.coef_california}, we observe that the log transformation mitigates the impact of heavy tails, resulting in similar performance between the non-private Huber and OLS estimators. In contrast, without this transformation, the response variable remains right-skewed, leading to notably different estimates between Huber and OLS. In both scenarios, the private Huber estimator closely aligns with the non-private Huber estimator. 

\begin{table}[ht]
\centering
\setlength{\tabcolsep}{2pt} 
\caption{Estimated regression coefficients for the California Housing dataset, comparing the $(0.5,10n^{-1.1})$-DP method, its non-private counterpart, and OLS.}
\scriptsize
\begin{tabular}{cccccccc}
\toprule
&  & Intercept & Med. income & Med. age & Housing med. age & Total rooms & Total population  \\
\midrule
$\log$(Med. Value)& Huber & 12.085 & 0.387 & 0.107 &$-0.091$  &0.163  &$-0.012$ \\
&DP Huber  & 12.065 & 0.401 & 0.091& $-0.187$&  0.154  &0.073\\
& OLS &12.085 & 0.419  &0.098& $-0.187$ & 0.422& $-0.187$\\ 
\midrule
Med. Value  & Huber & 8.274 & 3.283&  0.990& $-1.078$ & 1.494 & $-0.068$ \\
(in \$25,000) &DP Huber  & 7.785 & 3.227 & 0.791 &$-0.598$  &0.625 & 0.389\\
& OLS &8.274 & 3.493 & 0.927 &$-1.787$&  3.278& $-1.217$\\ 
\bottomrule
\end{tabular}
\label{tab.coef_california}
\end{table}

In the Communities and Crime dataset, the response variable is the normalized total number of violent crimes per 100000 population (\texttt{ViolentCrimesPerPop}), with 99 features characterizing 1994 communities.  
To evaluate the prediction performance of the proposed method, we vary $m$ from 500 to 1900.  With the privacy levels $(\epsilon,\delta)$ set to  $(0.9, \delta_m)$  and tuning parameters selected according to the principle in Section \ref{selection.high} based on $\mathcal{D}_{\rm tarin}$, we obtain the non-private sparse Huber estimator $\breve{\boldsymbol{\beta}}^{(T)}$, the sparse DP Huber estimator ${\boldsymbol{\beta}}^{(T)}$ obtained in Algorithm \ref{alg:DPHuberhigh}, and the sparse DP LS estimator ${\boldsymbol{\beta}}^{(T)}_{{\rm LS}}$ derived from Algorithm 4.2 in \citeS{cai2021costS}. Similar to the procedure in  Section~\ref{selection.high}, we obtain an initial estimator for the DP algorithm in each estimation. The process is repeated 300 times, yielding Figure \ref{fig:comm_RD}, which shows that while maintaining privacy, the sparse DP  Huber method approaches its non-private counterpart as $m$ increases and outperforms the  sparse DP LS method \citepS{cai2021costS}.

\begin{figure}[h]  
    \centering      \includegraphics[width=0.5\textwidth]{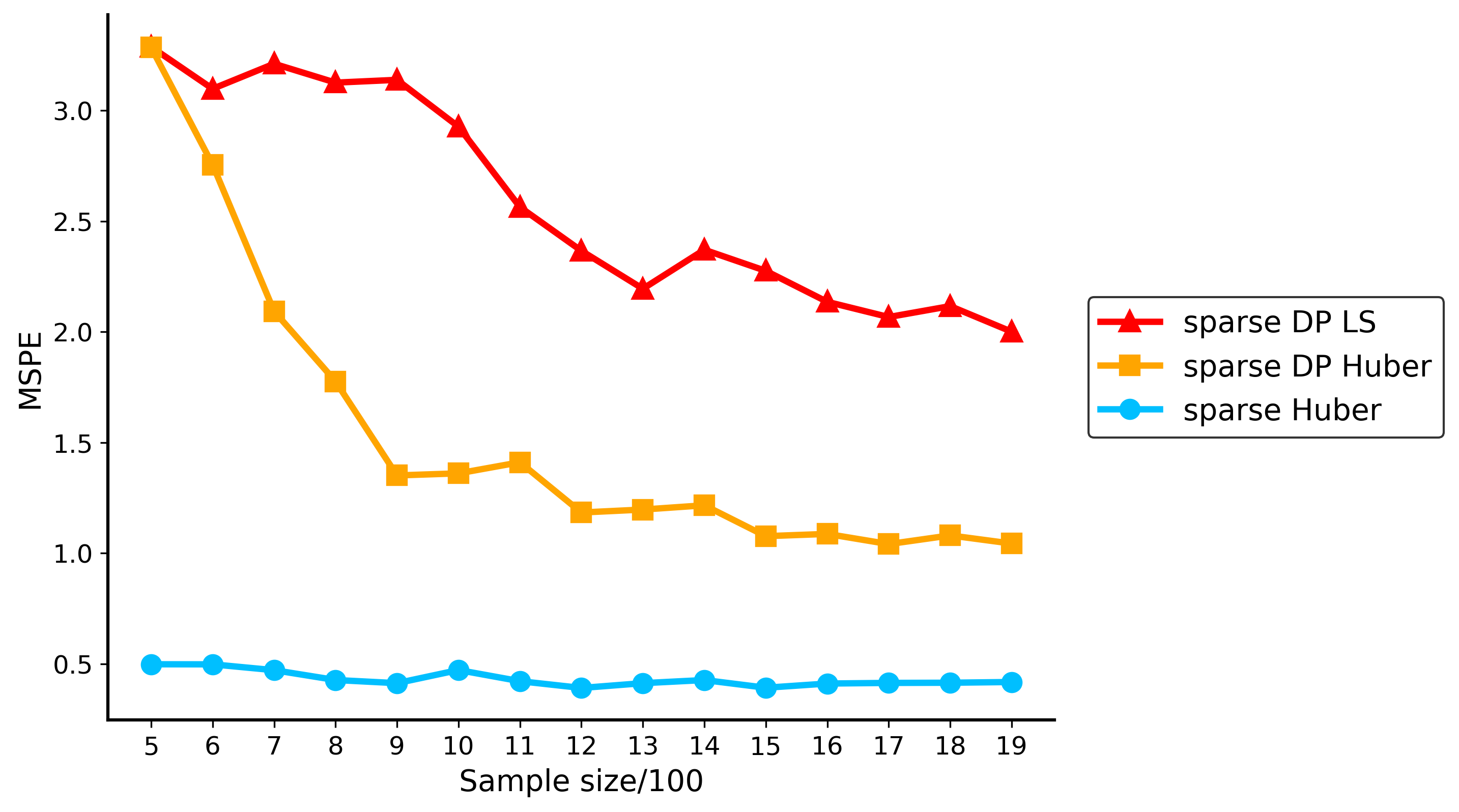}
    \caption{Plots of the average MSPE, based on 300 replications, comparing the sparse DP Huber estimator, the sparse DP LS estimator, and the non-private sparse Huber estimator under privacy levels  $(0.9, \delta_m)$ and $s=5$, as $m$ increases from 500 to 1900.
    }   
    \label{fig:comm_RD}  
\end{figure}

\section{Review on Adaptive Huber Regression}
\label{review-Huber-reg}

Recall the settings in Section~\ref{sec:adaptive-huber}. Under a sub-Gaussian condition on the covariates, \citeS{sun2020adaptiveS} showed that for any $z\geq 0$, the Huber estimator with $\tau = \tau_0 \sqrt{n/(p+z)}$, for some $\tau_0 \geq \sigma_0$, satisfies
$\|  \hat \bbeta_\tau - \bbeta^* \|_2  \lesssim \tau_0 \sqrt{(p+z)/n}$ with probability at least $1-2 e^{-z}$, provided that $n\gtrsim p+z$. Under sparse models in high-dimensional settings, where $s^*:=\| \bbeta^* \|_0 \ll \min\{n,p\}$, we consider the sparsity-constrained Huber estimator 
\begin{align}
     \hat{\bbeta}_\tau(s)   \in \arg\min_{\bbeta\in \RR^p :\, \| \bbeta \|_0 \leq s }  \hat \cL_\tau(\bbeta)\,   ,  \label{def:huber.lasso}
\end{align} 
where $s\geq 1$ is a prespecified sparsity level that should be greater than the true sparsity $s^*$. When the random covariate vectors are sub-Gaussian or sub-exponential, the finite sample properties of $\hat \bbeta_\tau$ given in \eqref{def:m.est} have been well-studied \citepS{chen2020robustS,sun2020adaptiveS}. For the sparsity-constrained Huber estimator $\hat \bbeta_\tau(s)$ given in \eqref{def:huber.lasso}, due to the non-convex nature of the sparsity-constrained optimization problem, a common practice is to find an approximate solution that possesses desirable statistical properties. To this end, a popular family of algorithms based on iterative hard thresholding (IHT) has been proposed for compressed sensing and sparse regression \citepS{blumensath2009iterativeS, jain2014iterativeS}. These methods directly project the gradient descent update onto the underlying non-convex feasible set. The analysis of these algorithms requires the loss function to satisfy certain geometric properties, such as restricted strong convexity and restricted strong smoothness \citepS{jain2014iterativeS}. These conditions, however, do not hold for the empirical Huber loss. Therefore, the statistical properties of the output of the IHT algorithm when applied to \eqref{def:huber.lasso} remain unclear.

A computationally more efficient approach is arguably the $\ell_1$-penalized estimator, defined as $\arg\min_{\bbeta \in \RR^p} \{  \hat \cL_\tau (\bbeta) + \lambda \| \bbeta \|_1 \}$, where $\lambda>0$ is a regularization parameter. Its finite-sample properties, under sub-Gaussian covariates and heavy-tailed noise variables, with or without symmetry, have been studied by \citeS{fanliwang2017S}, \citeS{loh2017statisticalS}  and \citeS{sun2020adaptiveS}. In the asymmetric case, the robustification parameter $\tau$ should be of the order $\sigma_0 \sqrt{n/\log p}$ to balance bias and robustness. The algorithm that solves $\ell_1$-penalized estimators often involves soft-thresholding. However, when combined with noise injection to preserve privacy, the privacy cost of the resulting estimator scales with $\sqrt{p}$, as studied in \citeS{liu2024DPS}. 
Therefore, in this work, we focus on privatizing the iterative hard thresholding algorithm.

\section{Proof of the Results in Section~\ref{sec:4.1}}

\subsection{  Intermediate Results for Theorem \ref{thm:low} }\label{sec.inter1}
 To begin with, we define the good events required to establish the convergence of $\bbeta^{(T)}$ obtained from Algorithm \ref{alg:DPHuberlow}.   Assume that the initial estimate $\bbeta^{(0)}$ satisfies $\bbeta^{(0)} \in {\hat \bbeta}  +  \BB^p(r_0)$ for some $r_0>0$. Given a truncation parameter $\gamma>0$, a curvature parameter $\phi_l>0$ and a smoothness parameter $\phi_u>0$, define the good events 
\begin{equation}
\begin{aligned}
&~~~~~~~~~~~~~~~~~~\cE_0   
    =  \bigg\{   \hat \bbeta \in \Theta\bigg(\frac{r_0}{2}\bigg) \bigg\} \cap  \bigg\{  \max_{i\in[n]} \| \bx_i \|_2 \leq \gamma \bigg\}\,  ,   \\ 
    &~~~\cE_1  
     = \bigg\{   \inf_{  (\bbeta_1 , \,\bbeta_2 ) \in \mathbb{N}(r_0) }  \frac{\hat \cL_\tau(\bbeta_1)  - \hat \cL_\tau (\bbeta_2) - \langle \nabla \hat \cL_\tau(\bbeta_2), \bbeta_1-\bbeta_2 \rangle}{ \| \bbeta_1 -\bbeta_2 \|_2^2}  \geq \phi_l \bigg\}\,,   \\
    &~~~~~\cE_2  
      = \bigg\{  \sup_{\bbeta_1,\, \bbeta_2 \in \RR^p} \frac{\hat \cL_\tau(\bbeta_1)  - \hat \cL_\tau (\bbeta_2) - \langle \nabla \hat \cL_\tau(\bbeta_2), \bbeta_1-\bbeta_2 \rangle}{\| \bbeta_1 -\bbeta_2 \|_2^2} \leq \phi_u   \bigg\}\, ,
\end{aligned}
\label{def:event.E0-2}
\end{equation}
where $\mathbb{N}(r_0) = \{ (\bbeta_1, \bbeta_2) :  \| \bbeta_2 - \bbeta_1 \|_2 \leq r_0,  \, {\rm and~either }~ \bbeta_1 \in \Theta(r_0/2) ~{\rm or}~ \bbeta_2 \in \Theta(r_0/2) \}$. On the event $\cE_0$, the non-private Huber estimator lies within a neighborhood of $\bbeta^*$. Event $\cE_1$ is associated with the restricted strong convexity of the empirical loss $\hat{\cL}_\tau(\cdot)$.  In the case of the squared loss with $\tau=\infty$, it holds that
$$
\hat \cL_{\infty}(\bbeta_1)  - \hat \cL_{\infty}(\bbeta_2) - \langle \nabla \hat \cL_{\infty}(\bbeta_2), \bbeta_1-\bbeta_2 \rangle  =  (\bbeta_1 - \bbeta_2) \bigg(\frac{1}{2n} \sn \bx_i \bx_i^\T \bigg)  (\bbeta_1 - \bbeta_2)\, .
$$
By taking $\phi_l=\lambda_{\min}\{ (2n)^{-1} \sn \bx_i \bx_i^\T \}$, it follows that
$$
\frac{\hat \cL_{\infty}(\bbeta_1)  - \hat \cL_{\infty}(\bbeta_2) - \langle \nabla \hat \cL_{\infty}(\bbeta_2), \bbeta_1-\bbeta_2 \rangle}{ \| \bbeta_1 - \bbeta_2 \|_2^2 } \geq \phi_l ~\mbox{for all}~ \bbeta_1 , \bbeta_2 \in \RR^p\,.
$$
For the Huber loss with $\tau>0$, strong convexity does not hold universally for all $(\bbeta_1, \bbeta_2)$. Therefore, it is necessary to restrict $(\bbeta_1, \bbeta_2)$  to the local region $\mathbb{N}(r_0)$. Event $\cE_2$ pertains to the global smoothness property, which holds for all $(\bbeta_1, \bbeta_2)$ as the Huber loss has a bounded second-order derivative that is well-defined everywhere except at two points.

Conditioned on the above good events, Theorem \ref{thm:dp.convergence} below establishes a high-probability bound for $\| \bbeta^{(T)} - \hat \bbeta \|_2$, where the probability is taken with respect to the Gaussian noise variables $\{\bg_t \}_{t=0}^{T-1}$ injected into the gradient descent.  The proof of Theorem \ref{thm:dp.convergence} is given in Section \ref{proof.thm:dp.convergence}.

\begin{theorem} \label{thm:dp.convergence}
Assume that $\bbeta^{(0)} \in \hat \bbeta + \BB^p(r_0)$ for some $r_0>0$, and the learning rate satisfies $\eta_0 \in (0, 1/(2\phi_u)] $. Let $T\geq 2\log\{r_0 n\epsilon(\sigma_0 p)^{-1}\}/\log\{(1-\rho)^{-1}\}$ with $\rho = (2  \phi_l\eta_0)^2 \in (0, 1)$. For any given $z> 0$, select the noise scale $\sigma>0$ such that
\begin{align}
   \frac{\sigma }{  \phi_l r_0   } \big\{\sqrt{p} + \sqrt{ 2(\log T + z) }\big\}     \leq 1 \, . \label{private.sample.size.requirement}
\end{align}
Then, conditioned on the event $\cE_0\cap\cE_1 \cap \cE_2$,  the DP Huber estimator $\bbeta^{(T)}$ obtained in Algorithm~\ref{alg:DPHuberlow} satisfies 
\begin{align}
    \|\bbeta^{(T)} - \hat \bbeta \|_2^2 \leq    \bigg(\frac{\sigma_0 p}{n\epsilon} \bigg)^2 +   \eta_0  \bigg(  \eta_0 +  \frac{1}{2\phi_l} \bigg)\bigg(\frac{2p}{\rho} + 3z \bigg) \sigma^2\,,  
\end{align}
with probability (over $\{ \bg_t\}_{t=0}^{T-1}$) at least $1-2e^{-z}$.
\end{theorem}

Theorem~\ref{thm:dp.convergence} requires that the initial estimate satisfies $\bbeta^{(0)}\in \hat \bbeta + \BB^p(r_0)$ for some prespecified $r_0>0$. This requirement is not restrictive, as $r_0$ can be chosen to be of constant order. Indeed, Theorem \ref{prop:initial_value} below demonstrates that for any initial value $\bbeta^{(0)}$, this condition is always satisfied after a few iterations, as long as the sample size $n$ is sufficiently large. The proof of Theorem \ref{prop:initial_value} is given in Section \ref{proof.prop:initial_value}.

\begin{theorem}\label{prop:initial_value}
For any initial value $\bbeta^{(0)} \in \RR^p$, let $R:=\|\bbeta^{(0)}-\hat \bbeta \|_2 \geq r_0$ and $\eta_0 \in (0, 1/(2\phi_u)] $ such that $T_0 :=  R^2(\eta_0\phi_l r_0^2)^{-1}\in \mathbb{N}_{+}$.   For any given $z>0$, select the noise scale $\sigma$ such that 
\begin{align}
     \frac{R\sigma}{  \phi_l r_0^2}\big\{\sqrt{p} + \sqrt{ 2(\log T_0 + z) }\big\}  \leq \frac{\sqrt{7}-1}{6} \, . \label{private.sample.size.requirement_initial}
\end{align}
Then, conditioned on the event $\cE_0\cap\cE_1 \cap \cE_2$, the  DP Huber estimator $\bbeta^{(T_0)}$ obtained in Algorithm~\ref{alg:DPHuberlow}  satisfies $\|\bbeta^{(T_0)}- \hat \bbeta \|_2\leq r_0$ with probability at least $1-e^{-z}$ (over $\{ \bg_t\}_{t=0}^{T_0-1}$).
\end{theorem}

The results in Theorems~\ref{thm:dp.convergence} and \ref{prop:initial_value} are deterministic in nature, providing upper bounds on the statistical errors of the DP Huber estimator obtained via noisy gradient descent, conditioned on certain favorable (random) events defined in \eqref{def:event.E0-2}.  In what follows, we demonstrate that these events occur with high probability under Assumptions~\ref{assump:design} and \ref{assump:heavy-tail}.

\begin{proposition}\label{prop:events1}
Let Assumptions \ref{assump:design} and \ref{assump:heavy-tail} hold. Then, for any given $\epsilon, z>0$ and $\tau_0\geq \sigma_0$, by selecting the robustification parameter $\tau = \tau_0\{n\epsilon(p+z)^{-1}\}^{1/(2+\iota)}$,  the event $\cE=\cE_0\cap \cE_1\cap \cE_2$ occurs and 
$$
\|  \hat \bbeta- \bbeta^*\|_2    \leq C_0   \bigg\{  \max\{ \sigma_\iota^{2+\iota} \tau_0^{-1-\iota} , \tau_0 \}   \bigg( \frac{p+z}{n \epsilon} \bigg)^{(1+\iota)/({2+\iota})} + \sigma_0 \sqrt{\frac{p+z} {n }} \bigg\}\,,
$$
with probability at least $1-5e^{-z}$, provided that $n\epsilon \geq C_1(p+z)$, and the parameters $(r_0,\gamma,\phi_l,\phi_u)$ in \eqref{def:event.E0-2} are chosen as 
\begin{align*}
    r_0 = \frac{\tau}{16\sqrt{\kappa_4\lambda_1}}\,,~~~\gamma = C_2\sqrt{p+\log n+z}\,,~~\phi_l =\frac{\lambda_p}{8}\,,~~\mbox{and}~~\phi_u=\lambda_1\,,
\end{align*}
where $C_0$ and $C_2$ are positive constants depending only on $(\upsilon_1,\lambda_1,\lambda_p,\kappa_4)$, and $C_1$ is a positive constant depending only on $(\upsilon_1,\lambda_1,\lambda_p,\kappa_4,\tau_0,\sigma_\iota)$. 
\end{proposition}

The proof of Proposition~\ref{prop:events1} is provided in Section~\ref{proof.e1}. With the intermediate results established above, we are now prepared to prove the main theorem in Section~\ref{sec:4.1}.

\subsection{Proof of Theorem \ref{thm:low}}

Notice that by choosing $z\asymp \log n$ in Proposition~\ref{prop:events1}, the event $\cE=\cE_0\cap \cE_1\cap \cE_2$ occurs with high probability, satisfying $\PP(\cE^c)=O(n^{-1})$. Consequently, Theorem~\ref{thm:low} follows directly from Theorem~\ref{thm:dp.convergence} and Proposition~\ref{prop:events1}.

\subsection{Proofs of the Intermediate Results}
\setcounter{lemma}{0}
\renewcommand{\thelemma}{\thesection.\arabic{lemma}}
\subsubsection{Technical Lemmas}

We begin by presenting several technical lemmas concerning the gradient and Hessian of the empirical loss, defined as
\begin{align}
\nabla \hat \cL_\tau(\bbeta)   = - \frac{1}{n} \sn   \psi_\tau ( y_i - \bx_i^\T \bbeta  )   \bx_i   ~~\mbox {and }~~ \nabla^2 \hat \cL_\tau(\bbeta)   = \frac{1}{n} \sn  \psi_\tau'( y_i - \bx_i^\T \bbeta  )  \bx_i \bx_i^\T\,. \label{sample.hessian}
\end{align}
Since the first-order derivative of the Huber loss is 1-Lipschitz continuous,  the empirical loss $\hat \cL_\tau(\cdot)$ is globally smooth (with high probability). 
To see this, recall from Assumption~\ref{assump:design} that the eigenvalues of $\bSigma=\EE(\bx_i \bx_i^\T)$ lie in $[\lambda_p, \lambda_1]$. Hence,
\begin{align}
\hat \cL_\tau(\bbeta_1) - \hat \cL_\tau(\bbeta_2)  \leq \langle \nabla \hat \cL_\tau (\bbeta_2), \bbeta_1 - \bbeta_2 \rangle + \frac{\lambda_1}{2}  \lambda_{\max}(  \hat \bS ) \cdot  \| \bbeta_1 - \bbeta_2 \|_2^2 \,, \label{global.smootheness}
\end{align}
where $\hat \bS : = n^{-1} \sn \bchi_i \bchi_i^\T$ with $\bchi_i = \bSigma^{-1/2} \bx_i$.   Furthermore, Huber's loss is locally strongly convex, as $\rho''_\tau(u) = 1$ for $u\in (-\tau , \tau)$.

\begin{lemma}\label{lem:gs}
Under Assumption \ref{assump:design}, for any $z\geq 0$, with probability at least $1-e^{-z}$, 
\begin{align*}
\hat \cL_\tau(\bbeta_1) - \hat \cL_\tau(\bbeta_2)  - \langle \nabla \hat \cL_\tau (\bbeta_2), \bbeta_1 - \bbeta_2 \rangle \leq   \lambda_1  \| \bbeta_1 - \bbeta_2 \|_2^2 \, , 
\end{align*}
holds uniformly over $\bbeta_1,\bbeta_2 \in \mathbb{R}^p$, provided that  $n\gtrsim  \upsilon_1^4 (p+z)$.
\end{lemma}

The proof of Lemma \ref{lem:gs} is given in Section \ref{proof.lem:gs}. The following lemma shows that, for a prespecified radius $r>0$, the empirical Huber loss is locally strongly convex (with high probability), provided that the robustification parameter $\tau$ is sufficiently large. Recall 
\begin{align}
    \kappa _4 := \sup_{\bu\in \mathbb{S}^{p-1}}   \EE \big(\langle \bSigma^{-1/2} \bx_i, \bu \rangle^4 \big)\,. \label{def:kappa4}
\end{align}

\begin{lemma} \label{lem:rsc}
Let Assumptions  \ref{assump:design} and \ref{assump:heavy-tail} hold. For any $r>0$, let the robustification parameter $\tau$ satisfy $\tau \geq 16 \max\{ \sigma_0 ,  (\kappa_4 \lambda_1)^{1/2} r \}$, where $\kappa _4$ is defined in \eqref{def:kappa4}. Then, it holds for any $z\geq 0$ that, with probability at least $1-e^{-z}$,
\begin{align*}
     \hat \cL_\tau(\bbeta_1)  - \hat \cL_\tau (\bbeta_2) - \langle \nabla \hat \cL_\tau(\bbeta_2), \bbeta_1-\bbeta_2 \rangle  \geq \frac{\lambda_p}{8} \| \bbeta_1 -\bbeta_2 \|_2^2\,,
\end{align*} 
holds uniformly for $\bbeta_1 \in \Theta(r/2)$ and $\bbeta_2 \in \bbeta_1 + \BB^p(r)$, provided that $n\gtrsim \kappa_4(\lambda_1/\lambda_p)^2 (  p + z)$.  The same bound also applies when $\bbeta_1 \in \bbeta_2 +  \BB^p(r)$ and $\bbeta_2 \in \Theta(r/2)$. 
\end{lemma}

The proof of Lemma \ref{lem:rsc} is given in Section \ref{proof.lem:rsc}.
Next, we focus on the  local growth property of the empirical Huber loss, particularly outside a neighborhood of $\hat \bbeta$.

\begin{lemma}  \label{lem:rsc.smootheness}
Conditioned on $\cE_0 \cap \cE_1$ with $\cE_0$ and $\cE_1$ defined in \eqref{def:event.E0-2}, we have
\begin{align}
  &  \hat \cL_\tau(\bbeta ) - \hat \cL_\tau(\hat \bbeta )    \geq   \phi_l  r_0    \| \bbeta - \hat \bbeta \|_2 \,,     \label{restricted.loss.difference} 
\end{align}
for all $\bbeta  \notin    \hat \bbeta + \BB^p(r_0)$. 
\end{lemma}

The proof of Lemma \ref{lem:rsc.smootheness} is given in Section \ref{proof.lem:rsc.smootheness}.
The following lemma provides upper bounds for the i.i.d. standard normal random vectors $\{ \bg_t\}_{t=0}^{T-1}$ in Algorithm \ref{alg:DPHuberlow}. In particular, inequality \eqref{weighted.chi-square.concentration} is a slightly improved version of the tail bound in Lemma~11 of \citeS{cai2020costS}.

\begin{lemma} \label{lem:max.normal.l2}
Let $\bg_0, \bg_1,\ldots, \bg_{T-1} \in \RR^p$ ($T \geq 1$) be independent standard multivariate normal random vectors. Then, for any $z\geq 0$,
\begin{align}
	 \PP \bigg\{  \max_{t\in\{0\}\cup [T-1]} \| \bg_t \|_2 \geq   p^{1/2} + \sqrt{2(\log T + z ) }   \bigg\} \leq e^{-z} \,.  \label{max.chi-square.concentration}
\end{align}
Moreover, for any $\rho\in (0,1)$, with probability at least $1-e^{-z}$, it holds that
\begin{align}
	\sum_{t=0}^{T-1} \rho^t \| \bg_t \|_2^2 \leq \frac{p}{1-\rho} + 2 \sqrt{\frac{p z}{1-\rho^2}} + 2 z   \leq \frac{2p}{1-\rho} +\bigg(\frac{1}{1+\rho} + 2 \bigg)z \,. \label{weighted.chi-square.concentration}
\end{align}
\end{lemma}

The proof of Lemma \ref{lem:max.normal.l2} is given in Section \ref{proof.lem:max.normal.l2}.
The following lemma restates Theorem~7.3 in \citeS{bousquet2003concentrationS} with slight modification, which is a refined version of Talagrand’s inequality.
\begin{lemma} \label{lem:bousquet}
    Let $\mathcal{F}$ be a countable set of functions from some Polish space $\mathcal{X}$ to $\mathbb{R}$, and $\xi_1,\ldots,\xi_n$ be independent random variables taking their values in $\mathcal{X}$. Assume that all functions $f$ in $\mathcal{F}$ are measurable, square-integrable, and satisfy $\sup_{x\in \cX} |f(x)| \leq B$, $\mathbb{E}  \{f (\xi_i ) \}=0$ for all $i\in[n]$, and $\sn \sup _{f \in \mathcal{F}} \EE \{f^2 (\xi_i )\} \leq K $. Define
    $$
    Q:=\sup _{f \in \mathcal{F}} \sn f (\xi_i ) \,. 
    $$
   For any $z\geq 0$, it holds that, with probability at least $1-e^{-z}$, 
    $$ 
    Q \leq \EE(Q)+\sqrt{ 2 K z + 4 B\EE(Q)z}+\frac{B z}{3} \leq \frac{5}{4} \EE(Q) + \sqrt{2 K z} +  \frac{13 }{3}Bz\, .
    $$
\end{lemma}

\subsubsection{Proof of Theorem~\ref{thm:dp.convergence}}\label{proof.thm:dp.convergence}

Define $\cE= \cE_0 \cap \cE_1 \cap \cE_2$, where $\cE_0$, $\cE_1$ and $\cE_2$ are the events given in \eqref{def:event.E0-2}. To control the random perturbations in noisy gradient descent, for some $\rho\in (0,1)$ to be determined, define 
\begin{align}
\cF =   \bigg\{   \max_{ t\in \{0\}\cup [T-1]} \| \bg_t \|_2 \leq  \zeta_T \bigg\} \cap  \bigg\{ \sum_{t=0}^{T-1} (1-\rho)^t  \| \bg_{T-1-t} \|_2^2 \leq 2\rho^{-1} p + 3z  \bigg\}  \,,  \label{def:G}
\end{align}
where $ \zeta_T  := p^{1/2} + \sqrt{2(\log T +z ) }$. It follows from Lemma~\ref{lem:max.normal.l2} that $\PP( \cF)  \geq 1-2 e^{-z}$.
In the following, we prove the result by conditioning on the event $\cE  \cap \cF$.  Starting with an initial value $\bbeta^{(0)} \in  \hat \bbeta + \BB^p(r_0)$, Proposition \ref{prop:crude.bound} shows that all successive iterates will remain within the local neighborhood $ \hat \bbeta + \BB^p(r_0)$. The proof of Proposition \ref{prop:crude.bound} is given in Section \ref{proof.prop:crude.bound}.

\begin{proposition}  \label{prop:crude.bound}
Under the conditions of Theorem~\ref{thm:dp.convergence}, and conditioning on $\cE \cap \cF$, all iterates $\bbeta^{(t)}$, $t\in[T]$, remain within the local region $\hat \bbeta  + \BB^p(r_0)$.
\end{proposition}

Next, we establish a contraction property for the noisy gradient descent iterates. Define
\begin{align*}
     \wt \bbeta^{(t+1)}  = \bbeta^{(t)} -  \eta_0  \nabla \hat \cL_\tau  (\bbeta^{(t)}) ~~\mbox{and}~~ \bh_t = \sigma \bg_t\, , \ \ t \in \{0\}\cup [T-1]\,,
\end{align*}
and note that $\bbeta^{(t+1)} = \wt \bbeta^{(t+1)} + \eta_0 \bh_t$ conditioned on the event $\cE_0$. Here, $\sigma>0$ is either $\sigma_{{\rm dp}}$  or $\sigma_{{\rm gdp}}$ from \eqref{dp.noise.scale}.
Under condition \eqref{private.sample.size.requirement}, Proposition~\ref{prop:crude.bound} ensures that $\bbeta^{(t)} \in \hat \bbeta  + \BB^p(r_0)$ for all $t\in \{0\}\cup[T]$. Similarly, it can be shown that the non-private gradient descent iterates $\wt \bbeta^{(t)}$ also remain within the ball $\hat \bbeta  + \BB^p(r_0)$ for $t\in[T]$.
For simplicity, set
\begin{align*}
	\bdelta^{(t)} = \bbeta^{(t )} - \hat  \bbeta ,  \quad \wt \bdelta^{(t )} = \wt \bbeta^{(t)} - \hat \bbeta ~~\mbox{ for }~ t \in \{0\}\cup [T]\,.
\end{align*}
For any $\eta \in (0, 1]$, at each iteration, we bound  $\| \bdelta^{(t+1)} \|_2 = \| \wt  \bdelta^{(t+1)} + \eta_0 \bh_t \|_2$  as follows:  
\begin{align}
  \| \bdelta^{(t+1)}   \|_2^2  
	& =  \| \wt \bdelta^{(t+1)}  \|_2^2 + \eta_0^2 \| \bh_t \|_2^2 + 2\eta_0 \langle \wt  \bdelta^{(t+1)}  ,  \bh_t \rangle \nn \\
	& \leq  ( 1 + \eta ) \| \wt \bdelta^{(t+1)}   \|_2^2 + (1+ \eta^{-1} ) \eta_0^2 \| \bh_t \|_2^2 \nn \\
	& = (1+  \eta  )  \| \bdelta^{(t)}   \|_2^2   + (1+ \eta^{-1}  ) \eta_0^2 \| \bh_t \|_2^2   \label{case2.1} \\
	&~~~~~+ 2\eta_0 (1+\eta)      \underbrace{ \bigg\{ \frac{\eta_0}{2}   \| \nabla \hat \cL_\tau(\bbeta^{(t)}  )\|_2^2 -  \langle \bbeta^{(t)} -  \hat \bbeta , \nabla \hat \cL_\tau (\bbeta^{(t)} ) \rangle\bigg\} }_{  \Pi_1  }  \, .  \nn
\end{align}
To bound the term $\Pi_1$, we use the local strong convexity and smoothness properties of $\hat \cL_\tau(\cdot)$. From \eqref{def:event.E0-2}, we observe that
\begin{align*}
 \hat \cL_\tau( \hat \bbeta ) - \hat \cL_\tau(\bbeta^{(t)})  -  \langle \nabla \hat \cL_\tau(\bbeta^{(t)}) , \hat \bbeta - \bbeta^{(t)} \rangle  & \geq \phi_l  \| \bbeta^{(t)} - \hat \bbeta \|_2^2\,,  \\
 \hat \cL_\tau(\wt \bbeta^{(t+1 )}) - \hat \cL_\tau(  \bbeta^{(t)} ) - \langle \nabla \hat \cL_\tau( \bbeta^{(t)} ) , \wt \bbeta^{(t+1)} - \bbeta^{(t)} \rangle & \leq    \phi_u  \| \wt \bbeta^{(t+1)} - \bbeta^{(t)} \|_2^2 \, .
\end{align*} 
Together, these upper and lower bounds imply
\begin{align*}
  0 & \leq   \hat \cL_\tau(\wt \bbeta^{(t+1 )})  -  \hat \cL_\tau( \hat \bbeta )  \\  
 & \leq \langle \nabla \hat \cL_\tau(\bbeta^{(t)})  , \wt \bbeta^{(t+1)} - \hat\bbeta   \rangle  +\phi_u  \| \wt \bbeta^{(t+1)} - \bbeta^{(t)} \|_2^2 -  \phi_l \| \bbeta^{(t)} -  \hat\bbeta \|_2^2    \\
 & = \langle  \nabla \hat \cL_\tau(\bbeta^{(t)})  , \wt \bbeta^{(t+1)} -  \hat\bbeta   \rangle  +     \eta_0^2 \phi_u  \|  \nabla \hat \cL_\tau(\bbeta^{(t)} ) \|_2^2 - \phi_l \| \bbeta^{(t)} - \hat \bbeta \|_2^2  \\
  & = \langle  \nabla  \hat \cL_\tau(\bbeta^{(t)})  ,  \bbeta^{(t)} -  \hat\bbeta   \rangle   - \eta_0 ( 1 -  \phi_u  \eta_0  )    \| \nabla \hat \cL_\tau(\bbeta^{(t)} ) \|_2^2    -\phi_l  \| \bbeta^{(t)} - \hat \bbeta \|_2^2 \\
  & \leq \langle  \nabla  \hat \cL_\tau(\bbeta^{(t)})  ,  \bbeta^{(t)} -  \hat\bbeta   \rangle   -  \frac{\eta_0}{2}  \| \nabla \hat \cL_\tau(\bbeta^{(t)} ) \|_2^2    -\phi_l  \| \bbeta^{(t)} - \hat \bbeta \|_2^2 \,,
\end{align*}
where we used the learning rate condition $\eta_0 \leq 1/(2\phi_u)$ in the final step.
Substituting this into \eqref{case2.1} gives
\begin{align*}
   \| \bdelta^{(t+1)}   \|_2^2 
& \leq   \big\{  1 + \eta -  2\phi_l  \eta_0    (1+\eta )  \big\} \| \bdelta^{(t)}   \|_2^2 +  (  1 +  \eta^{-1}  )\eta_0^2 \| \bh_t \|_2^2   \nn \\
& =   (1+\eta)   (  1 - 2   \phi_l   \eta_0   )  \| \bdelta^{(t)}  \|_2^2 +  (  1 +  \eta^{-1}  ) \eta_0^2\| \bh_t \|_2^2  \, .  \nn
\end{align*}
Taking $\rho = \eta^2$ and $\eta = 2   \phi_l \eta_0$, we conclude that
\begin{align*}
  \| \bdelta^{(t+1)}   \|_2^2 \leq  (1-\rho )   \| \bdelta^{(t)}  \|_2^2 +   (  1 + \eta^{-1} ) \eta_0^2\| \bh_t \|_2^2\, , \ \  t\in \{0\}\cup [T-1] \, . \nn
\end{align*}
This recursive bound further implies
\begin{align}
\| \bbeta^{(T)} - \hat  \bbeta \|_2^2  & \leq  (1-\rho )^Tr_0^2  +   (1+ \eta^{-1} )  \eta_0^2 \sigma^2  \sum_{t=0}^{T-1} (1-\rho )^t \| \bg_{T-1-t} \|_2^2  \, . \label{convergence.rate1}
\end{align}
Due to $ \sum_{t=0}^{T-1} (1-\rho )^t \| \bg_{T-1-t} \|_2^2 \leq 2 \rho^{-1} p + 3z$ on the event $\cF$,
as long as 
$$
T  \geq  \frac{2 \log \{ { r_0 n \epsilon }( \sigma_0 p)^{-1}\}}{\log \{(1-\rho)^{-1}\}}   \,,
$$
  the final  iterate $\bbeta^{(T)}$, conditioned on $\cE$, satisfies the bound
\begin{align*}
\| \bbeta^{(T)} -\hat \bbeta \|_2    \leq   \sqrt{ \bigg(\frac{\sigma_0 p}{n \epsilon }\bigg)^2 +   \eta_0^2 (1+ \eta^{-1} )  (2 \rho^{-1} p + 3z) \sigma^2 }    \nn\,, 
\end{align*} 
with probability (over $\{ \bg_t\}_{t=0}^{T-1}$) at least $1-2e^{-z}$. This concludes the proof.  \qed

\subsubsection{Proof of Theorem~\ref{prop:initial_value}}\label{proof.prop:initial_value}

Note that, conditioned on $\cE_0$, $\bbeta^{(t+1)} =\bbeta^{(t)} - \eta_0 \nabla \hat \cL_\tau (\bbeta^{(t)})  + \eta_0  \bh_t$ with $\bh_t =  \sigma \bg_t$,  for $t\in\{0\}\cup [T-1]$. In addition to the event $\cF$ given in \eqref{def:G}, we define 
\begin{align}\label{def:G_0}
\cF_0 =   \bigg\{   \max_{ t\in\{0\}\cup [T_0-1]} \| \bg_t \|_2 \leq  \zeta_{T_0} = p^{1/2} + \sqrt{2(\log T_0 +z ) } \bigg\}\, .
\end{align}
Note that $\PP( \cF_0)  \geq 1- e^{-z}$ by Lemma~\ref{lem:max.normal.l2}. When conditioning on $\cF_0$, we have
$$
\max_{ t\in\{0\}\cup [T_0-1]}\|\bh_t\|_2\leq  \zeta_{T_0} \sigma =:r_{{\rm priv}}\,.
$$ 
Following  \eqref{pf:prop:main_derivation} in the proof of Proposition~\ref{prop:crude.bound}, we know that, conditioning on $\cE_0\cap\cE_2\cap\cF_0$, for any $t\in\{0\}\cup [T_0-1]$, it holds that
\begin{align}
\label{pf:prop:main_derivation_copy}
\hat \cL_\tau(  \bbeta^{(t+1)} ) -  \hat \cL_\tau( \hat \bbeta  ) \leq \frac{1}{2 \eta_0} \|\bbeta^{(t)}  - \hat \bbeta  \|_2^2 - \frac{1}{2\eta_0 } \| \bbeta^{(t+1)} - \hat \bbeta  \|_2^2 +  r_{{\rm priv}}     \|  \bbeta^{(t+1)} -\hat  \bbeta  \|_2\,,
\end{align}
provided $\eta_0 \leq 1/(2 \phi_u)$, regardless of whether $\bbeta^{(t)}  \in \hat \bbeta + \BB^p(r_0)$. Summing \eqref{pf:prop:main_derivation_copy} for $t\in\{0\}\cup [T_0-1]$, we obtain
\begin{align}\label{pf:prop:initial_value:eq0}
& \sum_{t=1}^{T_0}\{\hat \cL_\tau(  \bbeta^{(t)} ) -  \hat \cL_\tau( \hat \bbeta  )\} \nn \\
&~~~~ \leq 
\frac{1}{2 \eta_0}\bigg( \|\bbeta^{(0)}  - \hat \bbeta  \|_2^2-  \| \bbeta^{(T_0)} - \hat \bbeta  \|_2^2 + 2T_0 \eta_0  r_{{\rm priv}}  \max_{t\in[T_0]}\|  \bbeta^{(t)} -\hat  \bbeta  \|_2\bigg)\,.
\end{align}
Next, we aim to bound $ \max_{ t\in [T_0]}\|  \bbeta^{(t)} -\hat  \bbeta  \|_2$. Noting that the left-side of \eqref{pf:prop:main_derivation_copy}, $\hat \cL_\tau(  \bbeta^{(t+1)} ) -  \hat \cL_\tau( \hat \bbeta  )$, is non-negative since $ \hat \bbeta $ is a global minimizer of $ \hat \cL_\tau(\cdot)$, we have
\begin{align*}
\|\bbeta^{(t+1)}  - \hat \bbeta  \|_2\leq \|\bbeta^{(t)}  - \hat \bbeta  \|_2+ 2 \eta_0  r_{{\rm priv}} \,,
\end{align*}
for any $t\in \{0\}\cup [T_0-1]$, which further implies
\begin{align*}
\max_{t\in [T_0]}\|  \bbeta^{(t)} -\hat  \bbeta  \|_2\leq \|\bbeta^{(0)}  - \hat \bbeta  \|_2+ 2T_0 \eta_0 r_{{\rm priv}}\,.
\end{align*}
Substituting this bound on $ \max_{  t\in [T_0]}\|  \bbeta^{(t)} -\hat  \bbeta  \|_2$ into the right-side of \eqref{pf:prop:initial_value:eq0}, it follows that
\begin{align}\label{pf:prop:initial_value:eq1}
\frac{1}{T_0}\sum_{t=1}^{T_0}\{\hat \cL_\tau(  \bbeta^{(t)} ) -  \hat \cL_\tau( \hat \bbeta  )\}\leq 
\frac{1}{2 \eta_0T_0} \|\bbeta^{(0)}  - \hat \bbeta  \|_2^2+  r_{{\rm priv}}  \big(\|\bbeta^{(0)}  - \hat \bbeta  \|_2+2 T_0 \eta_0 r_{{\rm priv}}\big) \,.
\end{align}
To proceed, we claim that for any $  t\in \{0\}\cup [T_0-1]$,
\begin{align}\label{pf:prop:initial_value:claim}
\hat \cL_\tau(  \bbeta^{(t+1)} ) -  \hat \cL_\tau( \bbeta^{(t)}   )\leq \eta_0 r_{{\rm priv}}^2\,.
\end{align}
Indeed, under $\cE_2$, we have
\begin{align*}
\hat \cL_\tau(  \bbeta^{(t+1)} ) -  \hat \cL_\tau( \bbeta^{(t)}   )
&\leq \langle  \nabla \hat \cL_\tau (\bbeta^{(t)} ), \bbeta^{(t+1)} -\bbeta^{(t)}  \rangle +  \phi_u  \| \bbeta^{(t+1)} -\bbeta^{(t)}  \|_2^2\\
&=-\eta_0\big\|\nabla \hat \cL_\tau (\bbeta^{(t)})  \big\|_2^2+\langle  \nabla \hat \cL_\tau (\bbeta^{(t)} ), \eta_0 \bh_t \rangle +\phi_u  \big\|- \eta_0 \nabla \hat \cL_\tau (\bbeta^{(t)})  + \eta_0 \bh_t \big\|_2^2\\
&\overset{({\rm i})}{\leq}-\frac{\eta_0}{2}\big\|\nabla \hat \cL_\tau (\bbeta^{(t)})  \big\|_2^2+(1-2\phi_u\eta_0)\langle  \nabla \hat \cL_\tau (\bbeta^{(t)} ), \eta_0 \bh_t \rangle +\phi_u \eta_0^2 \|\bh_t\|_2^2\\
&\overset{({\rm ii})}{\leq}\bigg\{\phi_u+\frac{(1-2\phi_u\eta_0)^2}{2\eta_0}\bigg\}\eta_0^2 \|\bh_t\|_2^2\overset{({\rm iii})}{\leq} \eta_0 r_{{\rm priv}}^2\,,
\end{align*}
where inequality (i) follows from the assumption $\eta_0 \leq 1/(2 \phi_u)$; in inequality (ii), we use the inequality $\langle  \nabla \hat \cL_\tau (\bbeta^{(t)} ),\eta_0 \bh_t \rangle\leq c\|\nabla \hat \cL_\tau (\bbeta^{(t)} )\|_2^2+ \eta_0^2\|\bh_t\|_2^2/(4c)$, with $c={\eta_0}\{2(1-2\phi_u\eta_0)\}^{-1}$ when $1-2\phi_u\eta_0>0$ (noting that  inequality (ii) still holds when $1-2\phi_u\eta_0=0$); inequality (iii) follows from $\phi_u\leq 1/(2\eta_0)$ and $\max_{ t\in \{0\}\cup [T_0-1]}\|\bh_t\|_2\leq r_{{\rm priv}}$, conditioning on $\cF_0$.

Therefore, the claim \eqref{pf:prop:initial_value:claim} is valid, and it further implies
\begin{align*}
\hat \cL_\tau(  \bbeta^{(T_0)} ) -  \hat \cL_\tau&(\hat \bbeta   )\leq \frac{1}{T_0}\sum_{t=1}^{T_0}\big\{\hat \cL_\tau(  \bbeta^{(t)} ) -  \hat \cL_\tau( \hat \bbeta  )\big\}+ T_0 \eta_0 r_{{\rm priv}}^2\\
&~\overset{\text{by \eqref{pf:prop:initial_value:eq1}}}{\leq }\frac{1}{2 \eta_0T_0} \|\bbeta^{(0)}  - \hat \bbeta  \|_2^2+  r_{{\rm priv}}  \big(\|\bbeta^{(0)}  - \hat \bbeta  \|_2+3 T_0\eta_0 r_{{\rm priv}}\big)  \\
&~~~~ = \frac{R^2 }{2 \eta_0T_0} +  R  r_{{\rm priv}} + 3 T_0  \eta_0 r_{{\rm priv}}^2  \, .
\end{align*}
In what follows, set $\Delta=\phi_l r_0^2$ and $T_0 = {R^2}({\eta_0\Delta})^{-1}$. The right-hand side then becomes
$$
 \frac{\Delta}{2} +   R \zeta_{T_0} \sigma   + \frac{3}{\Delta}  (   R \zeta_{T_0} \sigma  )^2\,. 
$$
Provided that 
$$
 R\zeta_{T_0}\sigma  \leq \frac{\sqrt{7}-1}{6} \Delta \,,
$$
it holds that
\begin{align} \label{pf:prop:initial_value:last_eq}
\hat \cL_\tau(  \bbeta^{(T_0)} ) -  \hat \cL_\tau(\hat \bbeta   )\leq \Delta=\phi_lr_0^2\,.
\end{align}
Now, we claim that if \eqref{pf:prop:initial_value:last_eq} holds,  then $\|\bbeta^{(T_0)} -\hat \bbeta  \|_2\leq r_0$ conditioned on $\cE_0\cap\cE_1$. To see this, suppose instead that $\|\bbeta^{(T_0)} -\hat \bbeta  \|_2>r_0$. By Lemma \ref{lem:rsc.smootheness}, it follows that $\hat \cL_\tau(\bbeta^{(T_0)} ) - \hat \cL_\tau(\hat \bbeta )    \geq   \phi_l r_0      \| \bbeta^{(T_0)} - \hat \bbeta \|_2 > \phi_l r_0^2$,
which contradicts \eqref{pf:prop:initial_value:last_eq}. This implies $\bbeta^{(T_0)} \in \hat \bbeta+\BB^p(r_0)$, thereby completing the proof. \qed

\subsubsection{Proof of Proposition~\ref{prop:events1}}\label{proof.e1}

To derive the convergence of the non-private estimator $\hat \bbeta$, we follow the proof of Theorem~2.1 in \citeS{chen2020robustS} with slight modifications.  For any $z\geq 0$, let $\tau = \tau_0 \{{n\epsilon }/({p+z} )\}^{1/(2+\iota)}$, where $\tau_0 \geq \sigma_0$. Then, it follows that with probability at least $1-2 e^{-z}$,
\begin{align*}
 \|\hat \bbeta-\bbeta^*\|_{\bSigma}  & \lesssim \frac{\sigma_{\iota}^{2+\iota}}{\tau^{1+\iota}} +  \upsilon_1 \sigma_0 \sqrt{\frac{p+z}{n }} + \upsilon_1  \tau\frac{p+z}{n} \\
 & \lesssim   \max\{ \sigma_\iota^{2+\iota} \tau_0^{-1-\iota}, \tau_0 \}
 \bigg(\frac{p+z}{n \epsilon} \bigg)^{(1+\iota)/(2+\iota)} +  \sigma_0 \sqrt{\frac{p+z}{n}}\,, 
\end{align*}
as long as $n\gtrsim (p+z)/\epsilon$. This, in turn, implies
$$
\|\hat \bbeta -\bbeta^*\|_2 \leq C_0  \bigg\{ \max\{ \sigma_\iota^{2+\iota}\tau_0^{-1-\iota}, \tau_0 \}  \bigg(\frac{p+z}{n \epsilon} \bigg)^{(1+\iota)/(2+\iota)} +  \sigma_0  \sqrt{\frac{p+z}{n}} \bigg\}\,,
$$
where $C_0>0$ is a constant depending only on $(\lambda_p, \upsilon_1)$.

 By Lemmas 5.2 and 5.3 of \citeS{Vershynin2012S},   $\|\bSigma^{-1/2}  \bx_i\|_2\leq 2\max_{\bu\in\cN_{1/2}}|\langle \bu, \bSigma^{-1/2}  \bx_i \rangle|$, where  $\cN_{1/2}$ is the $1/2$-net of $\mathbb{S}^{p-1}$  satisfying  $|\cN_{1/2}| \leq 5^p$. Then, by Assumption \ref{assump:design}, we have
\begin{align*}
    \PP(\|\bSigma^{-1/2}  \bx_i\|_2 >z) \leq \PP\bigg(2\max_{\bu\in\cN_{1/2}}|\langle \bu, \bSigma^{-1/2}  \bx_i \rangle| >z\bigg)\leq 2\cdot5^p e^{-z^2/(8\upsilon_1^2)}\,, 
\end{align*}
for any $z\geq 0$. This, in turn, implies that with probability at least $1-e^{-z}$,
$$
\| \bSigma^{-1/2}  \bx_i\|_2    \leq   2\sqrt{2} \upsilon_1 \sqrt{p\log 5+z+ \log 2}\, .
$$
As a result, it follows that $\max_{  i\in [n]}\|   \bx_i\|_2 \lesssim   \upsilon_1 \lambda_1^{1/2} (\sqrt{p+z + \log n} )$ with probability at least $1-e^{-z}$, as desired.

The last two bounds follow immediately from Lemmas \ref{lem:rsc} and   \ref{lem:gs}. This concludes the proof of Proposition~\ref{prop:events1}. \qed

\subsubsection{Proof of Proposition~\ref{prop:crude.bound}}\label{proof.prop:crude.bound}
 
Conditioned on $\cE_0$, note that $\bbeta^{(0)} \in \hat \bbeta + \BB^p(r_0)$ and $\bbeta^{(t+1)} =\bbeta^{(t)} - \eta_0 \nabla \hat \cL_\tau (\bbeta^{(t)})  + \eta_0\sigma  \bg_t$ for $t\in \{0\}\cup [T-1]$. Assume $ \|\bbeta^{(t)}  -\hat \bbeta  \|_2  \leq r_0$ for some $t\geq 0$.  Proceeding via proof by contradiction,   suppose $\| \bbeta^{(t+1)} -  \hat \bbeta  \|_2 > r_0$. From Lemma~\ref{lem:rsc.smootheness}, conditioning  on $\cE_0\cap \cE_1$, it holds that $\phi_l   r_0    \|\bbeta^{(t+1)} -  \hat \bbeta \|_2 \leq  \hat \cL_\tau(  \bbeta^{(t+1)}  ) -  \hat \cL_\tau ( \hat \bbeta ) $.
For the right-hand side, we have
\begin{align}\label{pf:prop:main_derivation}
    &   \hat \cL_\tau(  \bbeta^{(t+1)} ) -  \hat \cL_\tau( \hat \bbeta  ) =  \hat \cL_\tau(  \bbeta^{(t+1)} ) - \hat \cL_\tau(\bbeta^{(t)} ) + \hat \cL_\tau(\bbeta^{(t)} ) -  \hat \cL_\tau( \hat\bbeta  ) \nn \\
 &~~~ \stackrel{({\rm i})}{\leq } \langle  \nabla \hat \cL_\tau (\bbeta^{(t)} ), \bbeta^{(t+1)} -\bbeta^{(t)}  \rangle +  \phi_u  \| \bbeta^{(t+1)} -\bbeta^{(t)}  \|_2^2 - \langle \nabla \hat \cL_\tau(\bbeta^{(t)} ) , \hat \bbeta  -\bbeta^{(t)}  \rangle \nn\\
 &~~~ = \frac{1}{\eta_0} \langle    \bbeta^{(t)}  -   \bbeta^{(t+1)}   ,   \bbeta^{(t+1)} - \hat \bbeta  \rangle   + \phi_u  \|  \bbeta^{(t+1)} -\bbeta^{(t)}  \|_2^2  +  \sigma \langle   \bg_t  ,   \bbeta^{(t+1)} - \hat \bbeta  \rangle  \nn \\ 
 &~~~ = \frac{1}{2 \eta_0} \|\bbeta^{(t)}  - \hat \bbeta  \|_2^2 - \frac{1}{2\eta_0 } \|  \bbeta^{(t+1)} - \hat \bbeta  \|_2^2 - \frac{1}{2 \eta_0} \|  \bbeta^{(t+1)} -\bbeta^{(t)}  \|_2^2  \\ 
 &~~~~ ~~~+\phi_u  \|  \bbeta^{(t+1)} -\bbeta^{(t)}  \|_2^2 +  \sigma  \langle   \bg_t  ,   \bbeta^{(t+1)} - \hat \bbeta  \rangle  \nn\\
 &~~~ \stackrel{({\rm ii})}{\leq }  \frac{1}{2 \eta_0} \|\bbeta^{(t)}  -  \hat \bbeta   \|_2^2 - \frac{1}{2\eta_0 } \| \bbeta^{(t+1)} - \hat \bbeta  \|_2^2 + \sigma  \|  \bbeta^{(t+1)} -\hat \bbeta  \|_2 \cdot    \|  \bg_t \|_2 \nn \\
 & ~~~\stackrel{({\rm iii})}{\leq }   \frac{1}{2 \eta_0} \|\bbeta^{(t)}  - \hat \bbeta  \|_2^2 - \frac{1}{2\eta_0 } \| \bbeta^{(t+1)} - \hat \bbeta  \|_2^2 +   \zeta_T  \sigma    \|  \bbeta^{(t+1)} -\hat  \bbeta  \|_2\,, \nn
\end{align} 
where inequality (i) holds under $\cE_2$ and the convexity  of $\hat \cL_\tau(\cdot)$,
 inequality (ii) holds if $\eta_0 \leq 1/(2 \phi_u)$, and inequality (iii) uses conditioning on $ \cF$.  
Provided that $\zeta_T \sigma \leq  \phi_l  r_0$,  combining the above lower and upper bounds on $   \hat \cL_\tau(  \bbeta^{(t+1)} ) -  \hat \cL_\tau( \hat \bbeta  ) $  yields $\|  \bbeta^{(t+1)} - \hat \bbeta   \|_2  \leq  \|\bbeta^{(t)}  -\hat \bbeta   \|_2  \leq  r_0$,
which leads to a contradiction. Therefore, starting from an initial value $\bbeta^{(0)} \in  \hat \bbeta  + \BB^p (r_0)$,  and conditioning on the event $\cE \cap  \cF$ with suitably chosen parameters, we must have $\| \bbeta^{(t )} -\hat \bbeta  \|_2 \leq r_0$ for all $t\in    [T ]$. \qed

\subsection{Construction of Private Confidence Intervals}
\label{sec.CI}

In what follows, we present a Gaussian approximation result for DP Huber estimator $\bbeta^{(T)}$ obtained from Algorithm \ref{alg:DPHuberlow}, which serves as the foundation for constructing private confidence intervals. For the purpose of inference, we impose a lower bound variance assumption that $\EE(\varepsilon_i^2\,|\, \bx_i) \geq  \underline{\sigma}_0^2$ almost surely for some constant $\underline{\sigma}_0 \in (0, \sigma_0]$. This assumption complements the upper bound variance condition in Assumption~\ref{assump:heavy-tail}, which is required only for estimation.
Moreover, define
 \begin{align}\label{Xi}
 	\bXi= \bSigma^{-1}\bOmega\bSigma^{-1}~~\mbox{with}~~\bOmega= \EE(   \varepsilon_i^2 \bx_i \bx_i^{\T})  \,.
 \end{align}

\begin{theorem}
\label{thm.DPGA}
Let the assumptions in Theorem \ref{thm:low} hold. Then, the DP Huber estimator ${\bbeta}^{(T)}$ obtained from Algorithm \ref{alg:DPHuberlow}, with $\tau \asymp \sigma_\iota\{ n\epsilon/(p+\log n) \}^{1/(2+\iota)}$ for some  $\iota\in(0,1]$, satisfies 
\begin{align*}
	&~\sup_{\bu\in \RR^p}\sup_{x\in\RR}\big| \PP( \sqrt{n} \langle \bu/\|\bu\|_{\bXi},  {\bbeta}^{(T)}-\bbeta^* \rangle  \leq x )-\Phi(x)  \big| \lesssim  C_\iota     \frac{\sigma\sqrt{n}}{\underline{\sigma}_0} \sqrt{p+\log n}  \,,
\end{align*}
provided that $n \epsilon \gtrsim (\sigma_{\iota}/\underline{\sigma}_0)^{2+ 4/\iota} (p+\log n)$ and $\sigma \sqrt{p+  \log n} \lesssim \tau$, where $\sigma$ is either $\sigma_{{\rm dp}} $ or $ \sigma_{{\rm gdp}}$ as defined in \eqref{dp.noise.scale}, and $C_\iota >0$ depends only on $\iota$.
 \end{theorem}

Since the asymptotic covariance matrix $\bXi$ in \eqref{Xi} is unknown, directly constructing a confidence interval for $\langle \bu, \bbeta^* \rangle$ is infeasible. To address this, we develop a DP estimator of $\bXi$.  For some robustification parameter $\tau_1 >0$ and truncation parameter $\gamma_1>0$,  define
 \begin{align*}
 	\widehat{\bSigma}_{ \gamma_1 ,\epsilon}= \widehat\bSigma_{\gamma_1}+\varsigma_{1 }\bE ~~\mbox{and}~~\widehat\bOmega_{ \tau_1 , \gamma_1 ,  \epsilon}(\bbeta)=\widehat\bOmega_{{\tau_1}, \gamma_1}(\bbeta)+\varsigma_{2 }\bE\,,
 \end{align*}
where $\varsigma_{1 }$ and $\varsigma_{2}$ are noise scale parameters depending on $\epsilon$ or $(\epsilon, \delta)$, $\bE\in \RR^{p\times p}$ is a symmetric random matrix with independent standard normal entries in the upper-triangular and diagonal positions, and
 \begin{align}\label{def.var}
 \widehat\bSigma_{\gamma_1} = \frac{1}{n}\sn  \bx_i\bx_i^{\T}w_{\gamma_1}^2(\|\bx_i\|_2)~~\mbox{and}~~\widehat\bOmega_{{\tau_1},\gamma_1}(\bbeta) = \frac{1}{n}\sn \psi_{\tau_1}^2(y_i-\bx_i^{\T}\bbeta)\bx_i\bx_i^{\T}w_{\gamma_1}^2(\|\bx_i\|_2)\,.
 \end{align}
Lemma \ref{lem.dp} below ensures that for any $\tau_1,\gamma_1 > 0$ and $\bbeta \in \RR^p$, both  $\widehat{\bSigma}_{\gamma_1,\epsilon}$ and $\widehat\bOmega_{\tau_1,\gamma_1,\epsilon}(\bbeta)$ satisfy
 \begin{itemize}
 	\item[(1)] $\epsilon$-GDP with $ \varsigma_{1 }=  2\gamma_1^2 (n\epsilon)^{-1}$ and $\varsigma_{2 }=2\gamma_1^2\tau_1^2  (n\epsilon)^{-1}$;
 	\item[(2)] $(\epsilon,\delta)$-DP with $ \varsigma_{1 }=  2\gamma_1^2\sqrt{2\log(1.25/\delta)}(n\epsilon)^{-1}$ and $\varsigma_{2 }=2\gamma_1^2\tau_1^2 \sqrt{2\log(1.25/\delta)}(n\epsilon)^{-1}$.
 \end{itemize}
However, the perturbed matrices $\widehat{\bSigma}_{\gamma_1,\epsilon}$ and $\widehat\bOmega_{\tau_1,\gamma_1,\epsilon}(\bbeta)$ may not be positive semi-definite. We further project both matrices onto the cone of positive definite matrices, defined by $\{\bH:\bH \succeq \zeta \bI\}$, and obtain
 \begin{align*}
 	\widehat{\bSigma}_{\gamma_1,\epsilon}^{+}= \arg\min_{\bH \succeq \zeta \bI}\|\bH-	\widehat{\bSigma}_{\gamma_1,\epsilon} \| ~~\mbox{and}~~\widehat\bOmega_{\tau_1,\gamma_1,\epsilon}^{+}(\bbeta)= \arg\min_{\bH \succeq \zeta \bI}\|\bH-\widehat\bOmega_{\tau_1,\gamma_1,\epsilon}(\bbeta)\|\,,
 \end{align*}
 where $\zeta>0$ is a sufficiently small constant. 
 By the post-processing property, the projected matrices $\widehat{\bSigma}_{\gamma_1,\epsilon}^{+}$ and $\widehat\bOmega_{\tau_1,\gamma_1,\epsilon}^{+}(\bbeta)$ inherit the same privacy guarantees as $\widehat{\bSigma}_{\gamma_1,\epsilon}$ and $\widehat\bOmega_{\tau_1,\gamma_1,\epsilon}(\bbeta)$, respectively.
 Define
 \begin{align*}
 	\widetilde\bXi_{\tau_1,\gamma_1,\epsilon}(\bbeta)  =   (\widehat{\bSigma}_{\gamma_1,\epsilon }^{+} )^{-1}\widehat\bOmega_{\tau_1,\gamma_1,\epsilon }^{+}(\bbeta) (\widehat{\bSigma}_{\gamma_1,\epsilon }^{+} )^{-1}\,.
 \end{align*}

By the basic composition properties of DP, $\widetilde\bXi_{\tau_1,\gamma_1,\epsilon}(\bbeta)$ satisfies either $(3\epsilon, 3\delta)$-DP or $\sqrt{3}\epsilon$-GDP. 
Assuming that $\bbeta$ is a consistent estimator of $\bbeta^*$, the consistency of $\widetilde\bXi_{\tau_1,\gamma_1,\epsilon}(\bbeta)$ as an estimator of $\bXi$ is established in Lemma \ref{lem.varest}.

\begin{corollary}
\label{cor.dpGA}
Let the assumptions in  Theorem \ref{thm.DPGA} hold. Set $\gamma_1 \asymp \sqrt{p+\log n}$ and $\tau_1 \asymp \tau \asymp \sigma_\iota\{ n\epsilon/(p+\log n) \}^{1/(2+\iota)}$.  
Then, the  DP Huber estimator ${\bbeta}^{(T)}$ obtained from Algorithm \ref{alg:DPHuberlow}   satisfies 
 \begin{align*}
	& \sup_{\bu\in \RR^p}\sup_{x\in\RR}\big| \PP( \sqrt{n} \langle \bu/\|\bu\|_{\widetilde\bXi_{\tau_1,\gamma_1,\epsilon}(\bbeta^{(T)})},  {\bbeta}^{(T)}-\bbeta^* \rangle  \leq x )-\Phi(x)  \big|  \\
  & \leq R_{n}^* \asymp \frac{\sigma\sqrt{n}  }{\underline{\sigma}_0} \sqrt{p+\log n} + \underbrace{\frac{\tau_1 (p+\log n)}{\underline{\sigma}_0^{2}}     \bigg(\sigma + \frac{\sigma_0 }{\sqrt{n}} \bigg)   }_{R_{n }^{**}}\, , 
 \end{align*}
provided that $n \epsilon \gtrsim (\sigma_\iota / \underline{\sigma}_0 )^{2+4/\iota} (p+\log n)$, $\sigma \sqrt{p+\log n} \lesssim \tau$ and $R_{n }^{**}\ll 1$.
\end{corollary}

\begin{remark}\label{remark.inference}
 In the case of $\epsilon$-GDP, we set $\sigma=\sigma_{\rm gdp}\asymp \gamma \tau \sqrt{T}/(n\epsilon)$ with $\gamma \asymp \sqrt{p+\log n}$. Then, the error terms $R_n^*$ and $ R^{**}_n$ in Corollary \ref{cor.dpGA} are of order
 $$
 R_{n}^*  \lesssim   \frac{\sigma_{\iota}^2}{\underline{\sigma}_0^2}\sqrt{T n}  \bigg(\frac{p+\log n}{n\epsilon}\bigg)^{(1+\iota)/(2+\iota)} \,,
$$
and
$$
   R^{**}_n \asymp   \frac{\sigma_{\iota}^2}{\underline{\sigma}_0^2} \sqrt{ T ( p+\log n )}\bigg(\frac{p+\log n}{n\epsilon}\bigg)^{ \iota /(2+\iota)} + \frac{\sigma_{\iota}^2}{\underline{\sigma}_0^2}\sqrt{n }\bigg(\frac{p+\log n}{n\epsilon}\bigg)^{(1+\iota)/(2+\iota)} \, .
$$
Recalling that $T \asymp \log n$, asymptotic normality holds provided that $(p/\epsilon)^{2+2/\iota} \log^{1+2/\iota}(n) =o(n)$ as $n\to \infty$. In the case of $(\epsilon,\delta)$-DP, we have $\sigma=\sigma_{\rm  dp}\asymp\gamma \tau T\sqrt{\log(T/\delta)}/(n\epsilon)$. Following similar calculations, asymptotic normality holds under the dimension growth condition $(p/\epsilon)^{2+2/\iota} \log^{2+4/\iota}(n) \log^{1+2/\iota}(\delta^{-1}\log n) =o(n)$ as $n\to \infty$.
\end{remark}

\subsubsection{Technique Lemmas}
\label{tech.lems}

The proofs of Theorem \ref{thm.DPGA} and Corollary \ref{cor.dpGA} rely on the following technical lemmas, whose proofs are deferred to Section \ref{proof.lemmas}.

   Lemma \ref{lem.Bahadur} establishes the  Bahadur representation of $\widehat{\bbeta} - \bbeta^*$, which serves as a key component in demonstrating its asymptotic normality. 
The proof of Lemma \ref{lem.Bahadur} is given in Section \ref{proof.lem.Bahadur}.

 \begin{lemma}\label{lem.Bahadur}
 	Let Assumptions \ref{assump:design} and \ref{assump:heavy-tail} hold, and suppose $\tau \asymp \sigma_{\iota}\{ n\epsilon/(p+\log n) \}^{1/(2+\iota)}$.
    Then, with probability at least $1-Cn^{-1}$, we have
 	\begin{align*}
 		&~\bigg\|\bSigma^{1/2}(\widehat{\bbeta}-\bbeta^*)-\frac{1}{n}\sn \psi_{\tau}(\varepsilon_i )\bSigma^{-1/2}\bx_i\bigg\|_2 \\
 		&~~~~ \lesssim \sigma_\iota\bigg\{ \bigg( \frac{p+\log n}{n\epsilon} \bigg)^{(3+2\iota)/(2+\iota)} + \sqrt{\frac{p+\log n}{n}}   \bigg( \frac{p+\log n}{n\epsilon} \bigg)^{(1+\iota)/(2+\iota)} + \frac{p+\log n}{n}\bigg\}\,,
 	\end{align*}
 	provided that $n \epsilon \gtrsim   p+\log n$. 
 \end{lemma}

Building on Lemma \ref{lem.Bahadur}, the following result establishes the asymptotic normality of the Huber estimator $\widehat{\bbeta}$.
The proof is provided in Section~\ref{proof.lem.GA}. 
 
 \begin{lemma}\label{lem.GA}
 	Let Assumptions \ref{assump:design} and \ref{assump:heavy-tail} hold. 
    Then, the Huber estimator $\widehat{\bbeta}$ with $\tau \asymp \sigma_\iota\{ n\epsilon/(p+\log n) \}^{1/(2+\iota)}$ for some  $\iota\in(0,1]$  satisfies  
  \begin{align*}
 	&~\sup_{\bu\in \RR^p}\sup_{x\in\RR}\big| \PP( \sqrt{n} \langle \bu/\|\bu\|_{\bXi},  {\widehat\bbeta} -\bbeta^* \rangle  \leq x )-\Phi(x)  \big|\lesssim C_\iota' \frac{\sigma_{\iota}\sqrt{n}}{\underline{\sigma}_0}   \bigg( \frac{p+\log n}{n\epsilon} \bigg)^{(1+\iota)/(2+\iota)} \, ,
 \end{align*}
as long as $n \epsilon \gtrsim (\sigma_{\iota}/\underline{\sigma}_0)^{2+4 /\iota} (p+\log n)$, where $C_\iota'>0$ depends only on $\iota$. 
 \end{lemma} 
 
Since $\bXi= \bSigma^{-1}\bOmega\bSigma^{-1}$ is typically unknown, the next two lemmas study the properties of $\widehat\bSigma_{\gamma_1}$ and $\widehat\bOmega_{\tau_1,\gamma_1}(\bbeta)$, as defined in \eqref{def.var}, which serve as estimators of $\bSigma$ and $\bOmega$, respectively.
These results provide theoretical guarantees for the subsequent private estimation of $\bXi$. 
The proofs of Lemmas \ref{lem.dp} and \ref{lem.divvar} are given in Sections \ref{proof.lem.dp} and \ref{proof.lem.divvar}, respectively.

 \begin{lemma}\label{lem.dp}
 	Let $\bE\in \RR^{p\times p}$ be a symmetric random matrix whose upper-triangular and diagonal
 	entries are i.i.d. $\cN(0,1)$. Then, for any $\tau_1>0$, $\gamma_1 >0$, and $\bbeta\in\RR^p$, both $\widehat\bSigma_{\gamma_1}+2\gamma_1^2\sqrt{2\log(1.25/\delta)}   (n\epsilon)^{-1}\bE$ and $\widehat\bOmega_{\tau_1,\gamma_1}(\bbeta)+2\gamma_1^2\tau_1^2\sqrt{2\log(1.25/\delta)}  (n\epsilon)^{-1}\bE$ satisfy   $(\epsilon,\delta)$-DP. 
    Moreover, the estimates $\widehat\bSigma_{\gamma_1}+2\gamma_1^2  (n\epsilon)^{-1}\bE$ and $\widehat\bOmega_{\tau_1,\gamma_1}(\bbeta)+2\gamma_1^2\tau_1^2 (n\epsilon)^{-1}\bE$ satisfy $\epsilon$-GDP.
 \end{lemma}

 \begin{lemma}\label{lem.divvar}
 	Let Assumptions \ref{assump:design} and \ref{assump:heavy-tail} hold,  and suppose $\gamma_1 \asymp \sqrt{p+\log n}$.
    Then, with probability at least $1-Cn^{-1}$, we have $\|\widehat\bSigma_{ \gamma_1}-\bSigma\| \lesssim \sqrt{(p+\log n)/n}  $ and
 \begin{align*}
	\sup_{\bbeta\in\Theta(r)}	\|\widehat\bOmega_{\tau_1,\gamma_1}(\bbeta)- \bOmega \|  \lesssim &~ \tau_1 r \sqrt{p+\log n}+    \tau_1^2\frac{p+\log n}{n} + \frac{\sigma_{\iota}^{2+\iota}}{\tau_1^{\iota}}\,,
\end{align*}
provided that $n \gtrsim p+\log n$.
 \end{lemma}

 Lemma \ref{lem.varest} below demonstrates that $\widetilde\bXi_{\tau_1,\gamma_1,\epsilon}(\bbeta)$ serves as a valid and privacy-preserving estimator of $\bXi$, provided that $\|\bbeta-\bbeta^*\|_2$ is sufficiently small.
 The proof of Lemma \ref{lem.varest} is given in Section \ref{proof.lem.varest}.

  \begin{lemma}\label{lem.varest}
 Let Assumptions \ref{assump:design} and \ref{assump:heavy-tail} hold. Set $\gamma_1 \asymp \sqrt{p+\log n}$ and $\tau_1 \asymp \sigma_\iota \{ n \epsilon / (p+\log n) \}^{1/(2+\iota) }$. 
 Then, with probability at least $1-Cn^{-1}$,  
\begin{align*}
    \sup_{\bbeta\in\Theta(r)}\|	\widetilde\bXi_{\tau_1,\gamma_1,\epsilon}(\bbeta)- \bXi\| \lesssim  (\tau_1 r +\varsigma_2   )\sqrt{p+\log n}\, ,
\end{align*}
holds as long as $\sqrt{ ({p+\log n})/{n}} + \varsigma_{1 }\sqrt{p+\log n}\ll 1$.
 \end{lemma}

\subsubsection{Proof of Theorem \ref{thm.DPGA}}

First, note that 
 \begin{align}\label{div1}
 	&~ | \sqrt{n}\langle \bu/\|\bu\|_{\bXi}, {\bbeta}^{(T)}-\bbeta^* \rangle -    \sqrt{n} \langle \bu/\|\bu\|_{ \bXi}, {\widehat\bbeta} -\bbeta^* \rangle| 
   \leq  \sqrt{n}\frac{\|\bu\|_2}{\|\bu\|_{\bXi}}\|\bbeta^{(T)}-\widehat{\bbeta} \|_2 \,.
 \end{align}
 Applying Theorem \ref{thm:dp.convergence} in Section \ref{sec.inter1}, we have with probability at least $1-Cn^{-1}$ that
 \begin{align*} 
 	\|\bbeta^{(T)}-\widehat{\bbeta} \|_2\lesssim  \frac{\sigma_0 p }{n\epsilon}+\sigma\sqrt{p+\log n}\, ,
 \end{align*}
 as long as $\sigma \sqrt{p+\log T+\log n} \lesssim r_0$, where $\sigma \in \{ \sigma_{{\rm dp}} , \sigma_{{\rm gdp}} \}$ as defined in \eqref{dp.noise.scale}. Recall that $T \asymp \log n$, and  as shown in the proof of Lemma \ref{lem.GA}, we have $ \|\bu\|_{\bXi}   \geq \underline{\sigma}_0  \|\bu\|_{\bSigma^{-1}}$.
 Then, for any  $\bu\in \RR^p$,   it holds with probability at least $1-Cn^{-1}$ that
 \begin{align*}
 	\sqrt{n}\frac{\|\bu\|_2}{\|\bu\|_{\bXi}}\|\bbeta^{(T)}-\widehat{\bbeta} \|_2 \leq R_{n } \asymp  \frac{\sqrt{n}}{ \underline{\sigma}_0} \bigg(\frac{\sigma_0 p }{n\epsilon}+\sigma\sqrt{p+\log n}\bigg) \,,
 \end{align*}
 provided that $\sigma \sqrt{p+ \log n} \lesssim r_0$.

Next, by applying Lemma \ref{lem.GA}, for any  $\bu\in \RR^p$ and $x\in\RR$,
 \begin{align*}
 	\PP( \sqrt{n} \langle \bu/\|\bu\|_{\bXi},   {\bbeta}^{(T)}-\bbeta^* \rangle  \leq x ) \leq &~ \PP( \sqrt{n} \langle \bu/\|\bu\|_{ \bXi},   {\widehat\bbeta} -\bbeta^* \rangle  \leq x + R_{n }) +\frac{C}{n}\\
  	\leq &~ \Phi(x + R_{n })+  C_{\iota}'' \frac{\sqrt{n} \sigma_\iota}{ \underline{\sigma}_0}      \bigg( \frac{p+\log n}{n\epsilon} \bigg)^{(1+\iota)/(2+\iota)}\\
 	\leq &~ \Phi(x )+  R_{n }+  C_{\iota}''\frac{\sqrt{n} \sigma_\iota}{ \underline{\sigma}_0}     \bigg( \frac{p+\log n}{n\epsilon} \bigg)^{(1+\iota)/(2+\iota)}\,,
 \end{align*}
 provided that $n \epsilon \gtrsim (\sigma_{\iota}/ \underline{\sigma}_0 )^{2 + 4/\iota} (p+\log n)$ and $\sigma \sqrt{p +\log n} \lesssim r_0$, where $C_{\iota}''>0$ is a constant depending only on $\iota$.

 Putting together pieces, we conclude that for any  $\bu\in \RR^p$ and $x\in\RR$,
 \begin{align*}
 	&~\PP( \sqrt{n} \langle \bu/\|\bu\|_{\bXi},   {\bbeta}^{(T)}-\bbeta^* \rangle  \leq x ) -\Phi(x ) \\
 	&~~~~\lesssim  \sqrt{n}\bigg\{ \frac{\sigma}{ \underline{\sigma}_0}  \sqrt{p+\log n}+C_{\iota}''\frac{\sigma_{\iota}}{ \underline{\sigma}_0}     \bigg( \frac{p+\log n}{n\epsilon} \bigg)^{(1+\iota)/(2+\iota)}\bigg\}\,.
 \end{align*}
 A reversed inequality can be obtained via the same argument. To summarize, we obtain
 \begin{align*}
 	&~\sup_{\bu\in \RR^p}\sup_{x\in\RR}\big| \PP( \sqrt{n} \langle \bu/\|\bu\|_{\bXi},  {\bbeta}^{(T)}-\bbeta^* \rangle  \leq x )-\Phi(x)  \big|\\
 	&~~~~\lesssim \max\{ 1,  C_{\iota}'' \}\sqrt{n} \bigg\{   \frac{\sigma}{ \underline{\sigma}_0}  \sqrt{p+\log n}+   \frac{\sigma_\iota}{ \underline{\sigma}_0}  \bigg( \frac{p+\log n}{n\epsilon} \bigg)^{(1+\iota)/(2+\iota)} \bigg\}\,,
 \end{align*}
provided that $n \epsilon \gtrsim (\sigma_{\iota}/ \underline{\sigma}_0 )^{(4+2\iota)/\iota} (p+\log n)$ and $\sigma \sqrt{p+  \log n} \lesssim r_0$. Moreover, note that
\begin{align}\label{lowersigma}
    \sigma \sqrt{p+\log n} \geq \sigma_{{\rm gdp}} \sqrt{p+\log n} \asymp  \sigma_\iota \sqrt{T} \bigg(\frac{p+\log n}{n \epsilon } \bigg)^{(1+\iota)/(2+\iota)} \,,
\end{align}
indicating that the first term on the right-hand side of the Gaussian approximation inequality is the dominating one. This completes the proof of   Theorem \ref{thm.DPGA} by noting that $r_0\asymp \tau$, as stated in Theorem~\ref{thm:low}. \qed

\subsubsection{Proof of Corollary \ref{cor.dpGA}}

To being with, applying Theorem \ref{thm:low} with $\tau \asymp \sigma_{\iota}\{n\epsilon/(p+\log n)\}^{1/(2+\iota)}$ gives
\begin{align*}
\|\bbeta^{(T)} -   \bbeta^* \|_2   
   \lesssim    \frac{\sigma_0 p}{n\epsilon} +    \sigma\sqrt{p+ \log n} + \sigma_{\iota}   \bigg( \frac{p+\log n}{n \epsilon} \bigg)^{(1+\iota)/({2+\iota})} +  \sigma_{0}\sqrt{\frac{p+\log n} {n }} \,,
\end{align*}
with probability at least $1-Cn^{-1}$, provided that $n\epsilon \gtrsim p + \log n$. Moreover, it follows from \eqref{lowersigma}   that $\|\bbeta^{(T)} -   \bbeta^* \|_2  \leq r_1 \asymp  \sigma \sqrt{p+\log n} + \sigma_0 \sqrt{(p+\log n)/n}$ with high probability.

Throughout the rest of the proof, write $\widetilde\bXi =\widetilde\bXi_{\tau_1,\gamma_1,\epsilon}(\bbeta^{(T)})$ for simplicity.  Since $\EE(\varepsilon_i^2\,|\,\bx_i)\geq \underline{\sigma}_0^2$ almost surely, we have $\| \bu \|_{\bXi}^2 \geq \underline{\sigma}_0^2\|\bu\|_{\bSigma^{-1}}^2$ for any $\bu\in\RR^p$.  In addition, taking $r= r_1$ in Lemma~\ref{lem.varest}, and noting that under $\epsilon$-GDP,
$$
\varsigma_2 = \frac{2 (\gamma_1 \tau_1)^2   }{n \epsilon}  ~\mbox{ and }~
 \tau_1 \sigma \sqrt{p + \log n}  = \sqrt{T} \frac{2 \gamma \tau \tau_1 \sqrt{p+\log n} }{n \epsilon } .
$$
With $\gamma \asymp \gamma_1$, $\tau \asymp \tau_1$ and $T \asymp \log n$, it is easy to see that $\varsigma_2 \lesssim \tau_1 \sigma \sqrt{p + \log n}$. Similarly, the same bound holds for $(\epsilon, \delta)$-DP.  Then, applying Lemma \ref{lem.varest}, we obtain that  with probability at least $1-C n^{-1}$,
\begin{align*}
   &~ \bigg| \frac{ \|\bu\|_{\widetilde\bXi }^2 }{\|\bu\|_{\bXi}^2} -1  \bigg| \leq  \frac{1}{\|\bu\|_{\bXi}^2}\big|\|\bu\|_{\widetilde\bXi}^2-\|\bu\|_{\bXi}^2\big| \leq \frac{\|\bu\|_2^2}{\|\bu\|_{\bXi}^2}\|  \widetilde\bXi -\bXi  \| \leq    r_2 \asymp  \underline{\sigma}_0^{-2}    \tau_1 r_1 \sqrt{p+\log n} \,, 
\end{align*}
as long as $\sqrt{ ({p+\log n})/{n}} + \varsigma_{1 }\sqrt{p+\log n}\ll 1$. For a sufficiently large sample size such that $r_2 \ll 1$ (which, in turn, implies $\varsigma_1 \sqrt{p+\log n} \ll 1$),  Theorem \ref{thm.DPGA} and Lemma~A.7 in \citeS{SZ2015S} imply 
\begin{align*}
 \sup_{\bu\in \RR^p}\sup_{x\in\RR}\big| \PP( \sqrt{n} \langle \bu/\|\bu\|_{\widetilde\bXi },  {\bbeta}^{(T)}-\bbeta^* \rangle  \leq x )-\Phi(x)  \big| \lesssim   \frac{\sigma\sqrt{n}  }{ \underline{\sigma}_0}  \sqrt{p+\log n} + r_2 \, , 
 \end{align*}
provided that $n \epsilon \gtrsim (\sigma_\iota / \underline{\sigma}_0 )^{2+4/\iota} (p+\log n)$ and $\sigma \sqrt{p+\log n} \lesssim \tau$. This proves the claim. \qed

\setcounter{lemma}{0}
\renewcommand{\thelemma}{\thesection.\arabic{lemma}}  
\setcounter{theorem}{0}
\renewcommand{\thetheorem}{\thesection.\arabic{theorem}}  
\setcounter{definition}{0}
\renewcommand{\thedefinition}{\thesection.\arabic{definition}}
\setcounter{proposition}{0}
\renewcommand{\theproposition}{\thesection.\arabic{proposition}}
\section{Proofs for the Results in Section \ref{sec:HD}}

\subsection{Intermediate Results in Section \ref{sec:HD}}
For ease of presentation, we introduce several notations. Let $\cS = \{ S \subseteq [p]: 1\leq |S| \leq 2s \}$. Let 
$\mathbb{H}(s):= \{ \bbeta\in\mathbb{R}^p:\|\bbeta\|_0\leq s \}$ and $\mathbb{N}(s, r_0) = \{ (\bbeta_1 , \bbeta_2 ) \in \mathbb{N}(r_0) : \bbeta_1, \bbeta_2 \in \mathbb{H}(s)\}$, which is a sparse version of the local set $\mathbb{N}(r_0)$ appeared in 
\eqref{def:event.E0-2}. Given a truncation parameter $\gamma>0$, a curvature parameter $\phi_l>0$, a smoothness parameter $\phi_u>0$,  a convergence radius $r_0>0$, and a parameter $r_1>0$, we define the good events in high-dimensional settings, with a slight abuse of notation, as follows:
\begin{align} 
    &~~~~~~~~~~~~~~~~~~~~~~~~~~\cE_0     
   = \bigg\{  \max_{i\in[n]} \| \bx_i \|_\infty \leq \gamma \bigg\}\,,     \nn \\ 
    &~\cE_1  
     = \bigg\{   \inf_{ (\bbeta_1,\, \bbeta_2) \in \mathbb{N}(s, r_0) }  \frac{\hat \cL_\tau(\bbeta_1)  - \hat \cL_\tau (\bbeta_2) - \langle \nabla \hat \cL_\tau(\bbeta_2), \bbeta_1-\bbeta_2 \rangle}{ \| \bbeta_1 -\bbeta_2 \|_2^2}  \geq \phi_l   \bigg\}\,,  \nn  \\
     &~~~\cE_2 
     = \bigg\{\sup_{\bbeta_1,\, \bbeta_2\in \mathbb{H}(s)}\frac{\hat \cL_\tau(\bbeta_1)  - \hat \cL_\tau (\bbeta_2) - \langle \nabla \hat \cL_\tau(\bbeta_2), \bbeta_1-\bbeta_2 \rangle}{\| \bbeta_1 -\bbeta_2 \|_2^2} \leq \phi_u  \bigg\}\,,\label{def:event.HD.E0-4}
\\
    &~~~~~~\cE_3  
     = \bigg\{ \sup_{ \bbeta \in \Theta(r_0)\cap \mathbb{H}(s),\, S\in \cS}  \| \{\nabla\hat \cL_\tau(\bbeta)  - \nabla\hat \cL_\tau (\bbeta^*) \}_S \|_2   \leq  \frac{3}{2}\phi_u r_0 \bigg\}\,, \nn \\
     &~~~~~~~~~~~~~~~~~~~~~~\cE_4  
    =    \bigg\{  \sup_{S\in\cS} \| \{\nabla\hat \cL_\tau (\bbeta^*) \}_S  \|_2 \leq r_1  \bigg\}\,. \nn
\end{align}
In the event $\cE_0$, the truncation parameter $\gamma>0$ ensures that the   $\ell_\infty$-sensitivity  of each clipped gradient descent step is bounded. For bounded covariates (in magnitude), $\gamma$ can simply be chosen as the upper bound. For unbounded covariates that follow a sub-Gaussian distribution, we will show that if $\gamma\asymp\sqrt{\log(p\vee n)}$, then event $\cE_0$ occurs with high probability. Events $\cE_1$ and $\cE_2$ are related to the restricted strong convexity and smoothness of the empirical Huber loss in high dimensions, when restricted to sparse parameter sets. Together, the events $\cE_3$ and $\cE_4$ ensure that the empirical gradients at $\bbeta$ are well controlled within a local and sparse neighborhood of $\bbeta^*$. This property is crucial for characterizing the convergence region and establishing the error bounds of the estimator derived from  Algorithm~\ref{alg:DPHuberhigh}.

For each $t \in {0} \cap [T-1]$, let ${\bw^t_1, \ldots, \bw^t_s, \widetilde{\bw}^t}$ denote the Laplace noise variables generated by the noisy update step of  Algorithm \ref{alg:DPHuberhigh}, where $\tilde \bbeta^{(t+1)}$ from \eqref{noisy.gdnew} is used as the input to the NoisyHT algorithm. Conditioned on $\cE_0 \cap \cE_1 \cap \cE_2 \cap \cE_3 \cap \cE_4$, Theorem \ref{thm:dp.HD.convergence} provides a high-probability bound for $\| \bbeta^{(T)} - \bbeta^* \|_2$, where the probability is taken with respect to the Laplace noise variables $\{\bw^t_1, \ldots, \bw^t_s, \widetilde{\bw}^t\}_{t=0}^{T-1}$ injected into the gradient descent updates. The proof of Theorem \ref{thm:dp.HD.convergence} is given in Section \ref{proof.thm:dp.HD.convergence}.
 
\begin{theorem} \label{thm:dp.HD.convergence} 
    Assume that $\bbeta^{(0)} \in \mathbb{H}(s)\cap \Theta(r_0)$ for some $r_0>0$, and  the learning rate satisfies $\eta_0 = \eta/(2\phi_u) $ for some $\eta\in(0,1)$. Write $\lambda = 2\eta_0\gamma \tau /n>0$. For any $z >0$, let  $\lambda$ satisfy  
    $$
    \bigg( \frac{s T\lambda}{\epsilon} \bigg)^2 \{\log (p s T)+z\}^2    \log\bigg(\frac{T}{\delta}\bigg)   \leq c_0 r_0^2
    $$
    for some small constant $c_0>0$ depending only on $(\eta,\phi_l,\phi_u)$. Then, conditioned on $\cE_0 \cap \cE_1 \cap \cE_2 \cap \cE_3\cap \cE_4$, the sparse DP Huber estimator $\bbeta^{(T)}$  obtained in Algorithm~\ref{alg:DPHuberhigh}, with $T\geq 2\log\{r_0    n \epsilon (\sigma_0\log p)^{-1}\}/\log\{(1-\rho)^{-1}\} $, satisfies
        \begin{align*}
           \|\bbeta^{(T)}-\bbeta^*\|_2^2  & \leq \frac{r_1^2}{  \phi_l^2}   + \bigg(\frac{2\phi_u}{\phi_l}+\frac{1}{2}\bigg) \bigg( \frac{\sigma_0\log p}{  n\epsilon} \bigg)^2 + \frac{2 C_0 C_\eta}{\phi_l\rho } \bigg( \frac{s T\lambda}{\epsilon} \bigg)^2 \{\log (p s T)+z\}^2    \log\bigg(\frac{T}{\delta}\bigg)    \,  
        \end{align*}
    with probability at least $1-2e^{-z}$ (over $\{\bw^t_1, \ldots, \bw^t_s, \widetilde{\bw}^t\}_{t=0}^{T-1}$), provided that  
    $$
   s \geq s^*\max\big\{ 192(1+\eta^{-2})(\phi_u/\phi_l)^2 , 16\eta/(1-\eta) \big\}~\mbox{and}~r_1 \leq \phi_l r_0 / 4 \,,
    $$ 
    where $C_0>0$ is an absolute constant, $\rho  =  4\eta s^* / s  +  \eta^2\phi_l/(8\phi_u)$ and  $C_{\eta} = 4\phi_u\{ 5+4\eta(1-\eta)^{-1} +6(1-\eta)\eta^{-1}\}$.
\end{theorem}

 Theorem \ref{thm:dp.HD.convergence} establishes an upper bound on the statistical errors of the sparse DP Huber estimator obtained via noisy hard thresholding, conditioned on certain favorable (random) events defined in \eqref{def:event.HD.E0-4}. Recall that the noise scale is set as $\lambda = 2\eta_0\gamma \tau /n$ in Algorithm~\ref{alg:DPHuberhigh}, and the noise scale condition in Theorem~\ref{thm:dp.HD.convergence} is not restrictive. If the condition does not hold, $r_0$ is of a smaller order than the convergence rate of $\|\bbeta^{(T)} - \bbeta^*\|_2$, implying that the initialization $\bbeta^{(0)}$ is sufficiently good and further application of NoisyIHT becomes unnecessary. In what follows, we demonstrate that these events occur with high probability under Assumptions~\ref{assump:design} and \ref{assump:heavy-tail}. The proof of Proposition \ref{prop:events2} is given in Section \ref{proof.e2}.

\begin{proposition}\label{prop:events2} 
    Let Assumptions \ref{assump:design} and \ref{assump:heavy-tail} hold, and set the robustification parameter   $\tau\geq 16\sigma_0$. Then, for any $z\geq 1/2$, the event $\cE=\cE_0 \cap \cE_1 \cap \cE_2 \cap \cE_3 \cap \cE_4$ occurs with probability at least $1 - 5e^{-z}$, provided that $n\geq C_{0}'(s\log p +z)$, and the   parameters $(r_0, \gamma, \phi_l, \phi_u,r_1)$ in \eqref{def:event.HD.E0-4} are chosen as 
    \begin{align*}
    &~~~~ r_0 =  \frac{\tau}{16\sqrt{\kappa_4\lambda_1}}  \,,~~ \gamma = C_1' \sqrt{\log p+\log n+z}\,,~~ \phi_l= \frac{\lambda_p}{8}\,, ~~\phi_u=\lambda_1 \,,\\
    &~~~~~\mbox{and}~~r_1 = C_2' \bigg\{\frac{\sigma_{\iota}^{2+\iota}} {\tau^{1+\iota} }   +     \sigma_0 \sqrt{\frac{s\log (ep/s) + z}{n}} +  \tau \frac{s\log (ep/s) + z}{n}  \bigg\}\,,
\end{align*} 
   where $C_0'$, $C_1'$ and $C_2'$ are   positive constants depending only on $(\upsilon_1,\lambda_1,\lambda_p,\kappa_4)$.  
\end{proposition} 

\subsection{Proof of Theorem \ref{thm:high}}

Similar to the proof of Theorem \ref{thm:low}, the proof of Theorem~\ref{thm:high} follows directly from Theorem~\ref{thm:dp.HD.convergence} and Proposition~\ref{prop:events2}. \qed

\subsection{Proofs of the Intermediate Results}
\subsubsection{Technical Lemmas}

The first lemma ensures that Algorithm \ref{alg:NHT} is differentially private, provided that the vector $\bv = \bv(\bZ)$ exhibits bounded changes in value when any single datum in $\bZ$ is modified.

\begin{lemma}\label{lem.HD.noiseht}
If, for every pair of adjacent datasets $\bZ$ and $\bZ'$, we have $\|\bv(\bZ) - \bv(\bZ')\|_{\infty} < \lambda$, then for $0<\epsilon \leq 0.5$, $0<\delta \leq 0.011$, and $s \geq 10$, the NoisyHT algorithm (Algorithm \ref{alg:NHT}) with a Laplace noise scale of at least $2\epsilon^{-1}\lambda\sqrt{5s \log(1/\delta) } $ is $(\epsilon, \delta)$-DP. 
\end{lemma}

The proof of Lemma \ref{lem.HD.noiseht} is given in Section \ref{proof.lem.HD.noiseht}.
Lemmas \ref{lem:HD.gs} and \ref{lem:HD.rsc} below extend Lemmas \ref{lem:gs} and \ref{lem:rsc} to high-dimensional settings under sparsity. The proofs of Lemmas \ref{lem:HD.gs} and \ref{lem:HD.rsc} are given in Sections \ref{proof.lem:HD.gs} and \ref{proof.lem:HD.rsc}, respectively.

\begin{lemma}\label{lem:HD.gs}
Under Assumption~\ref{assump:design}, for any $z\geq0$, it holds with probability at least $1-e^{-z}$ that 
$$
    \widehat{\mathcal{L}}_\tau (\boldsymbol{\beta}_1 )-\widehat{\mathcal{L}}_\tau (\boldsymbol{\beta}_2 )- \langle\nabla \widehat{\mathcal{L}}_\tau (\boldsymbol{\beta}_2 ), \boldsymbol{\beta}_1-\boldsymbol{\beta}_2 \rangle \leq \lambda_1 \|\bbeta_1 - \bbeta_2\|_2^2  
$$
for all $\bbeta_1, \bbeta_2 \in \mathbb{H}(s)$, provided that $n \gtrsim \upsilon_1^4 (s\log p + z)$. 
\end{lemma}

\begin{lemma} \label{lem:HD.rsc}
    Let Assumptions~\ref{assump:design} and \ref{assump:heavy-tail} hold. For any $r>0$, let the robustification parameter $\tau$ satisfy $\tau \geq 16 \max\{\sigma_0,  (\kappa_4\lambda_1)^{1/2}r\}$, where $\kappa_4$ is defined in (\ref{def:kappa4}). Then, for any $z\geq 0$,  with probability at least $1-e^{-z}$, 
\begin{align*}
    \hat \cL_\tau(\bbeta_1)  - \hat \cL_\tau (\bbeta_2) - \langle \nabla \hat \cL_\tau(\bbeta_2), \bbeta_1-\bbeta_2 \rangle \geq  \frac{\lambda_p}{8}  \|\bbeta_1 - \bbeta_2\|_2^2 
\end{align*} 
holds uniformly for $\bbeta_1 \in \Theta(r/2)\cap \mathbb{H}(s)$ and $\bbeta_2 \in \{ \bbeta_1 + \BB^p(r)  \} \cap \mathbb{H}(s)$, provided that $n \gtrsim \kappa_4 \upsilon_1^2  (\lambda_1 / \lambda_p )^2 (s\log p + z)$. The same bound also applies when $\bbeta_2 \in \Theta(r/2)\cap \mathbb{H}(s)$ and $\bbeta_1 \in \{ \bbeta_2 + \BB^p(r) \} \cap \mathbb{H}(s)$. 
\end{lemma}

Next, we derive the   properties of the gradient of the empirical Huber loss at the true coefficient vector $\bbeta^*$. 

\begin{lemma}\label{lem:HD.huber.grad_diff.ubd}
   Under Assumption~\ref{assump:design}, for any $z\geq 1/2$ and $r>0$, it holds with probability at least $1-e^{-z}$ that
    $$
       \sup_{\bbeta \in \Theta(r)\cap\mathbb{H}(s),  \,  S\in \cS } \big\| \{ \nabla\hat \cL_\tau(\bbeta)  - \nabla\hat \cL_\tau (\bbeta^*) \}_S\big\|_2 \leq  \frac{3}{2}\lambda_1 r \,,
    $$
    provided that $n \gtrsim \upsilon_1^4 (s\log p +z)$. 
\end{lemma}

\begin{lemma}\label{lem:HD.huber.grad_oracle.ubd}
    Under Assumptions \ref{assump:design} and \ref{assump:heavy-tail}, for any $z>0$, it holds with probability at least $1-e^{-z}$ that
    $$
    \sup_{S\in \cS} \big\| \{ \nabla\hat \cL_\tau (\bbeta^*)\}_S \big\|_2 \lesssim  \lambda_1^{1/2} \bigg\{  \frac{\sigma_{\iota}^{2+\iota}}{\tau^{1+\iota}} + \sigma_0 \sqrt{\frac{s\log(ep/s) + z}{n}} + \tau \frac{s\log(ep/s) + z}{n} \bigg\} \,.
    $$
\end{lemma}

The proofs of Lemmas \ref{lem:HD.huber.grad_diff.ubd} and \ref{lem:HD.huber.grad_oracle.ubd} are given in Sections \ref{proof.lem:HD.huber.grad_diff.ubd} and \ref{proof.lem:HD.huber.grad_oracle.ubd}, respectively. Lemma \ref{lem:HD.max.Laplace} provides a high probability bound on the maximum scale of noise injected per iteration due to the `peeling' process of Algorithm~\ref{alg:DPHuberhigh}. The proof of Lemma \ref{lem:HD.max.Laplace} is given in Section \ref{proof.lem:HD.max.Laplace}.

\begin{lemma} \label{lem:HD.max.Laplace}
    Let $\bw^t_1, \ldots, \bw^t_s, \widetilde{\bw}^t \in \RR^p$ for $t\in \{0\}\cup [T-1]$, with some $T\geq 1$, be independent random vectors whose entries independently follow the Laplace distribution $\operatorname{Laplace}(\lambda_{\operatorname{dp}})$, where $ \lambda_{{\rm dp}} = 2\epsilon^{-1}T \lambda \sqrt{5s \log( T/\delta)}$. Let
     $$
    W_t := \sum_{i=1}^s  \|\bw^t_i\|_{\infty}^2+\|\widetilde{\bw}_{S^{t+1}}^t\|_2^2   \quad \text{for}~~ t\in \{0\}\cup [T-1]\,, 
    $$  
    where $S^{t+1} \subseteq  [p]$ with $|S^{t+1}|\leq s$, and $\widetilde{\bw}^t_{S^{t+1}} \in \RR^{p}$ is obtained by setting all entries of $\widetilde{\bw}^t$ to $0$, except those with indices in $S^{t+1}$. Then, there exists an absolute constant $C_0>0$ such that for any $z>0$, the event 
\begin{align*}
\cE_W = \bigg\{\max_{t\in \{0\}\cup [T-1]} W_t \leq C_0  \bigg( \frac{sT \lambda}{\epsilon} \bigg)^2 \{\log (p s T)+z\}^2   \log\bigg(\frac{T}{\delta}\bigg)   \bigg\}\,, 
\end{align*}
occurs with probability at least $1-2e^{-z}$ (over $\{\bw^t_1, \ldots, \bw^t_s, \widetilde{\bw}^t\}_{t=0}^{T-1}$).
\end{lemma} 

Lemma \ref{lem.cons} below  extends Lemma {\rm 1} from \citeS{jain2014iterativeS} and Lemma {\rm A.3} from \citeS{cai2021costS} by accounting for the noise in the private selection of the support set. The proof of Lemma \ref{lem.cons} is given in Section \ref{proof.lem.cons}.

\begin{lemma}\label{lem.cons}
    Let $\tilde P_s$ be defined as in Algorithm \ref{alg:NHT}. For any $\bxi \in \RR^p$, let $S$ be the index set determined by Algorithm \ref{alg:NHT} with $\bv=\bxi$. Moreover, let $\hat{\bxi}\in \mathbb{R}^p$ such that $\|\hat{\bxi}\|_0 = \smallhat{s}\leq s$. Then, for any $c>0$ and index set $\tilde{S}\subseteq [p]$ satisfying $S\subseteq \tilde{S}$,
 \begin{align}\label{cons.bound}
     \| \{ \tilde P_s(\bxi) - \bxi\}_{\tilde{S}} \|_2^2 \leq (1+c^{-1})\frac{|\tilde{S}|-s}{|\tilde{S}|-\smallhat{s}}\| 
\{\hat{\bxi}-\bxi\}_{\tilde{S}} \|_2^2 + 4(1+c)\sum_{i\in[s]}\|\bw_i\|_{\infty}^2\,,
 \end{align}
 where $\{\bw_1,\ldots,\bw_s\}$ are specified in Algorithm \ref{alg:NHT}.
\end{lemma}

\subsubsection{Proof of Theorem~\ref{thm:dp.HD.convergence}}\label{proof.thm:dp.HD.convergence}

Let $\cE : =  \cE_0 \cap  \cE_1 \cap  \cE_2\cap  \cE_3\cap \cE_4$, where $\cE_0$, $\cE_1$, $\cE_2$, $\cE_3$ and $\cE_4$ are given by \eqref{def:event.HD.E0-4}. By Lemma~\ref{lem:HD.max.Laplace}, for any $z\geq 0$ given, let the noise level $\lambda>0$ satisfy 
\begin{align}\label{bound.ro}
  \bigg( \frac{s T\lambda}{\epsilon} \bigg)^2 \{\log (p s T)+z\}^2  \log\bigg(\frac{T}{\delta}\bigg)   \leq c_0r_0^2
\end{align} 
for some constant $c_0>0$ to be determined. Then, with probability at least $1-2e^{-z}$, $\max_{t\in \{0\}\cup [T-1]}W_t \leq  (c_0 C_0) r_0^2$, where $C_0>0$ is the absolute constant appearing in Lemma~\ref{lem:HD.max.Laplace}. This conclusion, together with $\cE$, serves as a prerequisite for Proposition~\ref{prop:HD.convergence.region} below. Specifically, Proposition~\ref{prop:HD.convergence.region} shows that, starting with an initial value $\bbeta^{(0)} \in \bbeta^* + \BB^p(r_0)$, all subsequent iterates will remain within the local neighborhood under certain conditions.

\begin{proposition}  \label{prop:HD.convergence.region}

Assume that $\bbeta^{(0)} \in \bbeta^* + \BB^p(r_0)$, and let $\eta\in (0, 1)$.
Then, conditioned on $\cE$, we have $\|\bbeta^{(t)}-\bbeta^*\|_2 \leq r_0$ for all $t\in [T]$, provided that
$$
     s \geq  192s^* (1 + \eta^{-2})\frac{\phi_u^2}{\phi_l^2}  \,, \quad r_1  \leq  \frac{\phi_l r_0}{4} ~~\mbox{ and }~~ \max_{t\in \{0\}\cup [T-1]} W_t \leq \frac{\eta^2 \phi_l^2r_0^2}{512 \phi_u^2} \, ,
$$
where the quantity $r_1 $ is defined in \eqref{def:event.HD.E0-4}.
\end{proposition}

Conditioned on $\cE_0 \cap \cE_1 \cap \cE_2$, Proposition \ref{prop:HD:excess_loss} below provides an iterative bound on the excess risk, $\hat\cL_\tau(\bbeta^{(t)})-\hat\cL_\tau(\bbeta^*)$, which is a key step in controlling the statistical error.
\begin{proposition}\label{prop:HD:excess_loss}
    Assume that $\bbeta^{(t)} \in \mathbb{H}(s)\cap \Theta(r_0)$ for some $r_0>0$, and that the learning rate satisfies $\eta_0 = \eta/(2\phi_u)$ for some $\eta\in(0,1)$. Then, conditioned on the event $\cE_0\cap\cE_1 \cap \cE_2$, the $(t+1)$-th iterate $\bbeta^{(t+1)}$ satisfies 
        \begin{align*}
            \hat \cL_\tau(\bbeta^{(t+1)})-\hat \cL_\tau(\bbeta^{*})\leq \bigg( 1- \frac{4\eta s^*}{s}  -\frac{\eta^2\phi_l}{8\phi_u}  \bigg)    \{\hat \cL_\tau(\bbeta^{(t)})-\hat \cL_\tau(\bbeta^{*}) \} + C_{\eta}W_t
            \end{align*}
    provided that
    $s \geq  16s^* \max \{8\phi_u^2/(3\eta^2\phi_l^2),  \eta/(1-\eta) \}$, where $W_t =  \sum_{i=1}^s \|\bw^t_i\|_{\infty}^2+\|\widetilde{\bw}_{S^{t+1}}^t\|_2^2$ is defined in Lemma~\ref{lem:HD.max.Laplace}, and $C_{\eta} = 4\phi_u\{ 5+4\eta(1-\eta)^{-1} +6(1-\eta)\eta^{-1}\}$. 
\end{proposition} 

To apply Proposition~\ref{prop:HD.convergence.region}, we take $c_0 = (\eta \phi_l )^2/(512 C_0 \phi_u^2 )$ in 
\eqref{bound.ro}. By the definition of $\bbeta^{(t)}$ obtained in Algorithm \ref{alg:DPHuberhigh}, we know $\bbeta^{(t)}\in \mathbb{H}(s)$ for any $t\in[T]$. Then, with the initial estimate satisfying $\|\bbeta^{(0)} - \bbeta^*\|\leq r_0$, it follows from Propositions \ref{prop:HD.convergence.region} and   \ref{prop:HD:excess_loss} that 
\begin{align*}
\hat \cL_\tau(\bbeta^{(t+1)})-\hat \cL_\tau(\bbeta^*) \leq(1- \rho) \{ \hat \cL_\tau(\bbeta^{(t)})-\hat \cL_\tau( \bbeta^*) \} +C_{\eta}W_t\,, \quad t\in \{0\}\cup [T-1]\,, 
\end{align*}
where $\rho  =  4\eta s^* / s  +  \eta^2\phi_l / (8\phi_u)$ and $C_{\eta} = 4\phi_u\{5+4\eta(1-\eta)^{-1}+6(1-\eta)\eta^{-1}\}$.
Summing the above inequality over $t\in \{0\}\cup [T-1]$ yields 
\begin{align*}
\hat \cL_\tau(\bbeta^{(T)})-\hat \cL_\tau(\bbeta^*) \leq (1-\rho )^T \{ \hat \cL_\tau(\bbeta^{(0)})-\hat \cL_\tau( \bbeta^*) \} +C_{\eta}\sum_{t=0}^{T-1}(1-\rho )^{T-1-t}W_t\,.
\end{align*} 
Since $\| \bbeta^{(0)} - \bbeta^* \|_2 \leq r_0$, on the event $\cE_2\cap \cE_4$ with $r_1 \leq \phi_l r_0/4$, we have 
\begin{align*}
\hat \cL_\tau(\bbeta^{(0)})-\hat \cL_\tau( \bbeta^*)&\leq \langle \nabla \hat \cL_\tau(\bbeta^*),\bbeta^{(0)}-\bbeta^*\rangle+\phi_u\|\bbeta^{(0)}-\bbeta^*\|_2^2\\
&\leq \sup_{S\in \cS} \big\| \{\nabla\hat \cL_\tau (\bbeta^*) \}_S \big\|_2 \cdot \|\bbeta^{(0)}-\bbeta^*\|_2  + \phi_u \|\bbeta^{(0)}-\bbeta^*\|_2^2 \\
& \leq  r_0 r_1 +  \phi_u r_0^2   \leq (\phi_u  + \phi_l/4 ) r_0^2 \,.
\end{align*}
Consequently,
\begin{align*}
\hat \cL_\tau(\bbeta^{(T)})-\hat \cL_\tau(\bbeta^*) 
& \leq   (1-\rho )^T  (\phi_u  + \phi_l/4 ) r_0^2 + C_{\eta}\sum_{t=0}^{T-1}(1-\rho )^{T-1-t}W_t\,.
\end{align*}
On the other hand, on the event $\cE_1$ it holds
\begin{align}
  \hat \cL_\tau(\bbeta^{(T)})-\hat \cL_\tau(\bbeta^*) & \geq \langle \nabla \hat \cL_\tau(\bbeta^*),\bbeta^{(T)}-\bbeta^*\rangle + \phi_l\|\bbeta^{(T)}-\bbeta^*\|_2^2 \nn \\
  & \geq -\sup_{S\in \cS} \big\| \{\nabla\hat \cL_\tau (\bbeta^*) \}_S \big\|_2 \cdot  \|\bbeta^{(T)}-\bbeta^* \|_2 + \phi_l \|\bbeta^{(T)}-\bbeta^*\|_2^2 \nn \\
  & \geq  - r_1  \|\bbeta^{(T)}-\bbeta^* \|_2 + \phi_l \|\bbeta^{(T)}-\bbeta^*\|_2^2 \,. \nn 
\end{align}
Combining the above upper and lower bounds yields
\begin{align}
     \phi_l \|\bbeta^{(T)}-\bbeta^*\|_2^2  
    \leq &~ r_1  \|\bbeta^{(T)}-\bbeta^*\|_2  +  (1-\rho)^T  (\phi_u  + \phi_l/4 )  r_0^2 + C_\eta  \sum_{t=0}^{T-1}(1-\rho )^{T-1-t}W_t  \nn \\
    \leq & ~\frac{\phi_l}{2}\|\bbeta^{(T)}-\bbeta^*\|_2^2 + \frac{r_1^2}{2\phi_l} +     (1-\rho)^T  (\phi_u  + \phi_l/4 )  r_0^2 \nn\\
   &~+ C_\eta  \max_{t\in \{0\}\cup [T-1]} W_t \cdot  \sum_{k=0}^{\infty}(1-\rho )^k  \,. \nn 
\end{align}
Let $T\geq 2\log\{  r_0    n \epsilon (\sigma_0\log p )^{-1}\}/\log\{(1-\rho)^{-1}\}$ such that $(1-\rho)^T  r_0^2 \leq \{({ \epsilon n})^{-1}{\sigma_0\log p}\}^2$. Then,
\begin{align*}
   \|\bbeta^{(T)}-\bbeta^*\|_2^2 \leq \frac{r_1^2}{ \phi_l^2}   + 
 \bigg(\frac{2\phi_u}{\phi_l}+\frac{1}{2}\bigg) \bigg(\frac{\sigma_0\log p}{ n\epsilon}\bigg)^2 + \frac{2C_\eta}{\phi_l\rho } \max_{t\in \{0\}\cup [T-1]} W_t \,.
\end{align*}
This completes the proof of the theorem. \qed 

\subsubsection{Proof of Proposition \ref{prop:events2}} \label{proof.e2}
Let $\bu = \bSigma^{1/2}\be_j/\|\bSigma^{1/2}\be_j\|_2\in \mathbb{S}^{p-1}$, where   $\be_j$ is the unit vector with 1 in the $j$-th position and 0 elsewhere. Then by Assumption \ref{assump:design}, due to $\|\bSigma^{1/2}\be_j\|_2 \leq \lambda_1^{1/2}$, for any $z\geq 0$, 
 \begin{align}\label{tail.xij}
    \PP(|x_{i,j}|\geq   \upsilon_1 \lambda_1^{1/2} z) \leq \PP(|x_{i,j}|\geq \|\bSigma^{1/2}\be_j\|_2 \upsilon_1 z)= \PP(| \langle \bu, \bSigma^{-1/2} \bx_i\rangle | \geq \upsilon_1 z ) \leq 2 e^{-z^2/2}\,.
 \end{align} 
 Taking the union bound over $i\in [n]$ and $j\in [p]$, we have
$$
\mathbb{P}\bigg\{\max_{i\in[n]}\|\bx_i\|_{\infty} \geq   \upsilon_1\lambda_1^{1/2}\sqrt{2 \log(2np)+2 z}\bigg\} \leq e^{-z} \,.
$$
The last four bounds follow directly from Lemmas \ref{lem:HD.rsc}, \ref{lem:HD.gs}, \ref{lem:HD.huber.grad_diff.ubd}, and \ref{lem:HD.huber.grad_oracle.ubd}. This completes the proof of Proposition~\ref{prop:events2}. $\hfill\Box$

\subsubsection{Proof of Proposition~\ref{prop:HD.convergence.region}}

We provide a proof by induction to show that $\|\bbeta^{(t)} - \bbeta^*\|_2 \leq r_0$ for all $t \geq 0$. By assumption,  the base case $\|\bbeta^{(0)}-\bbeta^*\|_2 \leq r_0$ holds. For the inductive step, assume that $ \|\bbeta^{(t)}-\bbeta^*\|_2 \leq r_0$ for some $t\geq 0$. We aim to prove that $\|\bbeta^{(t+1)}-\bbeta^*\|_2 \leq r_0$ also holds. To this end, we bound $\hat \cL_\tau(\bbeta^{(t+1)}) - \hat \cL_\tau(\bbeta^*)$ from above and below in terms of $\| \bbeta^{(t+1)} - \bbeta^* \|_2$.

We begin by bounding  $\hat \cL_\tau(\bbeta^{(t+1)}) -  \hat \cL_\tau (\bbeta^*)$ from above. From \eqref{def:event.HD.E0-4}, it follows that
\begin{align*}
    \hat\cL_\tau(\bbeta^{(t+1)})-\hat\cL_\tau(\bbeta^{(t)})\leq&\, \langle\nabla\hat\cL_\tau(\bbeta^{(t)}),\bbeta^{(t+1)}-\bbeta^{(t)}\rangle+\phi_u\|\bbeta^{(t+1)}-\bbeta^{(t)}\|_2^2\,,\\
    \hat\cL_\tau(\bbeta^{(t)})-\hat\cL_\tau(\bbeta^*) \leq& -\langle\nabla\hat\cL_\tau(\bbeta^{(t)}),\bbeta^*-\bbeta^{(t)}\rangle-\phi_l\|\bbeta^{(t)}-\bbeta^*\|_2^2\,.
\end{align*} 
Together, these two bounds imply 
\begin{align} \hat\cL_\tau(\bbeta^{(t+1)})-\hat\cL_\tau(\bbeta^*)  \leq&~\langle\nabla\hat\cL_\tau(\bbeta^{(t)}),\bbeta^{(t+1)}-\bbeta^*\rangle+\phi_u\|\bbeta^{(t+1)}-\bbeta^{(t)}\|_2^2-\phi_l\|\bbeta^{(t)}-\bbeta^*\|_2^2 \nn  \\
        \leq&~ \phi_u\|\bbeta^{(t+1)}-\bbeta^{(t)}\|_2^2-\phi_l\|\bbeta^{(t)}-\bbeta^*\|_2^2 + \eta_0^{-1}\langle\bbeta^{(t)}-\bbeta^{(t+1)},\bbeta^{(t+1)}-\bbeta^*\rangle \nn \\
        &+ \eta_0^{-1}\langle\eta_0\nabla\hat\cL_\tau(\bbeta^{(t)})-\bbeta^{(t)}+\bbeta^{(t+1)},\bbeta^{(t+1)}-\bbeta^*\rangle \label{HD.conv.region.init} \\
        \leq&~ \{\phi_u-(2\eta_0)^{-1}\}\|\bbeta^{(t+1)}-\bbeta^{(t)}\|_2^2+\{(2\eta_0)^{-1}-\phi_l\} \|\bbeta^{(t)}-\bbeta^*\|_2^2 \nn \\
        & - (2\eta_0)^{-1}\|\bbeta^{(t+1)} - \bbeta^*\|_2^2 +\eta_0^{-1}\langle\bbeta^{(t+1)} -\tilde\bbeta^{(t+1)} , \bbeta^{(t+1)} - \bbeta^*\rangle\,,  \nn
\end{align}
where $\tilde\bbeta^{(t+1)}$ is specified in Algorithm \ref{alg:DPHuberhigh}. In the last step, we used the facts that $\widetilde{\bbeta}^{(t+1)} = \bbeta^{(t)} - \eta_0 \nabla\hat\cL_\tau(\bbeta^{(t)})$ conditioning on $\cE_0$, and 
\begin{align*}
  \langle\bbeta^{(t)}- \bbeta^{(t+1)} , \bbeta^{(t+1)} -\bbeta^*\rangle  = \frac{1}{2} \big\{ \|\bbeta^{(t)}-\bbeta^*\|_2^2 -  \|\bbeta^{(t+1)} -\bbeta^*\|_2^2 -  \|\bbeta^{(t+1)} -\bbeta^{(t)}\|_2^2 \big\} \,.
\end{align*} 
 
Taking $\eta_0 = \eta/(2\phi_u)$ for some $\eta\in(0, 1)$,  \eqref{HD.conv.region.init} implies
\begin{align} \label{HD.conv.region.init1}
         \hat\cL_\tau(\bbeta^{(t+1)})-\hat\cL_\tau(\bbeta^*)
        & \leq (\phi_u\eta^{-1}-\phi_l) r_0^2 - \phi_u\eta^{-1}\|\bbeta^{(t+1)} - \bbeta^*\|_2^2\nn\\ &~~~~ +2\phi_u\eta^{-1}\underbrace{\langle\bbeta^{(t+1)} -\tilde\bbeta^{(t+1)} , \bbeta^{(t+1)} - \bbeta^*\rangle}_{\Pi_2}\,.
\end{align} 
Write $S^{*}={\rm{supp}}(\bbeta^{*})$, and denote by $S^{t+1}$  the index set $S$ determined by Algorithm \ref{alg:NHT} when $\bv$ is selected as $\tilde\bbeta^{(t+1)}$. It is obvious that ${\rm{supp}}(\bbeta^{(t+1)})\subseteq S^{t+1}$. Since $\bbeta^{(t+1)} =\tilde P_s(\tilde\bbeta^{(t+1)} )+\tilde \bw^t_{S^{t+1}}$, it follows that 
\begin{align} \label{HD.conv.region.sub1}
      \Pi_2 & =\langle\tilde P_s(\tilde\bbeta^{(t+1)}) -\tilde\bbeta^{(t+1)},\bbeta^{(t+1)} - \bbeta^*\rangle + \langle \tilde \bw^t_{S^{t+1}} ,\bbeta^{(t+1)} - \bbeta^*\rangle\nn\\
        &\leq \frac{8\phi_u}{\eta\phi_l} \big\| \{\tilde P_s(\tilde\bbeta^{(t+1)}) -\tilde\bbeta^{(t+1)} \}_{S^{t+1}\cup S^*}\big\|_2^2 + \frac{\eta \phi_l}{16 \phi_u}  \|\bbeta^{(t+1)}-\bbeta^*\|_2^2 + \frac{8\phi_u}{\eta \phi_l} \|\tilde \bw^t_{S^{t+1}} \|_2^2\,.
\end{align}

Recall that $|S^*|=s^*$. By Lemma \ref{lem.cons}, for any constant $c>0$, it holds that
\begin{align*}
&\big\| \{\tilde P_s(\tilde\bbeta^{(t+1)}) -\tilde\bbeta^{(t+1)} \}_{S^{t+1}\cup S^*}\big\|_2^2\\
&~~~\leq (1+c^{-1}) \frac{|S^{t+1}\cup S^*|-s}{|S^{t+1}\cup S^*|-s^*}   \big\|\{\bbeta^*-\tilde \bbeta^{(t+1)}\}_{S^{t+1}\cup S^*}\big\|_2^2+4(1+c)\sum_{i=1}^s\|\bw_i^t\|_\infty^2\\
&~~~\leq (1+c^{-1} ) \frac{2s^*}{s} \Big[ \|\bbeta^*-\bbeta^{(t)} \|_2^2+\big\|(2\phi_u)^{-1}\eta \, \{\nabla \hat \cL_\tau (\bbeta^{(t)})\}_{S^{t+1}\cup S^*}\big\|_2^2\Big]+4(1+c)\sum_{i=1}^s\|\bw_i^t\|_\infty^2\,,
\end{align*} 
where we used the properties that  $|S^{t+1}\cup S^*|\leq s+s^*$  and $\tilde\bbeta^{(t+1)} = \bbeta^{(t)} - (2\phi_u)^{-1}\eta \, \nabla \hat \cL_\tau (\bbeta^{(t)})$ conditioned on $\cE_0$.  Conditioned further on $\cE_3\cap \cE_4$ with $r_1 \leq \phi_l r_0/4$, we have
\begin{align*}
& \big\| \{\nabla \hat \cL_\tau(\bbeta^{(t)}) \}_{S^{t+1}\cup S^*}\big\|_2  \leq   \frac{3}{2}\phi_u r_0  +  r_1  <  2 \phi_u r_0 \, . 
\end{align*} 
Combining the last two displays with the assumption $\|\bbeta^*-\bbeta^{(t)}\|_2\leq r_0$, we obtain
\begin{align*}
        \big\|\{\tilde P_s(\tilde\bbeta^{(t+1)}) -\tilde\bbeta^{(t+1)}\}_{S^{t+1}\cup S^*}\big\|_2^2   \leq (1+c^{-1}) (1+ \eta^2)\frac{2s^*}{s} r_0^2    +4(1+c)\sum_{i=1}^s\|\bw_i^t\|_\infty^2\,.
\end{align*}
Plugging this bound into \eqref{HD.conv.region.sub1} yields 
\begin{equation}\label{HD.conv.region.sub3}
    \begin{aligned}
        \Pi_2 & \leq \frac{\eta \phi_l}{16 \phi_u}   \|\bbeta^{(t+1)}-\bbeta^*\|_2^2+ (1+c^{-1})(1+ \eta^2)\frac{s^*}{s}\frac{16\phi_u}{\eta\phi_l}   r_0^2 + 4(1+c) \frac{8\phi_u}{\eta\phi_l} W_t\,, 
    \end{aligned}
\end{equation} 
where $W_{t}=\sum_{i =1}^s\|\bw_{i}^t\|_{\infty}^{2}+\|\tilde{\bw}_{S^{t+1}}^t\|_{2}^{2}$.

By plugging \eqref{HD.conv.region.sub3} into \eqref{HD.conv.region.init1}, it follows that
\begin{align}\label{HD.conv.region.loss.ubd}
        \hat \cL_\tau(\bbeta^{(t+1)}) - \hat \cL_\tau(\bbeta^*) 
        &\leq \bigg(\frac{\phi_l}{8}-\frac{\phi_u}{\eta}\bigg)   \|\bbeta^{(t+1)}-\bbeta^*\|_2^2+ 4(1+c)\frac{16\phi_u^2}{\eta^2\phi_l}   W_t\nn\\
        &~~~~+
        \bigg\{\frac{\phi_u}{\eta}-\phi_l+(1+c^{-1})(1+ \eta^2)\frac{s^*}{s}\frac{32\phi_u^2}{\eta^2\phi_l}   \bigg\}  r_0^2   \,.  
\end{align} 
On the other hand, the convexity of $\hat \cL_\tau(\bbeta)$ implies that, given $\cE_4$ holds,
\begin{equation}\label{HD.conv.region.loss.lbd}
    \begin{aligned}
        \hat \cL_\tau(\bbeta^{(t+1)}) -  \hat \cL_\tau (\bbeta^*) & \geq  \langle \nabla \hat \cL_\tau(\bbeta^*),\bbeta^{(t+1)}-\bbeta^*  \rangle \geq -   |   \langle \nabla \hat \cL_\tau(\bbeta^*),\bbeta^{(t+1)}-\bbeta^*  \rangle   | \\
        & \geq - \frac{\phi_l}{8} \|\bbeta^{(t+1)}-\bbeta^*\|_2^2 - \frac{2}{\phi_l}r_1^2\,.
    \end{aligned}
\end{equation}
Combining \eqref{HD.conv.region.loss.ubd} with \eqref{HD.conv.region.loss.lbd}, and rearranging the terms, it follows that 
\begin{align*}
      & \bigg(\frac{\phi_u}{\eta} - \frac{\phi_l}{4}\bigg)  \cdot \big(\|\bbeta^{(t+1)}-\bbeta^*\|_2^2 - r_0^2\big) \\ 
      &~~~ \leq  \bigg\{ -\frac{3\phi_l}{4} +(1+c^{-1})(1+ \eta^2)\frac{s^*}{s}\frac{32\phi_u^2}{\eta^2\phi_l}\bigg\} \cdot r_0^2  +   \frac{2r_1^2}{\phi_l}    +  4(1+c)\frac{16\phi_u^2}{\eta^2\phi_l} \cdot W_t \\ 
    & ~~~\leq \bigg\{ - \frac{5\phi_l}{8} +(1+c^{-1})(1+ \eta^2)\frac{s^*}{s}\frac{32\phi_u^2}{\eta^2\phi_l}\bigg\} \cdot r_0^2  +  4(1+c)\frac{16\phi_u^2}{\eta^2\phi_l} \cdot W_t\, .
\end{align*}
Provided that $s\geq 128(1+c^{-1}) (1 +  \eta^{-2}) (\phi_u/\phi_l)^2 s^*$ and
$$
\max_{t\in \{0\}\cup [T-1]}W_t \leq \frac{3 \eta^2 \phi_l^2 }{512(1+c) \phi_u^2 } r_0^2 \,,
$$ 
the right-hand side of the above inequality is less than or equals to zero.

Finally, by setting $c=2$, we obtain $\|\bbeta^{(t+1)}-\bbeta^*\|_2 \leq r_0$, provided that $s \geq 192(1+\eta^{-2})(\phi_u/\phi_l)^2 s^*$, $r_1\leq  \phi_l r_0/4$ and $\max_{t\in \{0\}\cup [T-1]}W_t \leq   (\eta \phi_l/\phi_u)^2 r_0^2/512$.   This completes the inductive step and, consequently, the proof.
\qed
	
\subsubsection{Proof of Proposition~\ref{prop:HD:excess_loss}}

For ease of notation, we write $S^{*}=\rm{supp}(\bbeta^{*})$, and denote by $S^{t}$ and $S^{t+1}$, respectively,  the index set $S$ determined by Algorithm \ref{alg:NHT} when $\bv$ is selected as $\tilde\bbeta^{(t)}$ and $\tilde\bbeta^{(t+1)}$. Moreover, let $I^{t}=S^t\cup S^{t+1}\cup S^*$ and $\boldsymbol{g}^{(t)}=\nabla \hat\cL_\tau(\bbeta^{(t)})$. Under $\cE_2$, it follows that 
\begin{align*}
      \hat \cL_\tau(\bbeta^{(t+1)})-\hat \cL_\tau(\bbeta^{(t)}) 
    &\leq \langle \bbeta^{(t+1)}-\bbeta^{(t)},\boldsymbol{g}^{(t)}\rangle+\phi_u\| \bbeta^{(t+1)}-\bbeta^{(t)}\|_2^2\\
    &=\phi_u\big\| \{\bbeta^{(t+1)}-\bbeta^{(t)}+\eta_0\boldsymbol{g}^{(t)} \}_{I^t}\big\|_2^2-\phi_u\eta_0^2\big\| \boldsymbol{g}^{(t)}_{I^t}\big\|_2^2\\
    &~~~+(1-2\eta_0\phi_u)\underbrace{\langle \bbeta^{(t+1)}-\bbeta^{(t)},\boldsymbol{g}^{(t)}\rangle}_{\Pi_3}\,.
\end{align*}
Lemma \ref{lem:HD:prop:1} provides an upper bound of $\Pi_3$, whose proof is given in Section \ref{proof.lem:HD:prop:1}.
\begin{lemma} \label{lem:HD:prop:1}
    Conditioned on the event $\cE_0$, we have
    $$
        \langle\bbeta^{(t+1)}-\bbeta^{(t)}, \boldsymbol{g}^{(t)}\rangle \leq -\frac{\eta_0}{4}\big\|\boldsymbol{g}^{(t)}_{S^{t+1} \cup S^t}\big\|_2^2+  \frac{12}{\eta_0} W_t\,, 
    $$ 
    for $t\in \{0\}\cup [T-1]$, where $W_t = \sum_{i=1}^s\|\bw^t_i\|_{\infty}^2+\|\widetilde{\bw}_{S^{t+1}}^t\|_2^2 $ is defined in Lemma~\ref{lem:HD.max.Laplace}. In the noiseless case where $W_t = 0$, it holds that
    $$
     \langle\bbeta^{(t+1)}-\bbeta^{(t)}, \boldsymbol{g}^{(t)}\rangle \leq -\frac{\eta_0}{2}\big\|\boldsymbol{g}^{(t)}_{S^{t+1} \cup S^t}\big\|_2^2\, .
    $$
\end{lemma}
Lemma~\ref{lem:HD:prop:1}, together with $\|\boldsymbol{g}^{(t)}_{I^t}\|_2^2=\|\boldsymbol{g}^{(t)}_{I^t\setminus (S^t\cup S^*)}\|_2^2+\| \boldsymbol{g}^{(t)}_{S^t\cup S^*}\|_2^2$, yields
\begin{align*}
     \hat \cL_\tau(\bbeta^{(t+1)})-\hat \cL_\tau(\bbeta^{(t)}) 
    \leq&\, \underbrace{\phi_u\big\| \{\bbeta^{(t+1)}-\bbeta^{(t)}+\eta_0\boldsymbol{g}^{(t)} \}_{I^t}\big\|_2^2-\phi_u\eta_0^2\big\| \boldsymbol{g}^{(t)}_{I^t\setminus (S^t\cup S^*)}\big\|_2^2}_{\Pi_4}\\
    & -\phi_u\eta_0^2\big\| \boldsymbol{g}^{(t)}_{S^t\cup S^*}\big\|_2^2-(1-2\phi_u\eta_0) \frac{\eta_0}{4}\big\| \boldsymbol{g}^{(t)}_{S^t\cup S^{t+1}}\big\|_2^2+(1-2\phi_u\eta_0)\frac{12}{\eta_0} W_t\,.   
\end{align*}
To upper bound $\Pi_4$, we introduce the following lemma, whose proof is given in Section \ref{proof.lem:HD:prop:2}.
\begin{lemma}\label{lem:HD:prop:2} 
Conditioned on the event $\cE_0$, for any $c_1 >0$, we have
    \begin{align*}
        & \big\|\{\bbeta^{(t+1)}-\bbeta^{(t)}+\eta(2\phi_u)^{-1} \boldsymbol{g}^{(t)}\}_{I^t}\big\|_2^2 - \eta_0^2\big\|\boldsymbol{g}^{(t)}_{I^t \setminus (S^t \cup S^*)}\big\|_2^2 \\ 
        & ~~~ \leq   (4 s^*/s)\big\|\{\bbeta^* - \bbeta^{(t)} + \eta_0 \boldsymbol{g}^{(t)}\}_{I^t}\big\|_2^2 + c_1 \eta_0^2\big\|\boldsymbol{g}^{(t)}_{I^t \setminus(S^t \cup S^*)}\big\|_2^2 + 4 (5+c_1^{-1})W_t
    \end{align*} 
     for $t\in \{0\}\cup [T-1]$, where $W_t = \sum_{i=1}^s \|\bw^t_i\|_{\infty}^2+\|\widetilde{\bw}_{S^{t+1}}^t\|_2^2  $,  is defined in Lemma~\ref{lem:HD.max.Laplace}. In the noiseless case where $W_t = 0$, it holds that
      \begin{align*}
         \big\|\{\bbeta^{(t+1)}-\bbeta^{(t)}+\eta(2\phi_u)^{-1} \boldsymbol{g}^{(t)}\}_{I^t}\big\|_2^2- \eta_0^2\big\|\boldsymbol{g}^{(t)}_{I^t \setminus (S^t \cup S^*)}\big\|_2^2  \leq   \frac{s^*}{s}\big\| \{\bbeta^* - \bbeta^{(t)} + \eta_0 \boldsymbol{g}^{(t)}\}_{I^t}\big\|_2^2\, . 
    \end{align*} 
\end{lemma} 
It follows from Lemma~\ref{lem:HD:prop:2} that  
\begin{align}\label{HD.prop1.sub0}
     \hat \cL_\tau(\bbeta^{(t+1)})-\hat \cL_\tau(\bbeta^{(t)}) 
        & \leq 4\phi_u (s^*/s)\underbrace{\big\| \{\bbeta^* - \bbeta^{(t)} +\eta_0 \boldsymbol{g}^{(t)} \}_{I^t}\big\|_2^2}_{\Pi_5} \nn\\
        &~~~+ c_1\phi_u\eta_0^2\big\|\boldsymbol{g}^{(t)}_{I^t \setminus(S^t \cup S^*)}\big\|_2^2 -\phi_u\eta_0^2\big\| \boldsymbol{g}^{(t)}_{S^t\cup S^*}\big\|_2^2 -(1-2\phi_u\eta_0)\frac{\eta_0}{4}\big\|\boldsymbol{g}^{(t)}_{S^t\cup S^{t+1}}\big\|_2^2\nn\\
        &~~~+ \bigg\{ 4\phi_u(5+c_1^{-1})+(1-2\phi_u\eta_0)\frac{12}{\eta_0}\bigg\}\cdot W_t\,.  
\end{align} 
Next, we aim to bound $\Pi_5$ from above. Under the assumption that  $\bbeta^{(t)} \in  \bbeta^* + \BB^p(r_0)$, by the local strong convexity condition in \eqref{def:event.HD.E0-4}, 
\begin{align*}
    \Pi_5 &=2\eta_0\langle \bbeta^*-\bbeta^{(t)}, \boldsymbol{g}^{(t)}  \rangle+
     \|\bbeta^*-\bbeta^{(t)} \|_2^2+\eta_0^2\big\| \boldsymbol{g}^{(t)}_{I^t}\big\|_2^2\\
    & \leq 2\eta_0 \{ \hat \cL_\tau(\bbeta^{*})-\hat \cL_\tau(\bbeta^{(t)})  \} +
    (1-2\eta_0\phi_l) \|\bbeta^*-\bbeta^{(t)} \|_2^2+\eta_0^2\big\| \boldsymbol{g}^{(t)}_{I^t}\big\|_2^2\,. 
\end{align*}
Combining the upper bound of $\Pi_5$ with \eqref{HD.prop1.sub0} yields 
\begin{equation}\label{HD.prop1.sub1}
    \begin{aligned}
        & \hat \cL_\tau(\bbeta^{(t+1)})-\hat \cL_\tau(\bbeta^{(t)}) \\
        &~~~\leq 8\phi_u(s^*/s)\eta_0 \cdot \{\hat \cL_\tau(\bbeta^{*})-\hat \cL_\tau(\bbeta^{(t)})\}+
        4\phi_u(s^*/s)(1-2\eta_0\phi_l) \cdot \|\bbeta^*-\bbeta^{(t)} \|_2^2 \\
        &~~~~~~~+ 4\phi_u(s^*/s)\eta_0^2\cdot \big\| \boldsymbol{g}^{(t)}_{I^t}\big\|_2^2 + c_1\phi_u\eta_0^2\cdot\big\| \boldsymbol{g}^{(t)}_{I^t\setminus(S^t\cup S^*)}\big\|_2^2\\
        &~~~~~~~-\phi_u\eta_0^2\cdot\big\|  \boldsymbol{g}^{(t)}_{S^t\cup S^*}\big\|_2^2-(1-2\phi_u\eta_0)(\eta_0/4)\cdot\big\| \boldsymbol{g}^{(t)}_{S^t\cup S^{t+1}}\big\|_2^2\\
        &~~~~~~~+\{ 4\phi_u(5+c_1^{-1})+(1-2\phi_u\eta_0 )(12/\eta_0)\}\cdot W_t\,. 
    \end{aligned}
\end{equation}
Note that $\| \boldsymbol{g}^{(t)}_{I^{t}} \|_2^2= \|\boldsymbol{g}^{(t)}_{I^{t}\setminus (S^t\cup S^*)} \|_2^2+ \| \boldsymbol{g}^{(t)}_{S^t\cup S^*} \|_2^2$ and $\| \boldsymbol{g}^{(t)}_{S^t\cup S^{t+1}} \|_2^2\geq  \|\boldsymbol{g}^{(t)}_{I^{t}\setminus (S^t\cup S^*)} \|_2^2$.
By taking $\eta_0 = \eta/(2\phi_u)$ for some $\eta\in (0, 1)$ and $c_1 = (1- \eta )/(4 \eta)$, \eqref{HD.prop1.sub1} becomes 
\begin{align*}
     \hat \cL_\tau(\bbeta^{(t+1)})-\hat \cL_\tau(\bbeta^{(t)})
         \leq &~ \frac{4 \eta s^*}{s} \{\hat \cL_\tau(\bbeta^{*})-\hat \cL_\tau(\bbeta^{(t)})\} + \frac{4 s^*}{s}   ( \phi_u -   \eta \phi_l ) \| \bbeta^* - \bbeta^{(t)} \|_2^2  \\
     & ~ + \frac{\eta^2}{\phi_u} \bigg(   \frac{  s^* }{ s } - \frac{  1-\eta }{ 16  \eta} \bigg)  \big\|\boldsymbol{g}^{(t)}_{I^t\setminus (S^t\cup S^*)} \big\|_2^2  + \frac{\eta^2}{\phi_u} \bigg( \frac{  s^*}{  s } - \frac{1}{4 } \bigg) \big\| \boldsymbol{g}^{(t)}_{S^t\cup S^*}\big\|_2^2 \\
     & ~+ C_\eta W_t \,,
\end{align*}
where $C_{\eta} = 4\phi_u\{ 5+4\eta(1-\eta)^{-1}+6(1-\eta)\eta^{-1} \}$. Under the assumption that
$$
    s \geq s^*\max\big\{128\phi_u^2/(3\eta^2\phi_l^2), 16\eta/(1-\eta) \big\}\,,  
$$
we further have
\begin{equation} \label{HD.prop1.sub2}
    \begin{aligned}
    \hat \cL_\tau(\bbeta^{(t+1)})-\hat \cL_\tau(\bbeta^{(t)}) 
    & \leq
     \frac{4 \eta s^*}{s}  \{\hat \cL_\tau(\bbeta^{*})-\hat \cL_\tau(\bbeta^{(t)})\}  \\
    &~~~~-  \frac{\eta^2}{8 \phi_u} \underbrace{\bigg( \big\| \boldsymbol{g}^{(t)}_{S^t\cup S^*}\big\|_2^2- \frac{3\phi_l^2}{4} \|\bbeta^*-\bbeta^{(t)} \|_2^2\bigg)}_{\Pi_6} + \, C_{\eta} W_t \,.
\end{aligned}
\end{equation}

It remains to lower bound the term $\Pi_6$. Under the assumption $\bbeta^{(t)} \in 
\bbeta^* + \BB^p(r_0)$, and conditioned on event $\cE_1$, we have
\begin{align*}
     \hat \cL_\tau(\bbeta^{(t)})-\hat \cL_\tau(\bbeta^{*}) 
     \leq&~ \langle \boldsymbol{g}^{(t)}, \bbeta^{(t)}-\bbeta^{*}\rangle-\phi_l \| \bbeta^{(t)}-\bbeta^{*} \|_2^2\\
    \leq &~\frac{1}{\phi_l} \big\| \boldsymbol{g}^{(t)}_{S^t\cup S^*}\big\|_2^2+  \frac{\phi_l}{4} \| \bbeta^{(t)}-\bbeta^{*} \|_2^2 - \phi_l \| \bbeta^{(t)}-\bbeta^{*} \|_2^2\\
    \leq &~   \frac{1}{\phi_l} \bigg( \big\| \boldsymbol{g}^{(t)}_{S^t\cup S^*}\big\|_2^2- \frac{3\phi_l^2}{4} \|\bbeta^*-\bbeta^{(t)} \|_2^2 \bigg) \,, 
\end{align*} 
which, in turn, implies $\Pi_6 \geq \phi_l \{\hat \cL_\tau(\bbeta^{(t)})-\hat \cL_\tau(\bbeta^*) \}$. Plugging this lower bound of $\Pi_6$ into \eqref{HD.prop1.sub2} yields
\begin{align*}
    \hat \cL_\tau(\bbeta^{(t+1)})-\hat \cL_\tau(\bbeta^{*})\leq \bigg( 1- \frac{4 \eta s^* }{s} - \frac{\eta^2\phi_l}{8 \phi_u} \bigg)  \{ \hat \cL_\tau(\bbeta^{(t)})-\hat \cL_\tau(\bbeta^{*}) \} +C_{\eta}W_t\,, 
\end{align*} 
as claimed.  \qed

\section{Proofs of Technical Lemmas}\label{proof.lemmas}

\subsection{Proof of Lemma~\ref{lem:gs}}\label{proof.lem:gs}

Define $\hat \bS = n^{-1} \sn \bchi_i \bchi_i^\T$ and $\bchi_i = \bSigma^{-1/2} \bx_i$, such that $\EE(\hat \bS) = \bI_p$. By \eqref{global.smootheness}, it suffices to show that 
	$
	\lambda_{\max}( \hat \bS ) \leq 2 
	$
	holds with probability at least $1-e^{-z}$. By Weyl's inequality and the proof of Theorem 6.5 in \citeS{wainwright2019highS}, it follows that with probability at least $1-e^{-z}$,
	\begin{align*}
	\lambda_{\max}( \hat \bS ) - 1 \leq \|\hat \bS - \bI\|\leq  C\upsilon_1^2\bigg(\sqrt{\frac{p+z}{n}}+\frac{p+z}{n}\bigg) 
	\end{align*}
	for some universal constant $C>0$. Therefore, when $n\gtrsim  \upsilon_1^4(p+z)$, we know that $\lambda_{\max}( \hat \bS )\leq 2$ with probability at least $1-e^{-z}$. \qed

\subsection{Proof of Lemma~\ref{lem:rsc}}\label{proof.lem:rsc}

For each pair 
$(\bbeta_1, \bbeta_2)$ of parameters, set $\bdelta=\bbeta_1-\bbeta_2$ and note that
\begin{align*}
	& \hat \cL_\tau(\bbeta_1)  - \hat \cL_\tau (\bbeta_2) - \langle \nabla \hat \cL_\tau(\bbeta_2), \bbeta_1-\bbeta_2 \rangle   \\ 
	&~~~ =  \int_0^1 \langle \nabla \hat \cL_\tau (\bbeta_2 + u \bdelta) -\nabla  \hat \cL_\tau(\bbeta_2), \bdelta \rangle \,{\rm d} u \\ 
	& ~~~=  \frac{1}{n} \sn \int_0^1  \big\{ \psi_\tau(  y_i - \langle  \bx_i ,  \bbeta_2 \rangle ) -\psi_\tau( y_i -  \langle \bx_i, \bbeta_2+u\bdelta \rangle  ) \big\}  \langle \bx_i, \bdelta \rangle \,{\rm d} u \\
	&~~~ \geq  \frac{1}{n} \sn \int_0^1  \big\{ \psi_\tau(  y_i - \langle  \bx_i ,  \bbeta_2 \rangle ) -\psi_\tau( y_i -  \langle \bx_i, \bbeta_2+u\bdelta \rangle  ) \big\}  \langle \bx_i, \bdelta \rangle  \mathbbm{1}(\cF_i  ) \, {\rm d} u \, ,
\end{align*}
where the last inequality follows from the fact that $\{\psi_{\tau}(a)-\psi_{\tau}(b)\}(a-b)\geq 0$ for any $a,b\in\RR$, and $\cF_i := \{ |\varepsilon_i | \leq \tau/4\} \cap \{ | \langle \bx_i, \bbeta_1 - \bbeta^* \rangle | \leq \tau/4 \} \cap \{ | \langle \bx_i, \bbeta_1 -\bbeta_2 \rangle | \leq  \tau_0 \| \bbeta_1 - \bbeta_2 \|_2/(2r) \}$ for some $0<\tau_0\leq \tau$ to be determined.
For each $i\in[n]$, it holds conditioned on $\cF_i$ that  $| y_i - \langle  \bx_i ,  \bbeta_2 \rangle | \leq \tau$ and $| y_i -  \langle \bx_i, \bbeta_2+u\bdelta \rangle| \leq \tau$ for all $\bbeta_1\in \Theta(r/2)$, $\bbeta_2 \in \bbeta_1 +\BB^p(r)$ and $u\in[0,1]$. Consequently,
\begin{align}
 & \hat \cL_\tau(\bbeta_1)  - \hat \cL_\tau (\bbeta_2) - \langle \nabla \hat \cL_\tau(\bbeta_2), \bbeta_1-\bbeta_2 \rangle\nn \\
 &~~~ \geq  \frac{1}{n} \sn   \int_0^1 u \,{\rm d} u  \cdot  \langle \bx_i, \bdelta \rangle^2 \mathbbm{1} ( \cF_i  ) =   \frac{1}{2 n} \sn    \langle \bx_i, \bdelta \rangle^2 \mathbbm{1}( \cF_i  ) \, .\label{Fi}
\end{align}
For any $R>0$, define  functions  $\varphi_R(u)=u^2 \mathbbm{1}(|u|\leq R/2) + \{ u\, {\rm{sign}}(u) -R \}^2 \mathbbm{1}( R/2 < |u| \leq R)$ and  $\phi_R(u)= \mathbbm{1}(|u|\leq R/2) + \{ 2 - (2u/R) {\rm{sign}}(u)\} \mathbbm{1} (R/2 < |u| \leq R)$, which are smoothed versions of $u\mapsto u^2 \mathbbm{1}(|u|\leq R)$ and $u\mapsto \mathbbm{1}(|u|\leq R)$.
Moreover,  note that $\varphi_{cR}(cu) = c^2 \varphi_{R}(u)$ for any $c > 0$ and $\varphi_0(u) = 0$. The right-hand side of the  inequality \eqref{Fi} can be further bounded from below by
\begin{align}
 & \frac{1}{2n} \sn   \chi_i  \varphi_{\| \bdelta \|_2  \tau_0/(2 r)} (\langle \bx_i, \bdelta \rangle  )   \phi_{\tau/4} (\langle \bx_i, \bbeta_1 - \bbeta^* \rangle )   \nn \\
 &~~~ = \frac{1}{2} \| \bdelta \|_2^2 \cdot 
  \underbrace{ \frac{1}{  n  } \sn   \chi_i   \varphi_{  \tau_0/(2r)} (\langle \bx_i, \bdelta /\| \bdelta \|_2 \rangle  )   \phi_{\tau /4} (\langle \bx_i, \bbeta_1 - \bbeta^* \rangle )  }_{ =: V_n(\bbeta_1, \bbeta_2) } \,, \label{def:Vn}
\end{align}
where $\chi_i = \mathbbm{1}(|\varepsilon_i|\leq \tau/4)$. We  will lower bound $\EE  \{V_n(\bbeta_1, \bbeta_2) \}$ and upper bound $-V_n(\bbeta_1, \bbeta_2)  + \EE \{ V_n(\bbeta_1, \bbeta_2) \}$, respectively, for $\bbeta_1 \in  \Theta(r/2)$ and $\bbeta_2 \in \bbeta_1 + \BB^p(r)$.  Letting $\bv =   \bdelta / \| \bdelta \|_2  \in \mathbb{S}^{p-1}$,  from Markov's inequality and Assumption \ref{assump:heavy-tail}, we see that
\begin{align*}
 \EE  \{V_n(\bbeta_1, \bbeta_2) \}& =\frac{1}{n}\sn\EE \big\{ \varphi_{\tau_0/(2 r)} ( \langle  \bx_i , \bv \rangle  )   \phi_{\tau /4} (\langle \bx_i, \bbeta_1 - \bbeta^*\rangle ) \cdot \chi_i \big\}  \\
 & \geq     \bigg( 1 - \frac{16\sigma_0^2}{\tau^2} \bigg)  \cdot \frac{1}{n}\sn \EE \bigg\{ \langle \bx_i, \bv \rangle^2  \mathbbm{1}\bigg(  | \langle \bx_i, \bv \rangle | \leq \frac{\tau_0}{4 r}\bigg)     \mathbbm{1}\bigg(  | \langle \bx_i, \bbeta_1 - \bbeta^* \rangle | \leq  \frac{\tau}{8} \bigg)\bigg\} \\ 
 & \geq    \bigg( 1 - \frac{16\sigma_0^2}{\tau^2} \bigg)   \cdot \frac{1}{n}\sn \EE  \bigg \{  \langle \bx_i, \bv \rangle^2 -     \langle \bx_i, \bv \rangle^2  \mathbbm{1}\bigg(  | \langle \bx_i, \bv \rangle | > \frac{\tau_0}{4 r}\bigg)  \\
 &~~~~~~~~~~~~~~~~~~~~~~~~~~~~~~~~~~~~~~~~~~~~-       \langle \bx_i, \bv \rangle^2  \mathbbm{1}\bigg(  | \langle \bx_i, \bbeta_1 - \bbeta^* \rangle | >  \frac{\tau}{8} \bigg)  
 \bigg \}  \, . 
\end{align*}
By the definition of $\kappa_4$ in \eqref{def:kappa4},
\begin{align*}
&\frac{1}{n}\sn \EE \bigg\{\langle \bx_i, \bv \rangle^2  \mathbbm{1}\bigg(  | \langle \bx_i, \bv \rangle | > \frac{\tau_0}{4 r}\bigg) \bigg\} \\
 &~~~ \leq \bigg(\frac{4r}{\tau_0}\bigg)^2\cdot \frac{1}{n}\sn \EE \big(\langle \bx_i, \bv \rangle^4\big) \leq \bigg(\frac{4r}{\tau_0}\bigg)^2   \kappa_4  \| \bv \|_{\bSigma}^4 \leq \kappa_4 \lambda_1 \bigg(\frac{4  r}{\tau_0}\bigg)^2 \| \bv \|_{\bSigma}^2 \,.
\end{align*}
Recall $\bbeta_1 \in  \Theta(r/2)$. Similarly, setting $\bs = (\bbeta_1 - \bbeta^*)/\|\bbeta_1 - \bbeta^*\|_2 \in \mathbb{S}^{p-1}$, it follows that
\begin{align*}
&\frac{1}{n} \sn \EE  \bigg\{\langle \bx_i, \bv \rangle^2  \mathbbm{1}\bigg(  | \langle \bx_i, \bbeta_1 - \bbeta^* \rangle | >  \frac{\tau}{8}\bigg)  \bigg\} 
= \frac{1}{n} \sn \EE  \bigg\{\langle \bx_i, \bv \rangle^2  \mathbbm{1}\bigg(  | \langle \bx_i, \bs \rangle | >  \frac{\tau}{8\|\bbeta_1 - \bbeta^*\|_2}\bigg)  \bigg\} \\
&~~~\leq   \bigg(\frac{4r}{\tau}\bigg)^2 \cdot \frac{1}{n}\sn \EE \big(\langle \bx_i, \bv \rangle^2\cdot \langle \bx_i, \bs \rangle^2\big)   \leq \kappa_4  \bigg(\frac{4r}{\tau}\bigg)^2 \| \bv \|_{\bSigma}^2 \| \bs \|_{\bSigma}^2   \leq \kappa_4  \lambda_1 \bigg(\frac{4r}{\tau}\bigg)^2 \| \bv \|_{\bSigma}^2 \,.
\end{align*}
Combining the above inequalities, we conclude that
\begin{align*}
 & \frac{1}{n}\sn \EE  \bigg \{  \langle \bx_i, \bv \rangle^2 -     \langle \bx_i, \bv \rangle^2  \mathbbm{1}\bigg(  | \langle \bx_i, \bv \rangle | > \frac{\tau_0}{4 r}\bigg)  -       \langle \bx_i, \bv \rangle^2  \mathbbm{1}\bigg(  | \langle \bx_i, \bbeta_1 - \bbeta^* \rangle | > \frac{\tau}{8} \bigg)  
 \bigg \}  \\
 &~~~ \geq \| \bv \|_{\bSigma}^2  \bigg\{  1 - \kappa_4 \lambda_1 \bigg(\frac{4  r}{\tau_0}\bigg)^2 - \kappa_4  \lambda_1 \bigg(\frac{4r}{\tau}\bigg)^2 \bigg\} \geq \lambda_p \bigg\{  1 - 2 \kappa_4 \lambda_1 \bigg(\frac{4  r}{\tau_0}\bigg)^2  \bigg\}\, ,
\end{align*} 
which further implies $ \EE  \{V_n(\bbeta_1, \bbeta_2) \}\geq  ( 1 - 16\sigma_0^2/\tau^2  )\cdot \{  1- 2\kappa_4\lambda_1( 4r/\tau_0)^2   \} \cdot \lambda_p$ for any $\bbeta_1 \in  \Theta(r/2)$ and $   \bbeta_2 \in \bbeta_1 + \BB^p(r)$. Consequently, we set $\tau_0 =  16 (\kappa_4 \lambda_1 )^{1/2}  r$, such that as long as $\tau \geq  16 \max\{  \sigma_0 , (\kappa_4 \lambda_1 )^{1/2}  r \}$,  
\begin{align}
    \inf_{  \bbeta_1 \in  \Theta(r/2) , \, \bbeta_2 \in \bbeta_1 + \BB^p(r)}  \EE\{ V_n (\bbeta_1 , \bbeta_2 ) \} \geq    c_0  \lambda_p \,,  \label{mean.Vn.lbd}
\end{align} 
where $c_0=(1-1/16)(1-1/8) \approx 0.82$.

Next, we upper bound the supremum
\begin{align*}
\Omega(r)    : = \sup_{ \bbeta_1 \in  \Theta(r/2) , \, \bbeta_2 \in \bbeta_1 + \BB^p(r) }   \big[   -  V_n(\bbeta_1, \bbeta_2 )  +   \EE  \{ V_n(\bbeta_1, \bbeta_2 ) \}\big] \,.
\end{align*}
{Write $V_n(\bbeta_1, \bbeta_2)-\EE \{V_n(\bbeta_1, \bbeta_2) \}  = n^{-1} \sn [v_i(\bbeta_1, \bbeta_2 )-\EE\{ v_i(\bbeta_1, \bbeta_2 )\}]$, where 
$$
    v_i(\bbeta_1, \bbeta_2 ) =   \chi_i  \varphi_{   \tau_0 /(2r)} (\langle \bx_i, \bv \rangle  )   \phi_{\tau /4} (\langle \bx_i, \bbeta_1 - \bbeta^* \rangle ) \,.  
$$
Notice that $\chi_i=\mathbbm{1}(|\varepsilon_i|\leq \tau/4)$, $0\leq \varphi_R(u)\leq R^2/4$ and  $0\leq \phi_R(u)\leq 1$. Then $ v_i(\bbeta_1, \bbeta_2 ) \in  [0,   \tau_0^2 / (4 r)^2] $, which implies $\sup_{\bbeta_1,\bbeta_2} |    v_i(\bbeta_1, \bbeta_2 )-\EE \{ v_i(\bbeta_1, \bbeta_2 )  \} | \leq   \tau_0^2 / (4 r)^2 $. Due to $\varphi_{R}(u)\leq u^2\mathbbm{1}(|u|\leq R)\leq u^2$,  it holds that
\begin{align*}
     \EE( [  v_i(\bbeta_1, \bbeta_2 )-\EE \{v_i(\bbeta_1, \bbeta_2 ) \}]^2) \leq \EE\{v_i^2(\bbeta_1, \bbeta_2 )\} \leq \EE(\langle \bx_i,\bv \rangle^4) \leq \kappa_4 \| \bv \|_{\bSigma}^4 \leq \kappa_4 \lambda_1^2
\end{align*}
for any $\bbeta_1$ and $ \bbeta_2$, which implies $\sup_{\bbeta_1, \bbeta_2} \EE( [  v_i(\bbeta_1, \bbeta_2 )-\EE \{v_i(\bbeta_1, \bbeta_2 ) \}]^2)    \leq \kappa_4 \lambda_1^2$. 
Applying Lemma~\ref{lem:bousquet}, we obtain that for any $z>0$,
\begin{align}
\Omega(r)  \leq \frac{5}{4} \EE \{ \Omega(r)\}  + \kappa_4^{1/2} \lambda_1 \sqrt{\frac{2    z}{n  }} + \frac{13}{3}\bigg(\frac{\tau_0 }{ 4r}\bigg)^2 \frac{ z}{ n }    \label{talagrand.concentration}
\end{align}
holds with probability at least $1-e^{-z}$. To bound the expectation $\EE \{\Omega(r)\} $, using Rademacher symmetrization and Lemma~4.5 in \citeS{ledoux2013probabilityS}, it holds
\begin{align}
\EE \{ \Omega(r)\} \leq 2 \cdot \sqrt{\frac{\pi}{2}} \cdot \EE \bigg(  \sup_{\bbeta_1, \bbeta_2}   \mathbb{G}_{\bbeta_1, \bbeta_2 }    \bigg) \,,  \label{gaussian.complexity}
\end{align}
where the supremum is taken over $( \bbeta_1 , \bbeta_2)  \in \{ \bbeta^* + \BB^p(r/2) \} \times \{ \bbeta_1 + \BB^p(r)\}$---including $(\bbeta^*, \bbeta^*)$, 
$$ 
\mathbb{G}_{\bbeta_1, \bbeta_2 }  = \frac{1}{n } \sn g_i  \chi_i  \varphi_{\tau_0 /(2 r)} ( \langle  \bx_i , \bv \rangle  )   \phi_{\tau /4} (\langle \bx_i, \bbeta_1 - \bbeta^*\rangle ) \,,  
$$
and $g_1, \ldots, g_n$ are independent standard normal variables. Conditional on $\{(y_i,\bx_i)\}_{i=1}^n$,  $\mathbb{G}_{\bbeta_1, \bbeta_2 } $ is a centered Gaussian process and $\mathbb{G}_{\bbeta^*, \bbeta^* }  = 0$. 
For any two admissible pairs $(\bbeta_1, \bbeta_2)$ and $(\bbeta_1',\bbeta_2')$, write $\bv =   (  \bbeta_1 - \bbeta_2 )/ \| \bbeta_1 - \bbeta_2 \|_2$, $\bv' = ( \bbeta_1' -\bbeta_2') / \| \bbeta'_1 - \bbeta'_2 \|_2$, and note that
\begin{align*}
\mathbb{G}_{\bbeta_1, \bbeta_2 } - \mathbb{G}_{\bbeta'_1, \bbeta'_2 }  & =  \GG_{\bbeta_1, \bbeta_2 } -  \GG_{\bbeta_1' ,  \bbeta'_1 +\bdelta } +  \GG_{\bbeta_1'  , \bbeta'_1 +\bdelta  }  - \GG_{\bbeta'_1, \bbeta'_2 }  \\
& = \frac{1}{n } \sn g_i  \chi_i   \varphi_{\tau_0/(2 r)} (\langle \bx_i, \bv \rangle ) \big\{ \phi_{\tau/4} ( \langle \bx_i, \bbeta_1 - \bbeta^* \rangle ) - \phi_{\tau/4} ( \langle \bx_i, \bbeta'_1 - \bbeta^* \rangle  ) \big\} \\
&~~~~ + \frac{1}{n } \sn g_i  \chi_i   \phi_{\tau/4} ( \langle \bx_i, \bbeta'_1 - \bbeta^* \rangle  )   \big\{ \varphi_{\tau_0 /(2 r)} (\langle \bx_i, \bv \rangle ) - \varphi_{\tau_0/(2 r)} (\langle \bx_i, \bv' \rangle )  \big\}\, .
\end{align*}
Recall that $\varphi_R$ and $\phi_R$ are, respectively, $R$- and $(2/R)$-Lipschitz continuous, and $\varphi_R(u) \leq (R/2)^2$. It follows that
\begin{align*}
 \EE^* \{ ( \GG_{\bbeta_1, \bbeta_2 } -  \GG_{\bbeta_1' ,  \bbeta'_1 +\bdelta }  )^2\} &\leq   \frac{1}{ n^2 }  \bigg( \frac{8}{\tau} \bigg)^2\bigg( \frac{\tau_0 }{4r} \bigg)^4 \sn \chi_i  \langle \bx_i, \bbeta_1 - \bbeta_1' \rangle^2 \\
 &\leq \frac{\tau_0^2 }{4   n^2 r^4 } \sn \chi_i  \langle \bx_i, \bbeta_1 - \bbeta_1'  \rangle^2 \,, \\ 
  \EE^*  \{(    \GG_{\bbeta_1'  , \bbeta'_1 +\bdelta  }  - \GG_{\bbeta'_1, \bbeta'_2 }  )^2 \}  &\leq \frac{1}{ n^2 } \bigg( \frac{\tau_0 }{2r} \bigg)^2  \sn  \chi_i  \langle \bx_i, \bv- \bv' \rangle^2 = \frac{\tau_0^2 }{4   n^2 r^2 }\sn   \chi_i\langle \bx_i, \bv - \bv' \rangle^2\, ,
\end{align*}
where the expectation  $\mathbb{E}^*$
  is applied only over $\{g_i\}_{i=1}^n$.  Together, the last two displays imply
\begin{align*}
\EE^*  \{( \mathbb{G}_{\bbeta_1, \bbeta_2 } - \mathbb{G}_{\bbeta'_1, \bbeta'_2 }  )^2\}  \leq  
  \frac{\tau_0^2 }{2  n^2 r^4 } \sn \chi_i  \langle \bx_i, \bbeta_1 -  \bbeta_1'  \rangle^2  +  \frac{\tau_0^2}{2 n^2r^2 }\sn  \chi_i \langle \bx_i, \bv - \bv' \rangle^2\,.
\end{align*}
Define another (conditional) Gaussian process $\{ \mathbb{Z}_{\bbeta_1, \bbeta_2} \}$ as 
\begin{align}
\mathbb{Z}_{\bbeta_1, \bbeta_2} =  \frac{\tau_0 }{\sqrt{2}  n r^2 } \sn g_i'  \chi_i \langle  \bx_i, \bbeta_1 - \bbeta^* \rangle + \frac{\tau_0}{\sqrt{2}  n r } \sn g''_i \chi_i \frac{\langle \bx_i, \bbeta_2 - \bbeta_1 \rangle }{\| \bbeta_2 - \bbeta_1 \|_2}\, ,\label{def:ep:Z}
\end{align}
where $\{ g'_i \}$ and $\{g_i''\}$ are two independent copies of $\{ g_i\}$. The above calculations show that
\begin{align*}
 \EE^{**} \{( \mathbb{G}_{\bbeta_1, \bbeta_2 } - \mathbb{G}_{\bbeta'_1, \bbeta'_2 }  )^2 \} \leq \EE^{**} \{( \mathbb{Z}_{\bbeta_1, \bbeta_2 } - \mathbb{Z}_{\bbeta'_1, \bbeta'_2 }  )^2 \} \,  ,
\end{align*}
where the expectation  $\mathbb{E}^{**}$
  is applied only over $\{g_i,g_i',g_i''\}_{i=1}^n$. Using Sudakov-Fernique's Gaussian comparison inequality (see, e.g., Theorem~7.2.11 in \citeS{vershynin2018highS}) gives
\begin{align}
 \EE^{**}   \bigg( \sup_{\bbeta_1, \bbeta_2}  \mathbb{G}_{\bbeta_1, \bbeta_2 }  \bigg)  
 \leq  \EE^{**}  \bigg( \sup_{\bbeta_1, \bbeta_2}   \mathbb{Z}_{\bbeta_1, \bbeta_2 }   \bigg)  \,. \label{gaussian.comparison}
\end{align}
The same bound also applies to unconditional expectations of the two suprema.  For the random process $\{ \mathbb{Z}_{\bbeta_1, \bbeta_2 } \}$, by the Cauchy-Schwarz inequality, we have
\begin{align*}
   \EE   \bigg( \sup_{\bbeta_1, \bbeta_2}   \mathbb{Z}_{\bbeta_1, \bbeta_2 }   \bigg)    
  \leq&\, \frac{\tau_0}{\sqrt{2} r^2} \sup_{\bbeta_1 \in \Theta(r/2) } \| \bbeta_1 -\bbeta^* \|_2  \cdot \EE \bigg(\bigg\| \frac{1}{n} \sn  g_i' \chi_i \bx_i \bigg\|_2\bigg) \\
 &~+  \frac{\tau_0}{\sqrt{2} r }  \EE\bigg(\bigg\| \frac{1}{n} \sn  g''_i \chi_i \bx_i \bigg\|_2\bigg)  \\
  \leq & \,\frac{ \tau_0}{2 r}     \sqrt{\frac{  \lambda_1  p}{2 n}} + \frac{ \tau_0 }{ r }  \sqrt{\frac{  \lambda_1 p}{  2 n}} =    24 \kappa_4^{1/2} \lambda_1 \sqrt{\frac{  p}{2 n}}\,,
\end{align*} 
where the last equality follows from $\tau_0 =  16 (\kappa_4 \lambda_1 )^{1/2}  r$. This, together with \eqref{talagrand.concentration},  \eqref{gaussian.complexity} and \eqref{gaussian.comparison}, implies that with probability at least $1-e^{-z}$,   
\begin{align}
	\Omega(r)   &  \leq \frac{5}{4} \EE \{\Omega(r)\} + \kappa_4^{1/2} \lambda_1 \sqrt{\frac{2z}{n}} +  \frac{13}{3}\bigg(\frac{\tau_0}{4r} \bigg)^2 \frac{z}{n}   \nn \\
    & \lesssim    \kappa_4^{1/2} \lambda_1 \sqrt{\frac{   p + z}{  n}}  + \kappa_4 \lambda_1 \frac{z}{n} \lesssim  \kappa_4^{1/2} \lambda_1 \sqrt{\frac{   p + z}{  n}}\,, \label{Omega.concentration}
 \end{align}
as long as $n\gtrsim   \kappa_4 z$.

Finally, combining  \eqref{def:Vn},  \eqref{mean.Vn.lbd} and \eqref{Omega.concentration}, we conclude that with probability at least $1-e^{-z}$, 
\begin{align*}
 &  \hat \cL_\tau(\bbeta_1)  - \hat \cL_\tau (\bbeta_2) - \langle \nabla \hat \cL_\tau(\bbeta_2), \bbeta_1-\bbeta_2 \rangle  \\
 &~~~ \geq   \frac{1}{2} \bigg(   c_0 \lambda_p - C  \kappa_4^{1/2}  \lambda_1 \sqrt{\frac{  p + z}{n}}   \, \bigg) \cdot  \| \bbeta_1 - \bbeta_2 \|_2^2  \geq c_1  \| \bbeta_1 - \bbeta_2 \|_2^2  
\end{align*}
holds uniformly for $\bbeta_1 \in \Theta(r/2)$ and $\bbeta_2 \in \bbeta_1+\BB^p(r)$ for some absolute constant $c_1 \in (0, \lambda_p/4)$,  provided that $n\gtrsim   \kappa_4  (\lambda_1/\lambda_p)^2( p + z )$ and $\tau \geq  16 \max\{ \sigma_0 ,   (\kappa_4 \lambda_1)^{1/2} r \}$. This completes the proof. \qed

\subsection{Proof of Lemma~\ref{lem:rsc.smootheness}}\label{proof.lem:rsc.smootheness}

Conditioned on $\hat \bbeta \in \Theta(r_0/2)$, set  $\bzeta = \bbeta - \hat \bbeta$. It follows from the first-order Taylor expansion that
\begin{align*}
\hat R_\tau( \bzeta ) := \hat \cL_\tau(\bbeta) - \hat \cL_\tau(\hat \bbeta) - \langle  \underbrace{  \nabla \hat \cL_\tau(\hat \bbeta) }_{= \textbf{0}}, \bbeta - \hat \bbeta \rangle  = \int_0^1 \langle \nabla \hat \cL_\tau(\hat \bbeta + u \bzeta ) - \nabla \hat \cL_\tau(\hat \bbeta) , \bzeta \rangle \,{\rm d} u\, .
\end{align*}
For any $v \in (0,1)$ and $u>0$,  by Lemma~C.1 in \citeS{sun2020adaptiveS},
\begin{align*}
\langle \nabla \hat \cL_\tau(\hat \bbeta + u \bzeta ) - \nabla \hat \cL_\tau(\hat \bbeta) ,  u\bzeta \rangle \geq  \frac{1}{v} \langle \nabla \hat \cL_\tau(\hat\bbeta + u v \bzeta ) - \nabla \hat \cL_\tau(\hat \bbeta) \,,  u v \bzeta \rangle  \,, 
\end{align*}
which in turn implies $\hat R_\tau( \bzeta ) \geq v^{-1} \hat R_\tau( v\bzeta )$.   Hence, for any $\bbeta \notin  \hat \bbeta + \BB^p(r_0 )$ such that $\| \bzeta \|_2 = \| \bbeta - \hat  \bbeta \|_2 > r_0$, by taking $v = r_0 /    \| \bbeta - \hat  \bbeta \|_2   \in (0,1)$ and $\bzeta_0 =   v\cdot ( \bbeta - \hat  \bbeta) \in \partial \BB^p(r_0 )$, we obtain  that
\begin{align*}
 \hat \cL_\tau(\bbeta) - \hat \cL_\tau(\hat \bbeta)  &  \geq      r_0^{-1}  \| \bbeta - \hat \bbeta \|_2 \cdot  \hat R_\tau( \bzeta_0 )   =      r_0^{-1} \| \bbeta - \hat  \bbeta \|_2   \big\{  \hat \cL_\tau(\hat \bbeta + \bzeta_0 ) - \hat \cL_\tau(\hat \bbeta)  \big\}\, .
\end{align*}
Conditioned on $\cE_0 \cap \cE_1$, $\hat \cL_\tau(\hat \bbeta + \bzeta_0 ) - \hat \cL_\tau(\hat \bbeta)   \geq  \phi_l    \| \bzeta_0 \|_2^2$.  Together, these two bounds imply
\begin{align*}
\hat \cL_\tau(\bbeta) - \hat \cL_\tau(\hat \bbeta)  \geq  \phi_l   r_0   \| \bbeta - \hat \bbeta \|_2 \,, 
\end{align*}
thus verifying \eqref{restricted.loss.difference}. \qed

\subsection{Proof of Lemma~\ref{lem:max.normal.l2}}
\label{proof.lem:max.normal.l2}

For each $t\in \{0\}\cup [T-1]$,  we apply the concentration inequality for Lipschitz functions of standard normal random variables and obtain that
\begin{align*}
 \PP \big ( \| \bg_t\|_2   \geq   p^{1/2 }  + \sqrt{2z} \big ) \leq e^{-z}	
\end{align*}
for any $z\geq 0$.
Combining this with the union bound (over $t\in \{0\}\cup [T-1]$) yields \eqref{max.chi-square.concentration}.

Notice that  $\| \bg_t \|_2^2$ follows the chi-square distribution $\chi_p^2$, which is a special case of the gamma distribution $\Gamma(p/2,2)$. The centered variable, $\| \bg_t \|_2^2- p$, is known to be sub-gamma with variance factor  $v=2p$ and  scale parameter $c=2$ \citepS{boucheron2013concentrationS}.  Let $Z= \sum_{t=0}^{T-1} \rho^t  ( \| \bg_t\|_2^2 -p )$. For each $t$ and $0<\lambda <1/c$,
\begin{align*}
 \log \EE \{e^{\lambda\rho^t  ( \| \bg_t\|_2^2 - p ) }\} \leq \frac{v \lambda^{2} \rho^{2t} }{2(1 - c \lambda\rho^t )} \leq \frac{v \lambda^2 \rho^{2t} }{2(1-c\lambda)}\,.
\end{align*}
By independence,
\begin{align*}
\log \EE (e^{\lambda Z } )= \sum_{t=0}^{T-1}  \log \EE \{e^{\lambda\rho^t  ( \| \bg_t\|_2^2 - p ) }\} \leq \frac{v \lambda^2  \sum_{t=0}^{T-1}\rho^{2t} }{2(1-c\lambda)} \leq \frac{v}{1-\rho^2}\cdot \frac{  \lambda^2  }{2(1-c\lambda)}\, .
\end{align*}
Therefore,  the centered variable $Z$ is sub-gamma with parameters $(v/(1-\rho^2), c)=(2p/(1-\rho^2), 2)$. Applying Chernoff's bound to $Z$ (see Section~2.4 of \citeS{boucheron2013concentrationS}) yields
\begin{align*}
 \PP \big\{    Z > 2 (pz)^{1/2} (1-\rho^2)^{-1/2}    + 2 z  \big\}   \leq e^{-z}  
\end{align*}
for any $z>0$.
This, combined with  $\sum_{t=0}^{T-1} \rho^t \leq 1/(1-\rho)$, proves \eqref{weighted.chi-square.concentration}.   \qed

\subsection{Proof of Lemma \ref{lem.Bahadur}}
\label{proof.lem.Bahadur}

  For $\bdelta\in \RR^p$, define the function $\Delta(\bdelta)= \bSigma^{-1/2}\{\nabla \widehat{\cL}_\tau(\bbeta^*+\bdelta)-\nabla \widehat{\cL}_\tau(\bbeta^* )\}-\bSigma^{1/2}\bdelta$. By the mean value theorem for vector-valued functions, 
 \begin{align*}
 	\EE\{\Delta(\bdelta)\} = &~ \bSigma^{-1/2}\bigg[ \int_{0}^1  \nabla^2  \EE\{ \widehat{\cL}_\tau(\bbeta^*+ u \bdelta )\}\,\mbox{d} u\bigg]  \bdelta-\bSigma^{1/2}\bdelta\\
    = &~ \int_{0}^1 [\bSigma^{-1/2}\nabla^2  \EE\{\widehat{\cL}_\tau(\bbeta^*+ u \bdelta ) \}- \bSigma^{1/2}]\,\mbox{d} u \cdot \bdelta\,.
 \end{align*}
Hence,  we have $\sup_{\|\bdelta\|_2\leq r}\|\EE\{\Delta(\bdelta)\}\|_2 \leq \sup_{\bbeta\in\Theta(r)}\| \bSigma^{-1/2}\nabla^2 \EE\{ \widehat{\cL}_\tau(\bbeta)\} - \bSigma^{1/2} \| \cdot r $ for any $r>0$, which further implies that
\begin{align}\label{bound.supdelta}
    \sup_{\|\bdelta\|_2\leq r}\| \Delta(\bdelta) \|_2 \leq \underbrace{\sup_{\|\bdelta\|_2\leq r}\| (1-\EE)
\{\Delta(\bdelta)\} \|_2}_{\rm I_1} + \underbrace{\sup_{\bbeta\in\Theta(r)}\| \bSigma^{-1/2}\nabla^2 \EE\{ \widehat{\cL}_\tau(\bbeta)\} - \bSigma^{1/2} \|}_{\rm I_2} \cdot r \,.
\end{align}
In what follows, we upper bound ${\rm I_1}$ and ${\rm I_2}$, respectively.

Note that ${\rm I_1} = \sup_{\bbeta\in\Theta(r)}\|(1-\EE)[\bSigma^{-1/2}\{\nabla \widehat{\cL}_\tau(\bbeta )-\nabla \widehat{\cL}_\tau(\bbeta^* )\}]\|_2$. Following the inequality above and (B.13) in the supplementary material of \citeS{chen2020robustS}, it can be similarly shown that, for any $z\geq 1/2$ and $r>0$, with probability at least $1-e^{-z}$,
\begin{align}\label{bound:I1}
    {\rm I_1} \lesssim  r \sqrt{\frac{p+ z}{n}}\,, 
\end{align} 
as long as $n\gtrsim p+z$.

For ${\rm I_2}$,  note that $\nabla^2  \EE \{ \widehat{\cL}_\tau(\bbeta) \} =n^{-1}\sn \EE [ \mathbbm{1}\{ |r_i(\bbeta)|\leq \tau\} \bx_i\bx_i^{\T} ]$ where $r_i(\bbeta)=y_i-\bx_i^{\T}\bbeta$. Then, 
 \begin{align*}
 	\bSigma^{-1/2}\nabla^2  \EE\{\widehat{\cL}_\tau(\bbeta)\}\bSigma^{-1/2} -\bI_{\rm p} = &~-\frac{1}{n}\sn \EE\{\mathbbm{1}\{ |r_i(\bbeta)|> \tau \} \bchi_i\bchi_i^{\T}\} \,,
 \end{align*}
where $\bchi_i = \bSigma^{-1/2}\bx_i$. By Markov's inequality, $\EE[\mathbbm{1}\{ |r_i(\bbeta)|> \tau \} \,|\,\bx_i] \lesssim \tau^{-2-\iota}\EE(|\varepsilon_i|^{2+\iota}\,|\,\bx_i)+\tau^{-2-\iota }|\bx_i^{\T}(\bbeta-\bbeta^*)|^{2+\iota }$. Under Assumptions \ref{assump:design} and \ref{assump:heavy-tail}, it follows that
\begin{align*}
    \sup_{\bbeta\in\Theta(r)}\bigg\| \frac{1}{n}\sn \EE\{\mathbbm{1}\{ |r_i(\bbeta)|> \tau \} \bchi_i\bchi_i^{\T}\} \bigg\| = &~\sup_{\bbeta\in\Theta(r)}\sup_{\bu \in \SSS^{p-1}}\bigg|\frac{1}{n}\sn \EE[\mathbbm{1}\{ |r_i(\bbeta)|> \tau \} (\bchi_i^{\T}\bu)^2]\bigg|\\
    \lesssim &~ \frac{\sigma_{\iota}^{2+\iota}}{\tau^{2+\iota}}+\frac{r^{2+\iota}}{\tau^{2+\iota}}\,,
\end{align*}
which in turn implies ${\rm I_2} \lesssim \sup_{\bbeta\in\Theta(r)}\|\EE\{\bSigma^{-1/2}\nabla^2  \widehat{\cL}_\tau(\bbeta)\bSigma^{-1/2}\} -\bI_{\rm p}\| \lesssim \sigma_{\iota}^{2+\iota}\tau^{-2-\iota}+ r^{2+\iota}\tau^{-2-\iota}$. Combining this with \eqref{bound.supdelta} and \eqref{bound:I1}, we conclude that for  any $z\geq 1/2$,
\begin{align*}
    \sup_{\|\bdelta\|_2\leq r}\| \Delta(\bdelta) \|_2 \lesssim r \bigg(\sqrt{\frac{p+ z}{n}} + \frac{\sigma_{\iota}^{2+\iota}}{\tau^{2+\iota}}+\frac{r^{2+\iota}}{\tau^{2+\iota}}\ \bigg)\,
\end{align*}
holds with probability at least $1-e^{-z}$ as long as $n\gtrsim p+z$.

Let $\widehat\bdelta=\widehat{\bbeta}-\bbeta^*$ and $\tau \asymp  \sigma_\iota  \{ n\epsilon/(p+\log n) \}^{1/(2+\iota)}$. By the proof of  Proposition \ref{prop:events1}, it follows that 
 \begin{align*} 
 	\|\widehat\bdelta\|_2=\|\widehat{\bbeta}-\bbeta^*\|_2\lesssim  \sigma_{\iota} \bigg( \frac{p+\log n}{n\epsilon} \bigg)^{(1+\iota)/(2+\iota)} + \sigma_0\sqrt{\frac{p+\log n}{n}} : = r_*\,
 \end{align*}
 with probability at least $1-Cn^{-1}$, provided that $n\gtrsim (p+\log n)/\epsilon$. Taking $r \asymp r_*$, and noting that $\nabla \widehat{\cL}_\tau(\widehat\bbeta )=0$, we have
 \begin{align*}
 	&~\bigg\|\bSigma^{1/2}(\widehat{\bbeta}-\bbeta^*)-\frac{1}{n}\sn \psi_{\tau}(\varepsilon_i )\bSigma^{-1/2}\bx_i\bigg\|_2 =\| \Delta(\widehat\bdelta)\|_2   \\
 	&~~~~\lesssim \bigg( \frac{p+\log n}{n\epsilon}+  \sqrt{\frac{p+\log n}{n}}\bigg)\bigg\{   \sigma_{\iota} \bigg( \frac{p+\log n}{n\epsilon} \bigg)^{(1+\iota)/(2+\iota)} + \sigma_0\sqrt{\frac{p+\log n}{n}}\bigg\}\\
 	&~~~~ \lesssim \sigma_\iota \bigg\{ \bigg( \frac{p+\log n}{n\epsilon} \bigg)^{(3+2\iota)/(2+\iota)} + \sqrt{\frac{p+\log n}{n}}   \bigg( \frac{p+\log n}{n\epsilon} \bigg)^{(1+\iota)/(2+\iota)} + \frac{p+\log n}{n}\bigg\}
 \end{align*}
 with probability at least $1-Cn^{-1}$, provided that $n\gtrsim    (p+\log n) /\epsilon$.  Hence, we complete the proof of Lemma   \ref{lem.Bahadur}.  \qed

\subsection{Proof of Lemma \ref{lem.GA}}
\label{proof.lem.GA}

To begin with, it follows from Lemma \ref{lem.Bahadur} that, with probability at least $1-Cn^{-1}$,
 \begin{align*}
 	&~	\bigg| \sqrt{n} \langle \bu,  \widehat{\bbeta}-\bbeta^* \rangle - \frac{1}{\sqrt{n}}\sn \psi_{\tau}(\varepsilon_i)\langle \bu, \bSigma^{-1}\bx_i \rangle \bigg| \leq \|\bu\|_{\bSigma^{-1}}\cdot \sqrt{n} R_{n,1} 
 \end{align*}
 holds uniformly over $\bu \in \RR^p$ as long as $n\gtrsim    (p+\log n) /\epsilon$, where
 \begin{align*}
 	R_{n,1}\asymp \sigma_\iota\bigg\{ \bigg( \frac{p+\log n}{n\epsilon} \bigg)^{(3+2\iota)/(2+\iota)} + \sqrt{\frac{p+\log n}{n}}   \bigg( \frac{p+\log n}{n\epsilon} \bigg)^{(1+\iota)/(2+\iota)} + \frac{p+\log n}{n}\bigg\}.
 \end{align*}

 For each $\bu \in \RR^p$, define independent random variables
 \begin{align*}
 	S_{i,\bu} =  \psi_{\tau}(\varepsilon_i)\langle \bu, \bSigma^{-1}\bx_i \rangle\,,~~i=1,\ldots,n\,.
 \end{align*}
 Under Assumptions \ref{assump:design} and \ref{assump:heavy-tail},  $	|\EE\{ \psi_{\tau}(\varepsilon_i)\,|\,\bx_i \}|  \leq \mathbb{E}\{(|\varepsilon_i|-\tau) \mathbbm{1}(|\varepsilon_i|>\tau) \,|\, \bx_i\} \leq \sigma_{\iota}^{2+\iota}\tau^{-1-\iota}$ and $\EE(|\langle \bu, \bSigma^{-1}\bx_i \rangle|)\leq \|\bu\|_{\bSigma^{-1}} $. By taking $\tau \asymp \sigma_\iota\{ n\epsilon/(p+\log n) \}^{1/(2+\iota)}$, we have 
 \begin{align}\label{bound1}
 	|\EE(S_{i,\bu} )| \leq \frac{\sigma_{\iota}^{2+\iota} \|\bu\|_{\bSigma^{-1}}}{\tau^{ 1+\iota}} 
\asymp  \|\bu\|_{\bSigma^{-1}}  \sigma_{\iota} \bigg( \frac{p+\log n}{n\epsilon} \bigg)^{(1+\iota)/(2+\iota)} \,.
 \end{align}
Substituting this into the previous bound implies that, with probability at least $1-Cn^{-1}$, 
 \begin{align}
 	&~~~\bigg| \sqrt{n} \langle \bu,  \widehat{\bbeta}-\bbeta^* \rangle - \frac{1}{\sqrt{n}}\sn \{ S_{i,\bu}-\EE(S_{i,\bu})\} \bigg| \leq \|\bu\|_{\bSigma^{-1}}\cdot \sqrt{n}R_{n,2}  \label{linear.dis}
 \end{align}
 holds uniformly over $\bu \in \RR^p$ as long as $n\gtrsim    (p+\log n) /\epsilon$, where
 \begin{align*}
 	R_{n,2}\asymp    \sigma_{\iota}   \bigg( \frac{p+\log n}{n\epsilon} \bigg)^{(1+\iota)/(2+\iota)}  \,. 
 \end{align*}

Let $Z_1\sim \cN(0, 1)$ and $\Upsilon_{\bu}^2=\var(S_{1,\bu})$. Applying the Berry-Esseen inequality (see, e.g., \citealpS{Shevtsova2014S}) to the i.i.d. random variables $\{S_{i,\bu}\}_{i=1}^n$, we obtain that
 \begin{align*}
 	\sup_{x\in \RR}\bigg| \PP\bigg[   \frac{1}{\sqrt{n}}\sn \{ S_{i,\bu}-\EE(S_{i,\bu})\}  \leq x \bigg ] - \PP( \Upsilon_\bu Z_1 \leq x ) \bigg| \leq C_{\iota}\frac{\EE\{|  S_{1,\bu}-\EE(S_{1,\bu}) |^{2+\iota}\}}{ \Upsilon_{\bu}^{2+\iota}n^{\iota/2}}\,,
 \end{align*}	
 where $C_{\iota}$ is a positive constant depending only on  $\iota\in (0,1]$. Together with \eqref{linear.dis} and the boundedness of the normal density function, we have, for any $x\in \RR$,
 \begin{align*}
 	\PP( \sqrt{n} \langle \bu,  \widehat{\bbeta}-\bbeta^* \rangle  \leq x ) \leq&~ \PP\bigg[ \frac{1}{\sqrt{n}}\sn \{ S_{i,\bu}-\EE(S_{i,\bu})\} \leq x+   \sqrt{n}R_{n,2}\|\bu\|_{\bSigma^{-1}}\bigg]+\frac{C}{n}\\
 	\leq &~ \PP(\Upsilon_\bu Z_1 \leq x+   \sqrt{n}R_{n,2}\|\bu\|_{\bSigma^{-1}})+C_{\iota}\frac{\EE\{|  S_{1,\bu}-\EE(S_{1,\bu}) |^{2+\iota}\}}{ \Upsilon_{\bu}^{2+\iota}n^{\iota/2}}+\frac{C}{n}\\
 	\leq &~ \PP(\Upsilon_\bu Z_1 \leq x )+\frac{  \sqrt{n}R_{n,2}\|\bu\|_{\bSigma^{-1}}}{\sqrt{2\pi}\Upsilon_{\bu}}+C_{\iota}\frac{\EE\{|  S_{1,\bu}-\EE(S_{1,\bu}) |^{2+\iota}\}}{ \Upsilon_{\bu}^{2+\iota}n^{\iota/2}}+\frac{C}{n}\,.
 \end{align*}

 Since $\psi_{\tau}(u) \leq |u|$ for any $u\in \RR$, by Assumption \ref{assump:heavy-tail}, it holds that
 \begin{align*}
 	\EE(|S_{1,\bu}|^{2+\iota})  \leq  \EE\{\EE( |\varepsilon_1 |^{2+\iota} \,|\,\bx_1 )|\langle \bu, \bSigma^{-1}\bx_1 \rangle|^{2+\iota}\}\leq \sigma_{\iota}^{2+\iota}\|\bu\|_{\bSigma^{-1}}^{2+\iota}\kappa_4^{(2+\iota)/4}\,,
 \end{align*} 
 which implies $\EE\{|  S_{1,\bu}-\EE(S_{1,\bu}) |^{2+\iota}\} \lesssim   \EE(|S_{1,\bu}|^{2+\iota}) \lesssim  \sigma_{\iota}^{2+\iota}\|\bu\|_{\bSigma^{-1}}^{2+\iota} $.  
 Note also that $\psi_{\tau}^2(u)=u^2-(u^2-\tau^2)\mathbbm{1}(|u|> \tau)$ and 
 $\EE(\varepsilon_i^2\,|\,\bx_i)\geq \underline{\sigma}_0^2$ almost surely. Then, we have  $ \|\bu\|_{\bXi}^2 = \bu^{\T}\bSigma^{-1}\EE( \varepsilon_i^2 \bx_i\bx_i^{\T} )\bSigma^{-1}\bu \geq \underline{\sigma}_0^2 \|\bu\|_{\bSigma^{-1}}^2$, and 
 \begin{align*}
 	\EE(S_{1,\bu}^2) =&~ (\bSigma^{-1/2}\bu)^{\T}  \EE\{ \psi_{\tau}^2(\varepsilon_1)\bchi_1\bchi_1^{\T} \}  (\bSigma^{-1/2}\bu)\\
 	=&~ (\bSigma^{-1/2}\bu)^{\T}  \EE(  \varepsilon_1^2 \bchi_1\bchi_1^{\T} )  (\bSigma^{-1/2}\bu)-(\bSigma^{-1/2}\bu)^{\T}  \EE\{ (\varepsilon_1^2-\tau^2)\mathbbm{1}(|\varepsilon_1|> \tau)\bchi_1\bchi_1^{\T} \}  (\bSigma^{-1/2}\bu)\\
 	= &~   \|\bu\|_{\bXi}^2-(\bSigma^{-1/2}\bu)^{\T}  \EE\{ (\varepsilon_1^2-\tau^2)\mathbbm{1}(|\varepsilon_1|> \tau)\bchi_1\bchi_1^{\T} \}  (\bSigma^{-1/2}\bu)\,, 
 \end{align*}
where $\bchi_1=\bSigma^{-1/2}\bx_1$. By Markov's inequality and Assumption \ref{assump:heavy-tail}, $\EE[\{ (\varepsilon_1^2-\tau^2)\mathbbm{1}(|\varepsilon_1|> \tau) \} \,|\,\bx_1]\leq \sigma_{\iota}^{2+\iota}\tau^{-\iota}$, and thus $(\bSigma^{-1/2}\bu)^{\T}  \EE\{ (\varepsilon_1^2-\tau^2)\mathbbm{1}(|\varepsilon_1|> \tau)\bchi_1\bchi_1^{\T} \}  (\bSigma^{-1/2}\bu)\leq \sigma_{\iota}^{2+\iota}\tau^{-\iota}\|\bu\|_{\bSigma^{-1}}^2$. Consequently,
 \begin{align}\label{second}
 	1- \frac{\sigma_{\iota}^{2+\iota}}{\underline{\sigma}_0^{2}\tau^{\iota}}	\leq \frac{\EE(S_{1,\bu}^2) }{  \|\bu\|_{\bXi}^2} \leq 1\,.
 \end{align}
 Moreover, by \eqref{bound1}, $|\EE(S_{i,\bu} )|^2 \leq \sigma_{\iota}^{4+2\iota}  \|\bu\|_{\bSigma^{-1}}^2/ \tau^{2+2\iota}$. It follows that
 \begin{align*} 
 	\bigg| \frac{\var(S_{1,\bu})}{  \|\bu\|_{\bXi}^2}-1 \bigg| \leq \frac{\sigma_{\iota}^{2+\iota}}{\underline{\sigma}_0^{2}\tau^{\iota}} + \frac{\sigma_{\iota}^{4+2\iota}}{\underline{\sigma}_0^{2}\tau^{2+2\iota}} \,.
 \end{align*}
 Recalling that $\tau \asymp \sigma_\iota \{ n\epsilon / (p+\log n) \}^{1/(2+\iota)}$ and letting $n \epsilon \gtrsim  (\sigma_{\iota}/ \underline{\sigma}_0)^{2+4/\iota} (p + \log n)$, this further implies
 \begin{align}\label{div.l2norm}
 	\frac{1}{2} \leq  \frac{ \Upsilon_{\bu}} {   \|\bu\|_{\bXi}}  \leq \frac{3}{2}\,.
\end{align}
Therefore, 
\begin{align*}
 	\frac{\EE\{|  S_{1,\bu}-\EE(S_{1,\bu}) |^{2+\iota}\}}{ \Upsilon_{\bu}^{2+\iota}n^{\iota/2}} \lesssim \frac{\sigma_{\iota}^{2+\iota}\|\bu\|_{\bSigma^{-1}}^{2+\iota} }{ \underline{\sigma}_0^{2+\iota}  \|\bu\|_{\bSigma^{-1}}^{2+\iota}n^{\iota/2}}=\frac{\sigma_{\iota}^{2+\iota}  }{ \underline{\sigma}_0^{2+\iota}   n^{\iota/2}}\,.
 \end{align*}
Substituting this into the right-hand side of the Berry–Esseen inequality yields
\begin{align*}
 	  \PP( \sqrt{n} \langle \bu,  \widehat{\bbeta}-\bbeta^* \rangle  \leq x )-\PP(\Upsilon_\bu Z_1 \leq x ) \lesssim \frac{\sigma_\iota}{\underline{\sigma}_0} \sqrt{n}    \bigg( \frac{p+\log n}{n\epsilon} \bigg)^{(1+\iota)/(2+\iota)}   +\frac{C_\iota \sigma_{\iota}^{2+\iota}  }{\underline{\sigma}_0^{2+\iota}  n^{\iota/2}}  
\end{align*}
for all $x\in \RR$ and  $\bu\in \RR^p$. A reversed inequality can be obtained by a similar argument.

Since 
\begin{align*} 
 	\bigg| \frac{\var(S_{1,\bu})}{ \|\bu\|_{\bXi}^2}-1 \bigg| \lesssim  \frac{\sigma_{\iota}^2}{\underline{\sigma}_0^{2} }  \bigg( \frac{p+\log n}{n \epsilon} \bigg)^{\iota/(2+\iota)}  \lesssim 1 \,, 
 \end{align*}
applying Lemma~A.7 in \citeS{SZ2015S} yields
$$
    \sup_{x\in \RR }  |  \PP(\Upsilon_\bu Z_1 \leq x) -\PP(  \| \bu \|_{\bXi} Z_1 \leq x)  | | \lesssim \frac{\sigma_{\iota}^2}{\underline{\sigma}_0^{2} }  \bigg( \frac{p+\log n}{n \epsilon} \bigg)^{\iota/(2+\iota)}\, . 
$$
Putting together the pieces, we conclude that
\begin{align*}
  \sup_{\bu\in \RR^p}\sup_{x\in\RR}\big| \PP( \sqrt{n} \langle \bu/\|\bu\|_{\bXi},  {\widehat\bbeta} -\bbeta^* \rangle  \leq x )- \PP(Z_1 \leq x) \big|  \lesssim  C'_\iota \frac{\sigma_\iota\sqrt{n}}{\underline{\sigma}_0}    \bigg( \frac{p+\log n}{n\epsilon} \bigg)^{(1+\iota)/(2+\iota)} 
\end{align*}
as long as $n\epsilon \gtrsim   (\sigma_{\iota}/\underline{\sigma}_0)^{2 + 4/\iota}  (p+\log n)$, where $C_\iota'>0$ is a constant depending on $\iota$. This completes the proof of the Lemma \ref{lem.GA}. \qed

\subsection{Proof of Lemma \ref{lem.dp}}
\label{proof.lem.dp}

Let $\bZ_n=\{(y_i,\bx_i)\}_{i=1}^n$ and $D=p(p+1)/2$.  Moreover, denote by $\bh_1(\bZ_n)$ and $\bh_2(\bZ_n)$, respectively, the $D$-dimensional vector that consists of the upper-triangular and diagonal entries of $\widehat\bSigma_{\gamma_1}=\widehat\bSigma_{\gamma_1}(\bZ_n)\in \RR^{p\times p}$ and $\widehat\bOmega_{\tau_1,\gamma_1}(\bbeta)=\widehat\bOmega_{\tau_1,\gamma_1}(\bbeta,\bZ_n)\in \RR^{p\times p}$.  Consider two datasets
 $\bZ_n$ and $\bZ_n'$  that differ by one datum, say $\bz_1=(y_1,\bx_1)\in \bZ_n$ versus $\bz_1'=(y_1',\bx_1')\in \bZ_n'$.

Notice that $w_{\gamma_1}(u)=\min\{\gamma_1/u,1 \}$ for any $u\geq 0$, and $\| \bx_1\bx_1^{\T} \|_{F}^2 = \sum_{j=1}^p\sum_{k=1}^p (x_{1,j}x_{1,k})^2 = \sum_{j=1}^px_{1,j}^2(\sum_{k=1}^p x_{1,k}^2)=\|\bx_1\|_2^4 $,
where $\|\cdot\|_{\rm F}$ denotes the Frobenius norm of a matrix. Then, it holds that 
\begin{align*}
    \| \bx_1\bx_1^{\T}w_{\gamma_1}^2(\|\bx_1\|_2)  \|_{\rm F} = \| \bx_1\bx_1^{\T}  \|_{\rm F}w_{\gamma_1}^2(\|\bx_1\|_2) = \|\bx_1\|_2^2 \cdot \min \bigg\{\frac{\gamma_1^2}{\|\bx_1\|_2^2},1\bigg\}\leq \gamma_1^2\,.
\end{align*}
Therefore, we obtain
 \begin{align*}
 	\| \bh_1(\bZ_n)-\bh_1(\bZ_n') \|_2 \leq&~ \|\widehat\bSigma_{\gamma_1}(\bZ_n)-\widehat\bSigma_{\gamma_1}(\bZ_n')  \|_{\rm F} \leq \frac{1}{n} \| \bx_1\bx_1^{\T}w_{\gamma_1}^2(\|\bx_1\|_2)-\bx_1'(\bx_1')^{\T}w_{\gamma_1}^2(\|\bx_1'\|_2) \|_{\rm F}\\
  \leq&~ \frac{1}{n}\big\{ \| \bx_1\bx_1^{\T}w_{\gamma_1}^2(\|\bx_1\|_2)  \|_{\rm F}+\|  \bx_1'(\bx_1')^{\T}w_{\gamma_1}^2(\|\bx_1'\|_2) \|_{\rm F}\big\} \leq \frac{2\gamma_1^2}{n}\,,
 \end{align*}
  Then, by  Lemma \ref{glmechanism} and Lemma \ref{gdp.mechanism}, $\bh_1(\bZ_n)+ 2\gamma_1^2 \sqrt{2\log(1.25/\delta)} (n\epsilon)^{-1}\bg$ and $\bh_1(\bZ_n)+ 2\gamma_1^2   (n\epsilon)^{-1}\bg$ with $\bg \sim \cN(\mathbf{0},\bI_D)$ are, respectively, $(\epsilon,\delta)$-DP and $\epsilon$-GDP. It follows from the post-processing property that $\widehat\bSigma_{\gamma_1}+ 2\gamma_1^2 \sqrt{2\log(1.25/\delta)} (n\epsilon)^{-1}\bE$ and $\widehat\bSigma_{\gamma_1}+ 2\gamma_1^2   (n\epsilon)^{-1}\bE$ are, respectively, also $(\epsilon,\delta)$-DP and $\epsilon$-GDP.

 For the second result in Lemma \ref{lem.dp}, note that $|\psi_{\tau_1}(u) | \leq \tau_1$ for any $u\in \RR$. Then, for any $\bbeta\in\RR^p$, 
 \begin{align*}
 	&\| \bh_2(\bZ_n)-\bh_2(\bZ_n') \|_2 \leq  \|\widehat\bOmega_{\tau_1,\gamma_1}(\bbeta,\bZ_n)-\widehat\bOmega_{\tau_1,\gamma_1}(\bbeta,\bZ_n')  \|_{\rm F}\\
    &~~~~\leq  \frac{1}{n}\| \psi_{\tau_1}^2(y_1-\bx_1^{\T}\bbeta)\bx_1\bx_1^{\T}w_{\gamma_1}^2(\|\bx_1\|_2) - \psi_{\tau_1}^2(y_1'-(\bx_1')^{\T}\bbeta)\bx_1'(\bx_1')^{\T}w_{\gamma_1}^2(\|\bx_1'\|_2)\|_{\rm F}\\
    &~~~~\leq   \frac{{\tau_1}^2}{n}\big\{\|  \bx_1\bx_1^{\T}w_{\gamma_1}^2(\|\bx_1\|_2)\|_{\rm F}+ \| \bx_1'(\bx_1')^{\T}w_{\gamma_1}^2(\|\bx_1'\|_2)\|_{\rm F} \big\}\leq \frac{2 (\gamma_1 \tau_1 )^2 }{n}\,,
 \end{align*}
 implying that, for any $\bbeta\in\RR^p$, $\widehat\bOmega_{\tau_1,\gamma_1}(\bbeta)+ 2  (\gamma_1 \tau_1 )^2 \sqrt{2\log(1.25/\delta)}  (n\epsilon)^{-1}\bE$ is   $(\epsilon,\delta)$-DP and  $\widehat\bOmega_{\tau_1,\gamma_1}(\bbeta)+ 2  (\gamma_1 \tau_1 )^2  (n\epsilon)^{-1}\bE$ is  $\epsilon$-GDP. This completes the proof of Lemma \ref{lem.dp}. \qed

 \subsection{Proof of Lemma \ref{lem.divvar}}
 \label{proof.lem.divvar}

By the proof of Proposition \ref{prop:events1},  we know that, for some sufficiently large constant $C_0>0$, the event   $\widetilde\cE_0=\{ \max_{i\in[n]}\|\bx_i\|_2 \leq C_0\upsilon_1\lambda_1^{1/2}\sqrt{p+ \log n} \}$  occurs with probability at least $1-Cn^{-1}$. By setting $\gamma_1 = C_0\upsilon_1\lambda_1^{1/2}\sqrt{p+ \log n}$, we have $w_{\gamma_1}^2(\|\bx_i\|_2)=1$ on the event $\widetilde\cE_0$. In the following, we proceed by conditioning on $\widetilde\cE_0$.

Define $r_i(\bbeta)=y_i-\bx_i^{\T}\bbeta$ such that $r_i(\bbeta^*) = \varepsilon_i$, and write
 \begin{align*} 
 	\widehat\bSigma_{\gamma_1} = \frac{1}{n}\sn  \bx_i\bx_i^{\T} ~~\mbox{and}~~\widehat\bOmega_{\tau_1,\gamma_1}(\bbeta) = \frac{1}{n}\sn \psi_{\tau_1}^2 ( r_i(\bbeta) )\bx_i\bx_i^{\T} \,.
 \end{align*}
By Theorem $6.5$ in \citeS{wainwright2019highS},   it holds with probability at least $1-Cn^{-1}$ that
 \begin{align*} 
 	\|\widehat\bSigma_{\gamma_1}-\bSigma\| \lesssim \upsilon_1^2\lambda_1\bigg( \sqrt{\frac{p+\log n}{n}}+\frac{p+\log n}{n} \bigg)\,.
 \end{align*}
To bound $\|\widehat\bOmega_{\tau_1,\gamma_1}(\bbeta)- \bOmega \|$, consider the following decomposition
 \begin{align}\label{div.omega}
 	  \|\widehat\bOmega_{\tau_1,\gamma_1}(\bbeta)- \bOmega \| \leq&~  \underbrace{\bigg\|\frac{1}{n}\sn \{\psi_{\tau_1}^2( r_i(\bbeta) )  - \psi_{\tau_1}^2 ( r_i(\bbeta^*))\}\bx_i\bx_i^{\T}\bigg\| }_{\rm I_1(\bbeta)}    +\underbrace{\bigg\| \frac{1}{n}\sn  (1-\EE)\{\psi_{\tau_1}^2( \varepsilon_i)\bx_i\bx_i^{\T}\}\bigg\| }_{\rm I_2}\nn\\
 	 &~+\underbrace{\bigg\| \frac{1}{n}\sn   \EE [\{\psi_{\tau_1}^2( \varepsilon_i)-\varepsilon_i^2\}\bx_i\bx_i^{\T}]\bigg\| }_{\rm I_3} \,.
 \end{align}
Starting from $\sup_{\bbeta\in\Theta(r)}{\rm I_1(\bbeta)}$, note that $\psi_{\tau_1}(\cdot)$ is 1-Lipschitz and $|\psi_{\tau_1}(u)|\leq \tau_1 $ for any $u\in\RR$. Therefore, $|\psi_{\tau_1}^2(r_i(\bbeta))  - \psi_{\tau_1}^2( r_i(\bbeta^*))|\leq |\psi_{\tau_1}(r_i(\bbeta))  - \psi_{\tau_1}( r_i(\bbeta^*))||\psi_{\tau_1}(r_i(\bbeta))  + \psi_{\tau_1}( r_i(\bbeta^*))| \leq 2\tau_1 \max_{i\in[n]}\|\bx_i\|_2\| {\bbeta} -\bbeta^*\|_2$, which further implies that
 \begin{align*}
 \sup_{\bbeta\in\Theta(r)}	{\rm I_1}(\bbeta) = &~  \sup_{\bbeta\in\Theta(r)}\sup_{\bu\in\SSS^{p-1}} \bigg |\frac{1}{n}\sn \{\psi_{\tau_1}^2(r_i(\bbeta))  - \psi_{\tau_1}^2( r_i(\bbeta^*))\} (\bx_i^{\T}\bu)^2\bigg | \\
 	\leq &~ 2\tau_1 r   \max_{i\in[n]}\|\bx_i\|_2 \sup_{\bu\in\SSS^{p-1}} \bigg |\frac{1}{n}\sn   (\bx_i^{\T}\bu)^2\bigg |= 2\tau_1 r   \max_{i\in[n]}\|\bx_i\|_2 \bigg\| \frac{1}{n}\sn  \bx_i\bx_i^{\T}\bigg\|\,.
 \end{align*}
 Moreover, it follows from   Assumption \ref{assump:design} that $ \|  {n}^{-1}\sn \bx_i\bx_i^{\T}  \|\leq \|{n}^{-1}\sn \bx_i\bx_i^{\T} -\bSigma\|+\|\bSigma\|\leq 2\lambda_1$ with probability at least $1-Cn^{-1}$, provided that $n\gtrsim \upsilon_1^4(p+\log n)$. Consequently,
 \begin{align}\label{bound.I1}
  \sup_{\bbeta\in\Theta(r)}	{\rm I_1}(\bbeta)  \lesssim &~  \tau_1 r \sqrt{p+\log n} \,
 \end{align}
 holds with probability at least $1-C n^{-1}$ as long as $n \gtrsim p+\log n$.  

Turning to ${\rm I_2}$, write $\bh_i = \psi_{\tau_1}(\varepsilon_i)\bx_i$. By Lemmas 5.2 and 5.4 of \citeS{Vershynin2012S}, there exists a $1/3$-net $\cN_{1/3}$ of $\SSS^{p-1}$ with cardinality $|\cN_{1/3}| \leq 7^p$, such that
 \begin{align*}
 	{\rm I_2} \leq 3 \max_{\bu\in\cN_{1/3}} \bigg |\frac{1}{n}\sn(1-\EE) (|\bh_i^{\T}\bu|^2)\bigg |\,.
 \end{align*}
Since $|\psi_{\tau_1}(u)|\leq \min\{ \tau_1 , |u| \}$ for any $u\in \RR$, for any integer $k\geq 2$, it holds that  
 \begin{align*}
 	\EE(|\bh_i^{\T}\bu|^{2k})\leq \tau_1^{2(k-2)}\EE\{\psi_{\tau_1}^4(\varepsilon_i)|\bx_i^{\T}\bu|^{2k}  \}\,.
 \end{align*}
Under Assumption \ref{assump:heavy-tail},  $\EE\{\psi_{\tau_1}^4(\varepsilon_i)\,|\,\bx_i \} \leq \tau_1^{2-\iota} \EE( |\varepsilon_i|^{2+\iota}\,|\,\bx_i )\leq \tau_1^{2-\iota}\sigma_{\iota}^{2+\iota}$.  
 By Assumption \ref{assump:design}, 
 \begin{align}\label{tail}
 	\PP(|\bx_i^{\T}\bu|>\upsilon_1\lambda_1^{1/2}z)\leq \PP(|\langle \bSigma^{1/2}\bu, \bSigma^{-1/2}\bx_i\rangle|>\upsilon_1\|\bSigma^{1/2}\bu\|_2z)\leq 2e^{-z^2/2} 
 \end{align}
for any $\bu \in \SSS^{p-1}$ and $z\geq 0$. This implies that for any $\bu \in \SSS^{p-1}$ and  each $k=2,3,\ldots$,
 \begin{align}\label{moments.x}
 	\EE(|\bx_i^{\T}\bu|^{2k}) = &~ \int_{0}^{\infty} \PP(|\bx_i^{\T}\bu|^{2k}>z)\,\mbox{d}z \leq 2(\upsilon_1\lambda_1^{1/2})^{2k}(2k)\int_{0}^{\infty}z^{2k-1}e^{-z^2/2}\,\mbox{d}z\nn\\
 	=&~ 2^{k+1}(\upsilon_1\lambda_1^{1/2})^{2k}k\Gamma(k) = \frac{k!}{2} 4(2\upsilon_1^2\lambda_1)^{k}\,.
 \end{align} 
 Then, we can conclude that
 \begin{align*}
 	 	\EE(|\bh_i^{\T}\bu|^{2k})\leq    \frac{k!}{2} 16\tau_1^{2-\iota}\sigma_{\iota}^{2+\iota}\upsilon_1^4\lambda_1^2 (2\upsilon_1^2\lambda_1\tau_1^2)^{k-2}\,.
 \end{align*}
It then follows from Bernstein's inequality (see, e.g., Proposition~2.9 in \citeS{MASSART2007S}) that, for any $z\geq 0$, 
 \begin{align*}
 	\bigg |\frac{1}{n}\sn(1-\EE) (|\bh_i^{\T}\bu|^2)\bigg | \lesssim \upsilon_1^2\lambda_1\bigg(  \sqrt{  \frac{\tau_1^{2-\iota}\sigma_{\iota}^{2+\iota} z}{n} }+ \frac{\tau_1^2z}{n} \bigg)
 \end{align*}
 holds with probability at least $1-2e^{-z}$.  Taking the union bound over $\bu\in\cN_{1/3}$ and setting $z\asymp p+\log n$, we obtain 
 \begin{align}\label{bound.I2}
 	{\rm I_2}  \lesssim    \sqrt{ \tau_1^{2-\iota}\sigma_{\iota}^{2+\iota} \frac{p+\log n}{n} }+ \tau_1^2\frac{p+\log n}{n} 
 \end{align}
 with probability at least $1-Cn^{-1}$.  

For ${\rm I_3}$, by the inequality above \eqref{second} in Section \ref{proof.lem.GA}, we also have ${\rm I_3} \lesssim \tau_1^{-\iota}\sigma_{\iota}^{2+\iota}$.  Combining this with \eqref{div.omega}, \eqref{bound.I1}  and \eqref{bound.I2}, it holds with probability at least $1-Cn^{-1}$ that 
 \begin{align*}
 \sup_{\bbeta\in\Theta(r)}	\|\widehat\bOmega_{\tau_1,\gamma}(\bbeta)- \bOmega \|  \lesssim &~ \tau_1 r \sqrt{p+\log n}+    \tau_1^2\frac{p+\log n}{n} + \frac{\sigma_{\iota}^{2+\iota}}{\tau_1^{\iota}}\,,
 \end{align*}
 provided that $n \gtrsim p+\log n$. This completes the proof of Lemma~\ref{lem.divvar}.  \qed

\subsection{Proof of Lemma \ref{lem.varest}}
\label{proof.lem.varest}

 First, by definition, for any $\epsilon>0$,
  \begin{align*}
  	\|\widehat{\bSigma}_{\gamma_1,\epsilon}^{+}-\widehat{\bSigma}_{\gamma_1,\epsilon} \|= \|\widehat{\bSigma}_{\gamma_1,\epsilon}^{+}- (\widehat{\bSigma}_{\gamma_1 } +\varsigma_{1 }\bE ) \|\leq  \|\widehat{\bSigma}_{\gamma_1 }- (\widehat{\bSigma}_{\gamma_1 } +\varsigma_{1 }\bE ) \|= \varsigma_{1 }\|\bE\|\,,
  \end{align*}
  which implies 
  \begin{align*}
  \|\widehat{\bSigma}_{\gamma_1,\epsilon}^{+}- {\bSigma}  \|\leq 	\|\widehat{\bSigma}_{\gamma_1,\epsilon}^{+}-\widehat{\bSigma}_{\gamma_1,\epsilon} \|+ \|\widehat{\bSigma}_{\gamma_1 }- {\bSigma}  \|+\varsigma_{1 }\|\bE\|\leq\|\widehat{\bSigma}_{\gamma_1 }- {\bSigma}  \|+2\varsigma_{1 }\|\bE\|\,.
  \end{align*}
 Corollary 4.4.8 in  \citeS{vershynin2018highS} implies that $\|\bE\|\lesssim \sqrt{p+\log n}$  with probability at least $1-Cn^{-1}$. Then, by Lemma \ref{lem.divvar}, it holds with probability at least $1-Cn^{-1}$ that 
  \begin{align*}
  	\|\widehat{\bSigma}_{\gamma_1,\epsilon}^{+}- {\bSigma}  \|\lesssim \sqrt{\frac{p+\log n}{n}} + \varsigma_{1 }\sqrt{p+\log n}\,,
  \end{align*}
 provided that $n\gtrsim p+\log n$. Analogously, by Lemma \ref{lem.divvar},  we have with probability at least $1-Cn^{-1}$,
  \begin{align*}
  	&~\sup_{\bbeta\in\Theta(r)}\|\widehat\bOmega_{\tau_1,\gamma_1,\epsilon}^{+}(\bbeta)- \bOmega\| \leq  \sup_{\bbeta\in\Theta(r)}\|\widehat\bOmega_{\tau_1,\gamma_1 } (\bbeta)- \bOmega\|+2\varsigma_{2 }\|\bE\|\\
  	&~~~~\lesssim  \tau_1 r \sqrt{p+\log n}+    \tau_1^2\frac{p+\log n}{n} + \frac{\sigma_{\iota}^{2+\iota}}{\tau_1^{\iota}} + \varsigma_{2 }\sqrt{p+\log n}\,.
  \end{align*}

With the above preparations,  we are now equipped to bound $\|	\widetilde\bXi_{\tau_1,\gamma_1,\epsilon}(\bbeta) - \bXi\|$ for any $\bbeta\in\RR^p$ as follows:
  \begin{align*}
  	&~\|	\widetilde\bXi_{\tau_1,\gamma_1,\epsilon}(\bbeta) - \bXi\| = \|(\widehat{\bSigma}_{\gamma_1,\epsilon}^{+})^{-1}\widehat\bOmega_{\tau_1,\gamma_1,\epsilon}^{+}(\bbeta)(\widehat{\bSigma}_{\gamma_1,\epsilon}^{+})^{-1}- \bSigma^{-1}\bOmega\bSigma^{-1}\|\\
  	&~~~~\leq  \|(\widehat{\bSigma}_{\gamma_1,\epsilon}^{+})^{-1}-\bSigma^{-1}\|^2\|\widehat\bOmega_{\tau_1,\gamma_1,\epsilon}^{+}(\bbeta)-\bOmega\|+ 2\|(\widehat{\bSigma}_{\gamma_1,\epsilon}^{+})^{-1}-\bSigma^{-1}\|\|\widehat\bOmega_{\tau_1,\gamma_1,\epsilon}^{+}(\bbeta)-\bOmega\|\|\bSigma^{-1}\|\\
  	&~~~~~~~+\|(\widehat{\bSigma}_{\gamma_1,\epsilon}^{+})^{-1}-\bSigma^{-1}\|^2\| \bOmega\|+ 2\|(\widehat{\bSigma}_{\gamma_1,\epsilon}^{+})^{-1}-\bSigma^{-1}\|\| \bOmega\|\|\bSigma^{-1}\|+\|\widehat\bOmega_{\tau_1,\gamma_1,\epsilon}^{+}(\bbeta)-\bOmega\|\|\bSigma^{-1}\|^2\,.
  \end{align*}
 As long as $\sqrt{ ({p+\log n})/{n}} + \varsigma_{1 }\sqrt{p+\log n}$ is sufficiently small, $\|\widehat{\bSigma}_{\gamma_1,\epsilon}^{+} -\bSigma\| \ll 1$ with probability at least $1-Cn^{-1}$. By Assumption \ref{assump:design}, we  have $\|\bSigma^{-1}\| = O(1)$, and  hence $\|(\widehat{\bSigma}_{\gamma_1,\epsilon}^{+})^{-1}\| = O(1)$   with probability at least $1-Cn^{-1}$.  Consequently,
  \begin{align*}
  	\|(\widehat{\bSigma}_{\gamma_1,\epsilon}^{+})^{-1}-\bSigma^{-1}\| \leq \|(\widehat{\bSigma}_{\gamma_1,\epsilon}^{+})^{-1}\|\|\widehat{\bSigma}_{\gamma_1,\epsilon}^{+} -\bSigma\|\|\bSigma^{-1}\|\lesssim \sqrt{\frac{p+\log n}{n}}+ \varsigma_{1 }\sqrt{p+\log n}\ll 1\,.
  \end{align*}
 Note that $\|\bOmega\|= \|\EE(\varepsilon_i^2 \bx_i\bx_i^{\T})\|\lesssim \sigma_0^2$.  Then,  it holds with probability at least $1-Cn^{-1}$ that
  \begin{align*} 
  	&~\sup_{\bbeta\in\Theta(r)}\|	\widetilde\bXi_{\tau_1,\gamma_1,\epsilon}(\bbeta)- \bXi\|\\
  	&~~~~ \lesssim     \sigma_0^2\sqrt{\frac{p+\log n}{n}}+\tau_1 r \sqrt{p+\log n}+    \tau_1^2\frac{p+\log n}{n} + \frac{\sigma_{\iota}^{2+\iota}}{\tau_1^{\iota}}+ (\sigma_0^2\varsigma_{1 }+\varsigma_{2 })\sqrt{p+\log n}\,,
  \end{align*} 
  provided that $\sqrt{ ({p+\log n})/{n}} + \varsigma_{1 }\sqrt{p+\log n}\ll 1$. From the definitions of $(\varsigma_{1 }, \varsigma_{2 })$, we have
$\varsigma_{2 } \geq \sigma_0^2\varsigma_{1 }$ and 
$\varsigma_{2 } \gtrsim \tau_1^2 (p+\log n)/(n\epsilon)$. Thus, taking $\tau_1 \asymp \sigma_\iota \{ n \epsilon / (p+\log n) \}^{1/(2+\iota) }$, it holds that
\begin{align*}
    \sup_{\bbeta\in\Theta(r)}\|	\widetilde\bXi_{\tau_1,\gamma_1,\epsilon}(\bbeta)- \bXi\| \lesssim  (\tau_1 r +\varsigma_2   )\sqrt{p+\log n}\,.
\end{align*}
 This completes the proof of Lemma~\ref{lem.varest}. \qed

\renewcommand{\thealgorithm}{S.\arabic{algorithm}}
\setcounter{algorithm}{0}

\subsection{Proof of Lemma~\ref{lem.HD.noiseht}}\label{proof.lem.HD.noiseht}
We first introduce the `Private Max' algorithm (Algorithm \ref{alg:privatemax}), which identifies and reports the index of the value with the largest noisy magnitude.
\begin{algorithm}
\caption{Private Max} \label{alg:privatemax}
\vspace{.2cm}
\textbf{Input:} Dataset $\bZ$, vector-valued function  $\bv= \bv(\bZ)=(v_1,\ldots,v_p)^{\T} \in \RR^p$ with each  $\text{sens}_1(v_j)$  bounded by $\Delta$, and    privacy parameter $ \epsilon $.
\begin{algorithmic}[1]
 \renewcommand{\algorithmicindent}{0pt}
\FOR{$j\in[p]$} 
\STATE ~ Set $\tilde{v}_j = |v_j(\bZ)|+w_j$, where $w_j$ is independently sampled from $\text{Laplace}(3\Delta/\epsilon)$; 
\ENDFOR
\end{algorithmic}
\textbf{Output:} Return $j^*=\argmax_{j\in[p]} \tilde{v}_j$ and $v^*_{j^*}(\bZ):=v_{j^*}(\bZ)+w$, where $w$ is a fresh draw from $\text{Laplace}(3\Delta/\epsilon)$.
\end{algorithm}
Based on Algorithm \ref{alg:privatemax}, Algorithm \ref{alg:NHT} is equivalent to the following `Peeling' algorithm 
 (Algorithm \ref{alg:peeling}) with the Laplace noise scale  $\Lambda= 2\epsilon^{-1}\lambda\sqrt{5s\log(1/\delta)}$, where the parameter $\lambda$ satisfies $\|\bv(\bZ)-\bv(\bZ')\|_{\infty} < \lambda$ for every pair of adjacent datasets $\bZ$ and $\bZ'$.  
\begin{algorithm}
\caption{Peeling} \label{alg:peeling}
\vspace{.2cm}
\textbf{Input:} Dataset $\bZ$, vector-valued function  $\bv= \bv(\bZ)=(v_1,\ldots,v_p)^{\T}  \in \RR^p$, number of invocations $s$, and the Laplace noise scale $\Lambda$.   
\begin{algorithmic}[1]
 \renewcommand{\algorithmicindent}{0pt}
\STATE Initialize $\bv_1=\bv_1(\bZ)=\bv(\bZ)$;
\FOR{$j\in[s]$} 
\STATE ~ Let $(i_j,v^*_{i_j}(\bZ))$ be returned by `Private Max' applied to $(\bZ,\bv_j)$ with Laplace noise scale $\Lambda$; 
\STATE ~ Set $\bv_{j+1}={\rm remove}(v_{i_j},\bv_j) \in \mathbb{R}^{p-j}$, i.e., remove the element $v_{i_j}$ from $\bv_j$;
\ENDFOR
\end{algorithmic}
\textbf{Output:} the $s$-tuple $\{(i_1,v^*_{i_1}(\bZ)),\ldots,(i_s,v^*_{i_s}(\bZ))\}$.
\end{algorithm}
We claim that 
\begin{align}\label{privatemax.dp}
    \text{the `Private Max', as described  in Algorithm \ref{alg:privatemax}, is}~(\epsilon,0)\text{-DP}\,.
\end{align}
Let $\Lambda= 2\epsilon^{-1}\lambda\sqrt{5s\log(1/\delta)}$. By \eqref{privatemax.dp}, we know that each `Private Max' in `Peeling' is 
 $({3\epsilon}/\{2\sqrt{5s\log(1/\delta)}\},0)$-DP. Then, by Lemma 2.9 of \citeS{dwork2018dpfdrS}, `Peeling' with the Laplace noise scale set to $\Lambda =  2\epsilon^{-1} \lambda \sqrt{5s \log(1/\delta)}$ ensures $(\tilde{\epsilon}, \delta)$-DP with 
\begin{align*}
    \tilde\epsilon =&~ \frac{3\epsilon}{2\sqrt{5s\log(1/\delta)}}\sqrt{2s\log(1/\delta)}+ \frac{3s\epsilon}{2\sqrt{5s\log(1/\delta)}}\big(e^{\frac{3\epsilon}{2\sqrt{5s\log(1/\delta)}}}-1\big)\\ 
    = &~ \frac{3\sqrt{2} }{2\sqrt{5}}\epsilon+ \frac{3s\epsilon }{2\sqrt{5s\log(1/\delta)}}\big(e^{\frac{3\epsilon}{2\sqrt{5s\log(1/\delta)}}}-1\big) \,.
\end{align*} 
Since $f(x)=ax^{1/2}(e^{ax^{-1/2}}-1)$ is monotonically decreasing in $x\geq 0$ for any $a>0$, and under the conditions $\epsilon\leq 0.5$, $\delta\leq 0.011$ and $s\geq 10$, we have ${3s}\{2\sqrt{5s\log(1/\delta)}\}^{-1}[e^{{3\epsilon}/\{2\sqrt{5s\log(1/\delta)}\}}-1] < 0.0512$. Since $3\sqrt{2}/(2\sqrt{5}) + 0.0512 < 1$, it follows that $\tilde{\epsilon} \leq \epsilon$. To confirm  Lemma \ref{lem.HD.noiseht}, it remains to show that \eqref{privatemax.dp} holds.

\noindent 
\underline{Proof of \eqref{privatemax.dp}.} For any index $j \in [p]$ and any measurable set $S \subseteq \mathbb{R}$, it suffices to show that
\begin{align*}
    \frac{\PP(\tilde{v}_j~\text{is the largest and}~v_j(\bZ)+w\in S)}{\PP(\tilde{v}_j'~\text{is the largest and}~v_j(\bZ')+w\in S)}\leq e^{\epsilon}\,,
\end{align*}
where $\tilde{v}_j'$ corresponds to $\tilde{v}_j$ but is evaluated on an adjacent database $\bZ'$.   This inequality is equivalent to 
\begin{align}\label{dp.target}
    \frac{\PP(\tilde{v}_j~\text{is largest})\PP( v_j(\bZ)+w\in S\,|\,\tilde{v}_j~\text{is largest})}{\PP(\tilde{v}_j'~\text{is largest})\PP( v_j(\bZ')+w\in S\,|\,\tilde{v}_j'~\text{is largest})}\leq e^{\epsilon}\,.
\end{align}
First, observe that by assumption $|v_j(\bZ)-v_j(\bZ')|\leq \Delta$. Then, Lemma 2.3 of \citeS{dwork2018dpfdrS} shows that 
\begin{align}\label{dp.target1}
       \frac{ \PP( v_j(\bZ)+w\in S\,|\,\tilde{v}_j~\text{is largest})}{ \PP( v_j(\bZ')+w\in S\,|\,\tilde{v}_j'~\text{is largest})}\leq e^{\epsilon/3}\,.
\end{align}
Second, let $\{w_j\}_{j=1}^p$ be the Laplace random variables generated in Algorithm \ref{alg:privatemax} and $\bw_{-j}={\rm remove}(w_{j},\bw)$, where $\bw=(w_1,\ldots,w_p)^{\T}$. Then, proving
\begin{align*} 
    \frac{\PP(\tilde{v}_j~\text{is largest}) }{\PP(\tilde{v}_j'~\text{is largest}) }\leq e^{2\epsilon/3}\,,
\end{align*}
 is equivalent to show 
 \begin{align}\label{dp.target2.equiv}
     \frac{\PP(\tilde{v}_j~\text{is largest}  \,| \, \bZ, \bw_{-j})}{\PP(\tilde{v}_j'~\text{is largest}  \,| \, \bZ', \bw_{-j})}\leq e^{2\epsilon/3}\,.
 \end{align} 
Recall $\tilde{v}_j=|v_j(\bZ)|+w_j$ and $\tilde{v}_j'=|v_j(\bZ')|+w_j$. Define 
$$
w_* = \min\big[w\in \RR: |v_j(\bZ)|+w > |v_{k}(\bZ)|+w_k~~\forall ~k \in  [p]\setminus\{j\}\big]\,.
$$
Then, $\tilde{v}_j$ is the largest value in the database  $\bZ$ if and only if $w_j\geq w_*$. Due to $\max_{j\in[p]}|v_j(\bZ)-v_j(\bZ')|\leq \Delta$, we have $|v_j(\bZ')|+\Delta+w_*\geq |v_j(\bZ)|+ w_*> |v_{k}(\bZ)|+w_k\geq |v_{k}(\bZ')|-\Delta+w_k$ for any $k\in  [p]\setminus\{j\}$, which implies $|v_j(\bZ')|+2\Delta+w_*>|v_{k}(\bZ')| +w_k$ for any $k\in  [p]\setminus\{j\}$. Thus, $\tilde{v}_j'$ is the largest value in the database  $\bZ'$ if $w_j \geq 2\Delta+w_*$. Since $w_j \sim \text{Laplace}(3\Delta/\epsilon)$, it holds that
\begin{align*}
   &\PP(\tilde{v}_j'~\text{is largest}  \,| \, \bZ', \bw_{-j}) \geq \PP(w_j  \geq 2\Delta+w_*)\\
   &~~~\geq e^{-2\epsilon/3}\PP(w_j  \geq  w_*)=  e^{-2\epsilon/3}\PP(\tilde{v}_j~\text{is largest}  \, | \, \bZ, \bw_{-j})\,, 
\end{align*} 
which yields \eqref{dp.target2.equiv}. Combining this with \eqref{dp.target1} establishes \eqref{dp.target}, as desired. Hence, we complete the proof of Lemma \ref{lem.HD.noiseht}.
\qed

\subsection{Proof of Lemma~\ref{lem:HD.gs}}\label{proof.lem:HD.gs}

With a slight abuse of notation, we denote $\bSigma_{SS}$ as the submatrix of $\bSigma$ consisting of the rows and columns indexed by $S$. Recall $\cS = \{ S \subseteq [p]: 1\leq |S| \leq 2s \}$.
For any $\bbeta_1 ,  \bbeta_2 \in \mathbb{H}(s)$ such that $\bbeta_1\neq \bbeta_2$, we have $\| \bbeta_1 - \bbeta_2 \|_0 \leq 2s$, and hence,
\begin{align*}
& \widehat{\mathcal{L}}_\tau (\boldsymbol{\beta}_1 )-\widehat{\mathcal{L}}_\tau (\boldsymbol{\beta}_2 )- \langle\nabla \widehat{\mathcal{L}}_\tau (\boldsymbol{\beta}_2 ), \boldsymbol{\beta}_1-\boldsymbol{\beta}_2 \rangle \\
&~~~ =\int_0^1 \langle\nabla \widehat{\mathcal{L}}_\tau (\boldsymbol{\beta}_2+u (\bbeta_1 - \bbeta_2) )-\nabla \widehat{\mathcal{L}}_\tau (\boldsymbol{\beta}_2 ), \bbeta_1 - \bbeta_2 \rangle \,\mathrm{d} u \\
&~~~ =\frac{1}{n} \sum_{i=1}^n \int_0^1 \big\{  \psi_\tau (y_i- \langle\bx_i, \boldsymbol{\beta}_2 \rangle )-\psi_\tau (y_i- \langle\bx_i, \boldsymbol{\beta}_2+u (\bbeta_1 - \bbeta_2) \rangle ) \big\} \langle\bx_i, \bbeta_1 - \bbeta_2 \rangle\, \mathrm{d} u \\
&~~~ \leq \frac{1}{2n} \sum_{i=1}^n \langle\bx_i, \bbeta_1 - \bbeta_2 \rangle^2 \leq \frac{1}{2} \|\bbeta_1 - \bbeta_2\|_2^2 \cdot \sup_{\bu \in \mathbb{S}^{p-1}:\, \|\bu\|_0\leq 2s} \bu^{\T} \hat \bSigma \bu \\
&~~~ = \frac{1}{2}   \|\bbeta_1 - \bbeta_2\|_2^2  \cdot  \sup_{S\in \mathcal{S} }   \| \hat \bSigma_{SS}  \|\, , 
\end{align*}
where $\hat \bSigma = n^{-1} \sn \bx_i \bx_i^\T$ and the first inequality is from the fact that $\psi_{\tau}(\cdot)$ is 1-Lipschitz.

To control  $\sup_{S\in \mathcal{S} }   \| \hat \bSigma_{SS}  \|$, we consider the decomposition
$$
    \sup_{S\in \mathcal{S}} \|\widehat{\bSigma}_{SS}\| \leq \sup_{S\in \mathcal{S} } \|(\widehat{\bSigma}-\bSigma)_{SS}\| + \sup_{S\in \mathcal{S} } \|\bSigma_{SS}\| \leq \sup_{S\in \mathcal{S} } \|(\widehat{\bSigma}-\bSigma)_{SS}\| +  \lambda_1\,. 
$$
Fix a subset $S\in \cS$ and let $\bu_S \in \mathbb{S}^{p-1}$. Under Assumption~\ref{assump:design}, the projected variable $\bu_{S}^\T \bx_{S} = (\bSigma^{1/2}\bu_S)^\T \bSigma^{-1/2}\bx$ is sub-Gaussian, satisfying that $\PP( |\bu_{S}^\T \bx_{S} | \geq  \upsilon_1  \lambda_1^{1/2} z ) \leq 2 e^{-z^2/2}$ for any $z\geq 0$. By Theorem $6.5$ in \citeS{wainwright2019highS}, there exists an absolute constant $C>0$ such that, for any $z>0$, 
 \begin{align*}
       \| (\widehat{\bSigma}-\bSigma )_{SS} \|\leq C \upsilon_1^2 \lambda_1 \bigg( \sqrt{\frac{2s+z}{n}} + \frac{2s+z}{n} \bigg) 
 \end{align*}
holds with probability at least $1-e^{-z}$. Since $| \cS | \leq  \sum_{k=1}^{2s} {p\choose k} \leq \{{ep}({2s} )^{-1}\}^{2s}$, taking a union bound over all subsets in $\cS$ yields that with probability at least $1-e^{-z}$,
$$
    \sup_{S\in \cS}  \| (\widehat{\bSigma}-\bSigma )_{SS}  \| \leq C \upsilon_1^2 \lambda_1 \bigg[ 
    \sqrt{\frac{2s \log\{ep/(2s)\}  + 2s + z}{n}}+ \frac{2s \log\{ep/(2s)\} + 2s + z}{n} \bigg] \,.
$$
Provided that $n\geq C' \upsilon_1^4 \{ s \log(p/s) + z\}$ for a sufficiently large constant $C'$, the right-hand side of the above inequality is bounded further by $\lambda_1$. 

Combining the results, we conclude that with probability at least $1-e^{-z}$,
\begin{align*}
\widehat{\mathcal{L}}_\tau (\boldsymbol{\beta}_1 )-\widehat{\mathcal{L}}_\tau (\boldsymbol{\beta}_2 )- \langle\nabla \widehat{\mathcal{L}}_\tau (\boldsymbol{\beta}_2 ), \boldsymbol{\beta}_1-\boldsymbol{\beta}_2 \rangle  \leq  \frac{1}{2}   \|\bbeta_1 - \bbeta_2\|_2^2 \cdot \sup_{S\in \mathcal{S}}  \|\widehat{\bSigma}_{SS}  \| \leq  \lambda_1 \|\bbeta_1 - \bbeta_2\|_2^2\,,  
\end{align*} 
as claimed.  \qed

\subsection{Proof of Lemma~\ref{lem:HD.rsc}}\label{proof.lem:HD.rsc}

Let $\cG_1  =  \{(\bbeta_1, \bbeta_2) \in \mathbb{H}(s) \times \mathbb{H}(s): \|\bbeta_1 - \bbeta^*\|_2\leq r/2, \|\bbeta_2 - \bbeta_1\|_2\leq r \}$ and $\cG_2  =  \{(\bbeta_1, \bbeta_2) \in \mathbb{H}(s) \times \mathbb{H}(s): \|\bbeta_2 - \bbeta^*\|_2\leq r/2, \|\bbeta_2 - \bbeta_1\|_2\leq r \}$, such that $\mathbb{N}(s, r) = \cG_1 \cup \cG_2$. Without loss of generality, we assume $(\bbeta_1, \bbeta_2) \in\cG_1$ through out the proof, and bound $\widehat{\mathcal{L}}_\tau (\boldsymbol{\beta}_1 )-\widehat{\mathcal{L}}_\tau (\boldsymbol{\beta}_2 )- \langle\nabla \widehat{\mathcal{L}}_\tau (\boldsymbol{\beta}_2 ), \boldsymbol{\beta}_1-\boldsymbol{\beta}_2 \rangle $ from below, uniformly over $(\bbeta_1, \bbeta_2) \in\cG_1$. It is straightforward to show that a similar result holds for the case $(\bbeta_1, \bbeta_2) \in\cG_2$.

For some $0<\tau_0\leq \tau$ to be determined, define the events 
$$
    \cF_i  = \{ |\varepsilon_i | \leq \tau/4\} \cap \{ | \langle \bx_i, \bbeta_1 - \bbeta^* \rangle | \leq \tau/4 \} \cap \{ | \langle \bx_i, \bbeta_1 -\bbeta_2 \rangle | \leq  \tau_0 \| \bbeta_1 - \bbeta_2 \|_2/(2r) \} \,.
$$ 
Moreover, define $\chi_i  = \mathbbm{1} ( |\varepsilon_i | \leq \tau / 4 )$ for $i\in[n]$, and let $\bv  = (\bbeta_1 - \bbeta_2) /\|\bbeta_1 - \bbeta_2\|_2\in \mathbb{S}^{p-1}$. For any $R>0$, let  $\varphi_R(\cdot) $ and  $\phi_R(\cdot) $ be the functions defined in the proof of Lemma~\ref{lem:rsc}. Then, following the proof of Lemma~\ref{lem:rsc}, it can be similarly derived that
\begin{align*}
&  \hat \cL_\tau(\bbeta_1)  - \hat \cL_\tau (\bbeta_2) - \langle \nabla \hat \cL_\tau(\bbeta_2), \bbeta_1-\bbeta_2 \rangle \\
&~~~ \geq   \frac{1}{2 n} \sn  \langle \bx_i, \bbeta_1 - \bbeta_2 \rangle^2 \mathbbm{1}(\cF_i) \\ 
    &~~~ \geq \frac{1}{2n} \sn   \chi_i  \varphi_{\| \bbeta_1 - \bbeta_2 \|_2  \tau_0/(2 r)} (\langle \bx_i, \bbeta_1 - \bbeta_2 \rangle  )   \phi_{\tau/4} (\langle \bx_i, \bbeta_1 - \bbeta^* \rangle )   \nn \\
    & ~~~= \frac{1}{2} \| \bbeta_1 - \bbeta_2 \|_2^2 \cdot 
    \underbrace{ \frac{1}{  n  } \sn   \chi_i   \varphi_{  \tau_0/(2r)} (\langle \bx_i,\bv \rangle  )   \phi_{\tau /4} (\langle \bx_i, \bbeta_1 - \bbeta^* \rangle )  }_{ =: V_n(\bbeta_1, \bbeta_2) } \\
    &~~~ = \frac{1}{2} \|\bbeta_1 - \bbeta_2\|_2^2   \big[  \EE \{V_n (\boldsymbol{\beta}_1, \boldsymbol{\beta}_2 )\} +   V_n (\boldsymbol{\beta}_1, \boldsymbol{\beta}_2 ) - \EE \{V_n (\boldsymbol{\beta}_1, \boldsymbol{\beta}_2 ) \}     \big] \,.
\end{align*}
As in the proof of Lemma~\ref{lem:rsc}, by taking $\tau_0 = 16 (\kappa_4 \lambda_1)^{1/2} r$, it holds that, as long as $\tau \geq 16 \max\{ \sigma_0, (\kappa_4 \lambda_1)^{1/2} r \}$, 
$$\
\underset{(\bbeta_1, \bbeta_2)\in \cG_1}{\inf} \EE \{V_n (\boldsymbol{\beta}_1, \boldsymbol{\beta}_2 ) \}\geq  c_0 \lambda_p\,, 
$$
where $c_0 = (1- 1/16)(1- 1/8) \approx 0.82$.

Next, we aim to bound the supremum $\Omega(r): = \sup_{(\bbeta_1, \bbeta_2) \in \cG_1} (\EE - 1) \{V_n(\bbeta_1, \bbeta_2)\}$. By \eqref{talagrand.concentration}, it holds with probability at least $1-e^{-z}$ that
$$ 
    \Omega(r) \leq  \frac{5}{4} \EE \{\Omega(r)\} + \kappa_4^{1/2} \lambda_1 \sqrt{ \frac{2 z}{n} } +   \frac{13}{3} \bigg(\frac{\tau_0}{4 r} \bigg)^2\frac{z}{n}\,. 
$$
It remains to upper bound $\EE\{ \Omega(r)\}$. Define   $\mathbb{Z}_{\bbeta_1, \bbeta_2 } $ in the same way as \eqref{def:ep:Z}. In the proof of Lemma~\ref{lem:rsc} (low-dimensional case), a key step is to bound $\mathbb{E} ( \sup_{ \bbeta_1, \bbeta_2  } \mathbb{Z}_{\bbeta_1, \bbeta_2} )$. To achieve this, we utilize the inequality $| \langle \ba, \bb \rangle | \leq \| \ba \|_2 \| \bb \|_2$ for any two vectors $\ba$ and $\bb$. In high dimensions, using instead the bound $| \langle \ba, \bb \rangle | \leq \| \ba \|_1 \| \bb \|_\infty$, we have
\begin{align}
     \EE \bigg( \sup_{(\bbeta_1, \bbeta_2)\in \cG_1}   \mathbb{Z}_{\bbeta_1, \bbeta_2 } \bigg)
	 &\leq \frac{\tau_0}{\sqrt{2} r^2} \cdot \sup_{(\bbeta_1, \bbeta_2)\in \cG_1} \| \bbeta_1 -\bbeta^* \|_1  \cdot \EE \bigg(\bigg\| \frac{1}{n} \sn  g_i' \chi_i \bx_i \bigg\|_\infty\bigg) \nn \\
	 & \quad + \frac{\tau_0}{\sqrt{2} r } \cdot \sup_{(\bbeta_1, \bbeta_2)\in \cG_1} \frac{\| \bbeta_1 -\bbeta_2 \|_1}{\| \bbeta_1 -\bbeta_2 \|_2}  \cdot  \EE \bigg(\bigg\| \frac{1}{n} \sn  g''_i \chi_i \bx_i \bigg\|_\infty \bigg) \nn \\
	 &\leq  \frac{3\tau_0\sqrt{s}}{2 r }\cdot \EE \bigg(\bigg\| \frac{1}{n} \sn  g'_i \chi_i \bx_i \bigg\|_\infty\bigg)\,,  \label{bound:gaussian-process-over-G1}
\end{align}
where the last step is based on the facts that $\| \bbeta_1 -\bbeta^* \|_1\leq \sqrt{s+s^*}\| \bbeta_1 -\bbeta^* \|_2\leq  \sqrt{2s}{r}/{2}$ and $\|\bbeta_1 - \bbeta_2\|_1 \leq \sqrt{2s} \|\bbeta_1 - \bbeta_2\|_2$ for $(\bbeta_1, \bbeta_2)\in \cG_1$.

By \eqref{tail.xij}, it holds that $ \PP(|x_{i,j}|\geq   \upsilon_1 \lambda_1^{1/2} z) \leq   2 e^{-z^2/2}$ for any $z\geq 0$. Then for $m \geq 2$, we have
\begin{align*}
	\mathbb{E}( |x_{i,j} |^{m}) &= (\upsilon_1 \lambda_1^{1/2})^{m} m \int_{0}^{\infty} t^{m-1} \mathbb{P} ( |x_{i,j} | \geq \upsilon_1 \lambda_1^{1/2} t )\, \mathrm{d}t \\
	& \leq 2  (\upsilon_1 \lambda_1^{1/2} )^{m} m \int_{0}^{\infty} t^{m-1} e^{-t^{2} / 2} \mathrm{d}t=2^{m / 2} (\upsilon_1 \lambda_1^{1/2})^{m} m \Gamma(m/2)  
\end{align*} 
for all $i\in[n]$ and $j\in[p]$.  
We first control the moments $ \mathbb{E}(|g'_i\chi_i x_{i,j} |^{m})$  for $ m\geq 2$. When $m=2$,  $\mathbb{E} (|g'_{i} \chi_i x_{i,j}|^2) \leq \sigma_{j,j} \leq \lambda_1$, where $\sigma_{j,j}$ is the $j$-th diagonal entry of $\bSigma$. 
For $m\geq 3$, since $g^{\prime}_i \sim N(0,1)$ is independent of $\bx_i$, applying the Legendre duplication formula $\Gamma(s) \Gamma(s+1 / 2)=2^{1-2 s} \sqrt{\pi} \Gamma(2s)$, we have
\begin{align*}
    \mathbb{E}(|g'_i\chi_i x_{i,j} |^{m}) & \leq 2^{m / 2} \frac{\Gamma (\frac{m+1}{2} )}{\sqrt{\pi}} \cdot 2^{m / 2} (\upsilon_1 \lambda_1^{1/2})^{m} m \Gamma(m / 2) \\
    & =2 (\upsilon_1 \lambda_1^{1/2})^{m} m !=\frac{m !}{2} \cdot 4\upsilon_1^{2} \lambda_1  \cdot (\upsilon_1 \lambda_1^{1/2})^{m-2}\,. 
\end{align*} 
For any $j\in [p]$, we denote $S_j = \sum_{i=1}^{n} g'_{i} \chi_i  x_{i, j}$.  Since $S_j$ is symmetric, by the  Bernstein inequality,
$$
\log \mathbb{E} (e^{-\lambda S_{j}}) = \log \mathbb{E} (e^{\lambda S_{j}}) \leq \frac{4 \upsilon_1^{2} \lambda_1 n \lambda^{2}}{2 (1-\upsilon_1 \lambda_1^{1/2} \lambda)}\,, 
$$
for all $\lambda \in (0,1 / (\upsilon_1 \lambda_1^{1/2}) )$. By Corollary 2.6 in \citeS{boucheron2013concentrationS},  
\begin{align*}
     \EE\bigg( \bigg\| \frac{1}{n} \sn  g'_i \chi_i \bx_i \bigg\|_\infty\bigg) & =\mathbb{E} \bigg(\max _{j\in[p]} |S_{j} / n | \bigg) =\mathbb{E}\bigg(
     \max_{j\in[p]} \, \max\{S_{j} / n,  -S_{j} / n\} \bigg) \\
     & \leq  \upsilon_1  \lambda_1^{1/2} \bigg\{ 2 \sqrt{\frac{2 \log (2 p)}{n}}+\frac{\log (2 p)}{n} \bigg\}\,.
\end{align*}
Plugging this upper bound into inequality (\ref{bound:gaussian-process-over-G1}), we obtain 
\begin{align*}
      \EE \bigg( \sup_{(\bbeta_1, \bbeta_2)\in \cG_1}   \mathbb{Z}_{\bbeta_1, \bbeta_2 } \bigg) \leq  \frac{3\tau_0 \sqrt{s}}{2 r }\cdot \upsilon_1 \lambda_1^{1/2} \bigg\{ 2 \sqrt{\frac{2 \log (2 p)}{n}}+\frac{\log (2 p)}{n}\bigg\}\,. 
\end{align*}
Under the assumption $n\gtrsim \log p$, it follows that $\EE \{\Omega(r) \}\lesssim \kappa_4^{1/2}\upsilon_1  \lambda_1 \sqrt{  s(\log p)/n}$. 
Combining the bounds derived above, we conclude that with probability at least $1-e^{-z}$,
$$
    \Omega(r)  \lesssim  \kappa_4^{1/2} \upsilon_1  \lambda_1 \sqrt{\frac{s\log p}{n}} +  \kappa_4^{1/2} \lambda_1 \sqrt{\frac{z}{n}} + \kappa_4 \lambda_1 \frac{z}{n}\,.
$$
Let the sample size satisfy $n \gtrsim \kappa_4 \upsilon_1^2  (\lambda_1 / \lambda_p )^2 (s\log p + z)$ such that $ \Omega(r)  \leq  (c_0-1/4) \lambda_p$. This result, combined with the lower bound on $ \EE\{ V_n (\boldsymbol{\beta}_1, \boldsymbol{\beta}_2 )\}$, establishes that 
$$
 \hat \cL_\tau(\bbeta_1)  - \hat \cL_\tau (\bbeta_2) - \langle \nabla \hat \cL_\tau(\bbeta_2), \bbeta_1-\bbeta_2 \rangle   \geq \frac{1}{8} \lambda_p \| \bbeta_1 - \bbeta_2 \|_2^2  
$$
holds  uniformly over $(\bbeta_1, \bbeta_2) \in \cG_1$. This completes the proof of the lemma. \qed

\subsection{Proof of Lemma~\ref{lem:HD.huber.grad_diff.ubd}}\label{proof.lem:HD.huber.grad_diff.ubd}

For any  $\bbeta \in \Theta(r)\cap\mathbb{H}(s)$ and $S\in \cS$, we have 
\begin{align*}
& \big\|\EE[\{\nabla\hat \cL_\tau(\bbeta)  - \nabla\hat \cL_\tau (\bbeta^*)\}_S]\big\|_2 \\
&~~~=\sup_{\bu\in\mathbb{S}^{p-1} : \, {\rm supp}(\bu)=S} \EE \bigg[ \frac{1}{n}\sum_{i=1}^n \{ \psi_\tau(\varepsilon_i) - \psi_\tau(\varepsilon_i-\langle \bx_i,\bbeta-\bbeta^*\rangle) \}  \bx_{i,S}^\T \bu  \bigg] \\
&~~~\leq \| \bbeta -\bbeta^* \|_{\bSigma}  \sup_{\bu\in\mathbb{S}^{p-1} : \, {\rm supp}(\bu)=S}  \| \bu \|_{\bSigma} \leq \lambda_1 r  \,.
\end{align*}
Taking the supremum over all $\bbeta \in \Theta(r) \cap \mathbb{H}(s)$ and $S\in \cS$ gives
\begin{align}
\sup_{ \bbeta \in \Theta(r)\cap\mathbb{H}(s) , \, S\in \cS}  \big\|\EE[\{\nabla\hat \cL_\tau(\bbeta)  - \nabla\hat \cL_\tau (\bbeta^*)\}_S]\big\|_2 \leq \lambda_1 r \,.  \label{mean.sup.bound}
\end{align}

To upper bound $\sup_{S \in \cS, \,  \bbeta \in \Theta(r)\cap\mathbb{H}(s)}  \|(1-\EE)[\{\nabla\hat \cL_\tau(\bbeta)  - \nabla\hat \cL_\tau (\bbeta^*)\}_S ] \|_2$, first fix the subset $S \in \cS$. Following the proof of inequality (2.5) in \citeS{chen2020robustS}, it can be similarly shown that, for any $z\geq 1/2$ and $r>0$, with probability at least $1-e^{-z}$,
$$
   \sup_{ \bbeta \in \Theta(r)\cap\mathbb{H}(s)} \big\|(1-\EE)[\{\nabla\hat \cL_\tau(\bbeta)  - \nabla\hat \cL_\tau (\bbeta^*)\}_S ] \big\|_2  \leq C \lambda_1   \upsilon_1^2r \sqrt{\frac{s+ z}{n}}\,, 
$$
as long as $n\gtrsim s+z$. Recall that $|\cS| \leq \{{ep}({2s})^{-1}\}^{2s}$. Taking a union bound over all $S \in \cS$ yields that, with probability at least $1-e^{-z}$,
$$
\sup_{S\in \cS, \, \bbeta \in \Theta(r)\cap\mathbb{H}(s)} \big\|(1-\EE)[\{\nabla\hat \cL_\tau(\bbeta)  - \nabla\hat \cL_\tau (\bbeta^*)\}_S ] \big\|_2 \leq C \lambda_1  \upsilon_1^2 r \sqrt{\frac{s\log(e p/s) + z}{n}}\,. 
$$
Combining this with \eqref{mean.sup.bound}, it follows that as long as $n\gtrsim \upsilon_1^4 (s\log p + z)$, 
$$
   \sup_{S\in \cS , \,  \bbeta \in \Theta(r)\cap\mathbb{H}(s)} \big\| \{\nabla\hat \cL_\tau(\bbeta)  - \nabla\hat \cL_\tau (\bbeta^*)\}_S  \big\|_2 \leq \frac{3}{2} \lambda_1 r
$$
holds with probability at least $1-e^{-z}$, as desired. \qed
 
\subsection{Proof of Lemma~\ref{lem:HD.huber.grad_oracle.ubd}}\label{proof.lem:HD.huber.grad_oracle.ubd}

For each $S\in \cS$, note that 
$$
    \| \EE[\{\nabla\hat \cL_\tau (\bbeta^*)\}_S]   \|_2  =\sup_{\bu\in\mathbb{S}^{p-1}: \, {\rm supp}(\bu)=S }  \frac{1}{n} \sn \EE \{ -\psi_\tau(\varepsilon_i)\bx_{i, S}^\T \bu\}\,.
$$
Under Assumption \ref{assump:heavy-tail}, we have $\EE(\varepsilon_i \,|\, \bx_i) = 0$, and thus 
$$
    -\mathbb{E}\{\psi_\tau(\varepsilon_i  )\, |\, \bx_i \}=\mathbb{E}\{\varepsilon_i \mathbbm{1}(|\varepsilon_i |>\tau)-\tau \mathbbm{1}(\varepsilon_i>\tau )+\tau \mathbbm{1}(\varepsilon_i<-\tau) \,|\, \bx_i\}\,,
$$ 
which implies 
$$
    |\mathbb{E}\{\psi_\tau(\varepsilon_i) \,| \,\bx_i\}| \leq \mathbb{E}\{(|\varepsilon_i|-\tau) \mathbbm{1}(|\varepsilon_i|>\tau) \,|\, \bx_i\} \leq  \sigma_{\iota}^{2+\iota} \tau^{-1-\iota}\,.
$$
Then, by the Cauchy-Schwartz inequality and Assumption \ref{assump:design}, we have $\|\EE[ \{\nabla\hat \cL_\tau (\bbeta^*)\}_S] \|_2 \leq  \lambda_1^{1/2} \sigma_{\iota}^{2+\iota} \tau^{-1-\iota}$. Taking the supremum over $S\in \cS$ yields
\begin{align}
    \sup_{S\in \cS}\big\|  \EE [\{ \nabla\hat \cL_\tau (\bbeta^*) \}_S] \big\|_2 \leq  \lambda_1^{1/2} \sigma^{ 2+\iota}_{\iota} \tau^{-1-\iota}\,.
\label{mean.gradient.bound}
\end{align}

It remains to upper bound $\sup_{S\in \cS}  \|(1-\EE) [\{\nabla\hat \cL_\tau (\bbeta^*)\}_S] \|_2 = \sup_{S\in \cS}  \| \bv_S \|_2$, where $\bv = n^{-1} \sum_{i=1}^n (1-\EE)\{ \psi_{\tau}(\varepsilon_i) \bx_i \} \in \RR^p$. Let $S\subseteq [p]$ be a subset with cardinality $|S|=k$ for some $1\leq k\leq 2s$. Assumption~\ref{assump:design} ensures that, for any $\bu \in \mathbb{S}^{p-1}$ that is supported on $S$, $\bx_{i, S}^\T \bu = \bx_i^\T \bu$ is sub-Gaussian, satisfying
$\PP( | \bx_{i, S}^\T \bu | \geq \lambda_1^{1/2} \upsilon_1 z )  \leq 2 e^{-z^2/2}$ for any $z\geq 0$. Then, following the proof of (B.8) in the supplementary material of \citeS{chen2020robustS}, it can be shown that, with probability at least $1- e^{-z}$,
$$
     \| \bv_S \|_2 \lesssim \lambda_1^{1/2} \upsilon_1  \bigg(  \sigma_0 \sqrt{\frac{k+z}{n}} + \tau \frac{k+z}{n} \bigg) \,.
$$
Taking a union bound over all $S \in \cS$ with $|\cS|\leq \{{ep}({2s})^{-1}\}^{2s}$, we conclude that with probability at least $1-e^{-z}$,
$$
    \sup_{S \in \cS} \| \bv_S \|_2 \lesssim \lambda_1^{1/2} \upsilon_1  \bigg\{  \sigma_0 \sqrt{\frac{s \log(ep/s) +z}{n}} + \tau \frac{s \log(ep/s) +z}{n} \bigg\} \,.
$$
This, combined with \eqref{mean.gradient.bound}, proves the claim. \qed

\subsection{Proof of Lemma~\ref{lem:HD.max.Laplace}} \label{proof.lem:HD.max.Laplace}

By the tail bound of $\operatorname{Laplace}(\lambda_{\operatorname{dp}})$, for any $z>0$, we have
$$
\PP(| w^t_{i, j}|\geq \lambda_{\operatorname{dp}}z) \leq e^{-z} ~~\mbox{ and }~~ \PP(|\widetilde{w}^t_j|\geq \lambda_{\operatorname{dp}}z) \leq e^{-z}  
$$
for all $i \in[s]$, $j\in[p]$, and $t \in \{0\}\cup [T-1]$. Taking a union bound over $j \in[p]$, $i\in[s]$, and $t \in \{0\}\cup [T-1]$, for any $z>0$, we have that for any $z>0$, with probability at least $1-e^{-z}$, it holds that
$$\max _{ t \in \{0\}\cup [T-1]} \sum_{i=1}^s\|\bw^t_i\|_{\infty}^2 \leq s\lambda_{\operatorname{dp}}^2\{z+\log(spT)\}^2\,, $$ 
and with probability at least $1-e^{-z}$,
$$
\max _{ t \in \{0\}\cup [T-1]}  \| \widetilde{\bw}^t_{S^{t+1}} \|_2^2 \leq s\lambda_{\operatorname{dp}}^2\{z+\log(sT)\}^2\,.
$$ 
Combining the last two displays with $\lambda_{{\rm dp}} = 2 \epsilon^{-1}T \lambda\sqrt{5s \log( T/\delta)}$, it follows that there exists an absolute constant $C_0>0$ such that, for any $z>0$, with probability at least $1-2e^{-z}$, 
\begin{align*}
\max _{ t \in \{0\}\cup [T-1]} W_t \leq C_0 \bigg( \frac{sT \lambda}{\epsilon} \bigg)^2 \{\log (p s T)+z\}^2   \log\bigg(\frac{T}{\delta}\bigg)  \, , 
\end{align*}
as claimed. \qed

\subsection{Proof of Lemma \ref{lem.cons}}\label{proof.lem.cons}
If $\tilde{S}\setminus S = \emptyset$ (i.e., $|\tilde{S}|=s$), it is evident that $\| \{ \tilde P_s(\bxi) - \bxi\}_{\tilde{S}} \|_2^2 = 0 $, and thus \eqref{cons.bound} holds trivially. If $|\tilde{S}|>s$, let $\bxi=(\xi_1,\ldots,\xi_p)^{\T}$, and let $H$ be the index set of the top $s$ coordinates of $\bxi_{\tilde{S}}$ in magnitude. Then, we have
    \begin{align}\label{div}
        \| \{ \tilde P_s(\bxi) - \bxi\}_{\tilde{S}} \|_2^2 = &\sum_{j\in    \tilde{S} \setminus S } \xi_{j}^2=\sum_{j\in (\tilde{S} \setminus S)\cap H^{c}}\xi_{j}^2+\sum_{j\in (\tilde{S} \setminus S)\cap H}\xi_{j}^2\nn\\
        \overset{{\rm (i)}}{\leq } & \sum_{j\in (\tilde{S} \setminus S)\cap H^{c}}\xi_{j}^2+ (1+c^{-1})\sum_{j\in S\cap H^c}\xi_{j}^2 + 4(1+c)\sum_{i\in[s]}\|\bw_i\|_{\infty}^2\nn\\
        \leq & (1+c^{-1})\sum_{j\in  \tilde{S} \cap H^{c}}\xi_{j}^2+4(1+c)\sum_{i\in[s]}\|\bw_i\|_{\infty}^2\,,
    \end{align}
    where inequality (i) follows from Lemma 3.4 of \citeS{cai2021costS}. Notice that $\tilde{S}\cap H^c$ is the set of indices corresponding to the smallest absolute values in $\bxi_{\tilde{S}}$. Let $\hat\bxi=(\smallhat{\xi}_1,\ldots,\smallhat{\xi}_p)^{\T}$, $\hat{S}={\rm supp}(\hat{\bxi})$ and $\hat{S}_0 = \hat{S} \cap \tilde{S}$. Due to $\smallhat{s}_0 := |\hat{S}_0|\leq \smallhat{s}$, it holds that
    \begin{align*}
        \frac{1}{|\tilde{S}|-s}\sum_{j\in  \tilde{S} \cap H^{c}}\xi_{j}^2 \leq  \frac{1}{|\tilde{S}|-\smallhat{s}_0}\sum_{j\in  \tilde{S} \cap \hat{S}_0^{c}}\xi_{j}^2 \leq  \frac{1}{|\tilde{S}|-\smallhat{s}}\sum_{j\in  \tilde{S} \cap \hat{S}_0^{c}}(\smallhat{\xi}_j-\xi_{j})^2 \leq \frac{1}{|\tilde{S}|-\smallhat{s}}\| \{\hat{\bxi}-\bxi\}_{\tilde{S}} \|_2^2\,.
    \end{align*}
    This, together with \eqref{div}, implies \eqref{cons.bound}.\qed 
 
\subsection{Proof of Lemma~\ref{lem:HD:prop:1}} \label{proof.lem:HD:prop:1}
Conditioned on $\cE_0$, it holds that $\bbeta^{(t+1)} = \widetilde{P}_s(\widetilde{\bbeta}^{(t+1)}) + \widetilde{\bw}^t_{S^{t+1}} = (\bbeta^{(t)} - \eta_0 \boldsymbol{g}^{(t)} + \widetilde{\bw}^t)_{S^{t+1}}$.  
This, together with $(\bbeta^{(t+1)})_{S^t\setminus S^{t+1}} = \boldsymbol{0}$ and  $S^{t+1} \cap (S^t \setminus S^{t+1}) = \emptyset$, leads to  
\begin{align}\label{inner.prod.diff.grad.bound}
     \langle\bbeta^{(t+1)}-\bbeta^{(t)}, \boldsymbol{g}^{(t)}\rangle  
        & = \langle\bbeta^{(t+1)}-\bbeta^{(t)}, \boldsymbol{g}^{(t)}\rangle_{S^{t+1}} + \langle\bbeta^{(t+1)}-\bbeta^{(t)}, \boldsymbol{g}^{(t)}\rangle_{S^t\setminus S^{t+1}} \nn\\
        & = \langle\bbeta^{(t)} - \eta_0 \boldsymbol{g}^{(t)} + \widetilde{\bw}^t -\bbeta^{(t)}, \boldsymbol{g}^{(t)}\rangle_{S^{t+1}} - \langle \bbeta^{(t)}, \boldsymbol{g}^{(t)}\rangle_{S^t\setminus S^{t+1}} \\
        & =  - \eta_0 \langle \boldsymbol{g}^{(t)}, \boldsymbol{g}^{(t)}\rangle_{S^{t+1}} +  \langle  \widetilde{\bw}^t, \boldsymbol{g}^{(t)}\rangle_{S^{t+1}} - \langle \bbeta^{(t)}, \boldsymbol{g}^{(t)}\rangle_{S^t\setminus S^{t+1}} \nn\\
        & \leq - \{\eta_0 - (4 c_1)^{-1} \}\langle \boldsymbol{g}^{(t)}, \boldsymbol{g}^{(t)}\rangle_{S^{t+1}} + c_1 \langle \widetilde{\bw}^t, \widetilde{\bw}^t\rangle_{S^{t+1}} - \langle \bbeta^{(t)}, \boldsymbol{g}^{(t)}\rangle_{S^t\setminus S^{t+1}}\nn 
\end{align} 
for any constant $c_1>0$. Here, we write $\langle \ba, \bb \rangle_S=\ba_S^{\T}\bb_S$ for any $\ba, \bb \in \RR^p$ and $S\subseteq [p]$.

It is more involved in bounding $\langle \bbeta^{(t)}, \boldsymbol{g}^{(t)}\rangle_{S^t\setminus S^{t+1}}$ from below. Since $(S^{t+1} \setminus S^t )\subseteq S^{t+1}$, $(S^t \setminus S^{t+1} )\subseteq S^t$ and $|S^{t+1} \setminus S^t|=|S^t \setminus S^{t+1}|$, 
Lemma 3.4 in \citeS{cai2021costS} yields that, for any $c_2>0$, 
\begin{equation} \label{HD.prop2.inner.sub1}
\begin{aligned}
    \langle \widetilde{\bbeta}^{(t+1)}, \widetilde{\bbeta}^{(t+1)}\rangle_{S^t \setminus S^{t+1}} & \leq (1+c_2)\langle \widetilde{\bbeta}^{(t+1)}, \widetilde{\bbeta}^{(t+1)}\rangle_{S^{t+1} \setminus S^t} +4(1+c_2^{-1}) \sum_{i=1}^s\|\bw^t_i\|_{\infty}^2\\
    & = (1+c_2)\langle \eta_0 \boldsymbol{g}^{(t)}, \eta_0 \boldsymbol{g}^{(t)}\rangle_{S^{t+1} \setminus S^t} +4(1+c_2^{-1}) \sum_{i=1}^s\|\bw^t_i\|_{\infty}^2\,. 
\end{aligned}
\end{equation}
In the last step, we used the facts that $\widetilde{\bbeta}^{(t+1)}=\bbeta^{(t)}-\eta_0 \boldsymbol{g}^{(t)}$ and $\bbeta^{(t)}_{S^{t+1}\setminus S^t}= \boldsymbol{0}$. 

In the meanwhile,  
\begin{align}\label{HD.prop2.inner.sub2}
    \langle \widetilde{\bbeta}^{(t+1)}, \widetilde{\bbeta}^{(t+1)}\rangle_{S^t \setminus S^{t+1}} & = \langle \bbeta^{(t)}-\eta_0 \boldsymbol{g}^{(t)}, \bbeta^{(t)}-\eta_0 \boldsymbol{g}^{(t)}\rangle_{S^t \setminus S^{t+1}} \nn\\
    & = \langle \bbeta^{(t)}, \bbeta^{(t)}\rangle_{S^t \setminus S^{t+1}} + \langle \eta_0 \boldsymbol{g}^{(t)}, \eta_0 \boldsymbol{g}^{(t)}\rangle_{S^t \setminus S^{t+1}} - 2 \eta_0 \langle \bbeta^{(t)} ,  \boldsymbol{g}^{(t)}\rangle_{S^t \setminus S^{t+1}}\nn \\
    & \geq  \langle \eta_0 \boldsymbol{g}^{(t)}, \eta_0 \boldsymbol{g}^{(t)}\rangle_{S^t \setminus S^{t+1}} - 2 \eta_0 \langle \bbeta^{(t)} ,  \boldsymbol{g}^{(t)}\rangle_{S^t \setminus S^{t+1}}\,. 
\end{align}
Combining \eqref{HD.prop2.inner.sub1} and \eqref{HD.prop2.inner.sub2} yields 
\begin{align*}
    & -  \langle \bbeta^{(t)},  \boldsymbol{g}^{(t)}\rangle_{S^t \setminus S^{t+1}}   \leq(1+c_2) (\eta_0/2)\big\|\boldsymbol{g}^{(t)}_{S^{t+1} \setminus S^t}\big\|_2^2-(\eta_0/2)\big\|\boldsymbol{g}^{(t)}_{S^t \setminus S^{t+1}}\big\|_2^2+2(1+c_2^{-1})\eta_0^{-1} \sum_{i=1}^s\|\bw^t_i\|_{\infty}^2 \,.
\end{align*}
Plugging this to \eqref{inner.prod.diff.grad.bound} yields
\begin{align*}
     \langle\bbeta^{(t+1)}-\bbeta^{(t)}, \boldsymbol{g}^{(t)}\rangle  &\leq(1+c_2) (\eta_0/2)\big\|\boldsymbol{g}^{(t)}_{S^{t+1} \setminus S^t}\big\|_2^2-(\eta_0/2)\big\|\boldsymbol{g}^{(t)}_{S^t \setminus S^{t+1}}\big\|_2^2 \\
     & ~~~~-  \{\eta_0-(4 c_1)^{-1}\}\big\|\boldsymbol{g}^{(t)}_{S^{t+1}}\big\|_2^2+\{c_1+2(1+c_2^{-1})\eta_0^{-1}\} W_t\,, 
\end{align*}
where $W_t = \sum_{i=1}^s\|\bw^t_i\|_{\infty}^2+\|\widetilde{\bw}_{S^{t+1}}^t\|_2^2 $. Taking $c_1=2\eta_0^{-1}, c_2=1/4$, and using the equality
$\|\boldsymbol{g}^{(t)}_{S^{t+1}}\|_2^2=\|\boldsymbol{g}^{(t)}_{S^{t+1} \setminus S^t}\|_2^2+\|\boldsymbol{g}^{(t)}_{S^t \cap S^{t+1}}\|_2^2$,
it follows that 
\begin{align*}
      \langle\bbeta^{(t+1)}-\bbeta^{(t)}, \boldsymbol{g}^{(t)}\rangle 
    & \leq-\frac{\eta_0}{4}\big\|\boldsymbol{g}^{(t)}_{S^{t+1} \setminus S^t}\big\|_2^2-\frac{7\eta_0}{8}\big\|\boldsymbol{g}^{(t)}_{S^t \cap S^{t+1}}\big\|_2^2-\frac{\eta_0}{2}\big\|\boldsymbol{g}^{(t)}_{S^t \setminus S^{t+1}}\big\|_2^2+ \frac{12}{\eta_0} W_t \\
    & \leq-\frac{\eta_0}{4}\big\|\boldsymbol{g}^{(t)}_{S^{t+1} \cup S^t}\big\|_2^2+\frac{12}{\eta_0} W_t\,. 
\end{align*}
Hence, we obtain the first result of Lemma \ref{lem:HD:prop:1}. In the noiseless case, the second result of Lemma \ref{lem:HD:prop:1} can be proved by using similar arguments.
\qed

\subsection{Proof of Lemma~\ref{lem:HD:prop:2}} \label{proof.lem:HD:prop:2}
Since $\bbeta^{(t)}_{I^t \setminus(S^t \cup S^*)}=\mathbf{0}$, conditioned on $\cE_0$, we have 
$$
\eta_0^2\big\|\boldsymbol{g}^{(t)}_{I^t \setminus(S^t \cup S^*)}\big\|_2^2=\big\|(\bbeta^{(t)}-\eta_0 \boldsymbol{g}^{(t)} )_{I^t \setminus(S^t \cup S^*)}\big\|_2^2=\big\|\widetilde{\bbeta}^{(t+1)}_{I^t \setminus(S^t \cup S^*)}\big\|_2^2\,. 
$$
We will bound the right-hand side from below by Lemma 3.4 in \citeS{cai2021costS}. Recall $I^{t}=S^t\cup S^{t+1}\cup S^*$. Since $$|S^t \setminus S^{t+1}|=|S^{t+1} \setminus S^t| \geq|S^{t+1} \setminus(S^t \cup S^*)|=|I^t \setminus(S^t \cup S^*)|\,, $$ there exists a support set $R \subseteq S^t \setminus S^{t+1}$ with cardinality $|R|=|I^t \setminus(S^t \cup S^*)|$. Applying Lemma 3.4 in \citeS{cai2021costS} yields that, for any $c_1>0$, 
$$
\big\|\widetilde{\bbeta}^{(t+1)}_R\big\|_2^2 \leq(1+c_1)\big\|\widetilde{\bbeta}^{(t+1)}_{I^t \setminus(S^t \cup S^*)}\big\|_2^2+4(1+c_1^{-1}) \sum_{i=1}^s\|\bw^t_i\|_{\infty}^2\,. 
$$
As $R \subseteq S^t \setminus S^{t+1}$, we have $\bbeta^{(t+1)}_R=\mathbf{0}$ and hence
\begin{equation} \label{on_R}
\begin{aligned}
    & -\eta_0^2\big\|\boldsymbol{g}^{(t)}_{I^t \setminus(S^t \cup S^*)}\big\|_2^2  \leq -(1+c_1)^{-1}\big\|\{\bbeta^{(t+1)}-\bbeta^{(t)}+\eta_0 \boldsymbol{g}^{(t)}\}_R\big\|_2^2 + 4 c_1^{-1} \sum_{i=1}^s\|\bw^t_i\|_{\infty}^2\,.
\end{aligned} 
\end{equation} 
By (\ref{on_R}) and the facts that $(I^t \setminus R) \cup R = I^t$ and $(I^t \setminus R) \cap R = \emptyset$, 
\begin{align*}\label{HD.prop2.target}
& \big\|\{\bbeta^{(t+1)}-\bbeta^{(t)}+\eta_0 \boldsymbol{g}^{(t)}\}_{I^t}\big\|_2^2-\eta_0^2\big\|\boldsymbol{g}^{(t)}_{I^t \setminus(S^t \cup S^*)}\big\|_2^2 \\
&~~~~~ \leq  \big\|\{\bbeta^{(t+1)}-\bbeta^{(t)}+\eta_0 \boldsymbol{g}^{(t)}\}_{I^t\setminus R}\big\|_2^2 \\
&~~~~~~~~+ c_1(1+c_1)^{-1}\big\|\{\bbeta^{(t+1)}-\bbeta^{(t)}+\eta_0 \boldsymbol{g}^{(t)}\}_{R}\big\|_2^2 + 4 c_1^{-1} \sum_{i=1}^s\|\bw^t_i\|_{\infty}^2\,.  
\end{align*} 
We bound the terms on the right-hand side one by one. Inequality (\ref{on_R}) implies
\begin{equation}\label{HD.prop2.sub1.bound}
    \begin{aligned}
    & \big\|\{\bbeta^{(t+1)}-\bbeta^{(t)}+\eta_0 \boldsymbol{g}^{(t)}\}_{R}\big\|_2^2  \leq (1+c_1)\eta_0^2\big\|\boldsymbol{g}^{(t)}_{I^t \setminus(S^t \cup S^*)}\big\|_2^2+4 c_1^{-1}(1+c_1)\sum_{i=1}^s\|\bw^t_i\|_{\infty}^2\,.
\end{aligned}
\end{equation}
For any $c_2>0$, due to $\bbeta^{(t+1)}=\widetilde{P}_s(\widetilde{\bbeta}^{(t+1)})+\widetilde{\bw}_{S^{t+1}}^t$ and $\widetilde{\bbeta}^{(t+1)}=\bbeta^{(t)}-\eta_0\boldsymbol{g}^{(t)} $ conditioned on $\cE_0$, it holds that
\begin{equation}\label{HD.prop2.sub2}
    \begin{aligned}
    & \big\|\{\bbeta^{(t+1)}-\bbeta^{(t)}+\eta_0 \boldsymbol{g}^{(t)}\}_{I^t\setminus R}\big\|_2^2\\
    & ~~~\leq (1+c_2)\big\|\{\widetilde{P}_s(\widetilde{\bbeta}^{(t+1)})-\widetilde{\bbeta}^{(t+1)}\}_{I^t \setminus R}\big\|_2^2 +(1+c_2^{-1})\|\widetilde{\bw}_{S^{t+1}}^t\|_2^2\,.
    \end{aligned}
\end{equation}
Since $R \subseteq S^t \setminus S^{t+1}$, we have $S^{t+1} \subseteq I^t \setminus R$.
By Lemma \ref{lem.cons}, for any $c_3>0$, 
\begin{equation}\label{HD.prop2.sub3}
    \begin{aligned}
    & \big\|\{\widetilde{P}_s(\widetilde{\bbeta}^{(t+1)})-\widetilde{\bbeta}^{(t+1)}\}_{I^t \setminus R}\big\|_2^2  \\
    & ~~~\leq \big(1+c_3^{-1}\big)\frac{|I^t \setminus R|-s}{|I^t \setminus R|-s^*}  \big\|\{\bbeta^* - \widetilde{\bbeta}^{(t+1)}\}_{I^t\setminus R}\big\|_2^2+4(1+c_3) \sum_{i=1}^s\|\bw^t_i\|_{\infty}^2 \\ 
    &~~~ \leq(1+c_3^{-1})(s^*/s)\big\|\{\bbeta^* - \widetilde{\bbeta}^{(t+1)}\}_{I^t\setminus R}\big\|_2^2+4(1+c_3) \sum_{i=1}^s\|\bw^t_i\|_{\infty}^2\,, 
\end{aligned}
\end{equation}
where we used the facts $$|I^t \setminus R|=|I^t|-|R|=|I^t|-|I^t \setminus(S^t \cup S^*)|=|S^t \cup S^*| \leq s+s^*, $$ and $|I^t \setminus R| \geq s \geq s^*$. Plugging \eqref{HD.prop2.sub3} into \eqref{HD.prop2.sub2} yields 
\begin{align*} 
     \big\|\{\bbeta^{(t+1)}-\bbeta^{(t)}+\eta_0 \boldsymbol{g}^{(t)}\}_{I^t\setminus R}\big\|_2^2
    & \leq (1+c_2)(1+c_3^{-1}) (s^*/s) \big\|\{\bbeta^* - \widetilde{\bbeta}^{(t+1)}\}_{I^t\setminus R}\big\|_2^2 \\
    & \quad +4 (1+c_2)(1+c_3) \sum_{i=1}^s\|\bw^t_i\|_{\infty}^2+(1+c_2^{-1})\|\tilde{\bw}_{S^{t+1}}^t\|_2^2\,.  
\end{align*} 
As $\widetilde{\bbeta}^{(t+1)} = \bbeta^{(t)} - \eta_0 \boldsymbol{g}^{(t)}$ under $\cE_0$, by taking  $c_2 = c_3 = 1$, we have 
\begin{equation}\label{HD.prop2.sub2.bound}
    \begin{aligned} 
    & \big\|\{\bbeta^{(t+1)}-\bbeta^{(t)}+\eta_0 \boldsymbol{g}^{(t)}\}_{I^t\setminus R}\big\|_2^2 \\
    & ~~~\leq \frac{4s^*}{s}\big\|\{\bbeta^* - \bbeta^{(t)} + \eta_0 \boldsymbol{g}^{(t)}\}_{I^t\setminus R}\big\|_2^2 + 16 \sum_{i=1}^s\|\bw^t_i\|_{\infty}^2+2\|\tilde{\bw}_{S^{t+1}}^t\|_2^2\,. 
    \end{aligned} 
\end{equation}

Combining \eqref{HD.prop2.sub1.bound} with \eqref{HD.prop2.sub2.bound}, we conclude that
\begin{align*}
& \big\|\{\bbeta^{(t+1)}-\bbeta^{(t)}+\eta_0 \boldsymbol{g}^{(t)}\}_{I^t}\big\|_2^2-\eta_0^2\big\|\boldsymbol{g}^{(t)}_{I^t \setminus(S^t \cup S^*)}\big\|_2^2 \\ 
 & ~~~\leq   (4 s^*/s)\big\|\{\bbeta^* - \bbeta^{(t)} + \eta_0 \boldsymbol{g}^{(t)}\}_{I^t\setminus R}\big\|_2^2 +  c_1\eta_0^2\big\|\boldsymbol{g}^{(t)}_{I^t \setminus(S^t \cup S^*)}\big\|_2^2 \\
& ~~~\quad +20  \sum_{i=1}^s\|\bw^t_i\|_{\infty}^2+2\|\tilde{\bw}_{S^{t+1}}^t\|_2^2 + 4 c_1^{-1} \sum_{i=1}^s\|\bw^t_i\|_{\infty}^2 \\
& ~~~\leq  (4 s^*/s)\big\|\{\bbeta^* - \bbeta^{(t)} + \eta_0 \boldsymbol{g}^{(t)}\}_{I^t\setminus R}\big\|_2^2 + c_1\eta_0^2\big\|\boldsymbol{g}^{(t)}_{I^t \setminus(S^t \cup S^*)}\big\|_2^2 + 4(5+c_1^{-1})W_t \\ 
& ~~~\leq  (4 s^*/s)\big\|\{\bbeta^* - \bbeta^{(t)} + \eta_0 \boldsymbol{g}^{(t)}\}_{I^t}\big\|_2^2 +   c_1\eta_0^2\big\|\boldsymbol{g}^{(t)}_{I^t \setminus(S^t \cup S^*)}\big\|_2^2 + 4(5+c_1^{-1})W_t  
\end{align*}
for any $c_1>0$. Hence, we obtain the first result of Lemma \ref{lem:HD:prop:2}. In the noiseless case, the second result of Lemma \ref{lem:HD:prop:2} can be proved by using similar arguments.\qed

\section{Privacy Guarantees and Additional Numerical Results}
\label{privacy.guarantee}

\subsection{Privacy guarantees of the initialization step}

\subsubsection{Privacy guarantee of the selection of $\tau_0$ in \eqref{ridge.beta}}

Since $|\tilde{y}_i|\leq \log n$ for all $i\in[n]$, the $\ell_1$-sensitivities of $n^{-1}\sum_{i=1}^n \tilde{y}_i$ and $  n^{-1}\sum_{i=1}^n \tilde{y}_i^2 $ are bounded by $2n^{-1} \log n $ and $n^{-1}(\log n)^2$, respectively. Hence, by the Laplace mechanism, $m_1 = n^{-1}\sum_{i=1}^n\tilde{y}_i+w_1$ and $m_2 = n^{-1}\sum_{i=1}^n\tilde{y}_i^2+w_2$ are both $(\epsilon_{\rm init}/8,0)$-DP, where $w_1\sim {\rm Laplace}(16(n\epsilon_{\rm init})^{-1}\log n)$ and $w_2 \sim {\rm Laplace}(8(n\epsilon_{\rm init})^{-1}(\log n)^2)$. Therefore, by the standard composition and post-processing properties of differential privacy (see Lemmas \ref{post} and \ref{standardcompositiontheorem}), the quantity $m_2 - m_1^2$, and hence the resulting $\tau_0$, is $(\epsilon_{\rm init}/4,0)$-DP.  \qed

\subsubsection{Privacy guarantee of the output perturbation in Section \ref{selection.low}}

Define $F(\bbeta, {\bZ} )={n}^{-1}\sum_{i=1}^n\rho_{\tau_0}(y_i-\tilde{\bx}_i^{\T}\bbeta)+(\lambda_0/2) \|\bbeta\|_2^2$, where $\tilde{\bx}_{i} = (1, \tilde{\bx}_{i,-1}^{\T})^{\T}$ and $\tilde{\bx}_{i,-1} = \bx_{i,-1} \min \{ \sqrt{p}/(6\|\bx_{i,-1}\|_2),1 \} $. Let $\widehat\bbeta_{\bZ}$ denote the ridge regression estimator given by
\begin{align*}
         \widehat\bbeta_{\bZ} = \arg\min_{ \bbeta\in \mathbb{R}^p } F(\bbeta,\bZ)  \,.
\end{align*}
Similarly, let $\widehat\bbeta_{\bZ'}$ be the estimator based on $\bZ'$, where $\bZ$ and $\bZ'$ differ in exactly one observation, say $(y_1, {\bx}_1)\in \bZ$ versus $(y_1', {\bx}_1')\in \bZ'$. By optimality, $\nabla F(\widehat\bbeta_{\bZ},\bZ)=\nabla F(\widehat\bbeta_{\bZ'},\bZ')=\bzero$. Moreover, we can write 
\begin{align*}
    F(\bbeta,\bZ) = F(\bbeta,\bZ') +\underbrace{\frac{1}{n}\{ \rho_{\tau_0}(y_1-\tilde{\bx}_1^{\T}\bbeta) - \rho_{\tau_0}(y_1'-(\tilde{\bx}_1')^{\T}\bbeta) \}}_{h(\bbeta)}\,,
\end{align*}
where $\tilde{\bx}_{1}' = (1, (\tilde{\bx}_{1,-1}')^{\T})^{\T}$ and $\tilde{\bx}_{1,-1}' = \bx_{1,-1}' \min \{ \sqrt{p}/(6\|\bx_{1,-1}'\|_2),1 \} $. Define the difference term
$h(\bbeta)= n^{-1} \{ \rho_{\tau_0}(y_1-\tilde{\bx}_1^{\T}\bbeta)-\rho_{\tau_0}(y_1'-(\tilde{\bx}_1')^{\T}\bbeta) \}$. Then
\begin{align*}
    \underbrace{\nabla F(\widehat\bbeta_{\bZ},\bZ)}_{=\bzero} = &~\nabla F(\widehat\bbeta_{\bZ},\bZ') + \nabla h(\widehat\bbeta_{\bZ}) = \nabla F(\widehat\bbeta_{\bZ},\bZ') - \nabla F(\widehat\bbeta_{\bZ'},\bZ') +\underbrace{\nabla F(\widehat\bbeta_{\bZ'},\bZ')}_{=\bzero}  + \nabla h(\widehat\bbeta_{\bZ})\, ,
\end{align*}
which implies $\nabla F(\widehat\bbeta_{\bZ},\bZ') - \nabla F(\widehat\bbeta_{\bZ'},\bZ')= -\nabla h(\widehat\bbeta_{\bZ})$, and hence
\begin{align}\label{equa}
    \|\nabla F(\widehat\bbeta_{\bZ},\bZ') - \nabla F(\widehat\bbeta_{\bZ'},\bZ')\|_2 =  \|\nabla h(\widehat\bbeta_{\bZ})\|_2\,.
\end{align}

Because $F(\cdot,\bZ')$ is $\lambda_0$-strongly convex, we have
\begin{align*}      
F(\widehat\bbeta_{\bZ},\bZ') \geq &~F(\widehat\bbeta_{\bZ'},\bZ')+\langle \nabla F(\widehat\bbeta_{\bZ'},\bZ'),  \widehat\bbeta_{\bZ}- \widehat\bbeta_{\bZ'}\rangle +  \frac{\lambda_0}{2} \| \widehat\bbeta_{\bZ}- \widehat\bbeta_{\bZ'}\|_2^2\,,\\
         F(\widehat\bbeta_{\bZ'},\bZ') \geq &~F(\widehat\bbeta_{\bZ },\bZ')+\langle \nabla F(\widehat\bbeta_{\bZ},\bZ'),  \widehat\bbeta_{\bZ'}- \widehat\bbeta_{\bZ }\rangle +  \frac{\lambda_0}{2} \| \widehat\bbeta_{\bZ}- \widehat\bbeta_{\bZ'}\|_2^2\, .
\end{align*}
Subtracting the two inequalities gives
\begin{align*}
    \langle \nabla F(\widehat\bbeta_{\bZ },\bZ')-\nabla F(\widehat\bbeta_{\bZ' },\bZ'),  \widehat\bbeta_{\bZ}- \widehat\bbeta_{\bZ'}\rangle \geq  \lambda_0 \|\widehat\bbeta_{\bZ}- \widehat\bbeta_{\bZ'}\|_2^2\,.
\end{align*}
Applying Cauchy–Schwarz and using \eqref{equa},
\begin{align*}
    \|\widehat\bbeta_{\bZ}- \widehat\bbeta_{\bZ'}\|_2 \leq \frac{1}{ \lambda_0}\|\nabla F(\widehat\bbeta_{\bZ},\bZ') - \nabla F(\widehat\bbeta_{\bZ'},\bZ')\|_2= \frac{1}{ \lambda_0}\|\nabla h(\widehat\bbeta_{\bZ})\|_2\,.
\end{align*}

Because $\|\tilde{\bx}_i\|_2^2 \leq 1+ \|\tilde{\bx}_{i,-1}\|_2^2\leq 1+p/36$ for all $i\in[n]$, and $\sup_{u\in\mathbb{R}}|\psi_{\tau_0}(u)|\leq \tau_0$, we obtain
\begin{align*}
        \|\nabla h(\widehat\bbeta_{\bZ})\|_2=&~\frac{1}{n }\| \psi_{\tau_0}(y_1-\tilde{\bx}_1^{\T}\widehat\bbeta_{\bZ})\tilde{\bx}_1 -   \psi_{\tau_0}(y_1'-(\tilde{\bx}_1')^{\T}\widehat\bbeta_{\bZ})\tilde{\bx}_1'\|_2\\
        \leq &~ \frac{1}{n }\big\{\| \psi_{\tau_0}(y_1-\tilde{\bx}_1^{\T}\widehat\bbeta_{\bZ})\tilde{\bx}_1\|_2 +\|  \psi_{\tau_0}(y_1'-(\tilde{\bx}_1')^{\T}\widehat\bbeta_{\bZ})\tilde{\bx}_1'\|_2\big\}\\
        \leq &~ \frac{ 2 \tau_0 \sqrt{1+p/36}}{n }\,,
\end{align*}
which yields the $\ell_2$-sensitivity bound
\begin{align*}
    \|\widehat\bbeta_{\bZ}- \widehat\bbeta_{\bZ'}\|_2 \leq  \frac{2\tau_0 \sqrt{1+p/36}}{\lambda_0 n}\,.
\end{align*}
Since this bound holds uniformly over all neighboring datasets $\bZ$ and $\bZ'$, the $\ell_2$-sensitivity of the estimator in \eqref{ridge.beta} is at most $2\tau_0 (\lambda_0 n)^{-1}\sqrt{1+p/36}$. By Lemma \ref{glmechanism}, the perturbed estimator $\bbeta^{(0)} = \widehat\bbeta_{\bZ} + \tilde\sigma\cdot\cN(\bzero,\bI_p)$, with  $\tilde\sigma  = 8 \tau_0(3n\epsilon_{\rm init}\lambda_0)^{-1}\sqrt{1+p/36}\sqrt{2\log(1.25/\delta_{\rm init})}$, is $(3\epsilon_{\rm init}/4,\delta_{\rm init})$-DP. \qed

\subsubsection{Privacy guarantee  of the support recovery in Section \ref{selection.high}}

The $\ell_1$-sensitivity of each scalar query ${g}_j$ ($j\in \{2,\ldots,p\}$) is bounded by $\Delta=2 n^{-1} \sqrt{\log (pn)}$. Applying the Report Noisy Max procedure \citepS[Claim 3.9]{dwork2014algorithmicS} to $( {g}_2,\ldots, {g}_p)$ with sensitivity bound $\Delta$ and privacy budget $\epsilon_{\rm init}/(2s_0)$ allows us to select a single index $j_{(1)}\in \{2,\ldots,p\}$ in an $(\epsilon_{\rm init}/(2s_0),0)$-DP manner. 
     
Repeating this selection step $s_0$ times without replacement yields the support set $\hat{\cS}_0=(j_{(1)},\ldots,j_{(s_0 )})$. By the standard composition property (Lemma \ref{standardcompositiontheorem}), the cumulative privacy loss for selecting the entire support $\hat{\cS}_0$ is $(\epsilon_{\rm init}/2,0)$-DP.
\qed

\subsection{Additional simulation results for $(\epsilon,\delta)$-DP Huber estimators}

\subsubsection{Additional simulation results in low dimensions}

Tables \ref{tab.low.est_p5} and \ref{tab.low.est_p20} report the relative $\ell_2$-errors of the $(\epsilon,\delta)$-DP Huber estimator across different combinations of sample size and SNR, for dimensions $p\in\{5,20\}$.

\begin{sidewaystable}
\centering
\scriptsize
\setlength{\tabcolsep}{2.5pt} 
\caption{Average logarithmic relative $\ell_2$-error of $(\epsilon,\delta)$-DP Huber estimators across 300 repetitions for various combinations of $(a, b, n,\epsilon)$, with dimension $p = 5$.}
\renewcommand{\arraystretch}{1.3}
\begin{tabular}{ccc cccccccc cccccccc}
\toprule
&&&\multicolumn{8}{c}{Gaussian design} &\multicolumn{8}{c}{ Uniform design}  \\\cmidrule(lr){4-11} \cmidrule(lr){12-19}
&&&\multicolumn{4}{c}{$ \cN(0,1)$ noise} &\multicolumn{4}{c}{ $ t_{2.25}$ noise}& \multicolumn{4}{c}{$ \cN(0,1)$ noise} &\multicolumn{4}{c}{ $ t_{2.25}$ noise} \\\cmidrule(lr){4-7} \cmidrule(lr){8-11}\cmidrule(lr){12-15} \cmidrule(lr){16-19}
$a$&$b$ &$n$ & non-private&$\epsilon=0.3$&$\epsilon=0.5$&$\epsilon=0.9$&non-private&$\epsilon=0.3$&$\epsilon=0.5$&$\epsilon=0.9$&non-private&$\epsilon=0.3$&$\epsilon=0.5$&$\epsilon=0.9$&non-private&$\epsilon=0.3$&$\epsilon=0.5$&$\epsilon=0.9$\\
\midrule
0.5&0.5 &2500 &$-3.666$ & $-0.271$ & $-1.257$ & $-2.109$ & $-3.071$ & $-0.178$ & $-1.073$ & $-1.923$ & $-3.705$ & $-0.298$ & $-1.311$ & $-2.129$ & $-3.071$ & $-0.190$ & $-1.098$ & $-1.952$ \\
&&5000&$-4.018$ & $-1.536$ & $-2.190$ & $-2.679$ & $-3.357$ & $-1.360$ & $-1.996$ & $-2.478$ & $-4.011$ & $-1.561$ & $-2.179$ & $-2.721$ & $-3.346$ & $-1.429$ & $-2.009$ & $-2.477$ \\
&&10000&$-4.358$ & $-2.256$ & $-2.652$ & $-3.130$ & $-3.664$ & $-2.091$ & $-2.463$ & $-2.833$ & $-4.339$ & $-2.243$ & $-2.691$ & $-3.163$ & $-3.667$ & $-2.151$ & $-2.515$ & $-2.862$ \\
&1&2500&$-3.320$ & $-0.160$ & $-1.040$ & $-1.870$ & $-2.737$ & $-0.055$ & $-0.852$ & $-1.672$ & $-3.358$ & $-0.179$ & $-1.089$ & $-1.875$ & $-2.737$ & $-0.070$ & $-0.896$ & $-1.702$ \\
&&5000&$-3.671$ & $-1.329$ & $-1.956$ & $-2.471$ & $-3.020$ & $-1.130$ & $-1.734$ & $-2.233$ & $-3.664$ & $-1.330$ & $-1.968$ & $-2.532$ & $-3.011$ & $-1.182$ & $-1.781$ & $-2.229$ \\
&&10000&$-4.012$ & $-2.097$ & $-2.446$ & $-2.964$ & $-3.325$ & $-1.862$ & $-2.229$ & $-2.598$ & $-3.992$ & $-2.087$ & $-2.525$ & $-2.957$ & $-3.329$ & $-1.905$ & $-2.293$ & $-2.617$ \\
&2&2500&$-2.973$ & $-0.040$ & $-0.815$ & $-1.625$ & $-2.398$ & $0.118$ & $-0.571$ & $-1.381$ & $-3.012$ & $-0.037$ & $-0.859$ & $-1.597$ & $-2.398$ & $0.132$ & $-0.626$ & $-1.421$ \\
&&5000&$-3.324$ & $-1.071$ & $-1.688$ & $-2.260$ & $-2.680$ & $-0.866$ & $-1.458$ & $-1.967$ & $-3.318$ & $-1.096$ & $-1.710$ & $-2.315$ & $-2.671$ & $-0.899$ & $-1.504$ & $-1.966$ \\
&&10000&$-3.665$ & $-1.855$ & $-2.256$ & $-2.723$ & $-2.983$ & $-1.595$ & $-1.966$ & $-2.343$ & $-3.645$ & $-1.879$ & $-2.310$ & $-2.731$ & $-2.987$ & $-1.655$ & $-2.030$ & $-2.355$ \\
\midrule
1&0.5&2500&$-4.359$ & $-0.817$ & $-1.822$ & $-2.680$ & $-3.716$ & $-0.725$ & $-1.614$ & $-2.457$ & $-4.398$ & $-0.846$ & $-1.880$ & $-2.693$ & $-3.713$ & $-0.723$ & $-1.686$ & $-2.494$ \\
&&5000&$-4.711$ & $-2.126$ & $-2.685$ & $-3.121$ & $-4.009$ & $-1.927$ & $-2.503$ & $-2.915$ & $-4.704$ & $-2.158$ & $-2.696$ & $-3.140$ & $-3.995$ & $-1.973$ & $-2.521$ & $-2.917$ \\
&&10000&$-5.051$ & $-2.757$ & $-3.056$ & $-3.411$ & $-4.325$ & $-2.569$ & $-2.885$ & $-3.227$ & $-5.032$ & $-2.759$ & $-3.091$ & $-3.418$ & $-4.325$ & $-2.626$ & $-2.946$ & $-3.261$ \\
&1&2500&$-4.013$ & $-0.722$ & $-1.626$ & $-2.464$ & $-3.397$ & $-0.597$ & $-1.401$ & $-2.249$ & $-4.051$ & $-0.735$ & $-1.667$ & $-2.480$ & $-3.397$ & $-0.579$ & $-1.478$ & $-2.290$ \\
&&5000&$-4.364$ & $-1.919$ & $-2.529$ & $-3.001$ & $-3.687$ & $-1.711$ & $-2.312$ & $-2.754$ & $-4.357$ & $-1.937$ & $-2.546$ & $-3.028$ & $-3.674$ & $-1.762$ & $-2.336$ & $-2.754$ \\
&&10000&$-4.705$ & $-2.640$ & $-2.946$ & $-3.334$ & $-3.998$ & $-2.400$ & $-2.733$ & $-3.086$ & $-4.685$ & $-2.642$ & $-2.986$ & $-3.342$ & $-3.999$ & $-2.465$ & $-2.798$ & $-3.110$ \\
&2&2500&$-3.666$ & $-0.592$ & $-1.382$ & $-2.236$ & $-3.071$ & $-0.448$ & $-1.164$ & $-2.002$ & $-3.705$ & $-0.595$ & $-1.438$ & $-2.245$ & $-3.071$ & $-0.427$ & $-1.219$ & $-2.044$ \\
&&5000&$-4.018$ & $-1.665$ & $-2.324$ & $-2.820$ & $-3.357$ & $-1.478$ & $-2.075$ & $-2.552$ & $-4.011$ & $-1.710$ & $-2.327$ & $-2.850$ & $-3.346$ & $-1.512$ & $-2.107$ & $-2.551$ \\
&&10000&$-4.358$ & $-2.444$ & $-2.786$ & $-3.193$ & $-3.664$ & $-2.191$ & $-2.539$ & $-2.906$ & $-4.339$ & $-2.458$ & $-2.834$ & $-3.202$ & $-3.667$ & $-2.252$ & $-2.604$ & $-2.923$ \\
\midrule
2&0.5&2500&$-5.053$ & $-1.114$ & $-2.201$ & $-3.008$ & $-4.336$ & $-1.066$ & $-2.043$ & $-2.840$ & $-5.091$ & $-1.126$ & $-2.276$ & $-2.996$ & $-4.329$ & $-1.024$ & $-2.107$ & $-2.855$ \\
&&5000&$-5.404$ & $-2.499$ & $-2.951$ & $-3.259$ & $-4.637$ & $-2.354$ & $-2.831$ & $-3.174$ & $-5.397$ & $-2.518$ & $-2.946$ & $-3.251$ & $-4.620$ & $-2.419$ & $-2.823$ & $-3.172$ \\
&&10000&$-5.744$ & $-2.949$ & $-3.151$ & $-3.501$ & $-4.973$ & $-2.853$ & $-3.115$ & $-3.439$ & $-5.725$ & $-2.953$ & $-3.188$ & $-3.498$ & $-4.968$ & $-2.890$ & $-3.178$ & $-3.478$ \\
&1&2500&$-4.706$ & $-1.052$ & $-2.080$ & $-2.941$ & $-4.027$ & $-0.978$ & $-1.893$ & $-2.725$ & $-4.745$ & $-1.063$ & $-2.150$ & $-2.930$ & $-4.023$ & $-0.941$ & $-1.953$ & $-2.743$ \\
&&5000&$-5.057$ & $-2.386$ & $-2.905$ & $-3.237$ & $-4.324$ & $-2.200$ & $-2.738$ & $-3.110$ & $-5.051$ & $-2.414$ & $-2.901$ & $-3.231$ & $-4.309$ & $-2.273$ & $-2.738$ & $-3.107$ \\
&&10000&$-5.398$ & $-2.921$ & $-3.134$ & $-3.484$ & $-4.650$ & $-2.777$ & $-3.056$ & $-3.390$ & $-5.378$ & $-2.924$ & $-3.170$ & $-3.483$ & $-4.647$ & $-2.819$ & $-3.120$ & $-3.421$ \\
&2&2500&$-4.359$ & $-0.951$ & $-1.903$ & $-2.795$ & $-3.716$ & $-0.858$ & $-1.705$ & $-2.559$ & $-4.398$ & $-0.968$ & $-1.967$ & $-2.788$ & $-3.713$ & $-0.828$ & $-1.763$ & $-2.587$ \\
&&5000&$-4.711$ & $-2.214$ & $-2.801$ & $-3.181$ & $-4.009$ & $-2.021$ & $-2.596$ & $-3.011$ & $-4.704$ & $-2.247$ & $-2.798$ & $-3.182$ & $-3.995$ & $-2.074$ & $-2.608$ & $-3.008$ \\
&&10000&$-5.051$ & $-2.848$ & $-3.091$ & $-3.451$ & $-4.325$ & $-2.661$ & $-2.966$ & $-3.313$ & $-5.032$ & $-2.848$ & $-3.125$ & $-3.451$ & $-4.325$ & $-2.709$ & $-3.031$ & $-3.336$ \\

\bottomrule
\end{tabular}
\label{tab.low.est_p5}
\end{sidewaystable}

\begin{sidewaystable}
\centering
\scriptsize
\setlength{\tabcolsep}{2.5pt} 
\caption{Average logarithmic relative $\ell_2$-error of $(\epsilon,\delta)$-DP Huber estimators over 300 repetitions for different combinations of $(a, b, n,\epsilon)$, with dimension $p = 20$.}
\renewcommand{\arraystretch}{1.3}
\begin{tabular}{ccc cccccccc cccccccc}
\toprule
&&&\multicolumn{8}{c}{Gaussian design} &\multicolumn{8}{c}{ Uniform design}  \\\cmidrule(lr){4-11} \cmidrule(lr){12-19}
&&&\multicolumn{4}{c}{$ \cN(0,1)$ noise} &\multicolumn{4}{c}{ $ t_{2.25}$ noise}& \multicolumn{4}{c}{$ \cN(0,1)$ noise} &\multicolumn{4}{c}{ $ t_{2.25}$ noise} \\\cmidrule(lr){4-7} \cmidrule(lr){8-11}\cmidrule(lr){12-15} \cmidrule(lr){16-19}
$a$&$b$ &$n$ & non-private&$\epsilon=0.3$&$\epsilon=0.5$&$\epsilon=0.9$&non-private&$\epsilon=0.3$&$\epsilon=0.5$&$\epsilon=0.9$&non-private&$\epsilon=0.3$&$\epsilon=0.5$&$\epsilon=0.9$&non-private&$\epsilon=0.3$&$\epsilon=0.5$&$\epsilon=0.9$\\
\midrule
0.5&0.5 &2500 &$-3.593$ & $0.716$ & $0.058$ & $-0.925$ & $-2.978$ & $0.769$ & $0.138$ & $-0.737$ & $-3.576$ & $0.695$ & $0.019$ & $-0.918$ & $-2.967$ & $0.742$ & $0.147$ & $-0.755$ \\
&&5000&$-3.943$ & $-0.186$ & $-1.096$ & $-1.964$ & $-3.260$ & $-0.075$ & $-0.893$ & $-1.737$ & $-3.938$ & $-0.183$ & $-1.135$ & $-1.971$ & $-3.274$ & $-0.083$ & $-0.903$ & $-1.758$ \\
&&10000&$-4.288$ & $-1.326$ & $-2.019$ & $-2.399$ & $-3.562$ & $-1.175$ & $-1.800$ & $-2.216$ & $-4.285$ & $-1.398$ & $-2.006$ & $-2.408$ & $-3.565$ & $-1.179$ & $-1.807$ & $-2.241$ \\
&1&2500&$-3.246$ & $0.748$ & $0.104$ & $-0.811$ & $-2.667$ & $0.853$ & $0.244$ & $-0.586$ & $-3.229$ & $0.730$ & $0.068$ & $-0.808$ & $-2.656$ & $0.827$ & $0.247$ & $-0.606$ \\
&&5000&$-3.596$ & $-0.124$ & $-0.965$ & $-1.800$ & $-2.942$ & $0.044$ & $-0.732$ & $-1.546$ & $-3.592$ & $-0.130$ & $-1.000$ & $-1.806$ & $-2.957$ & $0.035$ & $-0.735$ & $-1.566$ \\
&&10000&$-3.942$ & $-1.176$ & $-1.868$ & $-2.301$ & $-3.240$ & $-0.985$ & $-1.620$ & $-2.068$ & $-3.938$ & $-1.237$ & $-1.861$ & $-2.311$ & $-3.243$ & $-0.994$ & $-1.625$ & $-2.093$ \\
&2&2500&$-2.900$ & $0.817$ & $0.190$ & $-0.661$ & $-2.346$ & $0.964$ & $0.373$ & $-0.407$ & $-2.883$ & $0.772$ & $0.147$ & $-0.661$ & $-2.335$ & $0.944$ & $0.380$ & $-0.427$ \\
&&5000&$-3.250$ & $-0.027$ & $-0.796$ & $-1.587$ & $-2.616$ & $0.192$ & $-0.538$ & $-1.322$ & $-3.245$ & $-0.037$ & $-0.827$ & $-1.592$ & $-2.631$ & $0.194$ & $-0.543$ & $-1.342$ \\
&&10000&$-3.595$ & $-0.989$ & $-1.662$ & $-2.137$ & $-2.911$ & $-0.772$ & $-1.404$ & $-1.880$ & $-3.592$ & $-1.044$ & $-1.652$ & $-2.150$ & $-2.915$ & $-0.778$ & $-1.406$ & $-1.905$ \\
\midrule
1&0.5&2500&$-4.286$ & $0.590$ & $-0.038$ & $-1.117$ & $-3.586$ & $0.589$ & $-0.048$ & $-1.003$ & $-4.269$ & $0.542$ & $-0.087$ & $-1.110$ & $-3.574$ & $0.565$ & $-0.034$ & $-1.027$ \\
&&5000&$-4.636$ & $-0.286$ & $-1.313$ & $-2.167$ & $-3.881$ & $-0.256$ & $-1.183$ & $-2.047$ & $-4.631$ & $-0.281$ & $-1.336$ & $-2.175$ & $-3.896$ & $-0.254$ & $-1.189$ & $-2.065$ \\
&&10000&$-4.982$ & $-1.562$ & $-2.183$ & $-2.502$ & $-4.192$ & $-1.481$ & $-2.069$ & $-2.439$ & $-4.978$ & $-1.624$ & $-2.176$ & $-2.508$ & $-4.195$ & $-1.485$ & $-2.078$ & $-2.461$ \\
&1&2500&$-3.939$ & $0.600$ & $-0.023$ & $-1.065$ & $-3.283$ & $0.618$ & $-0.016$ & $-0.926$ & $-3.923$ & $0.553$ & $-0.073$ & $-1.060$ & $-3.272$ & $0.591$ & $-0.001$ & $-0.951$ \\
&&5000&$-4.289$ & $-0.267$ & $-1.248$ & $-2.113$ & $-3.572$ & $-0.214$ & $-1.091$ & $-1.949$ & $-4.285$ & $-0.263$ & $-1.271$ & $-2.121$ & $-3.586$ & $-0.213$ & $-1.099$ & $-1.968$ \\
&&10000&$-4.635$ & $-1.489$ & $-2.140$ & $-2.485$ & $-3.878$ & $-1.373$ & $-1.987$ & $-2.385$ & $-4.632$ & $-1.548$ & $-2.134$ & $-2.492$ & $-3.881$ & $-1.377$ & $-1.994$ & $-2.407$ \\
&2&2500&$-3.593$ & $0.620$ & $0.002$ & $-0.985$ & $-2.978$ & $0.659$ & $0.034$ & $-0.826$ & $-3.576$ & $0.574$ & $-0.046$ & $-0.983$ & $-2.967$ & $0.632$ & $0.048$ & $-0.851$ \\
&&5000&$-3.943$ & $-0.233$ & $-1.150$ & $-2.005$ & $-3.260$ & $-0.152$ & $-0.973$ & $-1.810$ & $-3.938$ & $-0.229$ & $-1.171$ & $-2.013$ & $-3.274$ & $-0.152$ & $-0.981$ & $-1.830$ \\
&&10000&$-4.288$ & $-1.374$ & $-2.048$ & $-2.440$ & $-3.562$ & $-1.234$ & $-1.866$ & $-2.298$ & $-4.285$ & $-1.431$ & $-2.044$ & $-2.448$ & $-3.565$ & $-1.238$ & $-1.871$ & $-2.321$ \\
\midrule
2&0.5&2500&$-4.979$ & $0.290$ & $-0.342$ & $-1.382$ & $-4.199$ & $0.275$ & $-0.372$ & $-1.297$ & $-4.962$ & $0.259$ & $-0.380$ & $-1.402$ & $-4.185$ & $0.255$ & $-0.357$ & $-1.328$ \\
&&5000&$-5.329$ & $-0.568$ & $-1.607$ & $-2.462$ & $-4.505$ & $-0.569$ & $-1.508$ & $-2.377$ & $-5.324$ & $-0.568$ & $-1.621$ & $-2.466$ & $-4.529$ & $-0.573$ & $-1.533$ & $-2.399$ \\
&&10000&$-5.675$ & $-1.861$ & $-2.465$ & $-2.776$ & $-4.822$ & $-1.828$ & $-2.396$ & $-2.746$ & $-5.671$ & $-1.922$ & $-2.450$ & $-2.789$ & $-4.830$ & $-1.828$ & $-2.408$ & $-2.767$ \\
&1&2500&$-4.633$ & $0.292$ & $-0.338$ & $-1.357$ & $-3.890$ & $0.280$ & $-0.361$ & $-1.252$ & $-4.616$ & $0.261$ & $-0.376$ & $-1.378$ & $-3.878$ & $0.261$ & $-0.346$ & $-1.283$ \\
&&5000&$-4.983$ & $-0.561$ & $-1.574$ & $-2.441$ & $-4.192$ & $-0.553$ & $-1.450$ & $-2.316$ & $-4.978$ & $-0.562$ & $-1.588$ & $-2.445$ & $-4.211$ & $-0.557$ & $-1.475$ & $-2.339$ \\
&&10000&$-5.328$ & $-1.825$ & $-2.452$ & $-2.772$ & $-4.506$ & $-1.761$ & $-2.355$ & $-2.723$ & $-5.325$ & $-1.886$ & $-2.438$ & $-2.785$ & $-4.511$ & $-1.761$ & $-2.366$ & $-2.745$ \\
&2&2500&$-4.286$ & $0.295$ & $-0.331$ & $-1.314$ & $-3.586$ & $0.289$ & $-0.343$ & $-1.188$ & $-4.269$ & $0.264$ & $-0.369$ & $-1.335$ & $-3.574$ & $0.270$ & $-0.329$ & $-1.218$ \\
&&5000&$-4.636$ & $-0.550$ & $-1.518$ & $-2.389$ & $-3.881$ & $-0.528$ & $-1.370$ & $-2.217$ & $-4.631$ & $-0.551$ & $-1.532$ & $-2.393$ & $-3.896$ & $-0.532$ & $-1.395$ & $-2.241$ \\
&&10000&$-4.982$ & $-1.762$ & $-2.418$ & $-2.763$ & $-4.192$ & $-1.665$ & $-2.283$ & $-2.681$ & $-4.978$ & $-1.820$ & $-2.405$ & $-2.776$ & $-4.195$ & $-1.665$ & $-2.295$ & $-2.703$ \\

\bottomrule
\end{tabular}
\label{tab.low.est_p20}
\end{sidewaystable}

\subsubsection{Additional simulation results in high dimensions}

Figure~\ref{fig.high.eps_p5000} illustrate how the logarithmic relative $\ell_2$-error varies with the sample size for both the $(\epsilon, \delta)$-DP sparse Huber estimator and its non-private counterpart, with $p=5000$ and $\epsilon \in \{0.5, 0.9\}$. Table~\ref{tab.high.est_p5000} summarizes a comparison between the sparse DP Huber estimator and the sparse DP LS estimator for $p=5000$ and $\epsilon=0.5$.

 \begin{figure}[h]
    \centering
\includegraphics[width=0.8\textwidth]{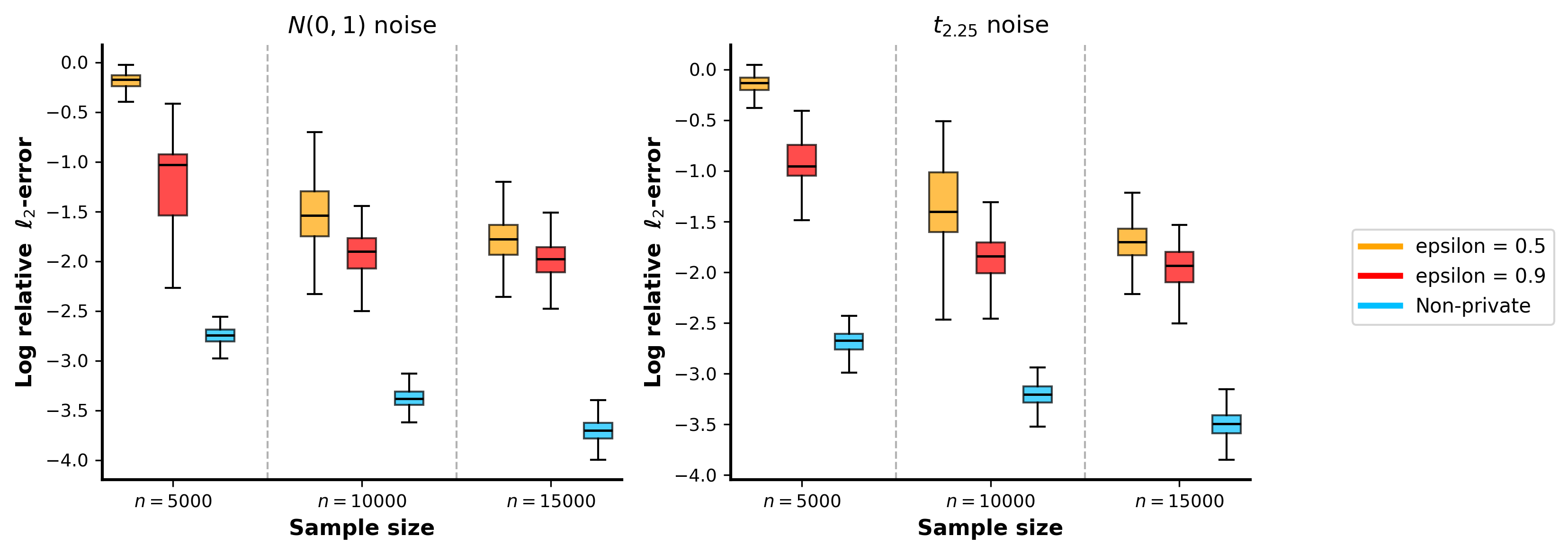}
    \caption{Boxplots of the logarithmic relative $\ell_2$-error over 300 repetitions for private and non-private sparse Huber estimators across sample sizes, with $p = 5000$ and $s^* = 10$.}
\label{fig.high.eps_p5000}
\end{figure}

\begin{table}[ht]
\centering
\footnotesize
\caption{Average logarithmic relative $\ell_2$-error of the slope coefficients over 300 repetitions across sample sizes, under privacy levels  $(0.5, 10n^{-1.1})$ with $p=5000$.}
\renewcommand{\arraystretch}{1.1}
\resizebox{0.6\textwidth}{!}{
\begin{tabular}{c c c c ccc }
\toprule
\multirow{2}{*}{Noise} & \multirow{2}{*}{$n$}  & \multirow{2}{*}{ non-private sparse  Huber} & \multirow{2}{*}{sparse DP Huber} & \multicolumn{3}{c}{ sparse DP LS} \\
\cline{5-7}
& & & & $C_R=0.1$ & $C_R=0.5$ & $C_R=1$ \\
\midrule
 $\cN(0,1)$       & 5000 & $-2.730$ &$-0.142$  & $0.080$ & $0.120$ & $0.170$ \\

&10000 & $-3.353$  &$-1.495$  & $-0.127$ & $0.032$ & $0.092$ \\

&15000 & $-3.678$   &$-1.784$  & $-0.388$ & $-0.071$ & $0.043$ \\

\midrule
$t_{2.25}$&5000 & $-2.667$  &$-0.100$  & $0.070$ & $0.133$ & $0.167$ \\

&10000 & $-3.183$  &$-1.338$  & $-0.124$ & $0.035$ & $0.099$ \\

&15000 & $-3.474$  &$-1.718$  & $-0.366$ & $-0.084$ & $0.046$ \\
\bottomrule
\end{tabular}}
\label{tab.high.est_p5000}
\end{table}

\subsection{Additional simulation results for $\epsilon$-GDP Huber estimators}
 
\subsubsection{Selection of tuning parameters}

To ensure that the overall mechanism is 
$\epsilon$-GDP, we divide the privacy budget between the initialization step and the main algorithm (Algorithm~\ref{alg:DPHuberlow}). Under the GDP framework, the composition of $T$ mechanisms with privacy parameters $\{\epsilon_t\}_{t=1}^T$ results in an overall $\epsilon$-GDP guarantee with $\epsilon= \sqrt{\epsilon_1^2+\cdots+\epsilon_T^2}$. Accordingly, we set $\epsilon_{\rm init}= \epsilon/ \sqrt{8}$ and $\epsilon_{\rm main}=\sqrt{7} \epsilon/\sqrt{8}$. Since $\sqrt{\epsilon_{\rm init}^2+\epsilon_{\rm main}^2} =\epsilon$, the combined procedure satisfies the target $\epsilon$-GDP guarantee.

\noindent
\underline{{\textbf{Initialization of}} $\bbeta^{(0)}$.} As in the $(\epsilon,\delta)$-DP setting, we obtain $\widehat\bbeta^{(0)}$ via \eqref{ridge.beta}. In the initialization step, the privacy budget $\epsilon_{\rm init}$ is split evenly: $\epsilon_{\rm init}/\sqrt{2}$ is used for selecting $\tau_0$, and $ \epsilon_{\rm init}/\sqrt{2}$ for the output-perturbation step. By the GDP composition rule, the initialization step as a whole is $\epsilon_{\rm init}$-GDP.

\begin{itemize}
    \item Selection of $\tau_0$. For each $i\in[n]$, define $\tilde{y}_i = \min\{ \log n, \max\{ -\log n ,y_i \} \}$. Let $m_1 = n^{-1}\sum_{i=1}^n\tilde{y}_i+4(n\epsilon_{\rm init})^{-1}\log n \cdot \cN(0,1)$ and $m_2 = n^{-1}\sum_{i=1}^n\tilde{y}_i^2+2(n\epsilon_{\rm init})^{-1}\log^2 n \cdot \cN(0,1)$. Set $\tau_0 = \sqrt{ m_2 - m_1^2}$ if $m_2 - m_1^2 > 0$ and $\tau_0 = 2$ otherwise. Since $m_1$ and $m_2$ are each $(\epsilon_{\rm init}/2)$-GDP, the GDP composition rule implies that the procedure for selecting $\tau_0$ is $(\epsilon_{\rm init}/\sqrt{2})$-GDP.

    \item Output perturbation. Let $\tilde{\sigma}=  2\sqrt{2}B\tau_0( n\epsilon_{\rm init}\lambda_0)^{-1}$ with $B= \sqrt{1+p/36}$. We then set $\bbeta^{(0)} = \widehat\bbeta^{(0)} + \tilde{\sigma}\cdot\cN(\bzero,\bI_p)$, which is $(\epsilon_{\rm init}/\sqrt{2})$-GDP.
\end{itemize}

\noindent
\underline{{\textbf{Selection of other parameters}}.} 
The parameters $(\eta_0,\gamma,T,\tau)$ in the $\epsilon$-GDP setting are selected in the same way as in the $(\epsilon,\delta)$-DP case; see Section \ref{selection.low} for further details.

\subsubsection{Simulation results}

Tables \ref{tab.low.est_GDP_p5}--\ref{tab.low.est_GDP_p20}  summarize the relative $\ell_2$-errors of the $ \epsilon$-GDP Huber estimator across various combinations of sample size and SNR for dimensions $p \in \{ 5, 10 , 20 \}$.

{\small{
 \setlength{\bibsep}{6pt}
\bibliographystyleS{my_temp} 
\bibliographyS{refs_final}
}}

\begin{sidewaystable}
\centering
\scriptsize
\setlength{\tabcolsep}{2.5pt} 
\caption{Average logarithmic relative $\ell_2$-error of the $\epsilon$-GDP Huber estimators over 300 repetitions across combinations of $(a, b, n,\epsilon)$, with dimension $p = 5$.}
\renewcommand{\arraystretch}{1.3}
\begin{tabular}{ccc cccccccc cccccccc}
\toprule
&&&\multicolumn{8}{c}{Gaussian design} &\multicolumn{8}{c}{ Uniform design}  \\\cmidrule(lr){4-11} \cmidrule(lr){12-19}
&&&\multicolumn{4}{c}{$ \cN(0,1)$ noise} &\multicolumn{4}{c}{ $ t_{2.25}$ noise}& \multicolumn{4}{c}{$ \cN(0,1)$ noise} &\multicolumn{4}{c}{ $ t_{2.25}$ noise} \\\cmidrule(lr){4-7} \cmidrule(lr){8-11}\cmidrule(lr){12-15} \cmidrule(lr){16-19}
$a$&$b$ &$n$ & non-private&$\epsilon=0.3$&$\epsilon=0.5$&$\epsilon=0.9$&non-private&$\epsilon=0.3$&$\epsilon=0.5$&$\epsilon=0.9$&non-private&$\epsilon=0.3$&$\epsilon=0.5$&$\epsilon=0.9$&non-private&$\epsilon=0.3$&$\epsilon=0.5$&$\epsilon=0.9$\\
\midrule
0.5&0.5 &2500 &$-3.666$ & $-2.708$ & $-3.152$ & $-3.448$ & $-3.071$ & $-2.613$ & $-2.997$ & $-3.207$ & $-3.705$ & $-2.674$ & $-3.127$ & $-3.408$ & $-3.071$ & $-2.563$ & $-2.960$ & $-3.180$ \\
&&5000&$-4.018$ & $-3.535$ & $-3.791$ & $-3.914$ & $-3.357$ & $-3.370$ & $-3.576$ & $-3.617$ & $-4.011$ & $-3.472$ & $-3.719$ & $-3.868$ & $-3.346$ & $-3.311$ & $-3.505$ & $-3.587$ \\
&&10000&$-4.358$ & $-4.142$ & $-4.236$ & $-4.294$ & $-3.664$ & $-3.836$ & $-3.894$ & $-3.893$ & $-4.339$ & $-4.081$ & $-4.195$ & $-4.278$ & $-3.667$ & $-3.803$ & $-3.905$ & $-3.922$ \\
&1&2500&$-3.320$ & $-2.471$ & $-2.869$ & $-3.127$ & $-2.737$ & $-2.326$ & $-2.716$ & $-2.927$ & $-3.358$ & $-2.456$ & $-2.872$ & $-3.116$ & $-2.737$ & $-2.293$ & $-2.700$ & $-2.904$ \\
&&5000&$-3.671$ & $-3.243$ & $-3.463$ & $-3.569$ & $-3.020$ & $-3.076$ & $-3.277$ & $-3.315$ & $-3.664$ & $-3.198$ & $-3.415$ & $-3.540$ & $-3.011$ & $-3.047$ & $-3.229$ & $-3.295$ \\
&&10000&$-4.012$ & $-3.818$ & $-3.891$ & $-3.948$ & $-3.325$ & $-3.555$ & $-3.606$ & $-3.603$ & $-3.992$ & $-3.774$ & $-3.866$ & $-3.939$ & $-3.329$ & $-3.538$ & $-3.625$ & $-3.628$ \\
&2&2500&$-2.973$ & $-2.201$ & $-2.570$ & $-2.798$ & $-2.398$ & $-2.035$ & $-2.425$ & $-2.636$ & $-3.012$ & $-2.208$ & $-2.595$ & $-2.813$ & $-2.398$ & $-2.013$ & $-2.427$ & $-2.612$ \\
&&5000&$-3.324$ & $-2.938$ & $-3.134$ & $-3.222$ & $-2.680$ & $-2.772$ & $-2.965$ & $-3.003$ & $-3.318$ & $-2.915$ & $-3.105$ & $-3.208$ & $-2.671$ & $-2.772$ & $-2.943$ & $-2.992$ \\
&&10000&$-3.665$ & $-3.485$ & $-3.545$ & $-3.595$ & $-2.983$ & $-3.266$ & $-3.308$ & $-3.301$ & $-3.645$ & $-3.463$ & $-3.533$ & $-3.592$ & $-2.987$ & $-3.263$ & $-3.334$ & $-3.323$ \\
\midrule
1&0.5&2500&$-4.359$ & $-3.385$ & $-3.837$ & $-4.104$ & $-3.716$ & $-3.234$ & $-3.607$ & $-3.802$ & $-4.398$ & $-3.350$ & $-3.736$ & $-4.025$ & $-3.713$ & $-3.149$ & $-3.539$ & $-3.751$ \\
&&5000&$-4.711$ & $-4.265$ & $-4.457$ & $-4.549$ & $-4.009$ & $-4.005$ & $-4.188$ & $-4.216$ & $-4.704$ & $-4.135$ & $-4.356$ & $-4.490$ & $-3.995$ & $-3.906$ & $-4.089$ & $-4.164$ \\
&&10000&$-5.051$ & $-4.792$ & $-4.861$ & $-4.905$ & $-4.325$ & $-4.427$ & $-4.471$ & $-4.470$ & $-5.032$ & $-4.727$ & $-4.804$ & $-4.867$ & $-4.325$ & $-4.372$ & $-4.467$ & $-4.485$ \\
&1&2500&$-4.013$ & $-3.059$ & $-3.503$ & $-3.784$ & $-3.397$ & $-2.943$ & $-3.333$ & $-3.542$ & $-4.051$ & $-3.030$ & $-3.412$ & $-3.714$ & $-3.397$ & $-2.875$ & $-3.280$ & $-3.497$ \\
&&5000&$-4.364$ & $-3.911$ & $-4.131$ & $-4.246$ & $-3.687$ & $-3.711$ & $-3.911$ & $-3.946$ & $-4.357$ & $-3.785$ & $-4.029$ & $-4.193$ & $-3.674$ & $-3.632$ & $-3.824$ & $-3.906$ \\
&&10000&$-4.705$ & $-4.474$ & $-4.567$ & $-4.621$ & $-3.998$ & $-4.166$ & $-4.216$ & $-4.211$ & $-4.685$ & $-4.405$ & $-4.516$ & $-4.599$ & $-3.999$ & $-4.115$ & $-4.217$ & $-4.233$ \\
&2&2500&$-3.666$ & $-2.774$ & $-3.195$ & $-3.460$ & $-3.071$ & $-2.661$ & $-3.059$ & $-3.275$ & $-3.705$ & $-2.761$ & $-3.136$ & $-3.410$ & $-3.071$ & $-2.617$ & $-3.029$ & $-3.234$ \\
&&5000&$-4.018$ & $-3.586$ & $-3.798$ & $-3.915$ & $-3.357$ & $-3.418$ & $-3.618$ & $-3.657$ & $-4.011$ & $-3.482$ & $-3.713$ & $-3.868$ & $-3.346$ & $-3.372$ & $-3.557$ & $-3.629$ \\
&&10000&$-4.358$ & $-4.145$ & $-4.238$ & $-4.294$ & $-3.664$ & $-3.896$ & $-3.945$ & $-3.939$ & $-4.339$ & $-4.082$ & $-4.195$ & $-4.278$ & $-3.667$ & $-3.857$ & $-3.952$ & $-3.959$ \\
\midrule
2&0.5&2500&$-5.053$ & $-4.148$ & $-4.532$ & $-4.710$ & $-4.336$ & $-3.907$ & $-4.246$ & $-4.406$ & $-5.091$ & $-4.119$ & $-4.428$ & $-4.584$ & $-4.329$ & $-3.833$ & $-4.178$ & $-4.325$ \\
&&5000&$-5.404$ & $-4.927$ & $-5.016$ & $-5.062$ & $-4.637$ & $-4.626$ & $-4.761$ & $-4.772$ & $-5.397$ & $-4.778$ & $-4.882$ & $-4.938$ & $-4.620$ & $-4.521$ & $-4.642$ & $-4.677$ \\
&&10000&$-5.744$ & $-5.263$ & $-5.291$ & $-5.348$ & $-4.973$ & $-4.963$ & $-4.991$ & $-4.998$ & $-5.725$ & $-5.172$ & $-5.185$ & $-5.228$ & $-4.968$ & $-4.888$ & $-4.952$ & $-4.962$ \\
&1&2500&$-4.706$ & $-3.812$ & $-4.256$ & $-4.481$ & $-4.027$ & $-3.621$ & $-4.005$ & $-4.183$ & $-4.745$ & $-3.782$ & $-4.162$ & $-4.391$ & $-4.023$ & $-3.560$ & $-3.952$ & $-4.123$ \\
&&5000&$-5.057$ & $-4.671$ & $-4.806$ & $-4.864$ & $-4.324$ & $-4.384$ & $-4.540$ & $-4.549$ & $-5.051$ & $-4.530$ & $-4.693$ & $-4.782$ & $-4.309$ & $-4.299$ & $-4.443$ & $-4.487$ \\
&&10000&$-5.398$ & $-5.104$ & $-5.139$ & $-5.183$ & $-4.650$ & $-4.767$ & $-4.790$ & $-4.784$ & $-5.378$ & $-5.024$ & $-5.057$ & $-5.104$ & $-4.647$ & $-4.702$ & $-4.773$ & $-4.781$ \\
&2&2500&$-4.359$ & $-3.449$ & $-3.910$ & $-4.175$ & $-3.716$ & $-3.310$ & $-3.725$ & $-3.929$ & $-4.398$ & $-3.425$ & $-3.831$ & $-4.109$ & $-3.713$ & $-3.281$ & $-3.700$ & $-3.887$ \\
&&5000&$-4.711$ & $-4.331$ & $-4.511$ & $-4.597$ & $-4.009$ & $-4.098$ & $-4.273$ & $-4.295$ & $-4.704$ & $-4.192$ & $-4.408$ & $-4.539$ & $-3.995$ & $-4.050$ & $-4.209$ & $-4.262$ \\
&&10000&$-5.051$ & $-4.840$ & $-4.901$ & $-4.941$ & $-4.325$ & $-4.540$ & $-4.564$ & $-4.545$ & $-5.032$ & $-4.763$ & $-4.838$ & $-4.902$ & $-4.325$ & $-4.481$ & $-4.562$ & $-4.560$ \\

\bottomrule
\end{tabular}
\label{tab.low.est_GDP_p5}
\end{sidewaystable}

\begin{sidewaystable}
\centering
\scriptsize
\setlength{\tabcolsep}{2.5pt} 
\caption{Average logarithmic relative $\ell_2$-error of the $\epsilon$-GDP Huber estimators over 300 repetitions across combinations of $(a, b, n,\epsilon)$, with dimension $p = 10$.}
\renewcommand{\arraystretch}{1.3}
\begin{tabular}{ccc cccccccc cccccccc}
\toprule
&&&\multicolumn{8}{c}{Gaussian design} &\multicolumn{8}{c}{ Uniform design}  \\\cmidrule(lr){4-11} \cmidrule(lr){12-19}
&&&\multicolumn{4}{c}{$\cN(0,1)$ noise} &\multicolumn{4}{c}{ $ t_{2.25}$ noise}& \multicolumn{4}{c}{$ \cN(0,1)$ noise} &\multicolumn{4}{c}{$t_{2.25}$ noise} \\\cmidrule(lr){4-7} \cmidrule(lr){8-11}\cmidrule(lr){12-15} \cmidrule(lr){16-19}
$a$&$b$ &$n$ & non-private&$\epsilon=0.3$&$\epsilon=0.5$&$\epsilon=0.9$&non-private&$\epsilon=0.3$&$\epsilon=0.5$&$\epsilon=0.9$&non-private&$\epsilon=0.3$&$\epsilon=0.5$&$\epsilon=0.9$&non-private&$\epsilon=0.3$&$\epsilon=0.5$&$\epsilon=0.9$\\
\midrule
0.5&0.5 &2500 &$-3.624$ & $-2.431$ & $-2.982$ & $-3.308$ & $-3.018$ & $-2.249$ & $-2.760$ & $-3.068$ & $-3.603$ & $-2.448$ & $-2.930$ & $-3.303$ & $-3.026$ & $-2.251$ & $-2.710$ & $-3.058$ \\
&&5000&$-3.955$ & $-3.338$ & $-3.661$ & $-3.798$ & $-3.290$ & $-3.132$ & $-3.388$ & $-3.486$ & $-3.966$ & $-3.320$ & $-3.636$ & $-3.807$ & $-3.309$ & $-3.114$ & $-3.366$ & $-3.485$ \\
&&10000&$-4.310$ & $-3.993$ & $-4.136$ & $-4.205$ & $-3.578$ & $-3.678$ & $-3.775$ & $-3.802$ & $-4.292$ & $-3.936$ & $-4.104$ & $-4.169$ & $-3.570$ & $-3.677$ & $-3.781$ & $-3.805$ \\
&1&2500&$-3.277$ & $-2.177$ & $-2.691$ & $-2.989$ & $-2.692$ & $-1.980$ & $-2.474$ & $-2.776$ & $-3.256$ & $-2.191$ & $-2.643$ & $-2.992$ & $-2.703$ & $-1.982$ & $-2.430$ & $-2.777$ \\
&&5000&$-3.609$ & $-3.020$ & $-3.331$ & $-3.465$ & $-2.964$ & $-2.834$ & $-3.096$ & $-3.192$ & $-3.619$ & $-3.010$ & $-3.314$ & $-3.473$ & $-2.982$ & $-2.827$ & $-3.083$ & $-3.204$ \\
&&10000&$-3.963$ & $-3.677$ & $-3.816$ & $-3.884$ & $-3.249$ & $-3.394$ & $-3.492$ & $-3.518$ & $-3.945$ & $-3.617$ & $-3.782$ & $-3.852$ & $-3.238$ & $-3.398$ & $-3.504$ & $-3.526$ \\
&2&2500&$-2.931$ & $-1.926$ & $-2.405$ & $-2.672$ & $-2.360$ & $-1.708$ & $-2.190$ & $-2.479$ & $-2.909$ & $-1.938$ & $-2.365$ & $-2.684$ & $-2.372$ & $-1.714$ & $-2.152$ & $-2.492$ \\
&&5000&$-3.262$ & $-2.714$ & $-3.004$ & $-3.126$ & $-2.631$ & $-2.539$ & $-2.801$ & $-2.889$ & $-3.273$ & $-2.718$ & $-3.001$ & $-3.137$ & $-2.649$ & $-2.546$ & $-2.801$ & $-2.915$ \\
&&10000&$-3.617$ & $-3.359$ & $-3.486$ & $-3.545$ & $-2.913$ & $-3.112$ & $-3.201$ & $-3.224$ & $-3.598$ & $-3.311$ & $-3.456$ & $-3.514$ & $-2.901$ & $-3.121$ & $-3.217$ & $-3.233$ \\
\midrule
1&0.5&2500&$-4.317$ & $-3.043$ & $-3.635$ & $-3.941$ & $-3.650$ & $-2.821$ & $-3.367$ & $-3.663$ & $-4.296$ & $-3.056$ & $-3.576$ & $-3.924$ & $-3.649$ & $-2.825$ & $-3.315$ & $-3.642$ \\
&&5000&$-4.648$ & $-4.019$ & $-4.290$ & $-4.379$ & $-3.924$ & $-3.758$ & $-3.985$ & $-4.059$ & $-4.659$ & $-4.002$ & $-4.258$ & $-4.378$ & $-3.940$ & $-3.740$ & $-3.952$ & $-4.044$ \\
&&10000&$-5.003$ & $-4.568$ & $-4.667$ & $-4.717$ & $-4.223$ & $-4.257$ & $-4.331$ & $-4.346$ & $-4.985$ & $-4.520$ & $-4.621$ & $-4.646$ & $-4.218$ & $-4.246$ & $-4.323$ & $-4.333$ \\
&1&2500&$-3.970$ & $-2.739$ & $-3.318$ & $-3.645$ & $-3.335$ & $-2.558$ & $-3.099$ & $-3.407$ & $-3.949$ & $-2.750$ & $-3.260$ & $-3.637$ & $-3.341$ & $-2.564$ & $-3.051$ & $-3.397$ \\
&&5000&$-4.302$ & $-3.691$ & $-4.000$ & $-4.120$ & $-3.610$ & $-3.482$ & $-3.731$ & $-3.813$ & $-4.312$ & $-3.671$ & $-3.971$ & $-4.128$ & $-3.627$ & $-3.471$ & $-3.709$ & $-3.812$ \\
&&10000&$-4.656$ & $-4.309$ & $-4.437$ & $-4.494$ & $-3.902$ & $-4.023$ & $-4.103$ & $-4.116$ & $-4.638$ & $-4.261$ & $-4.402$ & $-4.445$ & $-3.897$ & $-4.016$ & $-4.102$ & $-4.110$ \\
&2&2500&$-3.624$ & $-2.459$ & $-3.005$ & $-3.322$ & $-3.018$ & $-2.298$ & $-2.821$ & $-3.125$ & $-3.603$ & $-2.469$ & $-2.949$ & $-3.319$ & $-3.026$ & $-2.307$ & $-2.780$ & $-3.131$ \\
&&5000&$-3.955$ & $-3.353$ & $-3.670$ & $-3.806$ & $-3.290$ & $-3.191$ & $-3.454$ & $-3.539$ & $-3.966$ & $-3.334$ & $-3.645$ & $-3.814$ & $-3.309$ & $-3.193$ & $-3.445$ & $-3.556$ \\
&&10000&$-4.310$ & $-4.000$ & $-4.142$ & $-4.209$ & $-3.578$ & $-3.765$ & $-3.847$ & $-3.859$ & $-4.292$ & $-3.945$ & $-4.109$ & $-4.173$ & $-3.570$ & $-3.761$ & $-3.851$ & $-3.859$ \\
\midrule
2&0.5&2500&$-5.010$ & $-3.666$ & $-4.184$ & $-4.422$ & $-4.278$ & $-3.338$ & $-3.886$ & $-4.119$ & $-4.989$ & $-3.704$ & $-4.194$ & $-4.473$ & $-4.267$ & $-3.411$ & $-3.914$ & $-4.194$ \\
&&5000&$-5.341$ & $-4.609$ & $-4.844$ & $-4.971$ & $-4.556$ & $-4.337$ & $-4.549$ & $-4.615$ & $-5.352$ & $-4.648$ & $-4.863$ & $-5.004$ & $-4.561$ & $-4.390$ & $-4.581$ & $-4.664$ \\
&&10000&$-5.696$ & $-5.111$ & $-5.262$ & $-5.365$ & $-4.861$ & $-4.867$ & $-4.938$ & $-4.954$ & $-5.678$ & $-5.092$ & $-5.217$ & $-5.288$ & $-4.855$ & $-4.873$ & $-4.939$ & $-4.944$ \\
&1&2500&$-4.664$ & $-3.395$ & $-3.954$ & $-4.235$ & $-3.964$ & $-3.071$ & $-3.628$ & $-3.879$ & $-4.642$ & $-3.431$ & $-3.966$ & $-4.279$ & $-3.957$ & $-3.142$ & $-3.658$ & $-3.954$ \\
&&5000&$-4.995$ & $-4.414$ & $-4.659$ & $-4.763$ & $-4.239$ & $-4.091$ & $-4.316$ & $-4.380$ & $-5.006$ & $-4.451$ & $-4.682$ & $-4.798$ & $-4.249$ & $-4.158$ & $-4.361$ & $-4.440$ \\
&&10000&$-5.349$ & $-4.961$ & $-5.082$ & $-5.150$ & $-4.543$ & $-4.669$ & $-4.726$ & $-4.722$ & $-5.331$ & $-4.947$ & $-5.055$ & $-5.102$ & $-4.536$ & $-4.681$ & $-4.733$ & $-4.719$ \\
&2&2500&$-4.317$ & $-3.070$ & $-3.619$ & $-3.937$ & $-3.650$ & $-2.772$ & $-3.304$ & $-3.573$ & $-4.296$ & $-3.099$ & $-3.633$ & $-3.985$ & $-3.649$ & $-2.839$ & $-3.339$ & $-3.649$ \\
&&5000&$-4.648$ & $-4.094$ & $-4.376$ & $-4.490$ & $-3.924$ & $-3.767$ & $-4.019$ & $-4.103$ & $-4.659$ & $-4.132$ & $-4.404$ & $-4.526$ & $-3.940$ & $-3.849$ & $-4.080$ & $-4.175$ \\
&&10000&$-5.003$ & $-4.722$ & $-4.835$ & $-4.882$ & $-4.223$ & $-4.410$ & $-4.472$ & $-4.469$ & $-4.985$ & $-4.714$ & $-4.824$ & $-4.858$ & $-4.218$ & $-4.434$ & $-4.485$ & $-4.469$ \\

\bottomrule
\end{tabular}
\label{tab.low.est_GDP_p10}
\end{sidewaystable}

\begin{sidewaystable}
\centering
\scriptsize
\setlength{\tabcolsep}{2.5pt} 
\caption{Average logarithmic relative $\ell_2$-error of the $\epsilon$-GDP Huber estimators over 300 repetitions across combinations of $(a, b, n,\epsilon)$, with dimension $p = 20$.}
\renewcommand{\arraystretch}{1.3}
\begin{tabular}{ccc cccccccc cccccccc}
\toprule
&&&\multicolumn{8}{c}{Gaussian design} &\multicolumn{8}{c}{ Uniform design}  \\\cmidrule(lr){4-11} \cmidrule(lr){12-19}
&&&\multicolumn{4}{c}{$ \cN(0,1)$ noise} &\multicolumn{4}{c}{ $ t_{2.25}$ noise}& \multicolumn{4}{c}{$ \cN(0,1)$ noise} &\multicolumn{4}{c}{ $ t_{2.25}$ noise} \\\cmidrule(lr){4-7} \cmidrule(lr){8-11}\cmidrule(lr){12-15} \cmidrule(lr){16-19}
$a$&$b$ &$n$ & non-private&$\epsilon=0.3$&$\epsilon=0.5$&$\epsilon=0.9$&non-private&$\epsilon=0.3$&$\epsilon=0.5$&$\epsilon=0.9$&non-private&$\epsilon=0.3$&$\epsilon=0.5$&$\epsilon=0.9$&non-private&$\epsilon=0.3$&$\epsilon=0.5$&$\epsilon=0.9$\\
\midrule
0.5&0.5 &2500 &$-3.593$ & $-1.906$ & $-2.583$ & $-3.136$ & $-2.978$ & $-1.760$ & $-2.436$ & $-2.879$ & $-3.576$ & $-1.951$ & $-2.613$ & $-3.116$ & $-2.967$ & $-1.755$ & $-2.403$ & $-2.872$ \\
&&5000&$-3.943$ & $-3.023$ & $-3.479$ & $-3.722$ & $-3.260$ & $-2.817$ & $-3.211$ & $-3.395$ & $-3.938$ & $-3.020$ & $-3.477$ & $-3.712$ & $-3.274$ & $-2.819$ & $-3.206$ & $-3.393$ \\
&&10000&$-4.288$ & $-3.814$ & $-4.029$ & $-4.126$ & $-3.562$ & $-3.528$ & $-3.698$ & $-3.756$ & $-4.285$ & $-3.796$ & $-4.017$ & $-4.110$ & $-3.565$ & $-3.525$ & $-3.685$ & $-3.757$ \\
&1&2500&$-3.246$ & $-1.712$ & $-2.313$ & $-2.822$ & $-2.667$ & $-1.559$ & $-2.180$ & $-2.601$ & $-3.229$ & $-1.751$ & $-2.343$ & $-2.801$ & $-2.656$ & $-1.554$ & $-2.148$ & $-2.594$ \\
&&5000&$-3.596$ & $-2.708$ & $-3.151$ & $-3.398$ & $-2.942$ & $-2.532$ & $-2.920$ & $-3.111$ & $-3.592$ & $-2.708$ & $-3.151$ & $-3.390$ & $-2.957$ & $-2.536$ & $-2.921$ & $-3.113$ \\
&&10000&$-3.942$ & $-3.492$ & $-3.719$ & $-3.829$ & $-3.240$ & $-3.246$ & $-3.425$ & $-3.489$ & $-3.938$ & $-3.473$ & $-3.707$ & $-3.816$ & $-3.243$ & $-3.242$ & $-3.412$ & $-3.489$ \\
&2&2500&$-2.900$ & $-1.514$ & $-2.058$ & $-2.519$ & $-2.346$ & $-1.348$ & $-1.924$ & $-2.317$ & $-2.883$ & $-1.549$ & $-2.088$ & $-2.499$ & $-2.335$ & $-1.345$ & $-1.895$ & $-2.313$ \\
&&5000&$-3.250$ & $-2.417$ & $-2.835$ & $-3.064$ & $-2.616$ & $-2.251$ & $-2.625$ & $-2.813$ & $-3.245$ & $-2.422$ & $-2.840$ & $-3.059$ & $-2.631$ & $-2.259$ & $-2.633$ & $-2.821$ \\
&&10000&$-3.595$ & $-3.175$ & $-3.394$ & $-3.500$ & $-2.911$ & $-2.958$ & $-3.139$ & $-3.203$ & $-3.592$ & $-3.159$ & $-3.385$ & $-3.490$ & $-2.915$ & $-2.956$ & $-3.127$ & $-3.202$ \\
\midrule
1&0.5&2500&$-4.286$ & $-2.334$ & $-3.172$ & $-3.732$ & $-3.586$ & $-2.198$ & $-2.981$ & $-3.437$ & $-4.269$ & $-2.394$ & $-3.207$ & $-3.725$ & $-3.574$ & $-2.203$ & $-2.960$ & $-3.449$ \\
&&5000&$-4.636$ & $-3.670$ & $-4.086$ & $-4.287$ & $-3.881$ & $-3.419$ & $-3.792$ & $-3.957$ & $-4.631$ & $-3.673$ & $-4.086$ & $-4.271$ & $-3.896$ & $-3.425$ & $-3.792$ & $-3.957$ \\
&&10000&$-4.982$ & $-4.360$ & $-4.537$ & $-4.626$ & $-4.192$ & $-4.094$ & $-4.245$ & $-4.298$ & $-4.978$ & $-4.344$ & $-4.514$ & $-4.593$ & $-4.195$ & $-4.099$ & $-4.233$ & $-4.299$ \\ 
 &1&2500&$-3.939$ & $-2.152$ & $-2.907$ & $-3.477$ & $-3.283$ & $-2.024$ & $-2.738$ & $-3.191$ & $-3.923$ & $-2.203$ & $-2.941$ & $-3.467$ & $-3.272$ & $-2.028$ & $-2.718$ & $-3.203$ \\
&&5000&$-4.289$ & $-3.383$ & $-3.840$ & $-4.060$ & $-3.572$ & $-3.159$ & $-3.549$ & $-3.720$ & $-4.285$ & $-3.388$ & $-3.842$ & $-4.050$ & $-3.586$ & $-3.169$ & $-3.556$ & $-3.728$ \\
&&10000&$-4.635$ & $-4.156$ & $-4.348$ & $-4.430$ & $-3.878$ & $-3.877$ & $-4.034$ & $-4.078$ & $-4.632$ & $-4.141$ & $-4.334$ & $-4.408$ & $-3.881$ & $-3.887$ & $-4.029$ & $-4.085$ \\
&2&2500&$-3.593$ & $-1.959$ & $-2.620$ & $-3.161$ & $-2.978$ & $-1.840$ & $-2.478$ & $-2.908$ & $-3.576$ & $-2.003$ & $-2.653$ & $-3.150$ & $-2.967$ & $-1.844$ & $-2.461$ & $-2.922$ \\
&&5000&$-3.943$ & $-3.053$ & $-3.521$ & $-3.760$ & $-3.260$ & $-2.872$ & $-3.263$ & $-3.447$ & $-3.938$ & $-3.061$ & $-3.525$ & $-3.753$ & $-3.274$ & $-2.886$ & $-3.279$ & $-3.461$ \\
&&10000&$-4.288$ & $-3.862$ & $-4.078$ & $-4.170$ & $-3.562$ & $-3.616$ & $-3.784$ & $-3.829$ & $-4.285$ & $-3.847$ & $-4.067$ & $-4.156$ & $-3.565$ & $-3.627$ & $-3.782$ & $-3.839$ \\
\midrule
2&0.5&2500&$-4.979$ & $-2.293$ & $-2.923$ & $-3.298$ & $-4.199$ & $-2.176$ & $-2.725$ & $-3.096$ & $-4.962$ & $-2.326$ & $-2.927$ & $-3.312$ & $-4.185$ & $-2.199$ & $-2.741$ & $-3.121$ \\
&&5000&$-5.329$ & $-3.599$ & $-3.886$ & $-3.985$ & $-4.505$ & $-3.376$ & $-3.650$ & $-3.802$ & $-5.324$ & $-3.646$ & $-3.903$ & $-4.004$ & $-4.529$ & $-3.416$ & $-3.693$ & $-3.839$ \\
&&10000&$-5.675$ & $-4.296$ & $-4.382$ & $-4.422$ & $-4.822$ & $-4.098$ & $-4.201$ & $-4.261$ & $-5.671$ & $-4.341$ & $-4.419$ & $-4.461$ & $-4.830$ & $-4.141$ & $-4.249$ & $-4.307$ \\ 
&1&2500&$-4.633$ & $-2.196$ & $-2.774$ & $-3.180$ & $-3.890$ & $-2.054$ & $-2.542$ & $-2.913$ & $-4.616$ & $-2.226$ & $-2.777$ & $-3.192$ & $-3.878$ & $-2.075$ & $-2.558$ & $-2.937$ \\
&&5000&$-4.983$ & $-3.428$ & $-3.777$ & $-3.924$ & $-4.192$ & $-3.150$ & $-3.454$ & $-3.646$ & $-4.978$ & $-3.473$ & $-3.793$ & $-3.942$ & $-4.211$ & $-3.188$ & $-3.496$ & $-3.682$ \\
&&10000&$-5.328$ & $-4.219$ & $-4.329$ & $-4.371$ & $-4.506$ & $-3.915$ & $-4.048$ & $-4.128$ & $-5.325$ & $-4.262$ & $-4.364$ & $-4.408$ & $-4.511$ & $-3.959$ & $-4.096$ & $-4.172$ \\
&2&2500&$-4.286$ & $-2.066$ & $-2.566$ & $-2.970$ & $-3.586$ & $-1.908$ & $-2.324$ & $-2.671$ & $-4.269$ & $-2.092$ & $-2.570$ & $-2.980$ & $-3.574$ & $-1.928$ & $-2.340$ & $-2.694$ \\
&&5000&$-4.636$ & $-3.156$ & $-3.547$ & $-3.781$ & $-3.881$ & $-2.864$ & $-3.180$ & $-3.417$ & $-4.631$ & $-3.197$ & $-3.563$ & $-3.796$ & $-3.896$ & $-2.899$ & $-3.220$ & $-3.451$ \\
&&10000&$-4.982$ & $-4.025$ & $-4.205$ & $-4.278$ & $-4.192$ & $-3.645$ & $-3.821$ & $-3.935$ & $-4.978$ & $-4.067$ & $-4.239$ & $-4.311$ & $-4.195$ & $-3.688$ & $-3.868$ & $-3.977$ \\

\bottomrule
\end{tabular}
\label{tab.low.est_GDP_p20}
\end{sidewaystable}

\end{document}